%% file: dissertation.tex
\documentclass[german,11pt]{book}

\usepackage{a4wide}           
\usepackage{babel}            
\usepackage{amsmath}          
\usepackage{amssymb,amsfonts,amsthm,amscd} 

\usepackage{stmaryrd}         
\usepackage{latexsym}         
\usepackage{fancyhdr}         
\usepackage[latin1]{inputenc} 
\usepackage{bbm}              
\usepackage{mathrsfs}         
\usepackage{eucal}            

\usepackage[ps,dvips]{xypic}  
\usepackage{diagxy}           
\usepackage{paralist}         
\usepackage{graphicx}         

\usepackage{longtable}        
\usepackage{exscale}          
\usepackage{array}            
\usepackage{makeidx}          

\usepackage{natbib}           
\bibpunct{[}{]}{;}{a}{}{, }   

\allowdisplaybreaks             

\small\normalsize
\newcommand{\Index}[1]          {{#1}\index{#1}}

\newcommand{\indexbf}[1]        {\index{#1|textbf}}

\newcommand{\bigX}   {\mathop{\text{\LARGE \upshape \textsf X}}\displaylimits}

\DeclareMathOperator{\id}                {{\mathsf{id}}}
\DeclareMathOperator{\Id}                {{\mathsf{Id}}}
\DeclareMathOperator{\supp}              {\mathsf{supp}}

\DeclareMathOperator{\tr}                {{\mathsf{tr}}}

\DeclareMathOperator{\ad}                {ad}
\DeclareMathOperator{\Ad}                {Ad}
\DeclareMathOperator{\bild}              {\mathsf{im}}
\DeclareMathOperator{\Deg}               {Deg}
\DeclareMathOperator{\Vol}               {Vol}
\DeclareMathOperator{\Diff}              {\mathsf{Diff}}

\newcommand{\Prim}[1]         {\mathsf{Prim}(#1)}
\newcommand{\Lie}[1][{}]      {\mathcal{L}_{\sss {#1}}}
\newcommand{\LieAlg}[1]       {\mathfrak{#1}}
\newcommand{\LieAlgd}[1]      {\mathfrak{#1}^{\ast}}

\newcommand{\Obj}             {\mathsf{Obj}}
\newcommand{\Morph}           {\mathsf{Morph}}
\newcommand{\Bimorph}         {\mathsf{2\textrm{-}Morph}}
\newcommand{\Def}             {\mathsf{Def}}

\newcommand{\cl}[1][{}]       {\mathsf{cl}_{#1}}
\newcommand{\clr}[1][{}]      {\mathsf{cl}^{r}_{#1}}
\newcommand{\pr}              {\mathsf{pr}}

\newcommand{\Pol}           {\mathsf{Pol}}
\newcommand{\DiffOp}        {\mathsf{DiffOp}}
\newcommand{\DiffOpPol}     {\mathsf{DiffOpPol}}

\newcommand{\tensor}[1][{}]       {\mathbin{\otimes_{\sss {#1}}}}
\newcommand{\tensorhat}[1][{}]    {\mathbin{\widehat{\otimes}_{\sss {#1}}}}
\newcommand{\tensorhatohne}       {\mathbin{\widehat{\otimes}}}
\newcommand{\tensortilde}[1][{}]  {\mathbin{\widetilde{\otimes}_{\sss {#1}}}}
\newcommand{\cross}[2]            {{#1}\!\rtimes\!{#2}}
\newcommand{\lcross}[2]           {{#1}\!\ltimes\!{#2}}
\newcommand{\ctimes}              {\! \otimes \!} 

\DeclareMathOperator{\dega}       {deg_{\sss a}}
\DeclareMathOperator{\degs}       {deg_{\sss s}}
\DeclareMathOperator{\degl}       {deg_{\sss \lambda}}

\renewcommand{\eqref}[1]   {{(\ref{#1})}}

\newcommand{\cc}[1]        {\overline{#1}}
\newcommand{\uu}[1]        {\underline{#1}}

\renewcommand{\to}         {\,{\rightarrow}\,}
\newcommand{\too}          {\longrightarrow}
\renewcommand{\ker}        {\mathsf{ker\,}}
\newcommand{\del}          {\partial}
\newcommand{\de}           {\,{\mathrm{d}}}
\newcommand{\deohne}       {\mathrm{d}} 
\newcommand{\im}           {\mathrm{i}}
\newcommand{\eu}           {\mathrm{e}}

\newcommand{\sss}          {\scriptscriptstyle}
\newcommand{\mbf}[1]       {\boldsymbol{#1}} 
\renewcommand{\mit}[1]     {\mathnormal{#1}}   
\newcommand{\msf}[1]       {\mathsf{#1}}       
\renewcommand{\mathbb}[1]  {\mathbbm{#1}}

\newcommand{\defpro}[1]    {\boldsymbol{#1}} 

\newcommand{\bfcal}[1]     {\boldsymbol{\mathcal{#1}}}

\newcommand{\lconn}[1][{}]    {\nabla{_{\!\! \sss {#1}}}}
\newcommand{\lconnE}[1][{}]   {\nabla^{\sss E}_{\!\! \sss {#1}}}
\newcommand{\lconnEnd}[2][{}] {\nabla^{\sss \mathsf{End} ({#2})}_{\!\!
      \sss {#1}}}

\newcommand{\END}[1]       {\ensuremath{\mathcal{E}\mathit{nd}(#1)}}
\newcommand{\End}[2][{}]   {\ensuremath{\mathsf{End}_{\sss {#1}}(#2)}}

\newcommand{\GR}[2]             {\ensuremath{\mathsf{#1}(#2)}}
\newcommand{\GRn}[2]            {\ensuremath{\mathsf{#1}_{0}(#2)}} 
\newcommand{\Sympl}             {\mathsf{Sympl}}
\newcommand{\universell}[1]     {\alg{U} \left({#1} \right)}
\newcommand{\universellC}[1]    {\alg{U}_{\sss \field{C}}\left({#1} \right)}
\newcommand{\universellR}[1]    {\alg{U}_{\sss \field{R}}\left({#1}
  \right)}
\newcommand{\universellRingC}[1]    {\alg{U}_{\sss \ring{C}}\left({#1} \right)}
\newcommand{\universellRingR}[1]    {\alg{U}_{\sss \ring{R}}\left({#1}
  \right)}

\newcommand{\universelll}[1]    {\ensuremath{\alg{U}_{\lambda}\left({#1} \right)}}

\newcommand{\llb}          {[[}
\newcommand{\rrb}          {]]}
\newcommand{\FP}           {\llb \lambda \rrb}

\newcommand{\HdeRham}[1][{}]    {\mathrm{H}_{\sss \mathrm{dR}}^{#1}}
\newcommand{\HChern}[1][{}]     {\mathrm{H}^{#1}}
\newcommand{\HCech}[1][{}]      {\mathrm{\check{H}}^{#1}}

\newcommand{\HPoisson}[1][{}]   {\mathrm{H}_{\sss \pi}^{#1}}

\newcommand{\field}[1]     {\mathbb{#1}}
\newcommand{\fieldf}[1]    {\field{#1} \llb \lambda \rrb}

\newcommand{\alg}[1]       {\mathcal{#1}}
\newcommand{\algf}[1]      {\mathcal{#1} \llb \lambda \rrb}

\newcommand{\defalg}[1]    {\bfcal{#1}}

\newcommand{\modul}[1]     {\alg{#1}} 
\newcommand{\modulf}[1]    {\alg{#1} \llb \lambda \rrb}

\newcommand{\ring}[1]      {\mathsf{#1}}
\newcommand{\ringf}[1]     {\mathsf{#1} \llb \lambda \rrb}

\newcommand{\qring}[1]     {\mathsf{\hat{#1}}}

\newcommand{\schnitt}[1]   {\Gamma^\infty \left({#1}\right)}
\newcommand{\schnittf}[1]  {\Gamma^\infty \left({#1}\right)\! \llb
  \lambda \rrb} 

\newcommand{\kat}[1]       {\mathbf{#1}}
\newcommand{\hilbert}[1]   {\mathfrak{#1}} 

\newcommand{\Cinf}[1]      {C^{\infty} \!\left({#1}\right)}
\newcommand{\Cinff}[1]     {C^{\infty} \!\left({#1}\right)\llb \lambda \rrb}
\newcommand{\Cinfc}[1]     {C^{\infty}_{0} \!\left({#1}\right)}
\newcommand{\Comega}[1]    {C^{\omega} \!\left({#1} \right)}
\newcommand{\Comegaf}[1]   {C^{\omega} \!\left({#1} \right) \llb \lambda \rrb}

\newcommand{\verband}[1][{}] {\mathcal{L}_{\sss {#1}}}

\newcommand{\SP}[2][{}]     {{\left\langle{{#2}}\right\rangle}^{#1}}
\newcommand{\rSP}[3][{}]    {{\left\langle{{#2}}\right\rangle}_{\!
    \alg{#3}}^{#1}} 
\newcommand{\rSPcr}[3][{}]    {{\left\langle{{#2}}\right\rangle}_{\! {#3}}^{#1}}  
\newcommand{\lSP}[3][{}]    {^{~}_{\alg{#2}}{\!\left\langle{{#3}}\right\rangle}^{#1}} 
\newcommand{\lSPcr}[3][{}]    {^{~}_{{#2}}{\!\left\langle{{#3}}\right\rangle}^{#1}} 

\newcommand{\SPf}[1]         {{\left\langle \!\!\!\left\langle {#1}
        \right\rangle \!\!\! \right\rangle}} 
 
\newcommand{\rSPf}[3][{}]         {{\left\langle \!\!\!\left\langle {#2}
        \right\rangle \!\!\! \right\rangle_{\! {\defalg{#3}}}^{#1}}} 

\newcommand{\rmod}[2]      {\alg{#1}_{\alg{#2}}}
\newcommand{\rmodo}[2]     {{#1}_{{#2}}}
\newcommand{\rmodf}[2]     {\alg{#1}_{\alg{#2}} \llb \lambda \rrb}

\newcommand{\lmod}[2]      {{_{\alg{#1}}}{\alg{#2}}}
\newcommand{\lmodo}[2]     {{_{{#1}}}{{#2}}}
\newcommand{\lmodf}[2]     {{_{\alg{#1}}}{\alg{#2}} \llb \lambda \rrb}

\newcommand{\bimod}[3]     {{_{\alg{#1}}}{\alg{#2}}_{\alg{#3}}}
\newcommand{\bimodo}[3]    {{_{#1}}{#2}_{#3}}
\newcommand{\bimodf}[3]    {{_{\alg{#1}}}{\alg{#2}}_{\alg{#3}} \llb \lambda \rrb} 
\newcommand{\bimodcross}[4][{H}]
{{_{\cross{\alg{#2}}{#1}}}{\cross{\alg{#3}}{#1}}_{\cross{\alg{#4}}{#1}}}
\newcommand{\bimodcrosso}[4][{H}]
{{_{\cross{\alg{#2}}{#1}}}{{\alg{#3}} \otimes {#1}}_{\cross{\alg{#4}}{#1}}}  

\newcommand{\defrmod}[2]  {\boldsymbol{\alg{#1}}_{\defalg{#2}}} 
\newcommand{\deflmod}[2]  {_{\defalg{#1}}\boldsymbol{\alg{#2}}} 
\newcommand{\defbimod}[3] {_{\defalg{#1}}\boldsymbol{\alg{#2}}_{\defalg{#3}}}

\newcommand{\lmodplus}[2]      {\left({\lmod{#1}{#2}},
      \lSP{\alg{#1}}{\cdot,\cdot} \right)}
\newcommand{\rmodplus}[2]      {\left({\rmod{#1}{#2}},
      \rSP{\cdot,\cdot}{\alg{#2}} \right)} 

\newcommand{\rmodplusne}[2]      {\left({\rmod{#1}{#2}} /
      \rmod{#1}{#2}^{\perp}, \rSP{\cdot,\cdot}{\alg{#2}} \right)} 

\newcommand{\bimodplus}[3]   {\left(\bimod{#1}{#2}{#3},
      \lSP{#1}{\cdot,\cdot},
      \rSP{\cdot,\cdot}{#3} \right)} 
\newcommand{\bimodrplus}[3]   {\left(\bimod{#1}{#2}{#3},
      \rSP{\cdot,\cdot}{#3} \right)} 

\newcommand{\srep}[1][{}]  {\sideset{^{\ast}}{_{\sss {#1}}}{\operatorname{\textrm{-}\mathsf{rep}}}}
\newcommand{\sRep}[1][{}]  {\sideset{^{\ast}}{_{\sss {#1}}}{\operatorname{\textrm{-}\mathsf{Rep}}}}
\newcommand{\smod}[1][{}]  {\sideset{^{\ast}}{_{\sss {#1}}}{\operatorname{\textrm{-}\mathsf{mod}}}}
\newcommand{\sMod}[1][{}]  {\sideset{^{\ast}}{_{\sss {#1}}}{\operatorname{\textrm{-}\mathsf{Mod}}}}

\newcommand{\rieffel}[2]   {\mathsf{#1}_{\sss {#2}}}

\newcommand{\katlmod}[1]   {{#1}{\operatorname{\textrm{-}\mathsf{mod}}}}
\newcommand{\katlMod}[1]   {{#1}{\operatorname{\textrm{-}\mathsf{Mod}}}}
\newcommand{\katrmod}[1]   {\operatorname{\mathsf{mod}\textrm{-}}{#1}}
\newcommand{\katrMod}[1]   {\operatorname{\mathsf{Mod}\textrm{-}}{#1}}

\newcommand{\CSpan}[1]
{{\ring{C}{\operatorname{\textrm{-}\mathsf{span}}}}\left\{ {#1} \right\}}

\newcommand{\spstd}        {\mathbin{\star_{\sss \mathsf{Std}}}}
\newcommand{\spweyl}       {\mathbin{\star_{\sss \mathsf{Weyl}}}}

\newcommand{\spwick}       {\mathbin{\star_{\sss \mathsf{Wick}}}}

\newcommand{\spkappa}      {\mathbin{\star_{\sss \kappa}}}
\newcommand{\spemkappa}    {\mathbin{\star_{\sss 1-\kappa}}}
\newcommand{\sptkappa}     {\mathbin{\star_{\sss \tilde{\kappa}}}}
\newcommand{\spgutt}       {\mathbin{\star_{\sss \mathsf{G}}}}

\newcommand{\ordstd}       {\varrho_{\sss \mathsf{Std}}}
\newcommand{\ordweyl}      {\varrho_{\sss \mathsf{Weyl}}}

\newcommand{\ordwick}      {\varrho_{\sss \mathsf{Wick}}}
\newcommand{\ordkappa}     {\varrho_{\sss \kappa}}
\newcommand{\ordtkappa}    {\varrho_{\sss \tilde{\kappa}}}
\newcommand{\ordzahlweyl}  {\varrho_{\sss \!1\! /\!2}}

\newcommand{\symbstd}      {\sigma_{\sss \mathsf{Std}}}
\newcommand{\symbweyl}     {\sigma_{\sss \mathsf{Weyl}}}
\newcommand{\symbkappa}    {\sigma_{\sss \kappa}}
\newcommand{\symbtkappa}   {\sigma_{\sss \tilde{\kappa}}}
\newcommand{\symbwick}     {\sigma_{\sss \mathrm{Wick}}}

\newcommand{\fpweyl}       {\mathbin{{\circ_{\sss \mathsf{Weyl}}}}}

\newcommand{\Nweyl}[1][{}] {N_{\sss \!1\! /\!2}^{#1}}

\newcommand{\bundle}[3]     {{#1} \stackrel{#2}{\longrightarrow} {#3}}

\newcommand{\Name}[1]         {{\sc{#1}}}
\newcommand{\Nameb}[2]        {{\sc{#1}} $(^\ast{#2})$}
\newcommand{\Namebd}[3]        {{\sc{#1}} $(^\ast${#2}, $\dagger${#3})}
\newcommand{\NameSection}[1]  {{\sc{#1}}} 

\newcommand{\zentrum}[2][{}]  {\sideset{}{_{\sss {#1}}}{\operatorname{\mathfrak{Z}}}(#2)}

\newcommand{\ideal}[1]        {\mathcal{#1}}

\newcommand{\Deltaop}         {\Delta^{\mathrm{\sss op}}}
\newcommand{\muop}            {\mu^{\mathrm{\sss op}}}

\newcommand{\neact}           {\mathbin{\triangleright}}  
\newcommand{\act}             {\mathbin{\triangleright}} 
\newcommand{\neactt}[1][{}]   {\mathbin{\triangleright^{\sss {#1}}}}

\newcommand{\neactb}          {\mathbin{\triangleright^{\sss {\msf{b}}}}} 
\newcommand{\ccact}           {\mathbin{\overline{\triangleright}}}

\newcommand{\neacta}          {\mathbin{\triangleright_{\sss {\msf{a}}}}}

\newcommand{\Hgutt}            {H_{\sss \mathsf{G}}}

\newcommand{\Hom}[1][{}]    {\mathsf{Hom}_{\sss {#1}}}

\newcommand{\Pic}      {\operatorname{\mathsf{Pic}}}

\newcommand{\starPic}  {\sideset{}{^{\ast}}{\operatorname{\mathsf{Pic}}}}
\newcommand{\strPic}   {\sideset{}{^{\mathrm{str}}}{\operatorname{\mathsf{Pic}}}}

\newcommand{\PicH}     {\sideset{}{_{\sss H}}{\operatorname{\mathsf{Pic}}}}
\newcommand{\strPicH}  {\sideset{}{^{\mathrm{str}}_{\sss H}}{\operatorname{\mathsf{Pic}}}}
\newcommand{\starPicH} {\sideset{}{^{\ast}_{\sss H}}{\operatorname{\mathsf{Pic}}}}

\newcommand{\KPic}     {\operatorname{\underline{\mathsf{Pic}}}}

\newcommand{\SPic}     {{\operatorname{\mathsf{SPic}}}}

\newcommand{\Iso}      {\operatorname{\mathsf{Iso}}}
\newcommand{\IsoH}     {\sideset{}{_{\sss H}}{\operatorname{\mathsf{Iso}}}}
\newcommand{\starIso}  {\sideset{}{^{\ast}}{\operatorname{\mathsf{Iso}}}}
\newcommand{\starIsoH} {\sideset{}{^{\ast}_{\sss H}}{\operatorname{\mathsf{Iso}}}}

\newcommand{\Bij}      {\operatorname{\mathsf{Bij}}}

\newcommand{\Aut}[1][{}] {\sideset{}{_{\sss {#1}}}{\operatorname{\mathsf{Aut}}}}
\newcommand{\AutH}     {\sideset{}{_{\sss H}}{\operatorname{\mathsf{Aut}}}}
\newcommand{\InnAut}   {\operatorname{\mathsf{InnAut}}} 
\newcommand{\starInnAut}{\sideset{}{^{\ast}}{\operatorname{\mathsf{InnAut}}}}

\newcommand{\starAut}  {\sideset{}{^{\ast}}{\operatorname{\mathsf{Aut}}}}
\newcommand{\starAutH} {\sideset{}{^{\ast}_{\sss H}}{\operatorname{\mathsf{Aut}}}}

\newcommand{\OutAut}     {\operatorname{\mathsf{OutAut}}} 
\newcommand{\starOutAut} {\sideset{}{^{\ast}}{\operatorname{\mathsf{OutAut}}}}

\newcommand{\starAlg}  {\sideset{^\ast}{}{\operatorname{\mathsf{Alg}}}}
\newcommand{\staralg}  {\sideset{^\ast}{}{\operatorname{\mathsf{alg}}}}

\newcommand{\PraeHilbert}[1] {\operatorname{\mathsf{Pr\ddot{a}Hilbert}}(#1)}

\newcommand{\KProj}         {\operatorname{\uu{\mathsf{Proj}}}}
\newcommand{\KstarProj}     {\sideset{}{^\ast}{\operatorname{\uu{{\mathsf{Proj}}}}}}
\newcommand{\KstrProj}      {\sideset{}{^\mathrm{str}}{\operatorname{\uu{\mathsf{Proj}}}}}
\newcommand{\KProjH}        {\sideset{}{_{\sss H}}{\operatorname{\uu{\mathsf{Proj}}}}}
\newcommand{\KstarProjH}     {\sideset{}{^\ast_{\sss H}}{\operatorname{\uu{\mathsf{Proj}}}}}
\newcommand{\KstrProjH}      {\sideset{}{^\mathrm{str}_{\sss H}}{\operatorname{\uu{\mathsf{Proj}}}}}

\newtheoremstyle{break}
  {}                   
  {}                   
  {\itshape}           
  {}                   
  {\bf \bf}          
  {.}                  
  {\newline}           
  {}                   

\theoremstyle{break}

\newtheorem{definition}{Definition}[section]
\newtheorem{satz}[definition]{Satz}

\newtheorem{lemma}[definition]{Lemma}
\newtheorem{dlemma}[definition]{Definiton und Lemma}
\newtheorem{dkonstruktion}[definition]{Definition und Konstruktion}
\newtheorem{korollar}[definition]{Korollar}

\newtheorem{proposition}[definition]{Proposition}

\newtheoremstyle{break2}
  {}                   
  {}                   
  {\upshape}           
  {}                   
  {\bfseries}          
  {.}                  
  {\newline}           
  {}                   

 \theoremstyle{break2}

\newtheorem{bemerkung}[definition]{Bemerkung}
\newtheorem{bemerkungen}[definition]{Bemerkungen}

\newtheorem{beispiel}[definition]{Beispiel}
\newtheorem{beispiele}[definition]{Beispiele}

\numberwithin{equation}{section}
\makeatletter
\def\cleardoublepage{\clearpage\if@twoside \ifodd\c@page\else
    \hbox{}
    \vspace*{\fill}

\begin{center}
\end{center}                                                
    
    \vspace{\fill}                                              %
    \thispagestyle{empty}                                       %
    \newpage                                                    %
    \if@twocolumn\hbox{}\newpage\fi\fi\fi}                      %
\makeatother  

\makeindex 

\begin{document}
\xyoption{curve}


\index{Intertwiner|see{Verschr"ankungsoperator}}
\index{Hopf-Aktion@\Name{Hopf}-Aktion|see{\Name{Hopf}-Wirkung}}
\index{Smash-Produkt|see{Cross-Produktalgebra}}
\index{Rechtsmodul|see{Modul}}
\index{Linksmodul|see{Modul}}
\index{Operator!adjungierbarer|see{Abbildungen}}
\index{Flip|see{Vertauschungsoperator}}
\index{pull-back|see{Zur\"uckziehen}}
\index{Stern-Algebra@$^\ast$-Algebra|see{Algebra}}
\index{Standardordnung|see{Ordnungsvorschrift}}
\index{Quantengruppe|see{\Name{Hopf}-Algebra}}

\pagestyle{empty}
\include{titelblatt} 

\frontmatter

\pagestyle{fancy}
\tableofcontents

\include{danke}
\include{einleitung}

\mainmatter
\pagestyle{fancy}
\include{morita}
\include{picard}
\include{cross}
\include{sternprodukte}

\include{defmorita}
\include{ausblick}

\begin{appendix}
\renewcommand{\chaptername}{Anhang}

\include{algebra}
\include{geometrie}

\end{appendix}

\backmatter

\pagebreak

\pagestyle{fancy}

\include{symbolverzeichnis}

\include{personenindex}


\addcontentsline{toc}{chapter}{Literaturverzeichnis}
\fancyhead[CE]{\slshape \nouppercase{Literaturverzeichnis}} 
\fancyhead[CO]{\slshape \nouppercase{Literaturverzeichnis}}
 
\bibliography{dqarticle,dqproceeding,dqprocentry,dqthesis,dqbook,dqmisc,dqpreprints}
\bibliographystyle{stefan6}

\addcontentsline{toc}{chapter}{Index}
\fancyhead[CE]{\slshape \nouppercase{Index}} 
\fancyhead[CO]{\slshape \nouppercase{Index}}
\printindex
\pagestyle{empty}
\cleardoublepage

\end{document}

%% file: titelblatt.tex
\begin{titlepage}
\begin{center}

\vspace{0.6cm}
\Huge \scshape
$H$-\"aquivariante Morita-\"Aquivalenz \\
\vspace{0.3cm}
und Deformationsquantisierung 
\vspace{4cm}

\Large
\mdseries

Inaugural-Dissertation\\
zur Erlangung des Doktorgrades\\
\vspace{2cm}
\large
vorgelegt von \\
\vspace{1cm}
{\LARGE
\scshape

Stefan Jansen\\}

\vspace{1cm}
\large
\mdseries
aus Aachen

\vspace{2cm}
\end{center}
\vfill
\begin{center}
    \begin{minipage}{15cm}
        \begin{minipage}{3.8cm}
            \center{\includegraphics[width=3.1cm]{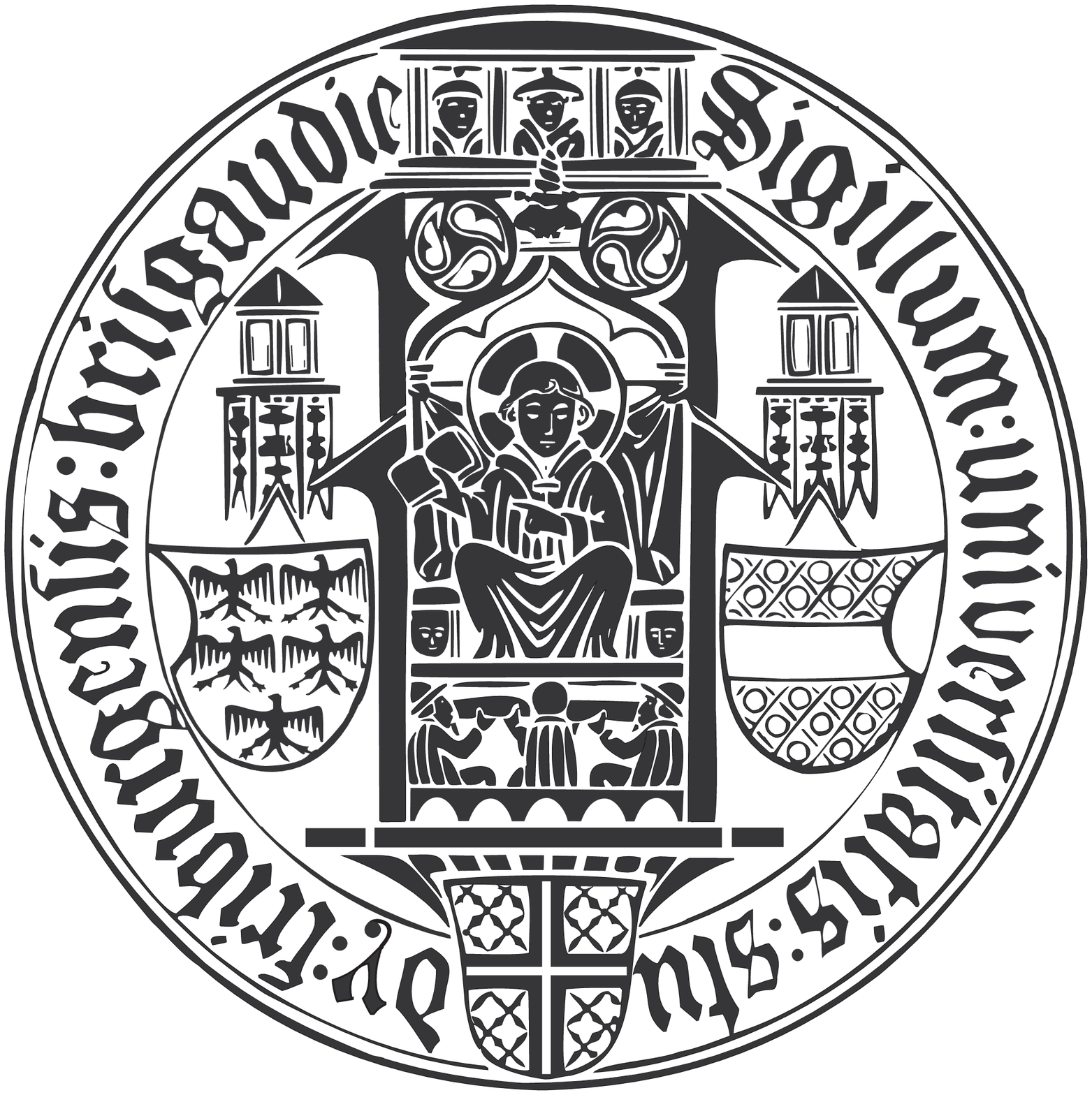}}
        \end{minipage}
        \hfill
        \begin{minipage}{11cm}
            \begin{flushleft}
                \scshape
                \Large
                \vspace{0.1cm}
                Albert-Ludwigs-Universit"at Freiburg\\
                Fakult"at f"ur Mathematik und Physik\\
                Physikalisches Institut \\
                November 2006\\
            \end{flushleft}
        \end{minipage}
    \end{minipage}
\end{center}

\pagebreak

~\\ \vfill

\begin{tabular}{ll}
Dekan: & Prof.~Dr.~\Name{J"org Flum}\\
Referent: & Prof.~Dr.~\Name{Hartmann R"omer}\\
Korreferent: & Prof.~Dr.~\Name{Wolfgang Soergel} \\
Tag der m"undlichen Pr"ufung: & 18.12.2006
\end{tabular}
\end{titlepage}


%% file: danke.tex
\chapter*{Danksagung}
\fancyhead[CE]{\slshape \nouppercase{Danksagung}} 
\fancyhead[CO]{\slshape \nouppercase{Danksagung}} 
\addcontentsline{toc}{chapter}{Danksagung}

Das Verfassen dieser Arbeit w"are ohne die Unterst"utzung
vieler Menschen nicht m"oglich gewesen. Ich m"ochte an dieser Stelle
all denjenigen meinen Dank aussprechen, die mich in den letzten
viereinhalb Jahren auf ihre Art weitergebracht,
unterst"utzt und ertragen haben. Besonders die folgenden Personen
haben unmittelbar Einflu"s auf das Gelingen der Arbeit genommen.\\

\vspace{-2mm}
An erster Stelle sei Professor \Name{Hartmann R"omer} genannt, der
mich herzlich in seine Arbeitsgruppe aufgenommen hat und stets
hilfsbereit war, wenn es irgendwelche Probleme zu bew"altigen
galt.\\  

\vspace{-2mm}
Ein besonderer Dank gilt \Name{Stefan Waldmann}, der extrem viel Geduld mit mir
aufgebracht hat, immer ein offenes Ohr f"ur meine Probleme hatte und
mir quasi immer den richtigen Tip geben konnte, wenn es mal nicht so
weiterging, wie ich mir das vorgestellt habe. 
Gerne denke ich an die lustigen und sch"onen Zwiebelkuchen- und Neuer
S"u"ser-Abende im Hause \Name{Waldmann}.  
Ohne die wissenschaftliche Betreuung durch \Name{Stefan} w"are diese
Arbeit in dieser Form sicherlich nicht denkbar gewesen w"are. \\

\vspace{-2mm}
In gleichem Ma"se geb"uhrt \Name{Nikolai Neumaier} Dank. Auch er war
stets bereit seine Zeit f"ur Diskussionen zu opfern, was immer eine
wissenschaftliche und menschliche Bereicherung f"ur mich war. Wie bei der Gestaltung von "Ubungszetteln setzt \Name{Nikolai} auch bei seinen
Kuchenkreationen Ma"sst"abe f"ur die kommenden Generationen. \\

\vspace{-2mm}
Ich danke den Doktoranden des achten Stocks, die alle mehr als nur Kollegen
geworden sind:\\
\Name{Svea Beiser} danke ich f"ur die sch"one Zeit, die wir
gemeinsam in einem B"uro verbracht haben, f"ur das Korrekturlesen,
f"ur die Versorgung mit Keksen und Schokolade, f"ur die \Name{Monk}-DVDs und
nat"urlich die Einladung zur Hochzeitsfeier... \\   
Obwohl \Name{Matteo Carrera} keinen fachlichen Beitrag zu der
vorliegenden Arbeit geleistet hat, danke ich ihm f"ur die viereinhalb
Jahre, die wir nun zusammen in der \glqq Abteilung
\Name{R"omer}\grqq{} verbracht haben. Ohne ihn 
w"are die Zeit hier deutlich weniger lustig ausgefallen. \\ 
Letzteres gilt auch f"ur \Name{Michael Carl}, dem ich au"serdem f"ur
die inspirativen Diskussionen danken m"ochte.\\  
\Name{Florian Becher} sei Dank f"ur das penible Korrekturlesen, die
t"aglichen Cartoons sowie den zahllosen Versuche mich mit
Automatenkaffee zu versorgen -- auch wenn ich seine Angebote nicht
einmal in Anspruch genommen habe...\\
Allen Mitarbeitern, Doktoranden und Diplomanden, die an dieser Stelle nicht genannt
wurden, danke ich f"ur die angenehme Zeit, die wir gemeinsam im
Physikhochhaus verbracht haben. \\ 

\vspace{-2mm}
F"ur die Hilfe beim Erstellen des Personenindexes danke ich \Name{Simone
  Gutt} und \Name{Martin Bordemann} herzlichst. Wer h"atte schon
gedacht, da"s \Name{Wick} Italiener war und \Name{Gian-Carlo} mit
Vorname hie"s? \\ 

\vspace{-2mm}
Ein ganz besonderer Dank gilt \Name{Alexandra Teynor}, die mir in den
letzten Monaten dieser Arbeit wichtigen Beistand geleistet hat. Danke f"ur den
Zwetschgendatschi. \\ 

\vspace{-2mm}
Zu guter Letzt m"ochte ich meinen Eltern \Name{Franz-Josef} und \Name{Irmgard
Jansen} danken, die mich immer unterst"utzten. \\ 

\hfill {\it Stefan Jansen}\\

\vspace{-6mm}

\hfill {\it Freiburg, im November 2006}\\


%% file: einleitung.tex
\chapter*{}
\vspace{4cm}
\begin{center}
{\it \glqq I learned to distrust all physical concepts as the basis
  for a theory. Instead one should put one's trust in a mathematical
  scheme, even if the scheme does not appear at first sight to be
  connected with physics. One should concentrate on getting
  interesting mathematics.\grqq} \\ 
\end{center}
\hfill \citet{marlow:1978a}  

\chapter*{Einleitung}
\fancyhead[CE]{\slshape \nouppercase{Einleitung}} 
\fancyhead[CO]{\slshape \nouppercase{Einleitung}} 
\addcontentsline{toc}{chapter}{Einleitung}

\section*{Motivation und Ziel}

Schon seit jeher haben sich die Physik und die Mathematik gegenseitig
in ihren Entwicklungen beeinflu"st. 
H"aufig war es die Physik, die neue mathematische Konzepte
ben"otigte, um die Forschung voranzutreiben, aber auch neue
Erkenntnisse in der Mathematik erweiterten den Horizont der
Physiker, so da"s es m"oglich war, neue physikalische Konzepte zu
finden. Beispielsweise brauchte \Name{Newton}
die Infinitesimal- und Integralrechnung, um die klassische Mechanik
mathematisch formulieren zu k"onnen oder \Name{Dirac} die
Distributionentheorie oder den Spektralkalk"ul von Operatoren auf
\Name{Hilbert}-R"aumen f"ur
die Quantenmechanik. Andererseits w"are es
\Name{Einstein} nicht m"oglich gewesen die Allgemeine
Relativit"atstheorie (ART) zu formulieren, wenn nicht im
19.~Jahrhundert die \Name{Riemann}sche Geometrie hinreichend
ausgearbeitet worden w"are. So gibt es von der Antike bis heute
zahllose Beispiele, die diese parallele Entwicklung der beiden Gebiete untermauern. 

Auch die vorliegende Arbeit wurde von uns mit dem Ziel verfa"st,
eine weitere Br"ucke zwischen der Physik und der Mathematik zu
schlagen. Motiviert durch physikalische Fragestellungen haben wir
versucht, neue mathematische Konzepte zu formulieren, die im
Rahmen der Physik ihre Anwendung finden k"onnen. Die drei
wesentlichen in dieser Arbeit behandelten Gebiete wollen wir kurz beleuchten und ihre Bedeutung
f"ur die Physik in aller K"urze darlegen.

\subsubsection*{Sternprodukte}

Sternprodukte sind ein Versuch, {\em klassische Theorien} (insbesondere
klassische Mechanik) zu
quantisieren und stellen damit eine alternative M"oglichkeit zur {\em
  kanonischen Quantisierung} dar. Man bezeichnet das Quantisieren
mittels Sternprodukten auch als {\em Deformationsquantisierung}. W"ahrend die kanonische
Quantisierung, die 1932 erstmals von \citet{vonneumann:1996a}
mathematisch ansprechend formuliert worden 
ist\footnote{Ein wichtiges Werk zur Quantenmechanik lieferte
  \Name{Dirac} schon 1930 \citep{dirac:1982a} -- also einige Zeit vor
  \Name{von Neumann}, jedoch war seine Herangehensweise an die
  mathematischen Konzepte der Quantenmechanik sehr klassisch gepr"agt
und daher nach heutiger Sicht der Dinge unzureichend.}, einen
funktional-analytischen Zugang bietet, sind 
Sternprodukte ein algebraischer bzw.~geometrischer Zugang zur
Quantisierung. 
Sternprodukte sind, wie wir in Kapitel~\ref{chapter:Sternprodukte}
sehen werden, direkt aus der klassischen Theorie abgeleitet und
daher ein sehr nat"urlicher Zugang zur Quantenmechanik. 
Eine Folge dessen ist, da"s der {\em klassische Limes} (\glqq $\hbar \to 0$\grqq )
im Rahmen der Deformationsquantisierung sehr gut zu verstehen ist, da
er aufgrund des konzeptuellen Aufbaus gleich mitgeliefert
wird. Man geht von einer Mannigfaltigkeit $M$ aus, die physikalisch dem
Phasenraum entspricht, und betrachtet die kommutative Algebra der
glatten, komplexwertigen Funktionen $\Cinf{M}$ mit dem punktweise
Produkt. Diese Algebra entspricht einer klassischen Observablenalgebra. 

Der "Ubergang zu einer Quantentheorie geschieht indem man diese
Algebra {\em deformiert}. Dies bedeutet das punktweise Produkt wird zu
einem assoziativen, jedoch nicht mehr kommutativen Produkt, dem {\em
  Sternprodukt} der Form
\begin{align*}
    f \star g = fg + \lambda C_{1}(f,g) + \lambda^{2} C_{2}(f,g) + ...
 \end{align*} 
Dabei bezeichnen wir $\lambda$ als den {\em Deformationsparameter} und
$C_{n}$ sind Bidifferentialoperatoren, die die Deformation \glqq kodieren\grqq. 
Elemente in der deformierten Algebra sind damit {\em formale Potenzreihen} der Form 
\begin{align*}
    f = \sum_{n=0}^{\infty} \lambda^{n} f_{n}.
\end{align*}
Das Produkt k"onnen wir als eine St"orungsreihe in $\lambda$ auffassen.
Unter dem klassischen Limes verstehen wir die Abbildung $\cl:f \mapsto
f_{0}$, was einem \glqq Vergessen der Quantenkorrekturen\grqq{} entspricht. 
Die Reihen bezeichnet man als {\em formal}, da es sich um ein algebraisches
Konzept handelt, und der Parameter $\lambda$ erstmal
ausschlie"slich der Ordnung dient. Erst in einem weiteren Schritt, n"amlich in einem
konvergenten Rahmen, wie er beispielsweise von 
\citet{beiser:2005a} betrachtet wird, tritt an die Stelle des formalen
Parameters eine reelle Zahl, das \Name{Planck}sche Wirkungsquantum
$\hbar$. Ein anderer Zugang zu konvergenten Sternprodukten ist
die {\em strikte Deformationsquantisierung}. Dazu verweisen wir auf die Arbeiten
von \citet{rieffel:1993a}, \citet{landsman:1998a}, \citet{bieliavsky:2002a}, 
\citet{heller.neumaier.waldmann:2006a:pre} und \citet{becher:2006a}. 

Der Vorteil des algebraischen Rahmens ist, da"s man nicht "uber
Konvergenz reden {\it mu"s}, also komplett auf
funktional-analytische Aspekte der Quantisierung verzichtet, und somit
eine gr"o"sere Klasse von Algebren betrachten kann. 

Bei der kanonischen Quantisierung ist das Bilden eines klassischen
Limes weitaus diffiziler und oft ist es
alles andere als trivial, einen "Ubergang zum Makroskopischen zu
finden. Ein einfaches Beispiel daf"ur ist der Impulsoperator im
Ortsraum $P_{i}=\tfrac{\hbar}{\im} \tfrac{\del}{\del q_{i}}$, der bei
\glqq Vernachl"assigung der Quantenkorrekturen\grqq{} nur die Null w"are.

Eine weitere Motivation, sich Sternprodukte anzusehen, ist die
Tatsache, da"s man auf s"amtlichen symplektischen sowie
\Name{Pois\-son}-Man\-nig\-fal\-tig\-kei\-ten Sternprodukte
konstruieren kann. Dies bedeutet, da"s man diese R"aume im Rahmen der
Deformationsquantisierung {\em immer} quantisieren kann. Die kanonische
Quantisierung versagt hier insofern, als da"s sie keineswegs mehr
kanonisch ist, sondern {\em ad hoc}-L"osungen verlangt.  
Der offensichtliche Nachteil der Deformationsquantisierung liegt im Fehlen eines
Spektralkalk"uls. Zwar wurde in
\citep{bayen.et.al:1977a} eine Methode vorgestellt einige einfache
Spektren zu berechnen, allerdings sind diese Ergebnisse
alles andere als befriedigend. Bis heute gibt es (leider) keine
sinnvolle L"osung dieses Problems, da man konvergente Sternprodukte
braucht und diese sehr schwierig zu behandeln sind. Dies liegt insbesondere daran, da"s
$\hbar$ eine dimensionsbehaftete Gr"o"se ist, deren Wert man durch
"andern der Einheiten beliebig manipulieren kann, und somit nicht \glqq
klein\grqq{} im Sinne einer St"orungsreihe ist.

\subsubsection{\Name{Morita}-Theorie und \Name{Picard}-Gruppoid}

Auch die \Name{Morita}-"Aquivalenz erfreut sich einer wichtigen
Anwendung in der Physik. Ein zentrales Objekt sowohl der klassischen
Mechanik wie auch der Quantenmechanik ist die {\em
  Observablenalgebra}. Observablenalgebren sind $^\ast$-Algebren, das
hei"st Algebren mit einer $^\ast$-Involution. Im speziellen ist das der klassischen
Mechanik eine Funktionenalgebra auf einem Phasenraum mit der komplexen
Konjugation und in der Quantenmechanik sind es (beschr"ankte) adjungierbare
Operatoren auf einem \Name{Hilbert}-Raum. 
Bei der \Name{Morita}-Theorie setzt man Objekte einer Kategorie, den Algebren,
zueinander in Beziehung indem man die Klasse der Morphismen, "uber die Algebramorphismen
hinaus, erweitert. Die Erweiterungen sind die sogenannte
{\em "Aquivalenzbimoduln}, und die Menge (oder Klasse) aller
"Aquivalenzbimoduln bezeichnet man dann als das 
{\em \Name{Picard}-Gruppoid}. Das \Name{Picard}-Gruppoid dient somit die
\Name{Morita}-"Aquivalenz auf eine elegante Art zu kodieren. Anders
ausgedr"uckt bedeutet dies, da"s man mittels der
\Name{Morita}-"Aquivalenz Algebren klassifizieren kann. 
Dies spielt in der Physik aber eine wichtige Rolle, denn eine Eigenschaft von
\Name{Morita}-"aquivalenten Objekten (also Algebren) ist, da"s sie eine "aquivalente
{\em Darstellungstheorie} haben. Dies spielt zum Beispiel in der
axiomatischen Quantenfeldtheorie nach \citet{haag.kastler:1964a} eine
Rolle, bei der man {\em Superauswahlregeln} 
untersucht. Diese geben an wieviele irreduzible Darstellungen eine
gegebene Ob\-ser\-va\-blen\-al\-ge\-bra hat. Mit Hilfe der
\Name{Morita}-Theorie kann man die Darstellungstheorie \glqq
komplizierter\grqq{} Algebren auf die Darstellungstheorie besser verstandener
zur"uckf"uhren. 

Eine mit der Darstellungstheorie von Sternproduktalgebren verbundene
Anwendung der \Name{Mo\-ri\-ta}-"Aqui\-va\-lenz ist 
der \Name{Dirac}-Monopol. Betrachtet man "Aquivalenzbimodul von Sternproduktalgebren, so stellt
sich die \Name{Dirac}sche Quantisierungsbedingung 
f"ur magnetische Monopole als eine Integralbedingung an die Klasse
der beiden Sternprodukte heraus, und nur dann, wenn beide
Sternproduktalgebren \Name{Morita}-"aquivalent sind, ist diese
Integralbedingung zu erf"ullen. Eine detaillierte Ausf"uhrung findet
man z.~B.~in \citep[Section~7.3]{waldmann:2005b} und
\citep[Section~4.2]{bursztyn.waldmann:2002a}.

\subsubsection{\Name{Hopf}-Algebren}

Im Zuge des Versuchs eine vereinheitlichende Theorie von
Quantentheorie und Gravitation (Quantengravitation) zu formulieren, wurden in der Physik 
{\em \Name{Hopf}-Algebren}, die man auch unter dem Begriff {\em
  Quantengruppen} findet, entwickelt und
diskutiert. \Name{Hopf}-Algebren stellen ein verallgemeinerndes
Konzept dar, um sich jenseits von Gruppen oder
\Name{Lie}-Algebren dem Begriff der {\em Symmetrie} zu n"ahern. Symmetrien spielen
in der Physik an vielen Stellen eine wichtige Rolle. In der
Festk"orperphysik sind es {\em diskrete Symmetrien}, die beispielsweise
Gitterstrukturen beschreiben (vgl.~z.~B.~\citep{Lax:2001a}). 
Wir sind jedoch eher an {\em kontinuierlichen Symmetrien}
interessiert, die sich in der Physik h"aufig in Form von
\Name{Lie}-Gruppen bzw.~\Name{Lie}-Algebren manifestieren. Einige in
der Physik auftretende Beispiele sind zum einen Erhaltungsgr"o"sen, die
nach \citet{noether:1918a} einer Symmetrie entsprechen
(Drehimpulserhaltung $\rightleftharpoons$ Rotationsinvarianz,
Impulserhaltung $\rightleftharpoons$ Translationsinvarianz,
Energieerhaltung $\rightleftharpoons$ Homogenit"at der Zeit, etc.) und
zum anderen Eichtheorien oder das Standardmodell. 
Da"s \Name{Hopf}-Algebren in diesem Zusammenhang ein interessantes Konzept sind, spiegelt
sich in der Tatsache wider, da"s sowohl die {\em Gruppenalgebren
von \Name{Lie}-Gruppen} als auch die {\em universell Einh"ullenden
von \Name{Lie}-Algebren} \Name{Hopf}-Algebren sind. Damit haben wir
die beiden f"ur die Physik wichtigen F"alle mittels eines \glqq
universelleren Formalismus\grqq{} abgedeckt, und uns zudem die Option auf
allgemeinere Symmetrien offengelassen.

\section*{Aufbau und Ergebnisse der Arbeit}

Die Arbeit ist in f"unf Kapitel sowie zwei Anh"ange eingeteilt. 

In Kapitel~\ref{chapter:Morita} wollen wir einen "Uberblick "uber die
\Name{Morita}-Theorie geben. Dabei werden wir -- aufbauend auf einer
f"ur das Verst"andnis wichtigen mathematischen Einleitung -- die drei
wichtigen Auspr"agungen der \Name{Morita}-Theorie diskutieren. Zuerst
die {\em ringtheoretische Theorie}, die zur Klassifizierung von
Darstellungen von Ringen entwickelt wurde und geschichtlich die Grundlage
f"ur die sp"ater entstandene {\em $^\ast$-\Name{Morita}-Theorie}
beziehungsweise {\em starke \Name{Morita}-Theorie} bildet. 
Wir besch"aftigen uns im wesentlichen mit den beiden letzteren, die
beide eine \Name{Morita}-Theorie f"ur $^\ast$-Algebren
darstellen. Damit klassifizieren sie eine f"ur die Physik interessante Kategorie
von Algebren, den Observablenalgebren. Das wichtige Ergebnis, da"s wir zur
\Name{Morita}-Theorie beitragen, ist eine Erweiterung auf den $H$-"aquivarianten
Fall, wobei $H$ eine \Name{Hopf}-$^\ast$-Algebra ist. Die
betrachteten $^\ast$-Algebren tragen nun eine, durch die Wirkung einer
\Name{Hopf}-$^\ast$-Algebra vorgegebene, Symmetrie. Wir haben
herausgearbeitet wie und unter welchen Umst"anden man diese
$H$-"aqui\-va\-ri\-an\-te \Name{Morita}-Theorie etablieren kann.

Eine konzeptionelle Fortf"uhrung des Kapitels~\ref{chapter:Morita} bildet
Kapitel~\ref{chapter:Picard}. Dort beleuchten  wir das
\Name{Picard}-Gruppoid bzw.~die \Name{Picard}-Gruppe n"aher. Auch hier
interessiert uns insbesondere der 
$H$-"aqui\-va\-ri\-an\-te Fall. Das wahrscheinlich wichtigste Ergebnis der
Arbeit ist die Untersuchung des Morphismus $\PicH \to \Pic$ vom
$H$-"aquivarianten \Name{Picard}-Gruppoid in das \glqq
gew"ohnliche\grqq{} \Name{Picard}-Gruppoid. Wir untersuchen den Kern sowie das Bild der 
Abbildung, und sind in der Lage den Kern vollst"andig zu charakterisieren und
das Bild unter gewissem Umst"anden angeben zu k"onnen.   
Den Kern der Abbildung $\PicH \to \Pic$ zu verstehen ist in diesem
Rahmen insofern ein sehr sch"ones Ergebnis, da die Existenz eines
f"ur die \Name{Morita}-"Aquivalenz notwendigen $H$-"aquivarianten "Aquivalenzbimoduls
ein Hebungsproblem ist. Im allgemeinen ist dies an eine kohomologische Obstruktionen
gebunden (siehe z.~B.~\citep{guillemin.ginzburg.karshon:2002a}).
Wir sind somit, mit den von uns entwickelten
Techniken, in der Lage, anzugeben wie das $H$-"aquivariante \Name{Picard}-Gruppoid aussieht. 
Das Bild der Abbildung $\PicH \to \Pic$ k"onnen angeben, wenn es auf
den jeweiligen $^\ast$-Algebren Impulsabbildungen gibt. 
Anschlie"send betrachten wir $H$-"aqui\-va\-ri\-an\-te \Name{Morita}-Invarianten,
die wir durch die Wirkung eines $H$-"aqui\-va\-ri\-an\-ten
\Name{Picard}-Gruppoids charakterisieren wollen. Hier haben wir eine
Formulierung mittels der Kategorientheorie angegeben, so da"s die
Struktur, die eine \Name{Morita}-Invariante auszeichnet, deutlich
wird. Mit diesen Mittel formulieren wir einige wichtige Beispiele
f"ur $H$-"aqui\-va\-ri\-an\-te \Name{Morita}-Theorie.

In Kapitel~\ref{sec:MoritaAequivalenzVonCrossProdukten} ist es uns nun
gelungen die $H$-\"aquivariante
\Name{Morita}-"Aquivalenz von Algeren und die gew"ohnliche
\Name{Morita}-"Aquivalenz von sogenannten {\em Cross-Produktalgebren} in Beziehung zu
setzen. Die Cross-Produktalgebra ist dabei eine $^\ast$-Algebra mit
den \glqq Informationen\grqq{} der $^\ast$-Algebra und der
\Name{Hopf}-$^\ast$-Algebra.  
 Wie wir herausgearbeitet haben, ist der
"Ubergang von einer $^\ast$-Algebra zur der korrespüondierenden
Cross-Produktalgebra funktoriell und liefert auf
\Name{Picard}-Gruppoidniveau einen 
Gruppoidmorphismus. Abschlie"sen werden wir das Kapitel mit einem
einfachen, jedoch instruktiven Beispiel.

In Kapitel~\ref{chapter:Sternprodukte} geben  wir eine Einf"uhrung in
das Konzept der Sternprodukte. Ausgehend von der Idee des Quantisierens
geben wir eine Motivation, sich mit Sternprodukten
auseinanderzusetzen. 
Dabei stellt ein flacher Phasenraum den einfachsten Fall dar, der die
bekannten Sternprodukte f"ur den $\field{R}^{2n}$ sowie den
$\field{C}^{n}$ liefert. Aufbauend darauf axiomatisieren wir
Sternprodukte im Sinne von \citet{bayen.et.al:1977a} und 
werden Sternprodukte f"ur symplektische Mannigfaltigkeiten
bzw.~\Name{Poisson}-Mannigfaltigkeiten definieren und deren wichtigste
Eigenschaften erl"autern, sowie grundlegende Definitionen angeben. 
Mit Hilfe der \Name{Fedosov}-Konstruktion werden wir vorf"uhren, wie
man Sternprodukte auf beliebigen symplektischen Mannigfaltigkeiten mit einem symplektischen
Zusammenhang konstruieren kann. Eine Erweiterung
dieser Konstruktion wird uns zur Deformation von Vektorb"undeln
f"uhren, die wir dazu noch in einem gruppeninvarianten Rahmen
formulieren werden. Dies ist insofern interessant, da 
es ein erstes (sehr spezielles, aber wichtiges) Beispiel f"ur die in
Kapitel~\ref{sec:MoritaAequivalenzdeformierterAlgebren} behandelte $H$-"aquivariante
\Name{Morita}-Theorie von deformierten Algebren liefert.    

Kapitel~\ref{sec:MoritaAequivalenzdeformierterAlgebren} verbindet die
Sternprodukte mit dem algebraischen Konzept einer
\Name{Hopf}-Wirkung. Wir definieren (rein algebraisch) 
$H$-"aquivariante Sternprodukte mit Hilfe von
Impulsabbildungen. Die Formulierung unterscheidet sich insofern
von bisherigen, da sie die Symmetrie durch eine \Name{Hopf}-$^\ast$-Algebra
zulassen. A posteriori geben wir dann eine geeignete
\Name{Hopf}-$^\ast$-Algebra an, so da"s die rein algebraische
Formulierung in die altbekannte, geometrisch motivierte Definition
"ubergeht.
Im weiteren f"uhren wir deformierte "Aquivalenzbimoduln ein, mit
deren Hilfe wir \Name{Morita}-"Aquivalenz von deformierten Algebren
(und damit insbesondere von Sternprodukten) beschreiben
k"onnen. Dar"uber gelangen wir zum \Name{Picard}-Gruppoid, dessen
Objekte deformierte Algebren sind. Besonders interessant ist die
{\em klassische Limes-Abbildung}, die auf \Name{Picard}-Gruppoidniveau
einen Morphismus in das \glqq klassische\grqq{} \Name{Picard}-Gruppoid
induziert.  

Die f"ur die Arbeit relevanten mathematischen Grundlagen haben wir in
einem Anhang zusammengetragen. Diese gehen an mancher Stelle "uber
das hinaus, was im Hauptteil der Arbeit gebraucht wird. 

In Anhang~\ref{chapter:AlgebraischeGrundlagen} findet man die algebraischen Grundlagen 
zu (geordneten) Ringen, formale Potenzreihen, \Name{Hopf}-$^\ast$-Algebren,
Cross-Produktalgebren, Verb"anden und projektive Moduln. Insbesondere
das Kapitel zu den Gruppen  $\GR{GL}{H,\alg{A}}$, $\GRn{GL}{H,\alg{A}}$,
  $\GR{U}{H,\alg{A}}$ und $\GRn{U}{H,\alg{A}}$ sei dem Leser ans Herz
  gelegt, da diese in Kapitel~\ref{sec:HaequivariantesPicardGruppoid}
  eine entscheidende Rolle spielen. In
  Anhang~\ref{chapter:GeometrischeGrundlagen} sind ein paar geometrische 
Grundlagen zusammengestellt. Interessant, da sie "uber das normale
Lehrbuchwissen hinausgehen, sind insbesondere die Kapitel "uber
symplektische Geometrie und $G$-invariante bzw.~$\LieAlg{g}$-invariante
Zusammenh"ange. 
Am Ende der Arbeit findet sich ein Symbolverzeichnis und ein
Personenindex, in dem wir alle namentlich erw"ahnten Personen
aufgelistet haben.


%% file: morita.tex
\chapter{\NameSection{Morita}-"Aquivalenz} 
\label{chapter:Morita}
\fancyhead[CO]{\slshape \nouppercase{\rightmark}} 
\fancyhead[CE]{\slshape \nouppercase{\leftmark}} 

\section{Grundlagen zur \Name{Morita}-Theorie}
\label{sec:MoritaEinfuehrendeDefinitionenUndSaetze}

\subsection{Pr"a-\Name{Hilbert}-R"aume, $^\ast$-Algebren und Positivit"at} 
\label{sec:PraeHilbertRaeumeUndSternAlgebren}

\subsubsection{Pr"a-\Name{Hilbert}-R"aume}
\label{sec:PraeHilbertRaeume}

Im folgenden wollen wir Pr"a-\Name{Hilbert}-R"aume definieren, um
sp"ater das verallgemeinernde Konzept der Pr"a-\Name{Hilbert}-Moduln
zu motivieren. Wir werden in dieser Arbeit eine leicht
verallgemeinernde Version von Pr"a-\Name{Hilbert}-R"aumen
ben"otigen, da wir nicht ausschlie"slich R"aume "uber
$\field{C}$ betrachten, sondern "uber einem Ring $\ring{C}$. Dabei
ist $\ring{C}=\ring{R}(\im)$ die komplexe Erweiterung des geordneten
Rings $\ring{R}$ mit $\im^2 = -1$. Diese Verallgemeinerung ist dadurch
motiviert, da"s wir uns sp"ater mit formal deformierten Algebren,
also insbesondere Sternprodukten, auseinandersetzen. Sternprodukte
sind $\fieldf{C}$-lineare
Abbildungen, und in diesem Formalismus k"onnen wir die Kategorie beibehalten, wenn
wir "uber formale Potenzreihen\index{formale Potenzreihe} reden wollen. Somit denken wir
insbesondere an die Ringe $\ring{R}=\field{R}$ und $\ring{R}=\fieldf{R}$. 
Es ist wichtig, da"s der Ring $\ring{R}$ geordnet ist, damit
wir das Konzept der Positivit"at implementieren k"onnen. Eine
detaillierte Einf"uhrung in das Konzept von (geordneten) Ringen,
sowie einfache und wichtige Beispiele haben wir in
Kapitel~\ref{sec:GeordeneteRinge} zusammengestellt. Eine weitere
Motivation, weshalb wir bei der Definition von
Pr"a-\Name{Hilbert}-R"aumen den Ring $\ring{C}$ verwenden wollen, ist,
da"s man bei der herk"ommlichen Definition wie sie in
Lehrb"uchern zu finden ist, nur die \glqq Ringeigenschaften\grqq{}
des K"orpers $\field{C}$ braucht. Somit ist unsere Definition eine in
gewisser Hinsicht sehr nat"urliche.

\begin{definition}[Pr"a-\Name{Hilbert}-Raum]
   \label{Definition:PraeHilbertRaum} 
   \indexbf{Pr\"a-\Name{Hilbert}-Raum}
    Sei $\alg{H}$ ein $\ring{C}$-Vektorraum mit einer
    Abbildung $\SP{\cdot,\cdot}: \alg{H} \times  \alg{H} \to \ring{C}$. Man
    nennt $\SP{\cdot,\cdot}$ ein {\em inneres Produkt}, falls die
    Abbildung $\SP{\cdot,\cdot}$ sesquilinear ist, d.~h.~f"ur alle
    $\phi,\psi,\chi \in \alg{H}$ und $s,t\in \ring{C}$ gilt 
    \begin{compactenum}
    \item $\SP{\phi, s\psi + t\chi} = s \SP{\phi,\psi} + t
        \SP{\phi,\chi}$,
    \item $\SP{\phi,\psi} = \cc{\SP{\psi,\phi}}$.
    \end{compactenum}

Ein inneres Produkt \indexbf{inneres Produkt} hei"st {\em positiv semidefinit}, falls
$\SP{\phi,\phi} \ge 0$ und {\em positiv definit}, \index{inneres Produkt!positiv definites} falls
$\SP{\phi,\phi} > 0$ f"ur alle $\phi\neq 0$.
Ein $\ring{C}$-Vektorraum mit einem positiv definiten inneren Produkt hei"st {\em
  Pr"a-\Name{Hilbert}-Raum}.   

Man nennt ein inneres Produkt {\em nichtausgeartet},\index{inneres Produkt!nichtausgeartetes} falls 
$\SP{\phi,\psi}=0$ f"ur alle $\psi \in \alg{H}$ impliziert, da"s
$\phi=0$ ist. Den {\em Ausartungsraum}\index{inneres Produkt!Ausartungsraum} von $\SP{\cdot, \cdot}$ bezeichnen wir mit

\begin{align}
    \label{eq:AusartungsraumSP}
 \alg{H}^{\perp} =\left\{\phi\in \alg{H} \,| \SP{\psi, \phi}=0 \;
     \text{f"ur alle $\psi \in \alg{H}$} \right\}. 
\end{align}
\end{definition}
Ein positives inneres Produkt kann damit nicht ausgeartet sein. 

Ein wichtiges Handwerkszeug wird die \Name{Cauchy-Schwarz}-Ungleichung
sein. 
\begin{satz}[{\Name{Cauchy-Schwarz}-Ungleichung}]
    \label{Satz:CauchySchwarzUngleichung}
\index{Ungleichung!\Name{Cauchy-Schwarz}}
\index{Cauchy-Schwarz-Ungleichung@\Name{Cauchy-Schwarz}-Ungleichung}
Sei $\alg{H}$ ein $\ring{C}$-Vektorraum mit positiv semidefinitem inneren
Produkt $\SP{\cdot, \cdot}$, dann gilt 
\begin{align}
    \label{eq:CauchySchwarzUngleichung}
    \SP{\phi, \psi}\SP{\psi,\phi} \le \SP{\phi,\phi}\SP{\psi,\psi}.
\end{align}
\end{satz}
\begin{proof}
Der Beweis von Satz~\ref{Satz:CauchySchwarzUngleichung} findet sich
beispielsweise f"ur den Fall $\ring{C}=\field{C}$ in
\citet{weidmann:2000a} oder \citet{lang:1997a}.  
\end{proof}
\begin{korollar}[\Index{Quotientenraum} eines semidefiniten inneren Produkts]
    \label{Korollar:QuotientenraumInneresProdukt}
    Sei $\alg{H}$ ein $\ring{C}$-Vektorraum mit einem positiv
    semidefiniten inneren Produkt $\SP{\cdot,\cdot}$. Das durch 
    \begin{align}
       \SP{[\phi],[\psi]} = \SP{\phi,\psi}  
    \end{align}
f"ur $[\phi], [\psi] \in \alg{H}/ \alg{H}^{\perp}$ definierte
innere Produkt ist ein wohldefiniertes positiv definites inneres
Produkt, und  $\alg{H}/ \alg{H}^{\perp}$ wird zu einem
Pr"a-\Name{Hilbert}-Raum.  
\end{korollar}
\begin{proof}
Das Korollar ist eine direkte Konsequenz der
\Name{Cauchy-Schwarz}-Ungleichung.  
\end{proof}

\begin{bemerkung}[Formale Potenzreihen]
Es kann vorkommen, da"s wir der Deutlichkeit halber auch von
$\ringf{R}$ reden werden. Hiermit soll herausgestellt werden, da"s es
sich beim Ring $\ring{R}$ um eine formale Potenzreihe
handelt. In gewisser Hinsicht ist diese Schreibweise tautologisch, da
der Ring $\ring{R}$ an sich bereits eine formale Potenzreihe sein kann. 
\end{bemerkung}

Wir wollen nun die {\em Kategorie der Pr"a-\Name{Hilbert}-R"aume} "uber
einem Ring $\ring{C}$ definieren. Die Objekte sind nat"urlich die
Pr"a-\Name{Hilbert}-R"aume, jedoch
m"ussen wir eine passende Wahl f"ur die Morphismen treffen. Als
geeignet werden sich die adjungierbaren Abbildungen zwischen
Pr"a-\Name{Hilbert}-R"aumen herausstellen. Motiviert ist dies
durch den Satz von \Name{Hellinger-Toeplitz}, der insbesondere in der
Quantenmechanik seine Anwendung findet.

\begin{definition}[\Name{Hilbert}-Raum]
    \label{Definition:HilbertRaum}
\index{Hilbert-Raum@\Name{Hilbert}-Raum}
Ein {\em \Name{Hilbert}-Raum} $\hilbert{H}$ ist ein vollst"andiger
Pr"a-\Name{Hilbert}-Raum "uber dem K"orper $\field{C}$. 
\end{definition}
Dabei bezeichen wir einen Raum als {\em vollst"andig}, wenn jede
\Name{Cauchy}-Folge konvergiert \citep{weidmann:2000a}.  
 
\begin{satz}[Satz von \Name{Hellinger-Toeplitz}]
    \label{Satz:HellingerToeplitzTheorem}
\index{Satz!Hellinger-Toeplitz@von \Name{Hellinger-Toeplitz}}
Seien $\hilbert{H}_{1}$und $\hilbert{H}_{2}$ \Name{Hilbert}-R"aume "uber
$\field{C}$, dann ist ein Operator $A:\hilbert{H}_{1} \to \hilbert{H}_{2}$
genau dann adjungierbar, wenn $A$ ein linearer und stetiger Operator ist.
\end{satz}
\begin{proof}
    Der Beweis findet sich in einschl"agiger Literatur, z.~B.~in
    \citep[S.~117]{rudin:1991a}. 
\end{proof}

\begin{bemerkung}[Funktional-analytische Aspekte]
F"ur einen tieferen Einblick in die funktional-analytischen Aspekt,
d.~h.~insbesondere des Spektralkalk"uls von Operatoren auf
\Name{Hilbert}-R"aumen, sei auf
\citep{werner:2002a,rudin:1991a,hirzebruch.scharlau:1991a,grossmann:1988a}
verwiesen. Da wir im weiteren nicht an Spektren oder anderen
funktional-analytischen Betrachtungen interessiert sind, werden wir
uns ausschlie"slich mit Pr"a-\Name{Hilbert}-R"aume besch"aftigen.
\end{bemerkung}

Die adjungierbaren Abbildung stellen somit eine besonders interessante
Klasse von Morphismen dar. Dies verleitet uns zu der folgenden Definition.

\begin{definition}[Adjungierbare Abbildungen]
 \label{Definition:AdjungierbareAbbildung}
 \index{Abbildung!adjungierbare}
   Seien $\alg{H}_{1}$ und $\alg{H}_{2}$ zwei
    Pr"a-\Name{Hilbert}-R"aume "uber $\ring{C}$ und $A:\alg{H}_{1}
    \to \alg{H}_{2}$ eine Abbildung. Man nennt $A$ {\em adjungierbar},
    falls es eine Abbildung $A^{\ast}:\alg{H}_{2} \to
    \alg{H}_{1}$ gibt, so da"s 
    \begin{align}
        \label{eq:AdjungierbareAbbildung}
        \SP{A\phi,\psi}_{\sss 2} = \SP{\phi,A^{\ast}\psi}_{\sss 1}
    \end{align}
f"ur alle $\phi \in \alg{H}_{1}$ und $\psi \in \alg{H}_{2}$. Man
nennt $A^{\ast}$ die {\em adjungierte Abbildung} zu $A$, und die Menge
aller adjungierbaren Abbildungen von $\alg{H}_{1}$ nach $\alg{H}_{2}$
bezeichnen wir mit 
$\alg{B}(\alg{H}_{1},\alg{H}_{2})$ bzw.~mit $\alg{B}(\alg{H})$ falls
$\alg{H}_{1}=\alg{H}_{2}=\alg{H}$.  
\end{definition}

\begin{lemma}[Adjungierte Abbildungen]
\label{Lemma:AdjungierteAbbildungen}
Seien $A,B \in \alg{B}(\alg{H}_{1},\alg{H}_{2})$ und $C \in
\alg{B}(\alg{H}_{2},\alg{H}_{3})$ $\ring{C}$-lineare Abbildungen. 
Die adjungierten Abbildungen $A^{\ast}, B^{\ast}, C^{\ast}$ sind
eindeutig bestimmte, $\ring{C}$-lineare Abbildungen und es gilt:
\begin{align}
    (aA+bB)^{\ast} = \cc{a}A^{\ast} + \cc{b}B^{\ast}, \quad
  (CA)^{\ast} = A^{\ast}C^{\ast}, \quad (A^{\ast})^{\ast}=A,
\end{align}
f"ur alle $a,b \in \ring{C}$. Ferner ist $\id=\id^{\ast} \in
\alg{B}(\alg{H})$. 
\end{lemma}
\begin{proof}
    Die Linearit"at sowie die Eindeutigkeit erh"alt man aus der
    Nichtausgeartetheit des inneren Produkts, die restlichen
    Eigenschaften durch simples Nachrechnen. 
\end{proof}

Damit gelangen wir zur Definition der Kategorie von Pr"a-\Name{Hilbert}-R"aumen
"uber einem Ring $\ring{C}$.

\begin{definition}[Kategorie der Pr"a-\Name{Hilbert}-R"aume]
\label{Definition:KategoriePraeHilbert}
\index{Kategorie!Prae-Hilbert-Raeume@der Pr\"a-\Name{Hilbert}-R\"aume}
Die Kategorie $\PraeHilbert{\ring{C}}$ der
Pr"a-\Name{Hilbert}-R"aume "uber $\ring{C}=\ring{R}(\im)$ ist
gegeben durch die Objekte, d.~h.~alle
Pr"a-\Name{Hilbert}-R"aume "uber dem Ring $\ring{C}$, sowie die
Morphismen zwischen den 
Objekten, die durch alle adjungierbaren Abbildungen zwischen je zwei
Pr"a-\Name{Hilbert}-R"aumen "uber $\ring{C}$ gegeben sind.
\end{definition}

\begin{definition}[Operatoren vom Rang Eins]
\index{Operator!vom Rang Eins}
Seien $\alg{H}_{1}$ und $\alg{H}_{2}$ Pr"a-\Name{Hilbet}-R"aume und
$\phi \in \alg{H}_{2}$, $\psi\in \alg{H}_{1}$, dann definiert man eine
lineare Abbildung $\Theta_{\phi,\psi}:\alg{H}_{1} \to \alg{H}_{2}$
durch
\begin{align}
    \label{eq:RangEinsOperator}
    \Theta_{\phi,\psi}\chi = \phi\SP{\psi,\chi}
\end{align}
f"ur $\chi \in \alg{H}_{1}$. Man nennt $\Theta_{\phi,\psi}$ einen
{\em Operator vom Rang Eins}. Linearkombinationen solcher Operatoren nennt man {\em Operatoren
von endlichem Rang}, \index{Operator!von endlichem Rang} deren
Gesamtheit man mit $\alg{F}(\alg{H}_{1},\alg{H}_{2})$ 
bezeichnet. Analog zu den adjungierbaren Operatoren sei $\alg{F}(
\alg{H},\alg{H}) =: \alg{F}(\alg{H})$.
\end{definition}
\begin{proposition}[Operatoren von endlichem Rang]
\label{Proposition:OperatorenVonEndlichemRang}
Seien $\alg{H}_{1}$, $\alg{H}_{2}$ und $\alg{H}_{3}$
Pr"a-\Name{Hilbert}-R"aume "uber dem Ring $\ring{C}$. Dann gilt
\begin{compactenum}
\item $\alg{F}(\alg{H}_{1},\alg{H}_{2}) \subseteq
    \alg{B}(\alg{H}_{1},\alg{H}_{2})$.
\item Die Operatoradjungtion 
\begin{align*} 
^\ast: \alg{B}(\alg{H}_{1},\alg{H}_{2})
    \to \alg{B}(\alg{H}_{2}\alg{H}_{1})
\end{align*}
 ist eine $\ring{C}$-antilineare Bijektion mit $\alg{F}(\alg{H}_{1},\alg{H}_{2})^\ast =
    \alg{F}(\alg{H}_{2},\alg{H}_{1})$.
\item $\alg{B}(\alg{H}_{2},\alg{H}_{3})\cdot
    \alg{F}(\alg{H}_{1},\alg{H}_{2}) \subseteq
    \alg{F}(\alg{H}_{1},\alg{H}_{3})$ und
    $\alg{F}(\alg{H}_{2},\alg{H}_{3}) \cdot
    \alg{B}(\alg{H}_{1},\alg{H}_{2}) \subseteq
    \alg{F}(\alg{H}_{1},\alg{H}_{3})$.     
\end{compactenum}
\end{proposition}
\begin{proof}
    Ein Beweis hierzu findet sich beispielsweise in
    \citep{waldmann:2004a:script}. 
\end{proof}

\subsubsection{$^\ast$-Algebren}

\index{Algebra!Stern-Algebra@$^\ast$-Algebra|textbf}

In diesem Kapitel werden wir das Konzept der $^\ast$-Algebren
einf"uhren. Dabei wird schnell klar, da"s dieses Konzept in
der Physik an vielen Stellen Verwendung findet, wie man an den
Beispielen \ref{Beispiele:SternAlgebren} sieht.
Mittels der $^\ast$-Algebren k"onnen wir {\em positive Funktionale} definieren,
denen eine ebenso wichtige Rolle zukommt, da sie die mathematische
Beschreibung von Zust"anden in der Physik entsprechen. 
\begin{definition}[$^\ast$-Algebra]
\label{Definition:SternAlgebra}
\index{Involution}
\index{Stern-Involution@$^\ast$-Involution}
\index{Antiautomorphismus!antilinearer}
   Sei $(\alg{A},\cdot)$ eine assoziative Algebra "uber dem Ring
   $\ring{C}$. Man nennt eine Algebra 
   $\alg{A}$ eine {\em $^\ast$-Algebra}, falls sie mit einem antilinearen
   Antiautomorphismus $^\ast:
   \alg{A} \to \alg{A}$ ausgestattet ist, 
   so da"s  
   \begin{align}
   (a^{\ast})^{\ast} = a, \quad (a\cdot b)^{\ast} = b^{\ast}\cdot
   a^{\ast},  
   \end{align} 
f"ur alle $a,b \in \alg{A}$. Ist die Algebra mit einem Einselement
$1_{\sss \alg{A}}$ ausgestattet, so gilt zus"atzlich $(1_{\sss
  \alg{A}})^{\ast} = 1_{\sss \alg{A}}$.
Wir bezeichnen die $^\ast$-Algebra mit
$(\alg{A},\cdot, ^\ast)$, bzw.~wenn keine Verwechslung m"oglich ist
nur mit $\alg{A}$ und unterdr"ucken wie gew"ohnlich die
Verkn"upfung und die Involution. 
\end{definition}
\begin{definition}[Kategorie der $^\ast$-Algebren]
\index{Kategorie!Stern-Algebren@der $^\ast$-Algebren}
Ein Morphismus von $^\ast$-Algebren ist ein Algebramorphismus $\phi:
\alg{A} \to \alg{B}$ f"ur den zus"atzlich $\phi(a^{\ast}) =
\phi(a)^{\ast}$ gilt. Auf diesem Wege erh"alt man die {\em Kategorie der
$^\ast$-Algebren (mit Einselement) "uber $\ring{C}$}, die man mit
$\staralg(\ring{C})$ (bzw.~$\starAlg(\ring{C})$)
bezeichnet.       
\end{definition}
Wie wichtig $^\ast$-Algebren sind, werden wir an den folgenden
Beispielen illustrieren. 
\begin{beispiele}[$^\ast$-Algebren]
\label{Beispiele:SternAlgebren}
~\vspace{-5mm}
\begin{compactenum}
    \item Der Ring $\ring{C}$ ist auf nat"urliche Weise eine
        $^\ast$-Algebra mit Einselement "uber sich selbst. Die $^\ast$-Involution ist die komplexe
        Konjugation $z \mapsto z^{\ast}=\cc{z}$. 
     \item F"ur jeden Pr"a-\Name{Hilbert}-Raum $\alg{H}$ bilden die
        adjungierbaren Operatoren $\alg{B}(\alg{H})$ eine
        $^\ast$-Algebra. Die Involution ist durch
        Adjungieren gegeben, und $\alg{F}(\alg{H})$ ist ein
        $^\ast$-Ideal, wie wir bereits in Proposition
        \ref{Proposition:OperatorenVonEndlichemRang} herausgestellt haben.
    \item Die Algebra $\alg{F}(\alg{H})$ ist auf nat"urliche Weise
         eine $^\ast$-Algebra (siehe
         \ref{Proposition:OperatorenVonEndlichemRang}). Sie besitzt
         i.~a.~kein Einselement.  
    \item Die $n\times n$-Matrizen "uber einem Ring $\ring{C}$ bilden
        die $^\ast$-Algebra $M_{n}(\ring{C})$. Die Verkn"upfung ist
        die "ubliche Matrixmultiplikation und die $^\ast$-Involution
        ist durch $(z_{ij})^{\ast}=(\cc{z_{ji}})$ gegeben. Ferner gilt
        $M_{n}(\ring{C}) \cong \alg{B}(\ring{C}^{n}) \cong
        \alg{F}(\ring{C}^{n})$. 
    \item Sind $\alg{A}$ und $\alg{B}$ zwei $^\ast$-Algebren, so ist
        auch $\alg{A} \otimes \alg{B}$ eine $^\ast$-Algebra. Die
        Involution der einzelnen Tensorfaktoren geschieht
        argumentweise $(a \otimes b)^{\ast} = a^{\ast} \otimes
        b^{\ast}$, die Multiplikation ist kanonisch $(a\otimes b)(a'
        \otimes b') = aa' \otimes bb'$, und beide Strukturen
        sind konsistent.
    \item Ist die Algebra $\alg{A}$ eine $^\ast$-Algebra, dann ist
        auch $M_{n}(\alg{A}):=\alg{A}\otimes M_{n}(\ring{C})$ eine $^\ast$-Algebra.
    \item Sei $M$ eine nichtkompakte Mannigfaltigkeit. Die Algebra
        $\Cinfc{M}$, der komplexen Funktionen mit 
        kompaktem Tr"ager auf $M$ mit der punktweisen
        komplexen Konjugation $\Cinfc{M} \ni f \mapsto
          f^{\ast}= \cc{f}$ bilden eine kommutative $^\ast$-Algebra, die
          allerdings kein Einselement besitzt.  
   \item Die Algebra eines \Name{Hermite}schen Sternprodukts \index{Sternprodukt!Hermitesches@\Name{Hermite}sches}
        $(\Cinff{M},\star,^{\cc{~}})$ auf einer
        \Name{Poisson}-Man\-nig\-fal\-tig\-keit $M$ ist eine $^\ast$-Algebra
        mit Einselement "uber dem Ring $\ring{C}=\fieldf{C}$. Die Involution
        ist durch die gew"ohnliche komplexe Konjugation gegeben, $(f
        \star g)^{\ast}= \cc{f \star g} = \cc{g} \star \cc{f} =
        g^{\ast} \star f^{\ast}$. 
        (wir werden sp"ater noch konkreter auf dieses Beispiel
        eingehen und verweisen auf Definitionen \ref{Definition:FormaleSternprodukte}
        und \ref{Definition:TypenVonSternprodukten}).  
\item Ein f"ur diese Arbeit wichtiges Beispiel sind die {\em
         Cross-Produktalgebren}, die wir sp"ater detailliert
       einf"uhren werden.\index{Cross-Produktalgebra}
       \index{Algebra!Cross-Produktalgebra} Dabei sei $\alg{A}$ eine
       $^\ast$-Algebra und $H$ eine \Name{Hopf}-$^\ast$-Algebra mit
       einer $^\ast$-Wirkung $\neact$ auf $\alg{A}$. Die
       Algebra $\cross{\alg{A}}{H}$ mit der Multiplikation 
       \begin{align*}
           (a \otimes h)\cdot (a' \otimes h') := a(h_{\sss (1)} \act a')
           \otimes h_{\sss (2)} h'
       \end{align*}
    und der $^\ast$-Involution $(a \otimes h)^\ast := h_{\sss (1)}^{\ast} \act a^{\ast} \otimes
          h_{\sss (2)}^{\ast}$ bilden eine $^\ast$-Algebra mit dem
          Einselement $1_{\sss \alg{A}} \otimes 1_{\sss H}$. Eine
            ausf"uhrliche Beschreibung, sowie Verweise auf geeignete
            Literatur sind in Kapitel~\ref{sec:CrossProdukte} zu finden.
\end{compactenum}
\end{beispiele}

\subsubsection{Positivit"at}
\label{sec:MoritaPositivitaetPraeHilbertSternalgebra}
\index{Positivitaet@Positivit\"at|textbf}
\begin{definition}[Positives Funktional und Zustand]
\label{Definition:PositivesFunktionalZustand}
\index{Funktional!positives}
Sei $\alg{A}$ eine $^\ast$-Algebra "uber einem Ring $\ring{C}$. Man
bezeichnet ein lineares Funktional $\omega: \alg{A} \to \ring{C}$ als
{\em positiv}, falls $\omega(a^{\ast}a) \ge 0$ f"ur alle $a \in \alg{A}$ ist. 
Ein positives Funktional f"ur das zus"atzlich $\omega(1_{\sss
  \alg{A}}) = 1_{\sss \ring{C}}$ gilt bezeichnet man als {\em
  \Index{Zustand}}.    
\end{definition}
\begin{bemerkungen}[Positives Funktional und Zustand]
~\vspace{-5mm}
\begin{compactenum}
  \item Die Definition \ref{Definition:PositivesFunktionalZustand}
      schlie"st auch deformierte Algebren, d.~h.~insbesondere
      Sternprodukte, ein. Im Fall von Sternprodukten wird ein
      Funktional $\omega$ $\fieldf{C}$-linear und nimmt Werte in
      $\fieldf{C}$ an. 
Elemente in dem Ring $\ring{R}=\fieldf{R}$ sind genau dann positiv,
wenn der erste nichtverschwindende Term positiv ist, d.~h.~ein Element 
        \begin{align}
            \label{eq:PositivitaetEinesElementsInR}
            \ring{R} \ni r = \sum_{n=n_{0}}^{\infty} \lambda^n r_{n}
        \end{align}
         ist genau dann positiv, wenn $r_{n_{0}} > 0$ im Sinne von
         Definition \ref{Definition:GeordneterRing}.   
   \item Wie wir unter den Beispielen \ref{Beispiele:SternAlgebren} gesehen
haben, sind nicht alle interessanten $^\ast$-Algebren mit einem Einselement
ausgestattet. Daher kann man in einigen F"allen nicht von Zust"anden
reden. Im wesentlichen werden wir uns auf $^\ast$-Algebren
mit einem Einselement beschr"anken. Bei $^\ast$-Algebra mit
Einselement "uber einem Ring $\ring{C}=\qring{C}$, der sogar ein
K"orper ist, kann man durch Normierung aus einem linearen Funktional
einen Zustand gewinnen. 
\end{compactenum}
\end{bemerkungen}
\begin{lemma}[\Name{Cauchy-Schwarz}-Ungleichung f"ur positive Funktionale]
    \label{Lemma:CauchySchwarzFunktional}
\index{Ungleichung!\Name{Cauchy-Schwarz}}
\index{Cauchy-Schwarz-Ungleichung@\Name{Cauchy-Schwarz}-Ungleichung}
Sei $\alg{A}$ eine $^\ast$-Algebra und $\omega:\alg{A}\to \ring{C}$
ein lineares positives Funktional, dann gilt
\begin{compactenum}
\item $\omega(a^{\ast} b) = \cc{\omega(b^{\ast} a)}$. Besitzt die
    Algebra ein Einselement $1_{\sss \alg{A}} \in \alg{A}$, so ist insbesondere
    $\omega(a^{\ast}) =\cc{\omega(a)}$ und aus $\omega(1_{\sss
      \alg{A}})=0$ folgt $\omega=0$.
\item $\omega(a^{\ast} b) \cc{\omega(a^{\ast} b)} \le \omega(a^{\ast}
    a)\omega(b^{\ast} b)$. 
\end{compactenum}
\end{lemma}
Mit Hilfe der linearen positiven Funktionale sind wir nun in der Lage
positive Algebraelemente zu definieren.  
\begin{definition}[Positive Algebraelemente -- die Mengen
    $\alg{A}^{+}$ und $\alg{A}^{++}$]
\index{positive Algebraelemente}
Sei $\alg{A}$ eine $^\ast$-Algebra "uber dem Ring $\ring{C}$. 
Man bezeichnet ein Element $a \in \alg{A}$ als {\em positiv},
    falls $\omega(a) > 0$ f"ur jedes beliebige positive Funktional
    $\omega: \alg{A} \to \ring{C}$. Die Menge aller dieser Elemente
    sei mit $\alg{A}^{+}$ bezeichnet:   
    \begin{align}
        \alg{A}^{+}=\left\{ a\in \alg{A} \,|\,\omega(a)> 0 \text{ f"ur alle
              positiven }\omega \right\}.
    \end{align}
Ferner definiert man  
    \begin{align}
        \alg{A}^{++}=\left\{ a\in \alg{A}\, \Big| \,
            a=\sum\nolimits_{i=1}^{n} \beta_i b^{\ast}_{i} b_{i}, \;
            \text{mit $0< \beta_{i} \in \ring{R}$ und $b_{i} \in
              \alg{A}$} \right\} .   
    \end{align} 
\end{definition}
Es liegt nah, positive Abbildungen zu betrachten, die positive
Elemente einer $^\ast$-Algebra auf positive Elemente einer weiteren
$^\ast$-Algebra abbilden.
\begin{definition}[Positive und vollst"andig positive Abbildungen]
\index{Abbildung!positive}
\index{Abbildung!vollstaendig positive@vollst\"andig positive}
 Seien $\alg{A}$ und $\alg{B}$ $^\ast$-Algebren. Man bezeichnet eine
 lineare Abbildung $\Phi:\alg{A} \to \alg{B}$ als {\em positiv}, wenn
 f"ur alle $a\in \alg{A}^{+}$ gilt $\Phi(a) \in \alg{B}^{+}$. Ferner
 nennt man $\Phi$ {\em vollst"andig positiv}, wenn die komponentenweise
 Erweiterung auf $\Phi: M_{n}(\alg{A}) \to M_{n}(\alg{B})$ positiv
 f"ur alle $n\in \field{N}$ ist. 
\end{definition}
\begin{bemerkungen}
\label{Bemerkung:PositiveElemente}
~\vspace{-5mm}
\begin{compactenum}
\item Offensichtlich ist $\alg{A}^{++} \subseteq \alg{A}^{+}$, da wir
aufgrund der Linearit"at von $\omega$ jedes
Element $a\in \alg{A}^{++}$ schreiben k"onnen als
\begin{align}
    \label{eq:A++UntermengeVonA+}
    \omega(a) & =  \omega \left(\sum\nolimits_{i} \beta_{i} b^{\ast}_{i} b_{i}
    \right) = \sum_{i} \underbrace{\beta_{i}}_{>0}
    \underbrace{\omega (b^{\ast}_{i} b_{i})}_{>0} >0.  
\end{align} 
Es kann jedoch Elemente in $\alg{A}^{+}$ geben, die sich nicht als
positive Linearkombination von Quadraten schreiben lassen. 
\item Bei $C^{\ast}$-Algebren ist $\alg{A}^{+}=\alg{A}^{++}$, da
jedes positive Element in einer $C^\ast$-Algebra eine eindeutige
Wurzel besitzt, so da"s $a=(\sqrt{a})^{2}$ geschrieben werden
kann. Wir verweisen auf 
\citep{sakai:1971a,kadison.ringrose:1997a,kadison.ringrose:1997b,landsman:1998a}.
\end{compactenum}
\end{bemerkungen} 


\subsection{Pr"a-\Name{Hilbert}-Moduln und Darstellungen} 

Wir wollen nun eine Verallgemeinerung des
Konzepts der Pr"a-\Name{Hilbert}-R"aume dar\-le\-gen, die {\em
  Pr"a-\Name{Hilbert}-Moduln}. Diese werden bei der Klassifikation
von Algebren mit Hilfe der \Name{Morita}-Theorie ein wichtiges Werkzeug
sein. 

\begin{definition}[Innerer Produktmodul und Pr"a-\Name{Hilbert} Modul]
\label{Definition:PraeHilbertModul}
\index{Prae-Hilbert-Modul@Pr\"a-\Name{Hilbert}-Modul}
\index{innerer Produktmodul}
Sei $\alg{D}$ eine $^\ast$-Algebra "uber dem Ring $\ring{C}$. Ein
{\em innerer Produktmodul} ist ein $\alg{D}$-Rechtsmodul mit einem inneren Produkt  
\begin{align}
    \label{eq:InneresProduktPraeHilbertModul}
    \rSP{\cdot,\cdot}{\alg{D}}: \rmod{H}{D} \times \rmod{H}{D} \to \alg{D},
\end{align}
mit den folgenden Eigenschaften.
\begin{compactenum}
\item Das innere Produkt $\rSP{\cdot,\cdot}{\alg{D}}$ ist
    $\ring{C}$-sesquilinear, d.~h.~linear im zweiten Argument und
    antilinear im ersten, so da"s
    $\rSP{x,ry+sz}{\alg{D}} = r \rSP{x,y}{\alg{D}} + s
    \rSP{x,z}{\alg{D}}$ f"ur alle $x,y,z \in \rmod{H}{D}$ und alle
    $r,s \in \ring{C}$.  
\item $\rSP{x, y}{\alg{D}} = \rSP[\ast]{y, x}{\alg{D}}$ f"ur alle $x, y \in \rmod{H}{D}$
\item Das innere Produkt ist $\alg{D}$-rechtslinear $\rSP{x,y \cdot
      d}{\alg{D}} = \rSP{x, y}{\alg{D}} d$ 
    f"ur alle $x, y \in \rmod{H}{D}$ und $d \in \alg{D}$.
\item Das innere Produkt $\rSP{\cdot,\cdot}{\alg{D}}$ ist
    nichtausgeartet, das hei"st f"ur $\rSP{x,\cdot}{D}=\rSP{y,
      \cdot}{D}$ ist $x=y$, und aus $\rSP{x, \cdot}{\alg{D}} \equiv 0$ folgt
    $x=0$. 
\end{compactenum}
Ferner spricht man von einem {\em Pr"a-\Name{Hilbert}-Modul}, wenn
zus"atzlich noch die folgende Bedingung erf"ullt ist:
\begin{compactenum}
\item[\it v.)] Das innere Produkt $\rSP{\cdot,\cdot}{\alg{D}}$ ist vollst"andig
    positiv. Dies bedeutet, da"s f"ur alle $n\in \field{N}$ und
    alle $x_{1},x_{2},\cdots, x_{n} \in \rmod{H}{D}$ die Matrix
    $(\rSP{x_{i},x_{j}}{\alg{D}}) \in M_{n}(\alg{D})^{+}$
    ist. 
\end{compactenum}
\end{definition}
\begin{definition}[Vollheit eines inneren Produkts]
\index{Vollheit}
    Sei $\alg{D}$ eine $^\ast$-Algebra "uber dem Ring $\ring{C}$ und
    $\rSP{\cdot,\cdot}{\alg{D}}:\rmod{H}{D} \times \rmod{H}{D} \to
    \alg{D}$ ein $\alg{D}$-wertiges inneres Produkt. Man nennt das
    innere Produkt {\em voll}, falls die $\ring{C}$-lineare H"ulle
    des inneren Produkts die ganze Algebra aufspannt, d.~h.~
    \begin{align}
        \label{eq:VollheitInneresProdukt}
        \CSpan{\rSP{x,y}{\alg{D}}\,|\,\text{f"ur }x,y \in
          \rmod{H}{D}}=\alg{D}. 
    \end{align}
\end{definition}

Nicht alle Algebren sind als Hilfsalgebren
geeignet, daher wollen wir definieren, wann eine Algebra
$\alg{D}$ f"ur unsere Zwecke {\em zul"assig} ist.

\begin{definition}[Zul"assige Algebra $\alg{D}$]
\label{Definition:ZulaessigeHilfsalgebra}
\index{Algebra!zulaessig@zul\"assige}
   Man nennt eine Algebra $\alg{D}$ {\em zul"assig}, falls f"ur
   jeden Pr"a-\Name{Hilbert}-Modul $\rmod{H}{D}$ das innere Produkt
   positiv definit ist, d.~h.~f"ur $\rSP{x,x}{\alg{D}}=0$ ist $x=0$.
\end{definition}
\begin{bemerkungen}[Innerer Produktmodul und Pr"a-\Name{Hilbert}-Modul]
\label{Bemerkung:PraeHilbertModul}
~\vspace{-5mm}
\begin{compactenum}
\item Die Definition f"ur einen inneren Produktmodul oder einen Pr"a-\Name{Hilbert}
    $\alg{B}$-Linksmoduln $\lmod{B}{H}$ mit einem $\ring{C}$-ses\-qui\-li\-ne\-aren
    inneren Produkt $\lSP{\alg{B}}{\cdot,\cdot}: \lmod{B}{H} \times
    \lmod{B}{H} \to \alg{B}$ verl"auft analog. Allerdings ist ein
    solches inneres Produkt $\ring{C}$-linear im ersten Argument und mit einer
      Linksmultiplikation von $\alg{B}$ vertr"aglich, so da"s die
      folgenden Bedingungen gelten
      $\lSP{\alg{B}}{rx+sy,z}=r \lSP{\alg{B}}{x,z}+s \lSP{\alg{B}}{y,z}$ und
     $\lSP{\alg{B}}{b\cdot x,y} = b \lSP{\alg{B}}{x,y}$ f"ur
     alle $b\in \alg{B}$, $x,y,z \in \lmod{B}{H}$ und $r,s \in
     \ring{C}$.
\item F"ur den Fall $\alg{D}=\ring{C}$ geht die Definition
    \ref{Definition:PraeHilbertModul} in die eines
    Pr"a-\Name{Hilbert}-Raums (Definition
    \ref{Definition:PraeHilbertRaum}) "uber. Einen Beweis f"ur die vollst"andige
    Positivit"at des inneren Produkts $\rSP{\cdot,\cdot}{\ring{C}}$
    f"ur Pr"a-\Name{Hilbert}-R"aume findet man in
    \citep[App.~A]{bursztyn.waldmann:2001a}. Es stellt sich heraus,
    da"s vollst"andige Positivit"at dort bereits aus der Positivit"at von
    $\rSP{\cdot,\cdot}{\ring{C}}$ folgt. 
\item Definition \ref{Definition:PraeHilbertModul} verallgemeinert
    desweiteren \Name{Hilbert}-Moduln "uber $C^{\ast}$-Algebren, wie
    sie in \citet{lance:1995a} behandelt werden. Die vollst"andige
    Positivit"at ist auch hier eine direkte Folge der
    Positivit"at des inneren Produkts \citep[Lemma 4.2]{lance:1995a}.
\item Die Definitionen zu adjungierbaren Abbildungen und Operatoren
    endlichen Rangs sind identisch zu denen in Kapitel
    \ref{sec:PraeHilbertRaeumeUndSternAlgebren}. Insbesondere sind die
    bei Pr"a-\Name{Hilbert}-Moduln die adjungierbaren Abbildung
    eindeutig, und es gilt f"ur alle $A \in \alg{B}(\rmod{H}{D})$
    Lemma~\ref{Lemma:AdjungierteAbbildungen}, was an der
    Nichtausgeartetheit des inneren Produkts liegt.
\item Eine Hilfsalgebra $\alg{D}$ ist immer dann zul"assig, wenn es
    gen"ugend viele positive lineare Funktionale $\omega: \alg{D} \to
    \ring{C}$ gibt, so da"s f"ur \Name{Hermite}sche Elemente
    $d=d^{\ast} \neq 0$ ein Funktional $\omega (d) \neq 0$
    existiert und $d+d=0$ ist, folgt da"s $d=0$ ist.
\end{compactenum}
\end{bemerkungen}
\begin{beispiel}[\Name{Hermite}sches Vektorb"undel]
\index{Vektorbuendel@Vektorb\"undel!Hermitesch@\Name{Hermite}sches}
    Ein wichtiges Beispiel f"ur die Pr"a-\Name{Hilbert}-Moduln
    kommt aus der Differentialgeometrie. Sei $\bundle{E}{\pi}{M}$ ein
    komplexes Vektorb"undel mit einer \Name{Hermite}schen Fasermetrik
    $h$, dann ist der $\Cinf{M}$-Rechtsmodul
    $\rmodo{\schnitt{E}}{\Cinf{M}}$ mit dem inneren Produkt
    $\SP{s,s^{\prime}}(x) := h_{\sss
      x}(s(x),s^{\prime}(x))$ mit $x\in M$ und $s,s^{\prime} \in
    \schnitt{E}$ ein Pr"a-\Name{Hilbert}-Modul. Die adjungierbaren
    Abbildungen sind die Schnitte im Endomorphismenb"undel
    $\alg{B}(\rmodo{\schnitt{E}}{\Cinf{M}}) = 
    \schnitt{\End{E}}$ mit der kanonischen Wirkung auf
    $\schnitt{E}$. Die $^\ast$-Involution wird durch die
    Fasermetrik $h$ induziert.   
\end{beispiel}
\begin{definition}[$^\ast$-Darstellung einer $^\ast$-Algebra]
\index{Darstellung!Sterndarstellung@$^{\ast}$-Darstellung}
\index{Stern-Darstellung@$^\ast$-Darstellung}
    Seien $\alg{A}$ und $\alg{D}$ $^\ast$-Algebren "uber dem Ring
    $\ring{C}$, und $\rmod{H}{D}$ ein
    Pr"a-\Name{Hilbert}-Modul. Eine {\em $^\ast$-Darstellung}
    $(\rmod{H}{D},\pi)$ ist das Paar bestehend aus einem Pr"a-\Name{Hilbert}-Modul
    $\rmod{H}{D}$ und einem $^\ast$-Ho\-mo\-mor\-phis\-mus $\pi: \alg{A} \to
    \alg{B}(\rmod{H}{D})$, d.~h.~es gilt f"ur alle $a,b \in
    \alg{A}$
    \begin{compactenum}
    \item $\pi(ab) = \pi(a)\pi(b)$,
    \item $\pi(a^{\ast}) = \pi(a)^\ast$.
    \end{compactenum}
Man bezeichnet eine $^\ast$-Darstellung $(\rmod{H}{D},\pi)$ ferner als
{\em stark nichtausgeartet}, wenn 
\begin{align}
\label{eq:StarkNichtausgearteteDarstellung}
    \pi(\alg{A}) \rmod{H}{D} = \rmod{H}{D}.
\end{align}
\end{definition}
Damit wird $\rmod{H}{D}$ zu einem $(\alg{B}(\rmod{H}{D}),
\alg{D})$-Bimodul (vgl.~Bemerkungen
\ref{Bemerkung:SternDarstellungen}). Eine $^\ast$-Darstellung $(\rmod{H}{D}, \pi$) ist immer stark
    nichtausgeartet, falls die dargestellte $^\ast$-Algebra $\alg{A}$ ein
    Einselement $1_{\sss \alg{A}} \in \alg{A}$ besitzt, da
    $\pi(1_{\sss \alg{A}})= \id_{\sss \rmod{H}{D}}$.     
\begin{definition}[Verschr"ankungsoperator (Intertwiner)]
\index{Verschraenkungsoperator@Verschr\"ankungsoperator}
    Seien $(\rmod{H}{D},\pi)$ und $(\rmod{K}{D}, \rho)$ zwei
    $^\ast$-Darstellungen einer $^\ast$-Algebra $\alg{A}$, so bezeichnet man
    eine Abbildung $T \in \alg{B}(\rmod{H}{D},
    \rmod{K}{D})$ als {\em Verschr"ankungsoperator} oder {\em
      Intertwiner}, falls 
    \begin{align}
        \label{eq:Intertwiner}
        T \pi (a) = \rho(a) T
    \end{align}
f"ur alle $a\in \alg{A}$ ist.
\end{definition}
\begin{definition}[Kategorie der $^\ast$-Darstellungen auf
    Pr"a-\Name{Hilbert}-Moduln] 
\index{Kategorie!@Sterndarstellungen auf Prae-Hilbert-Raeumen@der $^\ast$-Darstellungen auf
    Pr\"a-\Name{Hilbert}-Moduln}
Man bezeichnet die Kategorie der $^\ast$-Darstellungen der $^\ast$-Algebra $\alg{A}$ auf
inneren Produkt $\alg{D}$-Rechts\-modul mit $\smod[\alg{D}](\alg{A})$. Die
Objekte sind die $^\ast$-Darstellungen und 
die Morphismen die Ver\-schr"an\-kungs\-opera\-tor\-en. Ferner
bezeichnet man die Unterkategorie der $^\ast$-Darstellungen auf
Pr"a-\Name{Hilbert}-Moduln mit $\srep[\alg{D}](\alg{A})$. Sind die
$^\ast$-Darstellungen au"serdem noch
stark nichtausgeartet, so werden die Kategorien mit 
$\sMod[\alg{D}](\alg{A})$ bzw.~$\sRep[\alg{D}](\alg{A})$ bezeichnet.   
\end{definition}

\subsection{Moduln und Bimoduln}

Da f"ur die \Name{Morita}-Theorie die Bimoduln eine wichtige
Rolle spielen werden, werden wir die wichtigen Definitionen kurz zusammenfassen.
\begin{definition}[Kategorien der $\ring{R}$-Moduln]
\label{Definition:KategorieRModuln}
\index{Kategorie!R-Moduln@der $\ring{R}$-Moduln}
\index{Modul!R-Modul@$\ring{R}$-Modul}
    Sei $\ring{R}$ ein Ring. Man bezeichnet die Kategorie der
    $\ring{R}$-Linksmoduln mit $\katlmod{\ring{R}}$ und die
    Kategorie der $\ring{R}$-Rechtsmoduln mit
    $\katrmod{\ring{R}}$. Die Objekte der Kategorie sind die Moduln,
    die Morphismen die Modulmorphismen,
    \begin{align}
        \label{eq:ModulMorphismus}
        T:\lmodo{\ring{R}}{\alg{E}} \to \lmodo{\ring{R}}{\alg{E'}},
        \quad \lmodo{\ring{R}}{\alg{E}} \ni x \mapsto T(x) \in
        \lmodo{\ring{R}}{\alg{E'}},
    \end{align}
so da"s $T(r\cdot x) = r\cdot T(x)$ f"ur alle $r\in \ring{R}$ und alle
$x \in \lmodo{\ring{R}}{\alg{E}}$ ist. F"ur $\ring{R}$-Rechtsmoduln
ist die Definition analog. Kurz schreiben wir  
    \begin{align}
        \label{eq:LinksUndRechtsmoduln}
        \katlmod{\ring{R}} = \{\lmodo{\ring{R}}{\alg{E}}\, |
        \,\lmodo{\ring{R}}{\alg{E}} \,
        \text{ist $\ring{R}$-Linksmodul} \} \quad \text{und} \quad
        \katrmod{\ring{R}} = \{\rmodo{\alg{E}}{\ring{R}}\, |
        \,\rmodo{\alg{E}}{\ring{R}} \, \text{ist $\ring{R}$-Rechtsmodul} \}.
    \end{align}
Gilt zus"atzlich, da"s $\ring{R}\cdot \lmodo{\ring{R}}{\alg{E}} =
\lmodo{\ring{R}}{\alg{E}}$ bzw.~$ \rmodo{\alg{E}}{\ring{R}} \cdot
\ring{R} = \rmodo{\alg{E}}{\ring{R}}$ so bezeichnet man die
Kategorien mit $\katlMod{\ring{R}}$ bzw.~$\katrMod{\ring{R}}$. Analog 
definiert man f"ur eine gegebene Algebra $\alg{A}$ die Kategorien
$\katlmod{\alg{A}}$, $\katlMod{\alg{A}}$, $\katrmod{\alg{A}}$ und
$\katrMod{\alg{A}}$.     
\end{definition}
\begin{bemerkung}[Kategorie der $\ring{R}$-Moduln und $\alg{A}$-Moduln]
Im Fall von Ringen sind die Definitionen von $\katlMod{\ring{R}}$
bzw.~$\katrMod{\ring{R}}$ tautologisch, da
in unserer Definition \ref{Definition:Ring} jeder Ring ein Einselement
besitzt und damit jeder $\ring{R}$-Modul entweder in
$\katlMod{\ring{R}}$ oder $\katrMod{\ring{R}}$ ist. In der Literatur
  wird jedoch nicht immer gefordert, da"s $1_{\sss \ring{R}} \in \ring{R}$ ist
  (vgl.~beispielsweise \citep{gerritzen:1994a}).    

F"ur $\alg{A}$-Moduln ist diese Unterscheidung jedoch wichtig, da es
wichtige Beispiele f"ur Algebren ohne Einselement gibt
(vgl.~Beispiele \ref{Beispiele:SternAlgebren}). Beschr"ankt man sich
auf die Kategorien $\katlMod{\alg{A}}$ bzw.~$\katrMod{\alg{A}}$, so umgeht
man technische Schwierigkeiten, die beispielsweise in \citep{bursztyn.waldmann:2003a:pre}
beleuchtet werden. 
\end{bemerkung}

\begin{definition}[Bimodul]
\label{Definition:Bimodul}
\index{Bimodul|textbf}
Seien $\alg{A}$ und $\alg{B}$ Algebren "uber einem Ring
$\ring{C}$. Man bezeichnet einen 
$\alg{A}$-Rechtsmodul $\rmod{E}{A} \in \katrMod{\alg{A}}$, der gleichzeitig ein
$\alg{B}$-Links\-mo\-dul $\lmod{B}{E} \in \katlMod{\alg{B}}$ ist als
{\em $(\alg{B}, \alg{A})$-Bimodul} $\bimod{B}{E}{A}$, falls
zus"atzlich die Bedingung    
\begin{align}
(b \cdot x) \cdot a = b \cdot (x \cdot a)
\end{align}
f"ur alle $a \in \alg{A}$ sowie $b \in \alg{B}$ und $x \in
\bimod{B}{E}{A}$ erf"ullt ist.
\end{definition}

\begin{bemerkung}[Bimodul]
Wir haben bei der Definition bewu"st die Kategorien
$\katrMod{\alg{A}}$ bzw.~$\katlMod{\alg{B}}$ gew"ahlt um 
\glqq pathologische\grqq{} Bimoduln auszuschlie"sen. Sind die
Algebren mit einem Einselement ausgestattet, so m"ochten wir
insbesondere die Eigenschaften $1_{\sss \alg{B}} \cdot x = x$ und $x
\cdot 1_{\sss \alg{A}} = x$ f"ur alle $x \in \bimod{B}{E}{A}$.  
\end{bemerkung}

F"ur die sp"ater behandelte $^\ast$- bzw.~starke
\Name{Morita}-"Aquivalenz werden sich Bimoduln f"ur
$^\ast$-Algebren "uber dem Ring $\ring{C}$, die zudem mit
zwei inneren Produkten ausgestattet sind, als besonders n"utzlich
herausstellen. Daher wollen wir Bimoduln $\bimod{B}{E}{A}$ betrachten,
die mit einem $\alg{B}$-wertigen sowie einem $\alg{A}$-wertigen
inneren Produkt ausgestattet sind, und diese mit $\bimodplus{B}{E}{A}$
bezeichnen. Wir werden ferner fordern, da"s 
diese inneren Produkte mit den Modulstrukturen und miteinander
vertr"aglich sind.

\begin{definition}[Bimodulmorphismus]
\index{Bimodulmorphismus}
Seien $\alg{A},\alg{B}$ zwei Algebren "uber dem Ring $\ring{C}$ und
$\bimod{B}{E}{A}$ sowie $\bimod{B}{E^{\prime}}{A}$ zwei
$(\alg{B},\alg{A})$-Bimoduln. Man nennt die Abbildung  
\begin{align}
    \label{eq:BimodulMorphismus}
T: \bimod{B}{E}{A} \ni x \mapsto T(x) \in \bimod{B}{E^{\prime}}{A}  
\end{align}
einen {\em Bimodulmorphismus}, falls f"ur alle Elemente $x \in
\bimod{B}{E}{A}$, alle $b\in \alg{B}$ und $a\in \alg{A}$ 
\begin{align}
T(b\cdot x\cdot a)=b\cdot T(x)\cdot a
\end{align}
gilt. Desweiteren nennt man den Bimodulmorphismus {\em isometrisch}, falls
die Algebren $\alg{A}$ und $\alg{B}$ $^\ast$-Algebren sind und f"ur die inneren Produkte
\begin{align}
    \label{eq:IsometrischeBimodulmorphismen}
    \SP[\sss \alg{E}^{\prime}]{T(x),T(y)} = \SP[\sss \alg{E}]{x,y}
\end{align}
gilt.
\end{definition}

\begin{definition}[Kompatibilit"at Bimodul]
\label{Definition:KompatibilitaetBimodulMitLinkswirkung}
\index{Bimodul!kompatibel mit innerem Produkt}
Seien $\alg{A}$ und $\alg{B}$ $^\ast$-Algebren und sei weiter
$\bimod{B}{E}{A}$ ein $(\alg{B},\alg{A})$-Bimodul mit
$\alg{A}$-wertigem inneren Produkt, dann nennt man das innere Produkt {\em
  kompatibel} oder {\em vertr"aglich  mit der $\alg{B}$-Linksmodulstruktur}, falls  
\begin{align}
    \label{eq:KompatibilitaetBimodulMitLinkswirkung}
    \rSP{b\cdot x, y}{\alg{A}} = \rSP{x, b^{\ast} \cdot y}{\alg{A}}
\end{align}
f"ur alle $b\in \alg{B}$ und alle $x,y \in \bimod{B}{E}{A}$. Analog
ist die $\alg{A}$-Rechtswirkung kompatibel mit einem
$\alg{B}$-wertigen inneren Produkt falls 
\begin{align}
    \label{eq:KompatibilitaetBimodulMitRechtswirkung}
    \lSP{\alg{B}}{x, y\cdot a} = \lSP{\alg{B}}{x \cdot a^{\ast}, y}
\end{align}
f"ur alle $a\in \alg{A}$ und alle $x,y \in \bimod{B}{E}{A}$.
\end{definition}

Sei nun ein $(\alg{B},\alg{A})$-Bimodul $\bimod{B}{E}{A}$ mit einem $\alg{B}$-wertigen
und einem $\alg{A}$-wertigen inneren Produkt ausgestattet, so sind die
beiden inneren Produkte a priori v"ollig unabh"angig voneinander. Wir
ben"otigen sp"ater jedoch eine Vertr"aglichkeit der beiden inneren Produkte. 

\begin{definition}[Vertr"aglichkeit der inneren Produkte]
\label{Lemma:KompabilitaetDerInnrenProdukte}
\index{inneres Produkt!Vertraeglichkeit@Vertr\"aglichkeit}
Gegeben ein Bimodul $\bimodplus{B}{E}{A}$. Man nennt die beiden
inneren Produkte {\em kom\-pa\-ti\-bel} oder {\em vertr"aglich}, falls f"ur alle $x,y,z \in
\bimod{B}{E}{A}$ 
\begin{align}
    \label{eq:KompabilitaetInnereProdukte}
    \lSP{\alg{B}}{x,y} \cdot z = x \cdot \rSP{y,z}{\alg{A}}
\end{align}
gilt.
\end{definition}

\begin{bemerkung}
\label{Bemerkung:SternDarstellungen}
 Sei $(\rmod{H}{D},\pi)$ eine $^\ast$-Darstellung der
    $^\ast$-Algebra $\alg{A}$, so ist $\rmod{H}{D}$ auf kanonische
    Weise ein Bimodul, da $\alg{B}({\rmod{H}{D}})$ von links auf
    $\rmod{H}{D}$ wirkt und $\bimod{\alg{B}({\rmod{H}{D}})}{H}{D}$ zu
    einem $(\alg{B}({\rmod{H}{D}}), \alg{D})$-Bimodul macht. 
    Das innere Produkt $\rSP{\cdot,\cdot}{\alg{D}}$ ist per Definition
    \ref{Definition:AdjungierbareAbbildung} kompatibel mit der
    $\alg{B}({\rmod{H}{D}})$-Linksmodulstruktur. 
\end{bemerkung}

\subsection{Tensorprodukte und \Name{Rieffel}-Induktion von
  $^\ast$-Darstellungen} 
\label{sec:TensorprodukteRieffelInduktion}

Ziel dieses Kapitels ist es einen Funktor zwischen zwei Kategorien von
$^\ast$-Darstellungen anzugeben. Die folgenden Ideen gehen auf
\citet{rieffel:1974a,rieffel:1974b} zur"uck. Die Darstellung dieses
Kapitels orientiert sich im wesentlichen an
\citep{bursztyn.waldmann:2003a:pre}. 

Gegeben zwei $^\ast$-Algebren $\alg{A}$ und $\alg{B}$ "uber einem
Ring $\ring{C}$. Sei $\rmodplus{H}{B}$ ein
$\alg{B}$-Rechtsmodul mit einem $\alg{B}$-wertigen inneren Produkt und
$\bimodrplus{B}{E}{A}$ ein Bimodul mit einem $\alg{A}$-wertigen
inneren Produkt, das kompatibel mit der $\alg{B}$-Linkswirkung ist. Wir
betrachten nun das Tensorprodukt $\rmod{H}{B} \tensor[\alg{B}]
\bimod{B}{E}{A}$ "uber der Algebra $\alg{B}$.

\begin{lemma}[Inneres Produkt auf {$\rmod{H}{B} \tensor[\alg{B}]
    \bimod{B}{E}{A}$}]
\index{inneres Produkt!Tensorprodukt@auf Tensorprodukt}
    Auf dem Tensorprodukt $\rmod{H}{B} \tensor[\alg{B}]
    \bimod{B}{E}{A}$ existiert ein wohldefiniertes $\alg{A}$-wertiges
    inneres Produkt via
    \begin{align}
        \label{eq:InneresProduktAufTensorproduktalgebra}
        \rSP[\sss \alg{H}\otimes\alg{E}]{x_{1} \tensor[\alg{B}]
          y_{1}, x_{2} \tensor[\alg{B}] y_{2}}{\alg{A}} :=
        \rSP[\sss \alg{E}]{y_{1}, \rSP[\sss \alg{H}]{x_{1}, x_{2}}{\alg{B}}
          \cdot y_{2}}{\alg{A}}
    \end{align}
f"ur alle $x_{1},x_{2} \in \rmod{H}{B}$ und $y_{1},y_{2} \in
\bimod{B}{E}{A}$. 
\end{lemma}

\begin{proof}
    Da wir bisher keine Forderungen bez"uglich Positivit"at oder
    Ausgeartetheit an die inneren Produkte gestellt
    haben, m"ussen lediglich $\ring{C}$-Sesquilinearit"at, die Kompatibilit"at
    mit der $\alg{A}$-Rechts\-mo\-dul\-struk\-tur und das Verhalten unter der
    $^\ast$-Involution gepr"uft werden. Das innere Produkt ist jedoch
    so konstruiert, da"s diese Forderungen erf"ullt
    sind, vergleiche
    \citep{rieffel:1972a,rieffel:1974b,rieffel:1974a,bursztyn.waldmann:2003a:pre}.    
\end{proof}

\begin{definition}[Ausartungsraum von $\rmod{H}{B} \otimes_{\sss
      \alg{B}} \bimod{B}{E}{A}$] 
Wir bezeichnen den Ausartungsraum von $\rmod{H}{B} \tensor[\alg{B}]
    \bimod{B}{E}{A}$ bez"uglich des inneren Produktes $\rSP[\sss
    \alg{H} \otimes \alg{E}]{\cdot, \cdot}{\alg{A}}$ mit $(\rmod{H}{B}
    \tensor[\alg{B}] \bimod{B}{E}{A})^{\bot}$. 
\end{definition}

Auf dem Quotienten $(\rmod{H}{B} \tensor[\alg{B}]
\bimod{B}{E}{A}) / (\rmod{H}{B} \tensor[\alg{B}]
\bimod{B}{E}{A})^{\bot}$ existiert nun ein induziertes,
nichtausgeartetes inneres Produkt\index{inneres Produkt!induziertes},
das wir ebenfalls mit $\rSP[\sss 
\alg{H} \otimes \alg{E}]{\cdot, \cdot}{\alg{A}}$ bezeichnen. Damit
sind wir nun in der Lage das {\em innere Tensorprodukt $\tensorhatohne$} zu
definieren.

\begin{definition}[Inneres Tensorprodukt $\tensorhatohne$]
\index{inneres Tensorprodukt}
   Seien $\rmodplus{H}{B}$ ein $\alg{B}$-Rechtsmodul und $\bimodrplus{B}{E}{A}$
   ein mit der $\alg{B}$-Linkswirkung kompatibler Bimodul mit einem
   $\alg{A}$-wertigen inneren Produkt, so definiert man das {\em innere
     Tensorprodukt $\tensorhat[\alg{B}]$ "uber der Algebra $\alg{B}$} von $\rmodplus{H}{B}$ und
   $\bimodrplus{B}{E}{A}$ mittels  
   \begin{align}
       \label{eq:InneresTensorprodukt}
       \rmod{H}{B} \tensorhat[\alg{B}] \bimod{B}{E}{A} :=
       \left((\rmod{H}{B} \tensor[\alg{B}] \bimod{B}{E}{A}) /
           (\rmod{H}{B} \tensor[\alg{B}]
           \bimod{B}{E}{A})^{\bot}, \rSP[\sss \alg{H} \otimes
           \alg{E}]{\cdot,\cdot}{\alg{A}}  \right).
   \end{align}
\end{definition}

\begin{bemerkung}[Ausartungsr"aume bei innerem Tensorprodukt]
\label{Bemerkung:InneresProduktBimodul}
Sei nun $\bimod{B}{E}{A}$ ein Bimodul mit einem $\alg{A}$-wertigem und einem
$\alg{B}$-wertigen inneren Produkt, so ist erstmal nicht eindeutig, welchen
Ausartungsraum man herausteilen mu"s. Sind die beiden inneren
Produkte kompatibel (Gleichung
\eqref{eq:KompabilitaetInnereProdukte}), so sind beide
Ausartungsr"aume identisch. F"ur die $^\ast$- und starke
\Name{Morita}-"Aquivalenz sind genau diese Bimoduln von Interesse.  
Wir werden daher $\tensortilde$ statt $\tensorhatohne$ verwenden, wenn wir
ausdr"ucken wollen, da"s zwei kompatible innere Produkte involviert sind. Die
funktoriellen Eigenschaften von $\tensortilde$ und $\tensorhat$
unterscheiden sich nicht.     
\end{bemerkung}

\begin{lemma}[Kanonische Linkswirkung bei innerem Tensorprodukt]
 \label{Lemma:KanonischeLinkswirkungAufInneremTensorprodukt}
 Seien $\alg{A}$, $\alg{B}$ und $\alg{C}$ $^\ast$-Algebren "uber dem
 Ring $\ring{C}$, und $\bimodrplus{C}{H}{B}$
 bzw.~$\bimodrplus{B}{E}{A}$ sind mit der jeweiligen Linkswirkung 
 vertr"agliche Bimoduln mit innerem Produkt, dann hat $\bimod{C}{H}{B}
 \tensorhat[\alg{B}] \bimod{B}{E}{A}$ eine kanonische, mit dem
 inneren Produkt $\rSP[\sss \alg{H} \otimes \alg{E}]{\cdot,\cdot}{\alg{A}}$
 vertr"agliche $\alg{C}$-Linkswirkung.
\end{lemma}

\begin{proof}
Der Beweis ist eine einfache Rechnung. Sei nun $c\in \alg{C}$,
$x_{1},x_{2} \in \bimod{C}{H}{B}$ und $y_{1},y_{2} \in
\bimod{B}{E}{A}$, dann gilt
\begin{align*}
     \rSP[\sss \alg{H}\otimes\alg{E}]{c \cdot (x_{1} \tensor[\alg{B}]
          y_{1}), x_{2} \tensor[\alg{B}] y_{2}}{\alg{A}} & =
        \rSP[\sss \alg{H}\otimes\alg{E}]{(c \cdot x_{1}) \otimes_{\sss
            \alg{B}} y_{1}, x_{2} \tensor[\alg{B}]
          y_{2}}{\alg{A}} \\ & = \rSP[\sss \alg{E}]{y_{1}, \rSP[\sss
          \alg{H}]{c \cdot x_{1}, x_{2}}{\alg{B}} \cdot
          y_{2}}{\alg{A}} \\ & = \rSP[\sss \alg{E}]{y_{1}, \rSP[\sss
          \alg{H}]{x_{1}, c^{\ast} \cdot x_{2}}{\alg{B}} \cdot
          y_{2}}{\alg{A}} \\ & = \rSP[\sss
        \alg{H}\otimes\alg{E}]{x_{1} \tensor[\alg{B}] y_{1}, (c^{\ast} \cdot
          x_{2}) \tensor[\alg{B}] y_{2}}{\alg{A}} \\ & =
        \rSP[\sss \alg{H}\otimes\alg{E}]{x_{1} \tensor[\alg{B}]
          y_{1}, c^{\ast} \cdot (x_{2} \tensor[\alg{B}] y_{2})}{\alg{A}}. 
\end{align*}
\end{proof}
\begin{lemma}[Assoziativit"at von $\tensorhatohne$]
    \label{Lemma:AssoziativitaetInneresTensorprodukt}
Seien $\rmodplus{H}{B}$ ein $\alg{B}$-Rechtsmodul und $\bimodplus{B}{E}{A}$
bzw.~$\bimodplus{A}{F}{C}$ mit der jeweiligen Linkswirkung kompatible
Bimoduln, so existiert ein nat"urlicher isometrischer
Isomorphismus, so da"s $\tensorhatohne$ assoziativ wird:
\begin{align}
    \label{eq:AssoziativitaetInneresTensorprodukt}
    \left(\rmod{H}{B} \tensorhat[\alg{B}] \bimod{B}{E}{A} \right)
    \tensorhat[\alg{A}] \bimod{A}{F}{C} \cong \rmod{H}{B}
    \tensorhat[\alg{B}] \left ( \bimod{B}{E}{A} \tensortilde[\alg{A}]
        \bimod{A}{F}{C} \right).
\end{align}
\end{lemma}
\begin{proof}
    Die Assoziativit"at wird durch die Assoziativit"at des
    algebraischen Tensorprodukts induziert. Die Isometrie der inneren
    Produkte ist eine einfache Rechnung. 
\end{proof}

\begin{lemma}[Vertr"aglichkeit von $\tensorhatohne$ mit Morphismen]
\label{Lemma:VertraeglichkeittensorBmitMorphismen}
Seien die Bimoduln $\bimodrplus{B}{E}{A}$, $\bimodrplus{B}{E^{\prime}}{A}$ und
$\bimodrplus{A}{F}{C}$, $\bimodrplus{A}{F^{\prime}}{C}$ gegeben, wobei die
Linkswirkung jeweils kompatibel mit den inneren Produkten sei. Ferner
seien $S \in \alg{B}(\bimod{B}{E}{A}, \bimod{B}{E^{\prime}}{A})$ und $T \in
\alg{B}(\bimod{A}{F}{C}, \bimod{A}{F^{\prime}}{C})$. Das algebraische
Tensorprodukt $S \otimes_{\sss \alg{A}} T$ induziert einen
wohldefinierten Bimodulmorphismus $S \tensorhat[\alg{A}]
T:\bimod{B}{E}{A} \tensorhat[\alg{A}] \bimod{A}{F}{C} \to \bimod{B}{E^{\prime}}{A}
\tensorhat[\alg{A}] \bimod{A}{F^{\prime}}{C}$. Das Adjungieren erfolgt
komponentenweise und falls $S$ und $T$ isometrisch sind, dann ist
es auch $S \tensorhat[\alg{A}] T$.
\end{lemma}

Eine wichtige Frage f"ur das weitere Vorgehen ist nun, ob das
innere Tensorprodukt vollst"andig positive innere Produkte erh"alt. Wie von 
\citet{bursztyn.waldmann:2003a:pre} gezeigt, ist dies der Fall.

\begin{satz}[{\citep[Thm.~4.7]{bursztyn.waldmann:2003a:pre}}]
\label{Satz:KomplettePositivitaetDesInnerenProdukts}
 Sei $\rmodplus{H}{B}$ ein $\alg{B}$-Rechtsmodul und
 $\bimodrplus{B}{E}{A}$ ein $(\alg{B},\alg{A})$-Bimodul, und beide
 inneren Produkte seien vollst"andig positiv nach Definition
 \ref{Definition:PraeHilbertModul}, dann ist das innere Produkt
 $\rSP[{\sss \alg{H} \otimes \alg{E}}]{\cdot, \cdot}{\alg{A}}$ auf
 $\rmod{H}{B} \tensorhat[\alg{B}] \bimod{B}{E}{A}$ auch vollst"andig
 positiv.
\end{satz}

Mittels $\tensorhatohne$ haben wir somit f"ur feste $^\ast$-Algebren
$\alg{A}$, $\alg{B}$ und $\alg{C}$ 
einen Funktor zwischen den Kategorien der Moduln 

\begin{align}
\label{eq:FunktorFuersmod}   
    \tensorhat[\alg{B}]: \smod[\alg{B}](\alg{C}) \times
    \smod[\alg{A}](\alg{B}) \to \smod[\alg{A}](\alg{C}).   
\end{align}

gefunden. Im Fall von $^\ast$-Algebren mit Einselement sind die
Kategorien $\smod(\cdot)$ durch $\sMod(\cdot)$ zu ersetzen. Wegen Satz
\ref{Satz:KomplettePositivitaetDesInnerenProdukts} ist $\tensorhatohne$ auch
f"ur die Kategorien der $^\ast$-Darstellungen ein Funktor in dem
Sinne, da"s  

\begin{align}
    \tensorhat[\alg{B}]: \srep[\alg{B}](\alg{C}) \times
    \srep[\alg{A}](\alg{B}) \to \srep[\alg{A}](\alg{C}). 
\end{align}

Auch hier sind die Kategorien $\srep(\cdot)$ durch $\sRep(\cdot)$ zu ersetzen, falls
die $^\ast$-Algebren dies zulassen, sprich ein Einselement besitzen. 
Die beiden wichtigen Beispiele f"ur das innere Tensorprodukt sind die
{\em \Name{Rieffel}-Induktion} und der {\em Wechsel der
  Basisalgebra}, die wir uns im folgenden ansehen wollen. 

\begin{definition}[Nat"urliche Transformation, {\citep{kassel:1995a}}]
\index{Natuerliche Transformation@Nat\"urliche Transformation}
  Seien $F,G: \kat{A} \to \kat{B}$ zwei Funktoren. Eine {\em
    nat"urlichen Transformation $I$} ist eine Familie von Morphismen von
  $F$ nach $G$ (man schreibt $I: F \to G$), indiziert durch die
  Elemente der Kategorie $\kat{A}$, so da"s f"ur jeden Morphismus
  $\Morph(\kat{A}) \ni f: A \to A^{\prime}$ mit $A,A^{\prime} \in
  \Obj(\kat{A})$ das folgende Diagramm kommutiert:
\begin{equation}
\label{eq:KommutativesDiagramm:NatuerlicheTransformation}
\bfig
\square<600,600>[F(A)`G(A)`F(A^{\prime})`{G(A^{\prime}).};I(A)`F(f)`G(f)`I(A^{\prime})]
\efig
\end{equation}
Ist weiter $I(A)$ ein Isomorphismus von $\kat{B}$, so nennt man $I: F
\to G$ einen {\em nat"urlichen Isomorphismus}. 
\end{definition}

\begin{beispiele}[\Name{Rieffel}-Induktion, Wechsel der Basisalgebra]
\label{Beispiel:RieffelInduktion}
\index{Rieffel-Induktion@\Name{Rieffel}-Induktion|textbf}
\index{Wechsel der Basisalgebra}
 Seien $\alg{A}$, $\alg{B}$ und $\alg{D}$, $\alg{D}'$
 $^\ast$-Algebren "uber einem Ring
$\ring{C}$. Desweiteren sei $\bimod{B}{E}{A} \in
\srep[\alg{A}](\alg{B})$ und $\bimod{D}{G}{D'} \in \srep[\alg{D'}](\alg{D})$. 
\begin{compactenum}
\item Wir bezeichnen den Funktor 
\begin{align}
    \label{eq:RieffelInduktion}
    \rieffel{R}{\alg{E}} = \bimod{B}{E}{A}\tensorhat[\alg{A}] \cdot :
    \srep[\alg{D}](\alg{A}) \to \srep[\alg{D}](\alg{B}) 
\end{align}
als {\em \Name{Rieffel}-Induktion}. Dies bedeutet, auf einem Objekt der
Kategorie gilt
\begin{align*}
\rieffel{R}{\alg{E}}(\rmod{H}{D}) & = \bimod{B}{E}{A}
\tensorhat[\alg{A}] \rmod{H}{D}.  
\end{align*}
Angewendet auf einen Morphismus ergibt sich
\begin{align*}
\rieffel{R}{\alg{E}}(T) & = \id \tensorhat[\alg{A}] T, 
\end{align*}
f"ur $T\in \alg{B}(\alg{H},\alg{H}')$.    

\item Analog dazu kann man die Basisalgebra wechseln und den folgenden
    Funktor
\begin{align}
    \label{eq:WechselDerBasisAlgebra}
    \rieffel{S}{\alg{G}} = \cdot \tensorhat[\alg{D}] \bimod{D}{G}{D'}:
    \srep[\alg{D}](\alg{A}) \to \srep[\alg{D'}](\alg{A})  
\end{align}
angeben. Angewendet auf ein Objekt bedeutet dies
$\rieffel{S}{\alg{G}}(\bimod{A}{E}{D})= \bimod{A}{E}{D} \tensorhat[\alg{D}]
\bimod{D}{G}{D'}$ und auf einen Morphismus $\rieffel{S}{\alg{G}}(T) =
T \tensorhat[\alg{D}] \id$.
\end{compactenum}
\end{beispiele}

Aufgrund von Lemma~\ref{Lemma:AssoziativitaetInneresTensorprodukt}
kommutieren die beiden Funktoren $\rieffel{R}{\alg{E}}$ und
$\rieffel{S}{\alg{G}}$ bis auf eine nat"urliche
Transformation, so da"s $\rieffel{R}{\alg{E}}
\circ\rieffel{S}{\alg{G}} \cong \rieffel{S}{\alg{G}} \circ
\rieffel{R}{\alg{E}}$, was gleichbedeutend mit der Kommutativit"at
des folgenden Diagramms ist:

$$\bfig\square/->`->`->`->/<900,600>[{\srep[\alg{D}](\alg{A})}`{\srep[\alg{D'}](\alg{A})}`{\srep[\alg{D}](\alg{B})}`{\srep[\alg{D'}](\alg{B}).};{\rieffel{S}{\alg{G}}}`{\rieffel{R}{\alg{E}}}`{\rieffel{R}{\alg{E}}}`{\rieffel{S}{\alg{G}}}]\efig$$

\begin{bemerkung}[\Name{Rieffel}-Induktion Notation]
\label{Bemerkung:RieffelInduktionNotation}
Manchmal ist es zweckm"a"sig, die Notation der
\Name{Rieffel}-Induktion nicht zu formal zu schreiben. Gegeben eine
Darstellung $(\rmod{H}{D}, \pi)$, so bezeichnen wir auch kurz $\pi$
als Darstellung. Dementsprechend verstehen wir unter
$\rieffel{R}{\alg{E}} (\pi)$ die Darstellung, die wir durch Anwenden der
\Name{Rieffel}-Induktion
gem"a"s Beispiel \ref{Beispiel:RieffelInduktion} auf $(\rmod{H}{D}, \pi)$ erhalten. Ein
Anwenden von $\rieffel{R}{\alg{E}}$ auf $\pi$ als {\em Homomorphismus} ist
folglich nicht definiert. 
\end{bemerkung}

\section{\Name{Morita}-"Aquivalenz}
\label{sec:MoritaAequivalenz}
\index{Morita-Aequivalenz@\Name{Morita}-\"Aquivalenz}

In diesem Kapitel werden wir einen kurzen "Uberblick zu den 
Ideen der \Name{Morita}-Theorie geben. F"ur einen ausgiebigen "Uberblick verweisen wir auf
\citep{morita:1958a,bass:1968a,jacobson:1989a,lam:1999a} f"ur die
{\em ringtheoretischen
  \Name{Morita}-"Aquivalenz}\index{Morita-Aequivalenz@\Name{Morita}-\"Aquivalenz!ringtheoretische}auf \citep{ara:1999a,ara:1999b}, der  die  
 {\em {$^\ast$-}\Name{Morita}-"Aquivalenz}\index{Morita-Aequivalenz@\Name{Morita}-\"Aquivalenz!Stern@$^\ast$-} behandelt, sowie
 \citep{bursztyn.waldmann:2001a,bursztyn.waldmann:2001b,bursztyn.waldmann:2003a:pre} f"ur die {\em starke \Name{Morita}-"Aquivalenz}\index{Morita-Aequivalenz@\Name{Morita}-\"Aquivalenz!starke}. Desweiteren spielen die Arbeiten \citep{rieffel:1972a,rieffel:1974a,rieffel:1974b} eine wichtige Rolle, in denen die Grundlagen f"ur die $^\ast$- und starke \Name{Morita}-"Aquivalenz f"ur $C^{\ast}$-Algebren\index{Algebra!CStern-Algebra@$C^{\ast}$-Algebra} gelegt wurden.     
 
Motiviert ist die \Name{Morita}-"Aquivalenz durch die
Darstellungstheorie von Ringen bzw.~Algebren:
\Name{Morita}-"aquivalente Ringe bzw.~Algebren haben eine "aquivalente
Darstellungstheorien. Wie wir sehen werden, sind noch weitere Eigenschaften von Ringen
bzw.~Algebren unter \Name{Morita}-"Aquivalenz erhalten, siehe Satz
\ref{Satz:MoritaInvariantenRingtheoretisch} und Kapitel
\ref{sec:MoritaInvarianten}.   

Ausgangspunkt bilden die Kategorie der
Moduln\index{Kategorie!Moduln@der Moduln} f"ur Ringe bzw.~die
Kategorie der Algebren, deren "Aquivalenz als
Kategorien zu \Name{Morita}-"Aquivalenz der jeweiligen Ringe oder Algebren
f"uhrt. Der Unterschied bei den verschiedenen "Aquivalenzbegriffen
besteht darin, da"s der Ring bzw.~die Algebra (sowie der Funktor, der die
"Aquivalenz generieren wird) zus"atzliche
Strukturen tragen k"onnen. Dies werden wir in den n"achsten Kapiteln
genauer erl"autern.

\subsection{Ringtheoretische \Name{Morita}-"Aquivalenz} 
\label{sec:MoritaBimodulnRingtheoretisch}

Die grundlegende Idee der ringtheoretischen \Name{Morita}-Theorie ist
eine Klassifikation von Ringen mittels Ihrer Darstellungstheorie als
Endomorphismen von \Name{Abel}schen Gruppen. Dabei betrachten wir in
diesem Kapitel Ringe im Sinne von Definition \ref{Definition:Ring},
die nicht notwendigerweise geordnet sein m"ussen. 

\begin{definition}["Aquivalenz und Isomorphie von Kategorien]
\label{Definition:AequivalenzVonKategorien}
\index{Aequivalenz@\"Aquivalenz!Kategorien@von Kategorien}
\index{Isomorphie!Kategorien@von Kategorien}
    Man nennt zwei Kategorien $\kat{A}$ und $\kat{B}$ {\em
      "aquivalent}, falls es zwei kovariante Funktoren $F: \kat{A} \to
    \kat{B}$ und $G: \kat{B} \to \kat{A}$ gibt, so da"s das
    Hintereinanderausf"uhren bis auf einen nat"urliche
    Transformation die jeweiligen Identit"aten auf den Kategorien liefern
    \begin{align}
        \label{eq:AequivalenzVonKategorien}
        F \circ G \cong \Id_{\sss \kat{B}} \quad \text{und} \quad G \circ
        F \cong \Id_{\sss \kat{A}}.
    \end{align}
Ferner nennt man die beiden Kategorien {\em isomorph} wenn das
Hintereinanderausf"uhren der beiden Funktoren die Identit"at auf den
Kategorien ist
 \begin{align}
        \label{eq:IsomorphieVonKategorien}
        F \circ G = \Id_{\sss \kat{B}} \quad \text{und} \quad G \circ
        F = \Id_{\sss \kat{A}}.
    \end{align}
\end{definition}

Dabei bezeichnet $\Id$ den {\em Identit"atsfunktor}
\index{Identitaetsfunktor@Identit\"atsfunktor} auf der jeweiligen
Kategorie. 

\begin{definition}[\Name{Morita}-"Aquivalenz von Ringen]
    \label{Definition:MoritaAequivalenzRinge}
Man nennt zwei Ringe $\ring{R}$ und $\ring{S}$ {\em
  \Name{Morita}-"aquivalent}, falls die Kategorien $\katlMod{\ring{R}}$
und $\katlMod{\ring{S}}$ "a\-qui\-va\-lent im Sinne von Definition
\ref{Definition:AequivalenzVonKategorien} sind.  
\end{definition}

Die Frage nach der \Name{Morita}-"Aquivalenz zweier Ringe manifestiert
sich damit in der Suche nach zwei geeigneten Funktoren. Bei zwei
gegebenen Ringen ist es im allgemeinen sehr schwierig, diese
Funktoren zu finden. Wir wollen im weiteren ein Beispiel f"ur
\Name{Morita}-"aquivalente Ringe angeben. Sp"ater werden wir sehen,
da"s die \Name{Rieffel}-Induktion genau die Funktoren liefert um
\Name{Morita}-"Aquivalenz elegant zu formulieren.

\begin{beispiel}[{\citep[Thm.~17.20]{lam:1999a}}]
\label{Beispiel:MoritaRingUndMatrizenUeberRing}
Ein einfaches und wichtiges Beispiel ist ein Ring $\ring{R}$ und der
Ring $M_{n}(\ring{R})$, den $n\times n$-Matrizen "uber dem Ring
$\ring{R}$. Sei $V$  
ein $\ring{R}$-Linksmodul, dann definieren wir den
$M_{n}(\ring{R})$-Linksmodul durch $_{M_{n}(\ring{R})} F(V)=V^{n}$, wobei die Linkswirkung
mittels Matrixmultiplikation auf Vektoren gegeben ist. F"ur den Ring
$\ring{R}$ ist $V^{n}$, durch komponentenweise Multiplikation, ein
Rechtsmodul. Die andere Richtung verl"auft analog, und somit ist
"uber den Funktor $F: \katlMod{\ring{R}} \to
\katlMod{M_{n}(\ring{R})}$ eine "Aquivalenz von Kategorien gegeben
und $\ring{R}$ und $M_{n}(\ring{R})$ sind \Name{Morita}-"aquivalent.
\end{beispiel}

\begin{lemma}[\Name{Morita}-"Aquivalenz von kommutativen Ringen]
    Zwei kommutative Ringe sind genau dann
    \Name{Morita}-"aquivalent, wenn sie isomorph sind.
\end{lemma}
\begin{proof}
    Der Beweis findet sich in \citep[Cor.~18.42]{lam:1999a}.
\end{proof}

Nun stellt sich eine wichtige Frage: Welche tiefere Bedeutung hat
\Name{Morita}-"Aquivalenz? Ein erster Schritt ist nun nach
{\em
  \Name{Morita}-Invarianten} zu suchen, d.~h.~welche Eigenschaften 
der Ringe sind unter \Name{Morita}-"Aquivalenz erhalten und welche nicht. 
Das erste von uns behandelte Beispiel
\ref{Beispiel:MoritaRingUndMatrizenUeberRing} zeigt bereits, da"s 
weder {\em Dimension}\footnote{In dem angegebenen
  Beispiel~\ref{Beispiel:MoritaRingUndMatrizenUeberRing} kann man von
  einer Dimension im Sinn einer Vektorraumdimension sprechen.}  noch {\em Kommutativit"at} 
\Name{Morita}-Invarianten sein k"onnen.  

\begin{satz}[\Name{Morita}-Invarianten]
\label{Satz:MoritaInvariantenRingtheoretisch}
\index{Morita-Invariante@\Name{Morita}-Invariante} 
Die beiden Ringe $\ring{R}$ und $\ring{S}$ seien im Sinne von
 Definition \ref{Definition:MoritaAequivalenzRinge}
 \Name{Morita}-"aquivalent. Der Ring $\ring{R}$ ist nur genau dann 
 einfach, halbeinfach, \Name{Noether}sch, \Name{Artin}sch oder
 primitiv, wenn der Ring $\ring{S}$ einfach, halbeinfach, \Name{Noether}sch, \Name{Artin}sch oder
 primitiv ist. Desweiteren haben die beiden Ringe $\ring{R}$ und
 $\ring{S}$ $^\ast$-isomorphe Zentren $\zentrum{\ring{R}}$ und
 $\zentrum{\ring{S}}$ sowie isomorphe \Name{Jacobson}-Radikale 
 $J(\ring{R})$ und $J(\ring{S})$.
\end{satz}

\begin{proof}
    F"ur die Beweise verweisen wir auf \citep{lam:1999a,jacobson:1989a}.
\end{proof}

In Kapitel \ref{sec:MoritaInvarianten} werden wir uns im Rahmen der
(unter einer \Name{Hopf}-$^\ast$-Algebra $H$ "aquivarianten) $^\ast$- und
starken \Name{Morita}-"Aquivalenz mit 
\Name{Morita}-Invarianten auseinandersetzen.
 Es stellt sich heraus, da"s \Name{Morita}-"aquivalente
$^\ast$-Al\-ge\-bren {\em isomorphe $K$-Theorie\index{K-Theorie@$K$-Theorie}, isomorphe
  Verb"ande\index{Verband!Ideale@der Ideale} von Idealen} und  {\em isomorphe Zentren}
haben. Wie bereits erw"ahnt haben \Name{Morita}-"aquivalente
$^\ast$-Al\-ge\-bren auch {\em isomorphe Darstellungstheorien}, was
insbesondere in der Physik eine wichtige Anwendung findet.

\subsection{\Name{Morita}-"Aquivalenz und \Name{Rieffel}-Induktion}
\label{sec:RieffelInduktion}

Wir wollen nun das Konzept der \Name{Morita}-"Aquivalenz ein wenig
brauchbarer formulieren. Dazu wollen wir ein \glqq Rezept\grqq{} angeben,
Funktoren zu konstruieren, die notwendig sind um "uber
\Name{Morita}-"Aquivalenz zu entscheiden. Dies geht   
auf einen Satz von \citet{eilenberg:1960a} und
\citet{watts:1960a} zur"uck und wurde von
\citet{rieffel:1972a,rieffel:1974a,rieffel:1974b} sp"ater auf
$C^{\ast}$- und $W^{\ast}$-Algebren angewendet. In
\citep{landsman:1998a} findet man eine moderne und ausf"uhrliche Darstellung dessen. 
Die Idee besteht darin, den in Beispiel
\ref{Beispiel:MoritaRingUndMatrizenUeberRing} angegebenen Funktor
$F:\katlMod{\ring{R}} \to \katlMod{\ring{S}}$ auf eine kanonische
Weise angeben zu k"onnen, wenn er denn existiert. Dazu ben"otigen
wir Bimoduln, die wir bereits in Definition \ref{Definition:Bimodul} eingef"uhrt
haben. 


Mit Hilfe der Bimoduln sind wir in der Lage den Satz von
\citet{eilenberg:1960a} und \citet{watts:1960a} f"ur Ringe zu formulieren. 

\begin{satz}[Satz von \Name{Eilenberg}-\Name{Watts},
    \citep{eilenberg:1960a, cartan.eilenberg:1999a}]
    \label{Satz:EilenbergWattsTheorem}
\index{Satz!Eilenberg-Watts@von \Name{Eilenberg}-\Name{Watts}}
Seien $\ring{R}$ und $\ring{S}$ Ringe und $\katlMod{\ring{R}}$
bzw.~$\katlMod{\ring{S}}$ die Kategorien der 
$\ring{R}$-Links\-mo\-duln bzw.~der $\ring{S}$-Links\-mo\-duln. 
$F:\katlMod{\ring{R}} \to \katlMod{\ring{S}}$ sei ein kovarianter,
additiver\footnote{{\em Additiv} bedeutet in diesem Zusammenhang,
  da"s der Funktor die durch die Ringstruktur gegebene additive
  Struktur der Morphismen respektiert.} Funktor. 
$F(\ring{R})$ wird zu einem $\ring{S}$-Links\-mo\-dul und einem
$\ring{R}$-Rechts\-mo\-dul, so da"s
$\bimodo{\ring{S}}{F(\ring{R})}{\ring{R}}$ ein
$(\ring{S},\ring{R})$-Bimodul wird, und
\begin{align}
   \bimodo{\ring{S}}{F(\ring{R})}{\ring{R}} \otimes_{\sss \ring{R}}
   \cdot :\katlMod{\ring{R}} \to
    \katlMod{\ring{S}}  
\end{align}
ist ein additiver, kovarianter Funktor.
\end{satz}

\begin{lemma}[\citep{eilenberg:1960a}]
\label{Lemma:Eilenberg}
Sei $F: \katlMod{\ring{R}} \to \katlMod{\ring{S}}$ ein additiver,
kovarianter Funktor und $\lmodo{\ring{R}}{\alg{E}}$ ein
$\ring{R}$-Linksmodul.
\begin{compactenum}
\item Die Gruppe
$\Hom[\ring{R}](\lmodo{\ring{R}}{\alg{E}}, F(\ring{R}))$ ist f"ur
jeden $\ring{R}$-Linksmodul $\lmodo{\ring{R}}{\alg{E}}$ ein wohldefinierter
$\ring{S}$-Linksmodul.
\item $\Hom[\ring{R}](\cdot , F(\ring{R})):\katlMod{\ring{R}} \to
    \katlMod{\ring{S}}$ ist ein additiver, kontravarianter Funktor. 
\end{compactenum}     
\end{lemma}

\begin{proof}
Der Beweis findet sich in \citep{eilenberg:1960a,cartan.eilenberg:1999a}.
\end{proof}

Satz~\ref{Satz:EilenbergWattsTheorem} ist auf ringtheoretischen Niveau
die
\Name{Rieffel}-Induktion\index{Rieffel-Induktion@\Name{Rieffel}-Induktion},
die wir bereits in Beispiel~\ref{Beispiel:RieffelInduktion} kennengelernt
haben, allerdings mit der einer zus"atzlichen additiven Struktur,
die durch die Ringe gegeben ist.   

Wir wollen nun angeben, was wir unter einem {\em dualen
  Bimodul} zu einem gegebenen $(\ring{R},\ring{S})$-Bimodul
$\bimodo{\ring{R}}{\alg{E}}{\ring{S}}$ verstehen. Dieser
wird eine wichtige Rolle spielen, wenn wir mit Hilfe der Bimoduln
\Name{Morita}-"Aquivalenz charakterisieren wollen, bzw.~wenn wir 
sp"ater das \Name{Picard}-Gruppoid von Algebren betrachten. Wir k"onnen auf
nat"urliche Weise mittels Lemma~\ref{Lemma:Eilenberg} den dualen
$(\ring{S},\ring{R})$-Bimodul
$\bimodo{\ring{S}}{\alg{E}^{\ast}}{\ring{R}}$ konstruieren.

\begin{korollar}[Dualer Bimodul -- Ringe]
\label{Korollar:DualerBimodulRinge}
\index{Bimodul!dualer}
Seien $\ring{R}$ und $\ring{S}$ zwei Ringe und
$\bimodo{\ring{R}}{\alg{E}}{\ring{S}}$ sei ein $(\ring{R},\ring{S})$-Bimodul. Dann existiert
ein dazu {\em dualer $(\ring{S},\ring{R})$-Bimodul} $\bimodo{\ring{S}}{\alg{E}^{\ast}}{\ring{R}}$
mittels
\begin{align}
    \bimodo{\ring{S}}{\alg{E}^{\ast}}{\ring{R}}:=
    \Hom[\ring{R}](\bimodo{\ring{R}}{\alg{E}}{\ring{S}},\ring{R}).  
\end{align}
Der duale Bimodul besteht somit aus den $\ring{R}$-linearen Homomorphismen
vom urspr"unglichen Bimodul in den Ring $\ring{R}$. 
 \end{korollar}
\begin{proof}
Der Bimodul $\bimodo{\ring{S}}{\alg{E}^{\ast}}{\ring{R}}$ ist auf
nat"urliche Weise ein $\ring{S}$-Linksmodul
(vgl.~Lemma~\ref{Lemma:Eilenberg}) und ein
$\ring{R}$-Rechtsmodul. Die Modulstrukturen sind durch folgende
Konstruktion gegeben: sei $\phi \in
\bimodo{\ring{S}}{\alg{E}^{\ast}}{\ring{R}}$, $s \in \ring{S}$, $r \in \ring{R}$ und
$x \in \bimodo{\ring{R}}{\alg{E}}{\ring{S}}$, so ist $(\phi \cdot
r)(x) = \phi(r \cdot x)$ und $(s \cdot \phi)(x) = \phi(x \cdot s)$. Die
weiteren Modulstrukturen rechnet man leicht nach. 
\end{proof}

Nun k"onnen wir Definition \ref{Definition:MoritaAequivalenzRinge} so
deuten, da"s zwei Ringe $\ring{R}$ und $\ring{S}$ genau dann
\Name{Morita}-"aquivalent sind, wenn es gelingt einen \glqq geeigneten\grqq{}
$(\ring{R},\ring{S})$-Bimodul zu finden. In Korollar
\ref{Korollar:Aequivalenzbimodul} werden wir angeben, wann ein Bimodul
geeignet ist. 
 
\begin{korollar}[Bimodul und \Name{Morita}-"Aquivalenz]
\label{Korollar:Aequivalenzbimodul}
\index{Morita-Aequivalenz@\Name{Morita}-\"Aquivalenz|textbf}
Zwei Ringe $\ring{R}$ und $\ring{S}$ sind dann und nur dann
\Name{Morita}-"aquivalent, wenn es zwei (zueinander duale) Bimoduln 
$\bimodo{\ring{R}}{\alg{E}}{\ring{S}}$ und
$\bimodo{\ring{S}}{\alg{E}^{\ast}}{\ring{R}}$ gibt, so da"s
\begin{align}
   \bimodo{\ring{R}}{\alg{E}}{\ring{S}} \otimes_{\sss \ring{S}}
   \bimodo{\ring{S}}{\alg{E}^{\ast}}{\ring{R}} \cong
   \bimodo{\ring{R}}{\ring{R}}{\ring{R}} \quad \text{und} \quad 
   \bimodo{\ring{S}}{\alg{E}^{\ast}}{\ring{R}} \otimes_{\sss \ring{R}}
   \bimodo{\ring{R}}{\alg{E}}{\ring{S}} \cong
   \bimodo{\ring{S}}{\ring{S}}{\ring{S}}. 
\end{align}
\end{korollar}

\begin{definition}["Aqui\-va\-lenz\-bi\-mo\-dul]
    \label{Definition:Aequivalenzbimodul}
\index{Aequivalenzbimodul@\"Aquivalenzbimodul}
Man nennt einen Bimodul f"ur zwei Ringe $\ring{R}$ und $\ring{S}$
einen {\em "Aqui\-va\-lenz\-bi\-mo\-dul} oder einen {\em
  \Name{Morita}-"Aqui\-va\-lenz\-bi\-mo\-dul}, falls dieser Korollar
\ref{Korollar:Aequivalenzbimodul} gerecht wird.  
\end{definition}

Wir wollen nun den {\em Satz von \Name{Morita}} formulieren, dazu
ben"otigen wir zuvor folgende Definitionen.

\begin{definition}[Generator und Progenerator]
 \label{Definition:GeneratorProgenerator}
 Man nennt einen  $\ring{R}$-Rechtsmodul $\rmodo{\alg{E}}{\ring{R}}$ einen
\begin{compactenum}
 \item {\em \Index{Generator}}, falls jeder andere $\ring{R}$-Rechtmodul als
     einen Quotienten einer direkten Summe von Kopien des
     $\ring{R}$-Moduls erh"alt.  
 \item {\em \Index{Progenerator}}, falls er endlich erzeugt ist, projektiv und ein
 Generator ist.  
\end{compactenum}
\end{definition}

\begin{satz}[Satz von \Name{Morita}]
\label{Satz:MoritasTheorem}    
\index{Satz!Morita@von \Name{Morita}}
Zwei Ringe mit Einselementen $\ring{R}$ und $\ring{S}$ sind genau dann
\Name{Morita}-"aquivalent, wenn es einen Progenerator
$\ring{R}$-Rechtsmodul $\rmodo{\alg{E}}{\ring{R}}$ gibt, so da"s
$\ring{S} \cong \End[\ring{R}] {\rmodo{\alg{E}}{\ring{R}}}$. Ist
desweiteren $\bimodo{\ring{S}}{\alg{E}}{\ring{R}}$ ein
"Aquivalenzbimodul, dann ist der duale Bimodul durch  
$\bimodo{\ring{R}}{\alg{E}^{\ast}}{\ring{S}}:=
    \Hom[\ring{R}](\bimodo{\ring{S}}{\alg{E}}{\ring{R}},\ring{R})$
    gegeben. 
\end{satz}

\begin{definition}[Volle idempotente Elemente]
\label{Defintion:VolleIdempotenteElementeP}
Ein idempotentes Element $P \in M_{n}(\ring{R})$ nennt man {\em voll},
wenn die lineare H"ulle der Elemente der Form $TPS$, mit $T,S \in
M_{n}(\ring{R})$, ganz $M_{n}(\ring{R})$ ist.      
\end{definition}

\begin{bemerkungen}
~\vspace{-5mm}
\begin{compactenum}
\item F"ur volle Elemente $P$ schreiben wir die Definition
    \ref{Defintion:VolleIdempotenteElementeP} auch als: $M_{n}(\ring{R})PM_{n}(\ring{R})=
    M_{n}(\ring{R})$.
\item Sei $P$ idempotent. Ein endlich erzeugter projektiver $\ring{R}$-Modul
    $P\ring{R}^{n}$ ist genau dann ein Generator,
    wenn $P$ voll ist \citep[18.10(D)]{lam:1999a}.
\end{compactenum}
\end{bemerkungen}

Nun sind wir in der Lage eine alternative Formulierung von
\Name{Morita}-"Aquivalenz anzugeben.

\begin{satz}[Alternative Formulierung der \Name{Morita}-"Aquivalenz]
\label{Satz:MoritaAequivalenzUndProjektiveModul}
Zwei Ringe $\ring{R}$ und $\ring{S}$ sind genau dann
\Name{Morita}-"aquivalent, wenn es ein $n \in \field{N}$, sowie ein
volles, idempotentes $P \in M_{n}(\ring{R})$ gibt, so da"s
$S\cong P M_{n}(\ring{R}) P$. 
\end{satz}

Bevor wir uns nun endg"ultig
der $^\ast$- bzw.~starken \Name{Morita}-"Aquivalenz zuwenden, 
wollen wir das Beispiel~\ref{Beispiel:MoritaRingUndMatrizenUeberRing}
nochmal aufgreifen und in der oben erarbeiteten Sprache formulieren.

\begin{beispiel}["Aqui\-va\-lenz\-bi\-mo\-dul zu den Ringen $M_{n}(\ring{R})$
    und $\ring{R}$]   
\label{Beispiel:MoritaRingUndMatrizenUeberRing2}
Sei $\ring{R}$ ein Ring und $M_{n}(\ring{R})$ der Ring der $n\times
n$-Matrizen "uber $\ring{R}$, dann entspricht der in Beispiel
\ref{Beispiel:MoritaRingUndMatrizenUeberRing} angegebene Funktor dem
Bimodul $\bimodo{M_{n}(\ring{R})}{\ring{R}^{n}}{\ring{R}}$. Dabei
kann man Elemente in $\ring{R}^{n}$ als Spaltenvektoren
auffassen, die Linksmodulstruktur ist durch die gew"ohnliche Multiplikation
der Matrizen mit Vektoren gegeben, die Rechtsmodulstruktur ist die
komponentenweise Multiplikation mit Elementen aus $\ring{R}$. 
\end{beispiel}

\subsection{$^\ast$- und starke \Name{Morita}-"Aquivalenz}
\label{sec:SternStarkeMoritaAequivalenz}

Ein n"achster Schritt ist \Name{Morita}-"Aquivalenz f"ur
$^\ast$-Algebren zu formulieren. Diese zeichnen sich dadurch aus,
da"s sie aufgrund der $^\ast$-Involution mit mehr Struktur versehen
sind. Dadurch sind wir beispielsweise in der Lage "uber Positivit"at zu
sprechen. Im weiteren seinen $\alg{A}$ und $\alg{B}$ zwei $^\ast$-Algebren.

Der Bimodul $\bimod{B}{E}{A}$ ist mit zwei inneren Produkten $\lSP{B}{\cdot,\cdot}$ und
$\rSP{\cdot,\cdot}{A}$ versehen, von denen fordern wir noch nicht,
da"s $\lmodplus{B}{E}$ und $\rmodplus{E}{A}$ f"ur sich
Pr"a-\Name{Hilbert}-Moduln nach Definition
\ref{Definition:PraeHilbertModul} sind, allerdings wird dies bei der
sp"ater behandelten starken \Name{Morita}-Theorie der Fall sein. 

F"ur die \Name{Morita}-Theorie ist der duale Bimodul
wichtig. Aufgrund der $^\ast$-Struktur k"onnen wir den dualen Bimodul
zu $\bimodplus{B}{E}{A}$ auf eine kanonische Weise definieren.

\begin{definition}[Komplex-konjugierter Bimodul f"ur $^\ast$-Algebren]
\label{Definition:DualerBimodulSternAlgebren}
\index{Bimodul!komplex-konjugierter|textbf}
Seien $\alg{A}$ und $\alg{B}$ zwei $^\ast$-Algebren "uber einem Ring
$\ring{C}$, und $\bimodplus{B}{E}{A}$ sei ein Bimodul. Man definiert
den dazu {\em komplex-konjugierten Bimodul} $\bimodplus{A}{\cc{E}}{B}$ "uber
die Modulstrukturen  
\begin{gather}
    \label{eq:DualeBimodulstrukturenSternAlgebren}
     a \cdot \cc{x} := \cc{x  \cdot a^\ast}, \quad \cc{x} \cdot  b :=
     \cc{b^{\ast} \cdot x}, \\ 
     \lSP[\sss \alg{\cc{E}}]{A}{\cc{x},\cc{y}} := \rSP[\sss \alg{E}]{x,y}{A},
    \quad \rSP[\sss \alg{\cc{E}}]{\cc{x},\cc{y}}{B} := \lSP[\sss
    \alg{E}]{B}{x,y}. 
\end{gather}
Dabei bezeichnet man Elemente im komplex-konjugierten Bimodul durch die
Konjunktion $\cc{^{~}}$. Als Menge sind der komplex-konjugierte Bimodul und der
urspr"ungliche Bimodul identisch. 
\end{definition}

\begin{lemma}[Strukturen auf dualem Bimodul]
\label{Lemma:DualerBimodulSternAlgebren}
Der duale Bimodul in Definition
\ref{Definition:DualerBimodulSternAlgebren} ist wohldefiniert, und
er erbt die Strukturen des ur\-spr"ung\-li\-chen Bimoduls. 

\end{lemma} 
\begin{proof}
Wir m"ussen zeigen, da"s die Bimodulstrukturen auf dem dualen
Bimodul sinnvoll definiert sind.

\begin{compactenum}
\item Zuerst zeigen wir die (Links-) Modulstruktur.
\begin{align*}
    (a'a)\cc{x} = \cc{x (a'a)^{\ast}} = \cc{x (a^{\ast}
      {a'}^{\ast})} = \cc{(x a^{\ast}) {a'}^{\ast}} = a'
    \cc{(x a^{\ast})} = a' (a \cc{x})   
\end{align*}
Analog zeigt man die Konsistenz der Rechtsmodulstruktur, bzw.~der
Bimodulstruktur. 

\item Desweiteren gilt es zu zeigen, da"s die Strukturen der inneren Produkte
sich auf den dualen Bimodul "ubertragen. Aus $\rSP[\sss
\alg{E}]{x,y\cdot a}{A}= \rSP[\sss \alg{E}]{x,y}{A} a$ sowie
$\left(\rSP[\sss \alg{E}]{x,y}{\alg{A}} \right)^{\ast} = \rSP[\sss
\alg{E}]{y,x}{\alg{A}}$ und der
Definition \ref{Definition:DualerBimodulSternAlgebren} folgt
\begin{align*}
\lSP[\sss \cc{\alg{E}}]{\alg{A}}{a \cdot \cc{x},\cc{y}} = \lSP[\sss
\cc{\alg{E}}]{\alg{A}}{\cc{x\cdot a^{\ast}},\cc{y}} = \rSP[\sss
\alg{E}]{x\cdot a^{\ast},y}{\alg{A}}=(a^{\ast})^{\ast} \rSP[\sss
\alg{E}]{x,y}{\alg{A}}= a \lSP[\sss
\cc{\alg{E}}]{\alg{A}}{\cc{x},\cc{y}}.       
\end{align*}

\item Ist im Bimodul $\bimodplus{B}{E}{A}$ das $\alg{A}$-wertige innere
Produkt kompatibel mit der $\alg{B}$-Links\-ak\-tion (siehe Definition
\ref{Definition:KompatibilitaetBimodulMitLinkswirkung}), so ist dies
auch im dualen Bimodul $\bimodplus{A}{\cc{E}}{B}$ der Fall
\begin{align*}
\lSP[\sss \cc{\alg{E}}]{\alg{A}}{\cc{x}\cdot b,\cc{y}} = \lSP[\sss
\cc{\alg{E}}]{\alg{A}}{\cc{b^{\ast}\cdot x},\cc{y}} = \rSP[\sss
\alg{E}]{b^{\ast}\cdot x,y}{\alg{A}} = \rSP[\sss
\alg{E}]{x,b\cdot y}{\alg{A}} =\lSP[\sss
\cc{\alg{E}}]{\alg{A}}{\cc{x},\cc{b\cdot y}} = \lSP[\sss
\cc{\alg{E}}]{\alg{A}}{\cc{x},\cc{y}\cdot b^{\ast}}.     
\end{align*}
Analoges gilt f"ur das $\alg{B}$-wertige innere Produkt $\rSP[\sss
\cc{\alg{E}}]{\cdot, \cdot}{\alg{B}}$.  

\item Zu guter Letzt ist zu zeigen, da"s die Kompatibilit"at der beiden
inneren Produkte $\lSP[\sss \alg{E}]{\alg{B}}{\cdot, \cdot}$ und
$\rSP[\sss \alg{E}]{\cdot, \cdot}{\alg{A}}$ auch eine Kompatibilit"at
der beiden inneren Produkte auf dem dualen Bimodul mit sich bringt
\begin{align*}
    \cc{x} \cdot \rSP[\sss \cc{\alg{E}}]{\cc{y},\cc{z}}{\alg{B}} & =
    \cc{x} \cdot \lSP[\sss \alg{E}]{\alg{B}}{y,z} = \cc{x} \cdot
    \left(\lSP[\sss \alg{E}]{\alg{B}}{z, y}\right)^{\ast} = \cc{\lSP[\sss
      \alg{E}]{\alg{B}}{z, y} \cdot x} \\ & = \cc{z \cdot \rSP[\sss
      \alg{E}]{y,x}{\alg{A}}} = \left(\rSP[\sss
    \alg{E}]{y,x}{\alg{A}}\right)^{\ast}\cdot \cc{z} = \rSP[\sss
    \alg{E}]{x,y}{\alg{A}} \cdot \cc{z} = \lSP[\sss
    \cc{\alg{E}}]{\alg{A}}{\cc{x},\cc{y}} \cdot \cc{z}.       
\end{align*}
\end{compactenum}
Die Beweise zu {\it ii.)} und {\it iii.)} f"ur das $\alg{B}$-wertige innere Produkt $\rSP[\sss
\cc{\alg{E}}]{\cdot, \cdot}{\alg{B}}$ geschehen komplett analog zu den
aufgezeigten Rechnungen.

Es ist aufgrund der Definition der inneren Produkte auf dem dualen
Bimodul offensichtlich, da"s bei (nicht-)ausgearteten bzw.~vollst"andig
positiven Bimoduln $\bimodplus{B}{E}{A}$ auch der duale Bimodul
(nicht-)ausgeartet bzw.~vollst"andig positiv ist.   
\end{proof}

Nun k"onnen wir angeben, was $^\ast$-\Name{Morita}-"Aquivalenz
bzw.~{\em starke \Name{Morita}-"Aquivalenz} ist
\citep{rieffel:1972a,rieffel:1974a,rieffel:1974b,ara:1999a,ara:1999b,bursztyn.waldmann:2001a,bursztyn.waldmann:2001b}.

\begin{definition}[$^\ast$- und starker "Aqui\-va\-lenz\-bi\-mo\-dul]
\label{Definition:SternStarkerAequivalenzBimodul}

Seien $\alg{A}$und $\alg{B}$ $^\ast$-Algebren "uber einem Ring
$\ring{C}$ und $\bimod{B}{E}{A}$ ein $(\alg{B},\alg{A})$-Bimodul mit einem
$\alg{A}$-wertigen inneren Produkt $\rSP[\alg{E}]{\cdot,
  \cdot}{\alg{A}}$ und einem $\alg{B}$-wertigen inneren Produkt
$\lSP[\alg{E}]{\alg{B}}{\cdot,\cdot}$. Wir nennen
$\bimodplus{B}{E}{A}$ einen {\em
  $^\ast$-\Name{Morita}-"Aquivalenzbimodul} oder kurz {\em
  $^\ast$-"Aqui\-va\-lenz\-bi\-mo\-dul}, falls
die folgenden Bedingungen er\-f"ullt sind:
\begin{compactenum}
\item Die beiden inneren Produkte $\rSP[\alg{E}]{\cdot,
      \cdot}{\alg{A}}$ bzw.
    $\lSP[\alg{E}]{\alg{B}}{\cdot,\cdot}$ sind nichtausgeartet, voll
    und mit der $\alg{B}$-Wirkung bzw. $\alg{A}$-Wirkung kompatibel.
\item F"ur alle $x,y,z \in \bimod{B}{E}{A}$ sind die beiden inneren
    Produkte miteinander {kompatibel}, d.~h.~es gilt:\\ 
 $x \cdot \rSP[\alg{E}]{y,z}{\alg{A}} = \lSP[\alg{E}]{\alg{B}}{x,y}
  \cdot z$.
\item $\alg{B}\cdot \bimod{B}{E}{A} = \bimod{B}{E}{A}$ und
    $\bimod{B}{E}{A} \cdot \alg{A} = \bimod{B}{E}{A}$. 
\end{compactenum}

Sind zus"atzlich beide inneren Produkte vollst"andig positiv, das
hei"st $\lmodplus{B}{E}$ und $\rmodplus{E}{A}$ sind Pr"a-\Name{Hilbert}-Moduln, so nennt
man $\bimod{B}{E}{A}$ einen {\em starken
  \Name{Morita}-"A\-qui\-va\-lenz\-bi\-mo\-dul} oder kurz {\em starken "A\-qui\-va\-lenz\-bi\-mo\-dul}.   
\end{definition}

\begin{bemerkungen}[Komplex-konjugierter Bimodul]
\label{Bemerkung:DualerBimodulRingeAlgebren}
~\vspace{-5mm}
\begin{compactenum}
\item Im Fall von "Aquivalenzbimoduln f"ur $^\ast$-Algebren, sind
    die Definitionen von dualem und kom\-plex\--konjugiertem Bimodul
    miteinander identifizierbar, und es gilt aufgrund der Struktur 
des komplex-konjugierten Bimoduls in Definition \ref{Definition:DualerBimodulSternAlgebren} 
$\bimod{B}{\cc{\cc{\alg{E}}}}{A} = \bimod{B}{E}{A}$.

\item Sei $T: \bimod{B}{E}{A} \ni x \mapsto T(x) \in \bimod{B}{F}{A}$ ein
    Verschr"ankungsoperator, so definiert $\cc{T} (\cc{x}):= 
    \cc{T(x)}$ einen 
    Verschr"ankungsoperator\index{Verschraenkungsoperator@Verschr\"ankungsoperator} auf den dualen Bimoduln und $\cc{T} : \bimod{A}{\cc{E}}{B} \to \bimod{A}{\cc{F}}{B}$.   
\end{compactenum}
\end{bemerkungen}

\begin{definition}[$^\ast$- und starke \Name{Morita}-"Aquivalenz]
\index{Morita-Aequivalenz@\Name{Morita}-\"Aquivalenz!starke|textbf}
\index{Morita-Aequivalenz@\Name{Morita}-\"Aquivalenz!Stern@$^\ast$-|textbf}
\index{Aequivalenzbimodul@\"Aquivalenzbimodul|textbf}
Man nennt zwei $^\ast$-Algebren $\alg{A}$ und $\alg{B}$
 {\em $^\ast$-\Name{Morita}-"aquivalent} (bzw.~{\em stark
 \Name{Morita}-"aquivalent}), falls es einen
 $^\ast$-"Aqui\-va\-lenz\-bi\-mo\-dul (bzw.~einen starken
 "Aqui\-va\-lenz\-bi\-mo\-dul) gibt.  
\end{definition}

\begin{definition}[Idempotenz und Nichtausgeartetheit von
    $^\ast$-Algebren]
\index{Idempotenz} \index{Nichtausgeartetheit}
 Man nennt eine $^\ast$-Algebra $\alg{A}$ 
\begin{compactenum}
\item {\em nichtausgeartet}, wenn f"ur alle $b \in \alg{A}$ aus $ab =
    0$ oder aus $ ba=0$ folgt, da"s $a=0$.
\item {\em idempotent}, wenn Elemente der Form $ab$ die Algebra
    $\alg{A}$ aufspannen. 
\end{compactenum}
\end{definition}

Offensichtlich sind Algebren mit einem Einselement $1_{\sss \alg{A}}
\in \alg{A}$ immer idempotent.

\begin{lemma}[Der "Aquivalenzbimodul $\bimod{A}{A}{A}$]
    Der Bimodul $\bimod{A}{A}{A}$ ist genau dann ein $^\ast$- oder
    starker "Aquivalenzbimodul, wenn die $^\ast$-Algebra idempotent
    und nichtausgeartet ist. 
\end{lemma}
\begin{proof}
    Der Beweis findet sich in \citep{bursztyn.waldmann:2003a:pre}.
\end{proof}

\begin{lemma}[$^\ast$- und starke \Name{Morita}-"Aquivalenz ist
    "Aquivalenzrelation] 
    \label{Lemma:MoritaAequivalenzAequivalenzrelation}
$^\ast$- und starke \Name{Morita}-"Aquivalenz ist eine "Aquivalenzrelation.
\end{lemma}

\begin{proof}
Wir m"ussen Reflexivit"at, Symmetrie und Transitivit"at zeigen. 

Die Reflexivit"at geschieht mittels des Bimoduls $\bimod{A}{A}{A}$
mit den beiden kanonischen inneren Produkten 
\begin{align}
    \label{eq:ReflexivitaetInnereProdukte}
\lSP{\alg{A}}{a,b} = ab^{\ast} \qquad \rSP{a,b}{\alg{A}} = a^{\ast}b.    
\end{align}
F"ur die Symmetrie ben"otigen wir den dualen Bimodul
$\bimod{A}{\cc{E}}{B}$, den wir in Definition
\ref{Definition:DualerBimodulSternAlgebren} eingef"uhrt haben. Die
Transitivit"at zeigen wir durch $\tensorhatohne$, bzw.~durch die 
\Name{Rieffel}-Induktion. 

Ausf"uhrlicher findet man den Beweis in
\citep{ara:1999a,bursztyn.waldmann:2003a:pre}. 
\end{proof}

\begin{korollar}
\label{Korollar:SternUndStarkerAequivalenzbimodul}
Der Bimodul $\bimod{B}{E}{A}$ ist genau dann ein $^\ast$- bzw.~starker
\Name{Morita}-"Aquivalenzbimodul, falls 
    \begin{align}
        \label{eq:StarkeMoritaAequivalenzVonSternAlgebren}
\bimod{B}{E}{A} \tensortilde[\alg{A}] \bimod{A}{\cc{E}}{B} \cong
\bimod{B}{B}{B} \quad \text{und} \quad \bimod{A}{\cc{E}}{B}
\tensortilde[\alg{B}] \bimod{B}{E}{A} \cong \bimod{A}{A}{A}.
     \end{align}
\end{korollar}

\begin{proof}
Das Korollar ist eine direkte Konsequenz aus Satz
\ref{Lemma:MoritaAequivalenzAequivalenzrelation}
\citep[Prop.~7.2]{bursztyn.waldmann:2001a}. Die kanonischen
Isomorphismen sind gegeben durch die Abbildungen  
\begin{align*}
  x \otimes \cc{y} \mapsto \lSP{\alg{B}}{x,y} \quad \text{und} \quad
  \cc{x} \otimes y \mapsto \rSP{x,y}{\alg{A}}.  
\end{align*}
\end{proof}

\section{Die $H$-"aquivariante \Name{Morita}-Theorie}
\label{sec:HaequivarianteMoritaTheorie}
\index{Morita-Aequivalenz@\Name{Morita}-\"Aquivalenz!H-aequivariante@$H$-\"aquivariante}

Ein wichtiger Teil dieser Arbeit ist nun die Formulierung einer
unter einer \Name{Hopf}-$^\ast$-Wirkung
\index{Hopf-Wirkung@\Name{Hopf}-Wirkung} "aquivarianten
\Name{Morita}-Theorie. Dazu ben"otigen wir die Theorie der
\Name{Hopf}-Algebren, die wir in aller Ausf"uhrlichkeit in Anhang~\ref{sec:HopfAlgebren}
zusammengetragen haben. 

Die Forderung einer $^\ast$-Involution ist unumg"anglich, da wir insbesondere an $^\ast$-
und starker \Name{Morita}-"Aquivalenz interessiert sind. Mit einer
$^\ast$-Wirkung bezeichnen wir eine Wirkung der
\Name{Hopf}-$^\ast$-Algebra $H$, die zus"atzlich
\begin{align*}
    (h \act a)^{\ast} = S(h)^{\ast} \act a^{\ast}
\end{align*}
f"ur $h \in H$ und $a \in \alg{A}$ erf"ullt. Die
genauen Definitionen zu Wirkungen von \Name{Hopf}-Algebren befinden
sich in Anhang~\ref{sec:WirkungenHopfAufAlgebren}. Das Kapitel
orientiert sich an \citep{jansen.waldmann:2005a}. Im weiteren sei $H$ eine
\Name{Hopf}-$^\ast$-Algebra\index{Hopf-Stern-Algebra@\Name{Hopf}-$^\ast$-Algebra},
soweit nichts anderes explizit gesagt wird.  

\subsection{$H$-"aquivariante Bimoduln und Darstellungstheorie}

In diesem Kapitel wollen wir die $^\ast$-Darstellungstheorie f"ur
$^\ast$-Algebren $H$-"aquivariant formulieren. Dies bedeutet, da"s wir
eine vorgegebene {\em Symmetrie}\index{Symmetrie} in Form einer \Name{Hopf}-$^\ast$-Algebra
Wirkung implementieren. 

\begin{definition}[$H$-"aquivarianter Bimodul]
\label{Definition:HopfInvarianterBimodul}
\index{Bimodul!H-aequivarianter@$H$-\"aquivarianter|textbf}
Seien $\alg{A}$ und $\alg{B}$ zwei $^\ast$-Algebren und
$\bimod{B}{E}{A}$ ein Bimodul mit zwei vertr"aglichen inneren
Produkten. Ferner sei $(H,\neact)$ eine \Name{Hopf}-$^\ast$-Algebra, so da"s
$(H,\alg{A},\neact)$ und $(H,\alg{B},\neact)$ $H$-Linksmodulalgebren
nach Definition \ref{Definition:AxiomeWirkungHopf} sind. Wir nennen
eine Wirkung $(H,\neact)$ auf dem Bimodul $\bimod{B}{E}{A}$ {\em kompatibel mit dem  
Bimodul}, wenn folgende Bedingungen f"ur alle $x,y \in
\bimod{B}{E}{A}$, $a \in \alg{A}$ und $b \in \alg{B}$ 
erf"ullt sind:   
\begin{compactenum}
\item Die \Name{Hopf}-Wirkung ist mit den Bimodulstrukturen vertr"aglich,
    d.~h.~es gilt
    \begin{align}
        \label{eq:HopfWirkungKompatibelBimodulstruktur}
        h \act (b \cdot x \cdot a) = (h_{\sss (1)} \act b) \cdot
        (h_{\sss (2)} \act x) \cdot (h_{\sss (3)} \act a).
    \end{align}

\item Die \Name{Hopf}-Wirkung ist mit den inneren Produkten\index{inneres Produkt} vertr"aglich
\begin{align}
    \label{eq:WirkungAufSkalarprodukte}
    h \act \rSP{x,y}{\alg{A}}   = \rSP{S(h_{\sss (1)})^{\ast} \act x, h_{\sss
        (2)} \act y}{\alg{A}} \quad \text{und} \quad  h \act
    \lSP{\alg{B}}{x,y} = \lSP{\alg{B}}{h_{\sss (1)} \act x, S(h_{\sss (2)})^{\ast} \act y}.
\end{align}
\end{compactenum}
\end{definition}

\begin{bemerkung}
Da es sich  bei $(H,\alg{A},\neact)$ und $(H,\alg{B},\neact)$ um $H$-Modulalgebren
handelt, sind die Forderungen $h\act (bb'\cdot x) = (h_{\sss (1)} \act
b)(h_{\sss (2)} \act b') \cdot (h_{\sss (3)} \act x)$ und analog
$h\act (x\cdot aa') = (h_{\sss (1)} \act x) \cdot (h_{\sss (2)} \act a)
(h_{\sss (3)} \act a')$ automatisch erf"ullt wenn {\it i.)} gegeben ist.
\end{bemerkung}

\begin{lemma}[]
Definition \ref{Definition:HopfInvarianterBimodul} ist konsistent
definiert und mit allen Bimodul-Strukturen vertr"aglich.
\end{lemma}
\begin{proof}
Dazu m"ussen wir zeigen, da"s die in Gleichung
\eqref{eq:WirkungAufSkalarprodukte} gemachten
Kom\-pa\-ti\-bi\-li\-t"ats\-be\-din\-gun\-gen mit den
Bimodulstrukturen vertr"aglich sind.  
\begin{compactenum} 
 \item Die $H$-Wirkung auf das $\alg{A}$-wertige innere Produkt auf
     $\bimodplus{B}{E}{A}$ ist mit der $\alg{A}$-Rechtslinearit"at im
     zweiten Argument kompatibel.  
   \begin{align*}
        h \act \rSP{x,y\cdot a}{\alg{A}}  & = \rSP{S(h_{\sss
            (1)})^{\ast} \act x, h_{\sss (2)} \act (y\cdot
          a)}{\alg{A}} \\ & = \rSP{S(h_{\sss (1)})^{\ast} \act x, (h_{\sss
            (2)} \act y)(h_{\sss (3)} \act \cdot a)}{\alg{A}} \\ &  =
        \rSP{S(h_{\sss (1)})^{\ast} \act x, (h_{\sss (2)} \act
          y)}{\alg{A}} (h_{\sss (3)} \act \cdot a) \\ & = (h_{\sss (1)}
        \act  \rSP{x,y}{\alg{A}})(h_{\sss (2)} \act a) \\ & = h \act
        (\rSP{x,y}{\alg{A}} a).   
    \end{align*}
\item Die $H$-Wirkung auf das $\alg{A}$-wertige innere Produkt auf
    $\bimodplus{B}{E}{A}$ ist kompatibel mit der
    $\alg{B}$-Rechtsmodulwirkung.  
\index{inneres Produkt!H-Wirkung@mit $H$-Wirkung}
  \begin{align*}
      h \act \rSP{x,b \cdot y}{\alg{A}} & = \rSP{S(h_{\sss
            (1)})^{\ast} \act x, h_{\sss (2)} \act (b \cdot
          y)}{\alg{A}}\\ & = \rSP{S(h_{\sss (1)})^{\ast} \act x, (h_{\sss
            (2)} \act b)\cdot (h_{\sss (3)} \act  y)}{\alg{A}} \\ & =
        \rSP{ (h_{\sss (2)} \act b)^{\ast}\cdot (S(h_{\sss
            (1)})^{\ast} \act x), h_{\sss (3)} \act  y}{\alg{A}} \\ &
        = \rSP{ (S(h_{\sss 
            (2)})^{\ast} \act b^{\ast})\cdot (S(h_{\sss (1)})^{\ast} \act x),
          h_{\sss (3)} \act  y}{\alg{A}} \\ & =  \rSP{ S(h_{\sss
            (1)})^{\ast} \act (b^{\ast} \cdot x),
          h_{\sss (2)} \act  y}{\alg{A}} \\ & =  h \act \rSP{b^{\ast}
          \cdot x,y}{\alg{A}}.  
  \end{align*}
Dabei nutzt man im vorletzten Schritt, da"s f"ur die Antipode $S: H
\to H$ einer \Name{Hopf}-$^\ast$-Algebra $H$ gilt $(S \otimes S) \circ \Deltaop =
\Delta \circ S$ und f"ur die $^\ast$-Involution $(h \act b)^{\ast} = S(h)^{\ast} \act b^{\ast}$.  
\item Die $H$-Wirkung ist kompatibel mit der
    Kompatibilit"atsbedingung der beiden inneren Produkte.
  \begin{align*}
      h \act (x \cdot \rSP{y,z}{\alg{A}}) & = (h_{\sss (1)} \act x)
      \cdot (h_{\sss (2)} \act \rSP{y,z}{\alg{A}}) \\ &  = (h_{\sss (1)} \act x)
      \cdot \rSP{S(h_{\sss (2)})^{\ast} \act y , h_{\sss (3)} \act
        z}{\alg{A}} \\ & = \lSP{\alg{B}}{h_{\sss (1)} \act x,
        S(h_{\sss (2)})^{\ast} \act y} \cdot  (h_{\sss (3)} \act z) \\
    & = h_{\sss (1)} \act \lSP{\alg{B}}{x,y} \cdot (h_{\sss (2)} \act
      z) \\ & = h \act (\lSP{\alg{B}}{x,y} \cdot z).
\end{align*}
\end{compactenum}
Analog zu den in {\it i.)} und {\it ii.)} gemachten Rechnungen zeigt
man die Kompatibilit"aten des anderen inneren Produkts.
\end{proof}

\begin{lemma}["Aquivalente Formulierung der $H$-Wirkung auf innere Produkte]
Die Bedingungen in Gleichungen \eqref{eq:WirkungAufSkalarprodukte} sind
"aquivalent zu 
\begin{align}
    \label{eq:AeqivalenteFormulierungWirkungAufSkalarprodukte}
    \rSP{x,h \act y}{\alg{A}} = h_{\sss (2)} \act \rSP{h_{\sss
        (1)}^{\ast} \act x, y}{\alg{A}} \quad \text{und} \quad \lSP{\alg{B}}{x,
      h\act y}  = S(h_{\sss (2)})^{\ast} \act
    \lSP{\alg{B}}{h_{\sss (1)}^{\ast} \act x, y}.
\end{align}
\end{lemma}

\begin{proof}
Der Beweis sind je zwei einfache Rechnungen (vergleiche Definition
\ref{Definition:HopfSternAlgebra}).  
\begin{align*}
  h_{\sss (2)} \act \rSP{h_{\sss (1)}^{\ast} \act x, y}{\alg{A}} & =
  \rSP{(S(h_{\sss (2)(1)})^{\ast} h_{\sss (1)}^{\ast}) \act x, h_{\sss (2)(2)}
    \act y}{\alg{A}} \\ &= \rSP{(h_{\sss (1)}S(h_{\sss
      (2)(1)}))^{\ast} \act x, h_{\sss (2)(2)} \act y}{\alg{A}} \\ & =
  \rSP{\varepsilon(h_{\sss (1)})^{\ast} \act x, h_{\sss (2)} \act
    y}{\alg{A}} \\ & =\rSP{x, (\varepsilon(h_{\sss (1)}) h_{\sss (2)})
    \act y}{\alg{A}} \\ & = \rSP{x,h \act y}{\alg{A}} 
\intertext{sowie die Umkehrung}
 \rSP{S(h_{\sss (1)})^{\ast} \act x, h_{\sss (2)} \act y}{\alg{A}} & =
 h_{\sss (2)(2)} \act \rSP{h_{\sss (2)(1)}^{\ast} \act (S(h_{\sss
     (1)})^{\ast} \act x), y}{\alg{A}} \\ & = h_{\sss (2)(2)} \act
 \rSP{(S(h_{\sss (1)})h_{\sss (2)(1)})^{\ast} \act x, y}{\alg{A}} \\ &
=h_{\sss (2)} \act \rSP{\cc{\varepsilon}(h_{\sss (1)}) x,y}{\alg{A}}
\\ & = h \act \rSP{x,y}{\alg{A}}. 
\end{align*}
Analog rechnet man die Behauptung f"ur das $\alg{B}$-wertige innere
Produkt nach.
\end{proof}

\begin{bemerkung}[$H$-Wirkung auf Ausartungsraum]
\label{Bemerkung:HWirkungAusartungsraum}
 Sei $\rmodplus{E}{A}$ ein $\alg{A}$-Rechtsmodul mit einem ausgearteten
 inneren Produkt, so garantiert Gleichung
 \eqref{eq:AeqivalenteFormulierungWirkungAufSkalarprodukte}, da"s $H
 \act \alg{E}^{\perp}_{\alg{A}} \subseteq
 \alg{E}^{\perp}_{\alg{A}}$. Damit "ubersteht die $H$-Wirkung die
 Quotientenbildung $\rmod{E}{A}/ \alg{E}^{\perp}_{\alg{A}}$, und
 $\left( \rmod{E}{A}/ \alg{E}^{\perp}_{\alg{A}}, \rSP{\cdot,
       \cdot}{\alg{A}} \right)$ wird zu einem nichtausgearteten,
 $H$-"aquivarianten $\alg{A}$-Rechtsmodul mit innerem Produkt.  
\end{bemerkung}

Damit haben wir die Wirkung einer \Name{Hopf}-$^\ast$-Algebra auf
einem Bimodul mit inneren Produkten definiert. Um eine $H$-"aquivariante
\Name{Morita}-Theorie zu formulieren, m"ussen wir angeben, wie die
\Name{Hopf}-$^\ast$-Algebra auf dem dualen Bimodul operiert.
\begin{lemma}[Wirkung $\ccact$ auf dualem Bimodul
    $\bimod{A}{\cc{\alg{E}}}{B}$] 
\label{Lemma:WirkungAufDualemBimodul}
\index{Bimodul!dualer}
Seien $\alg{A}$ und $\alg{B}$ $^\ast$-Algebren und
$(\bimod{B}{E}{A},\neact)$ ein $H$-"aquivarianter
($^\ast$-, starker) \Name{Morita}-"Aqui\-va\-lenz\-bi\-mo\-dul, dann  
induziert die Wirkung $\neact$ eine Wirkung $\ccact$ auf dem dualen
Bimodul $\bimod{A}{\cc{E}}{B}$ durch

\begin{align}
    \label{eq:InduzierteWirkungAufDualemBimodul}
    h \ccact \cc{x}: = \cc{S(h)^{\ast} \act x}.  
\end{align}
\end{lemma}

\begin{proof}
Wir wollen zeigen, da"s es sich um eine Wirkung handelt,
vgl.~Definition \ref{Definition:AxiomeWirkungHopf}. Seien im weiteren
$g,h \in H$, $x \in \bimod{B}{E}{A}$ und die komplex konjugierten
Gr"o"sen aus dem dualen Bimodul $\bimod{A}{\cc{E}}{B}$ mit $\cc{x},
\cc{y}$ etc.~bezeichnet.
\begin{compactenum}
\item Hintereinanderausf"uhrung zweier Wirkungen ist wieder eine Wirkung
    der \Name{Hopf}-Algebra. 
    \begin{align*}
     g \ccact (h \ccact \cc{x}) & = g \ccact \cc{S(h)^{\ast} \act x} \\
     & = \cc{S(g)^{\ast} \act S(h)^{\ast} \act x} \\ & =
     \cc{(S(g)^{\ast} S(h)^{\ast}) \act x} \\ & =
     \cc{(S(h)S(g))^{\ast} \act x} \\ & = \cc{S(gh)^{\ast} \act x} \\
     & = (gh) \ccact \cc{x}.   
    \end{align*}
\item Desweiteren mu"s das Einselement in der \Name{Hopf}-Algebra
    eine triviale Wirkung haben.
    \begin{align*}
     1_{\sss H} \ccact \cc{x} & = \cc{S(1_{\sss H})^{\ast} \act x} =
     \cc{1_{\sss H}^{\ast} \act x} = \cc{1_{\sss H} \act x} = \cc {x}.  
    \end{align*}
\item Die Bimodulstrukturen m"ussen mit der Wirkung $\ccact$
    vertr"aglich sein. Dabei bleibt die Wirkung auf den Algebren
    $\alg{A}$ und $\alg{B}$ jeweils unangetastet.
    \begin{align*}
     h \ccact (a \cdot \cc{x}) & = h \ccact \cc{x \cdot a^{\ast}} \\  &
    = \cc{S(h)^{\ast} \act (x \cdot a^{\ast})} \\  & = \cc{(S(h_{\sss 
        (2)})^{\ast} \act x) \cdot (S(h_{\sss (1)})^{\ast} \act
      a^{\ast}) } \\& = \cc{(S(h_{\sss (2)})^{\ast} \act x) \cdot
      (h_{\sss (1)} \act a)^{\ast}} \\ & = (h_{\sss (1)} \act a) \cdot
    \cc{(S(h_{\sss (2)})^{\ast} \act x)} \\ &=  (h_{\sss (1)} \act a)
    \cdot (h_{\sss (2)} \ccact \cc{x}).    
    \end{align*}
Analog rechnet man die $\alg{B}$-Rechtswirkung auf
$\bimod{A}{E}{B}$ nach.

\item Nun gilt es noch zu zeigen, da"s die Wirkung konsistent
    mit den inneren Produkten des Bimoduls ist.
    \begin{align*}
    h \act \lSP[\sss \cc{\alg{E}}]{ \alg{A}}{\cc{x},\cc{y}} & =  h
    \act \rSP[\sss \alg{E}]{x,y}{\alg{A}} \\ & =  \rSP[\sss
  \alg{E}]{S(h_{\sss (1)})^{\ast} \act x, h_{\sss (2)} \act
  y}{\alg{A}} \\ & = \lSP[\sss \cc{\alg{E}}]{\alg{A}}{\cc{S(h_{\sss
      (1)})^{\ast} \act x}, \cc{h_{\sss (2)} \act y}} \\ & = \lSP[\sss
\cc{\alg{E}}]{\alg{A}}{h_{\sss (1)} \ccact \cc{x}, S(h_{\sss (2)})^{\ast}
  \ccact \cc{y}}   
    \end{align*}
Dabei haben wir im letzten Schritt genutzt, da"s aus Gleichung
\eqref{eq:InduzierteWirkungAufDualemBimodul} automatisch auch
$S(h)^{\ast} \ccact \cc{x} = \cc{h \act x}$ folgt, wie man durch eine
Substitution sehen kann. Wieder rechnet man auf "ahnliche Weise die
Konsistenz mit $\rSP[\sss \cc{\alg{E}}]{\cdot, \cdot}{\alg{B}}$ nach.
\end{compactenum}
\end{proof}

\subsubsection{$H$-"aquivariante Darstellungstheorie}
\label{sec:HaequivarianteDarstellungstheorie}
 
Seien $\alg{A}$ und $\alg{D}$ zwei $^\ast$-Algebren "uber dem Ring
$\ring{C}$, die beide mit einer festen \Name{Hopf}-$^\ast$-Algebra Wirkung
$(H,\neact)$ ausgestattet sind, so da"s $(H, \alg{A}, \neact)$ und
$(H, \alg{D}, \neact)$ $H$-Modulalgebren mit $^\ast$-Involution sind.

\begin{definition}[$H$-"aquivariante Darstellung]
\index{Darstellung!H-aequivariante@$H$-\"aquivariante}
Sei nun $\rmodplus{H}{D}$ ein $\alg{D}$-Rechtsmodul mit innerem
Produkt. Man bezeichnet eine $^\ast$-Darstellung $(\rmod{H}{D}, \pi)$
der $^\ast$-Algebra $\alg{A}$ als {\em $H$-"aquivariant}, wenn
\begin{align}
    \label{eq:HAeqivarianteSterndarstellung}
    \pi (h \act a)x = h_{\sss (1)} \act (\pi (a) S(h_{\sss (2)}) \act x),
\end{align}
f"ur alle $a\in \alg{A}$, $h\in H$  und $x \in \rmod{H}{D}$. Desweiteren nennt
man einen Verschr"ankungsoperator zwischen zwei $H$-"aquivarianten
Darstellungen $T:(\rmod{H}{D}, \pi) \to (\rmod{H'}{D},
\pi')$ {\em
  $H$-"aquivariant}\index{Verschraenkungsoperator@Verschr\"ankungsoperator!H-aequivarinater@$H$-\"aquivarianter}, falls f"ur alle $h\in H$ und $x \in 
\rmod{H}{D}$ gilt 
\begin{align}
    \label{eq:HAequivarianterIntertwiner}
    T(h \act x) = h\act T(x).
\end{align}
\end{definition}

Wir bezeichnen die Kategorie der $H$-"aquivarianten
$^\ast$-Darstellungen der $^\ast$-Algebra $\alg{A}$ auf 
$\alg{D}$-Rechts\-mo\-duln mit $\smod[\alg{D},H](\alg{A})$. Die
Unterkategorien der $H$-"aquivarianten $^\ast$-Darstellungen auf
Pr"a-\Name{Hil\-bert}-Mo\-duln bezeichnet man analog dazu mit $\sMod[\alg{D},H](\alg{A})$
sowie $\srep[\alg{D},H](\alg{A})$ und $\sRep[\alg{D},H](\alg{A})$.

\subsection{$H$-"aquivariante \Name{Morita}-"Aqui\-va\-lenz\-bi\-mo\-duln}
\label{sec:HaquivarianteMoritaAequivalenzbimoduln}

\subsubsection{$H$-"aquivariante Tensorprodukte und \Name{Rieffel}-Induktion}
\label{sec:HaquivarianteTensorprodukteundRieffelInduktion}

Gegeben einen $\alg{B}$-Rechtsmodul $\rmodplus{H}{B}$ mit einem
inneren Produkt und einen $(\alg{B},\alg{A})$-Bimodul
$\bimodplus{B}{E}{A}$ mit zwei inneren Produkten, wobei die inneren
Produkte jeweils mit beiden Algebra-Wirkungen vertr"aglich
seien. Desweiteren seien sowohl $\rmodplus{H}{B}$ als auch
der $(\alg{B},\alg{A})$-Bimodul $\bimodplus{B}{E}{A}$ $H$-"aquivariant.

Wir haben eine kanonische $H$-Wirkung auf den Elementen von
$\rmod{H}{B} \tensor[\alg{B}] \bimod{B}{E}{A}$ mittels
\begin{align}
    \label{eq:HWirkungAufProduktBimodul}
    h \act (x \tensor[\alg{B}] y) = (h_{\sss (1)} \act x)
    \tensor[\alg{B}] (h_{\sss (2)} \act y).
\end{align}

\begin{lemma}[$H$-Wirkung auf dem Modul $\rmod{H}{B} \otimes_{\sss
      \alg{B}} \bimod{B}{E}{A}$]
\label{Lemma:HWirkungAufDemTensorprodukt}
Die kanonische $H$-Wirkung auf dem Tensorprodukt $\rmod{H}{B}
\tensor[\alg{B}] \bimod{B}{E}{A}$ (Gleichung
\eqref{eq:HWirkungAufProduktBimodul}) ist mit dem inneren Produkt
$\rSP[\sss \alg{H} \otimes \alg{E}]{\cdot, \cdot}{\alg{A}}$
vertr"aglich, und macht $(\rmod{H}{B} \tensor[\alg{B}]
\bimod{B}{E}{A}) / (\rmod{H}{B} \tensor[\alg{B}]
\bimod{B}{E}{A})^{\perp}$ zu einem $H$-"aqui\-va\-ri\-an\-ten
$\alg{A}$-Rechtsmodul mit innerem Produkt. 
\end{lemma}

\begin{proof}
    \begin{align*}
        h \act  \rSP[\sss \alg{H} {\otimes} \alg{E}]{x
          {\tensor[\alg{B}]} y, x' {\tensor[\alg{B}]}
          y'}{\alg{A}} & = h 
        \act \rSP[\sss \alg{E}]{y,\rSP[\sss \alg{H}]{x,x'}{\alg{B}}
          \cdot y'}{\alg{A}} \\ & = \rSP[\sss \alg{E}]{S(h_{\sss
            (1)})^{\ast} \act y, h_{\sss(2)} \act \left( \rSP[\sss
              \alg{H}]{x,x'}{\alg{B}} \cdot y'\right)}{\alg{A}} \\ & =
        \rSP[\sss \alg{E}]{S(h_{\sss (1)})^{\ast} \act y,\rSP[\sss
          \alg{H}]{S(h_{\sss (2)})^{\ast} \act x, h_{\sss (3)} \act
            x'}{\alg{B}} \cdot  (h_{\sss (4)} \act y')}{\alg{A}} \\ & =
        \rSP[\sss \alg{H} {\otimes} \alg{E}]{S(h_{\sss
            (2)})^{\ast} \act 
          x {\tensor[\alg{B}]} S(h_{\sss (1)})^{\ast} \act y, (h_{\sss
            (3)} \act 
          x') {\tensor[\alg{B}]} (h_{\sss (4)} \act y')}{\alg{A}} \\ & =
        \rSP[\sss \alg{H} {\otimes} \alg{E}]{S(h_{\sss
            (1)})^{\ast} \act (x {\tensor[\alg{B}]} y), h_{\sss (2)}
          \act (x' {\tensor[\alg{B}]} y')}{\alg{A}}         
    \end{align*}
Dabei nutzt man im letzten Schritt die Eigenschaft $(S \otimes S)\circ
\Delta = \Deltaop \circ S$ der Antipode, vergleiche Proposition
\ref{Proposition:Antipode} {\it iii.)}. Dies zeigt, da"s das auf
$\rmod{H}{B} \tensor[\alg{B}] \bimod{B}{E}{A}$ induzierte innere
Produkt auch mit der $H$-Wirkung vertr"aglich ist. Die Wirkung
l"a"st sich via Bemerkung \ref{Bemerkung:HWirkungAusartungsraum}
auch auf den Quotienten 
$(\rmod{H}{B} \tensor[\alg{B}] \bimod{B}{E}{A}) / (\rmod{H}{B}
\tensor[\alg{B}] \bimod{B}{E}{A})^{\perp}$ und somit auf
$\tensorhat[\alg{B}]$ "ubertragen, wodurch $\left( (\rmod{H}{B}
    \tensorhat[\alg{B}] \bimod{B}{E}{A}), \rSP[\sss \alg{H} \otimes
\alg{E}]{\cdot,\cdot}{\alg{A}} \right)$ zu einem $H$-"aquivarianten
$\alg{A}$-Rechtsmodul mit innerem Produkt wird.
\end{proof}

Nun m"ussen wir analog zu
Lemma~\ref{Lemma:VertraeglichkeittensorBmitMorphismen} die
Vertr"aglichkeit der $H$-Wirkung mit den Morphismen zeigen.

\begin{lemma}[Vertr"aglichkeit der $H$-Wirkung mit Morphismen]
Seien $\bimodplus{B}{E}{A}$ und  $\bimodplus{B}{E^{\prime}}{A}$ zwei $H$-"aquivariante
Bimoduln und $\rmodplus{H}{B}$ bzw.~$\rmodplus{H^{\prime}}{B}$ seien
$H$-"aquivariante $\alg{B}$-Rechtsmoduln. Seien ferner $S:
\bimod{B}{E}{A} \to \bimod{B}{E^{\prime}}{A}$ und $T:\rmod{H}{B} \to
\rmod{H^{\prime}}{B}$ $H$-"aquivariante Verschr"ankungsoperatoren,
so ist die $H$-Wirkung mit der Tensorproduktbildung vertr"aglich.
\end{lemma}

\begin{proof}
Seien $x \in \rmod{H}{B}$ und $y \in \bimod{B}{E}{A}$, so gilt
\begin{align*}
    h \act (S(x) \tensor[\alg{B}] T(y)) & = (h_{\sss (1)} \act S(x))
    \tensor[\alg{B}] (h_{\sss (2)} \act T(y)) \\ & =  S(h_{\sss (1)} \act x)
    \tensor[\alg{B}] T(h_{\sss (2)} \act y) 
\end{align*}
Mit Lemma~\ref{Lemma:VertraeglichkeittensorBmitMorphismen} zeigt dies
die Vertr"aglichkeit der $H$-Wirkung mit den Morphismen.
\end{proof}

Damit haben wir eine funktorielle Abbildung konstruiert, und wir
k"onnen analog zu Gleichung \eqref{eq:FunktorFuersmod} schreiben

\begin{align}
    \label{eq:FunktorFuersmodHaequivariant}
\tensorhat[\alg{B}]: \smod[\alg{B},H](\alg{C}) \times
    \smod[\alg{A},H](\alg{B}) \to \smod[\alg{A},H](\alg{C}).
\end{align}
Das gewonnene innere Produkt $\tensorhatohne$ ist bis auf die kanonische
Isomorphie assoziativ (siehe Gleichung
\eqref{eq:AssoziativitaetInneresTensorprodukt}). Da $\tensorhatohne$ mit
der vollst"andigen Positivit"at der inneren Produkte vertr"aglich ist,
"ubertr"agt sich Gleichung \eqref{eq:FunktorFuersmodHaequivariant}
auch auf die Kategorien $\srep[\alg{B},H]$ bzw.~auf $\sMod[\alg{B},H]$
und $\sRep[\alg{B},H]$. Analog zu Beispiel
\ref{Beispiel:RieffelInduktion} k"onnen wir eine 
$H$-"aquivariante Version der \Name{Rieffel}-Induktion angeben.

\begin{definition}[$H$-"aquivariante \Name{Rieffel}-Induktion, Wechsel
    der Basisalgebra] 
\label{Definition:HAequivarianteRieffelInduktion}
\index{Rieffel-Induktion@\Name{Rieffel}-Induktion!H-aequivariante@$H$-\"aquivariante}
Seien $\alg{A}$, $\alg{B}$ und $\alg{D}$, $\alg{D}'$ $^\ast$-Algebren
"uber dem Ring $\ring{C}$, die alle mit einer Wirkung der
\Name{Hopf}-$^\ast$-Algebra $H$ im Sinne von Definition~\ref{Definition:AxiomeWirkungHopf}
versehen seien. Desweiteren seien 
$\bimod{B}{E}{A} \in \srep[\alg{A},H](\alg{B})$ und $\bimod{D}{G}{D'}
\in \srep[\alg{D'},H](\alg{D})$.  
\begin{compactenum}
\item Wir bezeichnen den $H$-"aquivarianten Funktor 
\begin{align}
    \label{eq:HAequivarianteRieffelInduktion}
    \rieffel{R}{\alg{E}} = \bimod{B}{E}{A} \tensorhat[\alg{A}] \cdot :
    \srep[\alg{D},H](\alg{A}) \to \srep[\alg{D},H](\alg{B}) 
\end{align}
als {\em $H$-"aquivariante \Name{Rieffel}-Induktion}. 

\item Desweiteren definieren wir den {\em $H$-"aquivarianten Wechsel
      der Basisalgebra} durch den Funktor
\index{Wechsel der Basisalgebra!H-aequivarianter@$H$-\"aquivarianter}
\begin{align}
    \label{eq:HAequivarianterWechselDerBasisAlgebra}
    \rieffel{S}{\alg{G}} = \cdot \tensorhat[\alg{D}] \bimod{D}{G}{D'}:
    \srep[\alg{D},H](\alg{A}) \to \srep[\alg{D'},H](\alg{A}).  
\end{align}
\end{compactenum}
Die Wirkungen auf Objekte und Morphismen sind die gleichen wie in den
Beispielen \ref{Beispiel:RieffelInduktion}. 
\end{definition}

\subsubsection{$H$-"aquivariante starke \Name{Morita}-"Aquivalenz}
\label{HaequivariantestarkeMoritaAequivalenz}
\index{Morita-Aequivalenz@\Name{Morita}-\"Aquivalenz!starke!H-aequivariante@$H$-\"aquivariante}
\index{Morita-Aequivalenz@\Name{Morita}-\"Aquivalenz!Stern@$^\ast$-!H-aequivariante@$H$-\"aquivariante}

\begin{definition}[$H$-"aquivarianter "Aqui\-va\-lenz\-bi\-mo\-dul]
\index{Aequivalenzbimodul@\"Aquivalenzbimodul!H-aequivarianter@$H$-\"aquivarianter}
    Seien $\alg{A}$ und $\alg{B}$ zwei $^\ast$-Algebren,
    $\bimodplus{B}{E}{A}$ ein $^\ast$-"Aqui\-va\-lenz\-bi\-mo\-dul und die Algebren
    wie auch der Bimodul seien mit einer $^\ast$-Wirkung der
    \Name{Hopf}-$^\ast$-Algebra $H$ vertr"aglich. Man nennt 
    $\bimod{B}{E}{A}$ einen {\em $H$-"aquivarianten
      $^\ast$-"Aqui\-va\-lenz\-bi\-mo\-dul}, wenn zus"atzlich die 
    Gleichungen \eqref{eq:WirkungAufSkalarprodukte} f"ur alle $h\in
    H$ und $x,y \in \bimod{B}{E}{A}$ erf"ullt sind. Sind
    zus"atzlich beide inneren Produkte vollst"andig positiv nennt man
    $\bimod{B}{E}{A}$ einen {\em $H$-"aquivarianten starken
      \Name{Mo\-ri\-ta}-"Aqui\-va\-lenz\-bi\-mo\-dul}.   
\end{definition}

\begin{definition}[$H$-"aquivariant \Name{Morita}-"aquivalente Algebren] 
\label{Definition:HAequivarianteMoritaAequivaleneAlgebren}
Gegeben seien zwei $^\ast$-Algebren $\alg{A}$ und $\alg{B}$, die beide mit einer
$H$-Wirkung einer \Name{Hopf}-$^\ast$-Algebra vertr"aglich sind. Man
nennt die beiden Algebren {\em $H$-"aquivariant
  $^\ast$-\Name{Morita}-"aquivalent} (bzw.~{\em $H$-"aquivariant
  stark \Name{Morita}-"aquivalent}) falls es einen
 $H$-"aquivarianten  
$^\ast$-\Name{Mo\-ri\-ta}-"Aqui\-va\-lenz\-bi\-mo\-dul 
(bzw.~$H$-"aqui\-va\-rian\-ten starken 
\Name{Mo\-ri\-ta}-"Aqui\-va\-lenz\-bi\-mo\-dul) f"ur die Algebren gibt.      
\end{definition}

\begin{satz}[$H$-"aquivariante \Name{Morita}-"Aquivalenz ist
    "Aquivalenzrelation] 
    \label{Satz:HAequivarianzIstAequivalenzrelation}
F"ur die idempotenten und nichtausgearteten $^\ast$-Algebren mit
$^\ast$-Wirkungen einer \Name{Hopf}-$^\ast$-Algebra $H$ ist $H$-"aquivariante
starke \Name{Morita}-"Aquivalenz eine "Aquivalenzrelation.  

Desweiteren sind $H$-"aquivariante $^\ast$-isomorphe $^\ast$-Algebren
$H$-"aqui\-va\-ri\-ant stark \Name{Morita}-"aquivalent und damit auch
$H$-"aqui\-va\-ri\-ant $^\ast$-\Name{Morita}-"aquivalent. 
\end{satz}

\begin{proof}
Wir haben bereits mit Lemma
\ref{Lemma:MoritaAequivalenzAequivalenzrelation} gezeigt, da"s
$^\ast$- und starke \Name{Morita}-"Aquivalenz eine
"Aqui\-va\-lenz\-re\-la\-tion sind. Es bleibt die
$H$-"Aqui\-va\-ri\-anz zu zeigen. F"ur eines der kanonischen inneren
Produkte auf $\bimod{A}{A}{A}$ ergibt sich  

\begin{align*}
    h \act \rSP{a,b}{\alg{A}} & = h \act (a^{\ast}b) = (h_{\sss (1)}
    \act a^{\ast})(h_{\sss (2)} \act b) = (S(h_{\sss (1)})^{\ast} \act
    a)^{\ast} (h_{\sss (2)}) \act b) \\ & = \rSP{S(h_{\sss (1)})^{\ast}
      \act a, h_{\sss (2)} \act b}{\alg{A}},
\end{align*}
und auf "ahnliche Weise verh"alt es sich mit dem
linksseitigen $\alg{A}$-wertigen inneren Produkt $\lSP{\alg{A}}{\cdot,
  \cdot}$.  

Um die Symmetrie zu zeigen ben"otigen wir die induzierte $H$-Wirkung auf dem
konjugierten Bimodul $\bimod{A}{\cc{E}}{B}$, die wir in Lemma
\ref{Lemma:WirkungAufDualemBimodul} eingef"uhrt haben. Die
Transitivit"at ist eine Folge von
Lemma~\ref{Lemma:HWirkungAufDemTensorprodukt} beziehungsweise der
$H$-"aquivarianten \Name{Rieffel}-Induktion, die wir in
Definition~\ref{Definition:HAequivarianteRieffelInduktion}
eingef"uhrt haben.   
\end{proof}


%% file: cross.tex
\chapter{\Name{Morita}-"Aquivalenz von Cross-Produktalgebren}
\label{sec:MoritaAequivalenzVonCrossProdukten}
\fancyhead[CO]{\slshape \nouppercase{\rightmark}} 
\fancyhead[CE]{\slshape \nouppercase{\leftmark}} 

\index{Algebra!Cross-Produktalgebra|(}
\index{Cross-Produktalgebra|(}

Ziel dieses Kapitels ist es Cross-Produktalgebren $\cross{\alg{A}}{H}$
zu studieren. Dabei
werden wir das \Name{Picard}-Gruppoid von Cross-Produktalgebren
genauer betrachten und eine Verbindung zu der $H$-"aqui\-va\-ri\-an\-ten
Theorie der zugrundeliegenden Algebra $\alg{A}$ herstellen. Eine
Einf"uhrung in die Cross-Produktalgebren findet sich in Anhang
\ref{sec:CrossProdukte} sowie in \citep{majid:1995a}.

\section{Cross-Produktalgebren von $^\ast$-Darstellungen}
\label{sec:CrossProduktevonSternDarstellungen}

Seien im weiteren $\alg{A}$ und $\alg{B}$ $^\ast$-Algebren mit
Einselement und $H$ sei eine \Name{Hopf}-$^\ast$-Algebra, so da"s
$(H,\alg{A},\neact)$ und $(H,\alg{B},\neact)$ $^\ast$-Linksmodulalgebren sind. Dann sind
$\cross{\alg{A}}{H}$ und $\cross{\alg{B}}{H}$ Cross-Pro\-dukt\-al\-ge\-bren
nach Definition \ref{Definition:CrossProdukt}, d.~h.~ist $\alg{A}$
eine $^\ast$-Algebra und $H$ eine \Name{Hopf}-$^\ast$-Algebra, dann
ist die Cross-Produktalgebra als Raum das Tensorprodukt $\alg{A}
\otimes H$, und Elemente sind von der Form $a \otimes h$, wobei $a \in
\alg{A}$ und $h\in H$ ist. Die Multiplikation zweier Elemente in
$\cross{\alg{A}}{H}$ ist gegeben durch
\begin{align*}
    (a \otimes h)\cdot (b \otimes g) = a\cdot (h_{\sss (1)} \act b)
    \otimes h_{\sss (2)}g,
\end{align*}
und die $^\ast$-Involution durch 
\begin{align*}
    (a \otimes h)^{\ast} = h_{\sss (1)}^{\ast} \act a^{\ast} \otimes
    h_{\sss (2)}^{\ast}.
   \end{align*}

\begin{lemma}[Isomorphie der Kategorien {$\smod[H](\alg{A})$} und {$\smod(\cross{\alg{A}}{H})$}]
    Die Kategorien $\smod[H](\alg{A})$ und $\smod(\cross{\alg{A}}{H})$
    sind isomorph. Auf den Objekten ist die Isomorphie gegeben
    durch 
    \begin{align}
        \label{eq:AequivalenzKategorienModHAundModAH}
        \smod[H](\alg{A}) \ni (\alg{H},\pi) \mapsto
        (\hat{\alg{H}},\hat{\pi}) \in \smod(\cross{\alg{A}}{H}), 
    \end{align}
wobei $\alg{H}=\hat{\alg{H}}$ als Pr"a-\Name{Hilbert}-Raum ist, und
$\hat{\pi}(a \otimes h)\phi = \pi(a)h \act \phi$ mit $\phi \in \alg{H}$ gilt. Auf den Morphismen
$T: (\alg{H}_{1},\pi_{1}) \to  (\alg{H}_{2},\pi_{2})$ ist die
Isomorphie durch die Identit"atsabbildung gegeben. Analoges gilt auch
f"ur $\sMod$, $\srep$ und $\sRep$.  
\end{lemma}

\begin{proof}
Der Beweis ist klar. Die Pr"a-\Name{Hilbert}-R"aume sind als R"aume
identisch und die Morphismen gehen durch den Identit"atsfunktor ineinander
"uber.
\end{proof}

\begin{proposition}[Positivit"at]
\index{Positivitaet@Positivit\"at!Cross-Produktalgebren@von Cross-Produktalgebren}
    Sei $\omega: \alg{A} \to \ring{C}$ ein
    $H$-invariantes\index{Funktional!positives!H-invariantes@$H$-invariantes},
    positives lineares Funktional und $\rho: H \to \ring{C}$ ein positives lineares Funktional\index{Funktional!positives}, dann ist auch $\omega \otimes \rho:
    \cross{\alg{A}}{H} \to \ring{C}$ positiv.
\end{proposition}

\begin{proof}
Sei $\sum_{i} a_{i} \otimes h_{i}$ gegeben. Eine einfache Rechnung,
die Invarianz von $\omega$ nutzend, gibt
\begin{align*}
    (\omega \otimes \rho)\left( \left( \sum a_{i} \otimes h_{i}\right)^{\ast}
    \left( \sum  a_{j} \otimes h_{j}\right) \right) = \sum_{i,j}
    \omega(a^{\ast}_{i}a_{j}) \rho(h_{i}^{\ast} h_{j}) \ge 0, 
\end{align*}
und wenn sowohl $\omega$ als auch $\rho$ positive lineare Funktionale
sind, sind sie auch vollst"andig positiv (im Sinne vollst"andig
positiver Abbildungen) \citep[Lemma 4.3]{bursztyn.waldmann:2001a}. 
\end{proof}

\begin{beispiele}[Positives Funktionale auf {$\cross{\alg{A}}{H}$}]
~\vspace{-5mm}
\index{Funktional!positives!Cross-Produktalgebra@auf Cross-Produktalgebra}
\begin{compactenum}
\item Ein erstes Beispiel f"ur ein positives Funktional auf
$\cross{\alg{A}}{H}$ ist $\omega \otimes \varepsilon:
    \cross{\alg{A}}{H} \to \ring{C}$, wobei $\omega$ positiv und
    $H$-invariant ist. Dabei bezeichnet $\varepsilon: H \to \ring{C}$ die Koeins
    der \Name{Hopf}-$^\ast$-Algebra.  
\item Allgemeiner gilt, da"s wenn $\chi: H \to \ring{C}$ ein
    {\it unit"arer Charakter}\index{Charakter!unitaerer@unit\"arer},
    d.~h.~ein $^\ast$-Homomorphismus ist,
    dann ist $\omega \otimes \chi: \cross{\alg{A}}{H} \to \ring{C}$
    ein positives Funktional. 
\end{compactenum}
\end{beispiele}

\begin{lemma}[{$(\cross{\alg{B}}{H},\cross{\alg{A}}{H})$-Bimodul}]
\label{Lemma:StrukturenCrossProduktalgebra}
 Sei $\bimod{B}{E}{A} \in \smod[\alg{A},H](\alg{B})$, dann wird
 $\alg{E} \ctimes H$ durch die beiden Verkn"upfungen     
 \begin{align}
     (b \ctimes g) \cdot (x \ctimes h) := (b \cdot g_{\sss (1)} \act
     x) \ctimes (g_{\sss (2)} h) \quad \text{und} \quad (x \ctimes g) \cdot
     (a \ctimes h) := (x \cdot g_{\sss (1)} \act a) \ctimes (g_{\sss
       (2)} h) 
 \end{align}
zu einem $(\cross{\alg{B}}{H},
\cross{\alg{A}}{H})$-Bimodul. Desweiteren definiert
\begin{align}
    \label{eq:InneresProduktAufCrossProdukt}
    \rSPcr[\sss \alg{E}\ctimes H]{x \ctimes g, y \ctimes
      h}{\cross{\alg{A}}{H}} := (g_{\sss (1)}^{\ast} \act \rSPcr[\sss
    \alg{E}]{x,y}{\alg{A}}) \ctimes  g_{\sss (2)}^{\ast} h   
\end{align}
ein $\cross{\alg{A}}{H}$-wertiges inneres Produkt\index{inneres Produkt!Cross-Produktalgebra@auf Cross-Produktalgebra} auf $\alg{E} \ctimes
H$, und es gilt 
\begin{align}
    \label{eq:KompatibilitaetFurInneresProdukt}
     \rSPcr[\sss \alg{E} \ctimes H]{(b\ctimes g) \cdot (x \ctimes h), y
       \ctimes k}{\cross{\alg{A}}{H}} = \rSPcr[\sss \alg{E} \ctimes
     H]{x \ctimes h, (b \ctimes g)^{\ast} \cdot (y \ctimes
       k)}{\cross{\alg{A}}{H.}}  
\end{align}
\end{lemma}

\begin{proof}
    Zuerst mu"s man zeigen, da"s es sich bei
    $\bimodo{\cross{\alg{B}}{H}}{\alg{E}\ctimes H}{\cross{\alg{A}}{H}}$
    um einen Bimodul handelt. Dazu rechnet man nach
    \begin{align*}
        \left( (b \ctimes h)(b' \ctimes h') \right) \cdot (x \ctimes
        g) & = \left( b(h_{\sss (1)} \act b') \ctimes h_{\sss
              (2)} h' \right) \cdot (x \ctimes g) \\ & =
        b(h_{\sss (1)} \act b')\cdot ((h_{\sss (2)} h')_{\sss
          (1)} \act x) \ctimes  (h_{\sss (2)} h')_{\sss
          (2)} g \\ & = b (h_{\sss (1)} \act b') \cdot (h_{\sss
          (2)(1)} \act h_{\sss (1)}' \act x) \ctimes h_{\sss (3)}
        h_{\sss (2)}' g \\ & =  b \cdot (h_{\sss (1)} \act (b' \cdot (h_{\sss
          (1)}' \act x))) \ctimes h_{\sss (2)} h_{\sss (2)}' g \\ & = (b
        \ctimes h) \cdot \left( b' (h_{\sss 
              (1)}' \act x) \ctimes h_{\sss (2)}' g \right) \\ & = (b
        \ctimes h) \cdot \left( (b' \ctimes h') \cdot (x \ctimes g)
        \right)  
      \end{align*}
und analog zeigt man die $\cross{\alg{A}}{H}$-Rechtswirkung
beziehungsweise die Vertr"aglichkeit der
$\cross{\alg{A}}{H}$-Rechtswirkung mit der
$\cross{\alg{B}}{H}$-Linkswirkung. Aufgrund der \glqq
symmetrischen\grqq{} Struktur reicht es, nur eine Wirkung zu zeigen. Desweiteren ist zu zeigen,
da"s das innere Produkt auf
$\bimodo{\cross{\alg{B}}{H}}{\alg{E}\ctimes H}{\cross{\alg{A}}{H}}$
sinnvoll definiert ist. Die $\ring{C}$-Sesquilinearit"at kann man gleich
ablesen. Die $^\ast$-Involution des inneren Produkts enspricht der
Vertauschung der Argumente, wie die folgende Rechnung zeigt.
 
\begin{align*}
    \left(\rSPcr[\sss \alg{E} \ctimes H]{x \ctimes g, x' \ctimes
      g'}{\cross{\alg{A}}{H}}\right)^{\ast} & = \left(g_{\sss (1)}^{\ast} \act
  \rSPcr[\sss \alg{E}]{x,x'}{\alg{A}} \ctimes g_{\sss (2)}^{\ast} g'
\right)^{\ast} \\ & = \left(
  g_{\sss (2)}^{\ast} g' \right)_{\sss (2)}^{\ast} \act \left( g_{\sss
    (1)}^{\ast} \act \rSPcr[\sss \alg{E}]{x,x'}{\alg{A}} \right)^{\ast}
\ctimes \left( 
  g_{\sss (2)}^{\ast} g' \right)_{\sss (2)}^{\ast} \\ & = {g}_{\sss
(1)'}^{\ast} \left(S(g_{\sss (1)}^{\ast}) g_{\sss (2)}^{\ast} 
\right)^{\ast} \act \left( \rSPcr[\sss
  \alg{E}]{x,x'}{\alg{A}} \right)^{\ast} \ctimes {g_{\sss
(2)}'}^{\ast} g_{\sss (3)} \\ & = {g}_{\sss (1)}'
\cc{\varepsilon(g_{\sss(1)}^{\ast})} \act \rSPcr[\sss
\alg{E}]{x',x}{\alg{A}} \ctimes {g}_{\sss (2)}' g_{\sss (2)} \\ & =
{g}_{\sss (1)}' \act \rSPcr[\sss \alg{E}]{x',x}{\alg{A}} \ctimes
{g}_{\sss (2)}' g  \\ & = \rSPcr[\sss \alg{E} \ctimes H]{x' \ctimes g', x \ctimes
g}{\cross{\alg{A}}{H.}}        
\end{align*}

Im vorletzten Schritt nutzt man die Linearit"at von $\otimes$ bez"uglich
Elementen in $\ring{C}$ sowie die Tatsache, da"s in einer
\Name{Hopf}-Algebra $\varepsilon(h_{\sss (1)}) h_{\sss (2)} = h$ ist.   

Bleibt zu zeigen, da"s das $\cross{\alg{A}}{H}$-wertige innere
Produkt mit der $\cross{\alg{A}}{H}$-Rechtswirkung vertr"aglich ist

\begin{align*}
    \rSPcr[\sss \alg{E} \ctimes H]{x \ctimes g, \left( x' \ctimes g'\right) \cdot
      \left( a\ctimes h \right)}{\cross{\alg{A}}{H}} & = \rSPcr[\sss
    \alg{E} \ctimes H]{x \ctimes g, x' \cdot \left(g_{\sss (1)}' \act
          a \right) \ctimes g_{\sss (2)}' h}{\cross{\alg{A}}{H}} \\ &
    = g_{\sss (1)}^{\ast} \act \rSPcr[\sss \alg{E}]{x, x' \cdot
      \left(g_{\sss (1)}' \act a\right)}{\alg{A}} \ctimes g_{\sss
      (2)}^{\ast} g_{\sss (2)}' h \\ & =  g_{\sss (1)}^{\ast} \act
    \left(\rSPcr[\sss \alg{E}]{x, x'}{\alg{A}} \left(g_{\sss (1)}'  
      \act a \right) \right) \ctimes g_{\sss (2)}^{\ast} g_{\sss (2)}' h
\\ & =  \left(g_{\sss  
      (1)}^{\ast} \act \rSPcr[\sss \alg{E}]{x, x'}{\alg{A}}\right)
    \left(g_{\sss (2)}^{\ast} g_{\sss (1)}' \act a \right) \ctimes g_{\sss
      (3)}^{\ast} g_{\sss (2)}' h   \\ & =  
      \left(g_{\sss (1)}^{\ast} \act \rSPcr[\sss \alg{E}]{x,x'}{\alg{A}}
      \ctimes g_{\sss (2)}^{\ast} g' \right) \left(a \ctimes 
          h \right) \\ & = 
      \rSPcr[\sss \alg{E} \ctimes H]{x \ctimes g, \left(x' \ctimes 
        g'\right)}{\cross{\alg{A}}{H}} \left(a \ctimes h \right). 
\end{align*}
\end{proof}

\begin{bemerkung}[{$\cross{\alg{B}}{H}$-wertiges} inneres Produkt]
Wir k"onnen auf $\alg{E} \otimes H$ auch ein $\cross{\alg{B}}{H}$-wertiges inneres Produkt
definieren, wenn der Bimodul $\bimod{B}{E}{A}$ mit einem
$\alg{B}$-wertigen inneren Produkt ausgestattet ist. Analog zu
Lemma \ref{Lemma:StrukturenCrossProduktalgebra} erbt das
$\cross{\alg{B}}{H}$-wertige innere Produkt alle Strukturen.
\end{bemerkung}

\begin{bemerkung}[Ausartungsraum von {$\alg{E} \otimes H$}]
    Es kann passieren, da"s das innere Produkt $\rSPcr[\sss \alg{E}
    \ctimes H]{\cdot,\cdot}{\cross{\alg{A}}{H}}$ ausgeartet ist. Ist
    dies der Fall, so geht man zum Quotienten
    \begin{align}
        \label{eq:QuotientCrossProduktBimodul}
        \cross{\alg{E}}{H} = \alg{E} \otimes H / (\alg{E} \otimes H)^{\perp}
    \end{align}
"uber, der mit der Bimodulstruktur und den inneren Produkten
vertr"aglich ist. 
\end{bemerkung}

\begin{satz}[Vollst"andige Positivit"at und Nichtausgeartetheit,
    {\citep[Lemma 6.6]{jansen.waldmann:2005a}}]
\label{Satz:VollstaendigePositivitaetundNichtausgeartetheitvonCrossProduktalgebra}
    Sei $\bimod{B}{E}{A} \in \srep[\alg{A},H](\alg{B})$, dann ist das
    innere Produkt $\rSPcr[\sss \alg{E} \ctimes H]{\cdot,
      \cdot}{\cross{\alg{A}}{H}}$ vollst"andig positiv, woraus folgt,
    da"s
    $\bimodcross{B}{E}{A} \in \srep[\cross{\alg{A}}{H}]
    (\cross{\alg{B}}{H})$ eine $^\ast$-Darstellung auf einem
    Pr"a-\Name{Hilbert}-Modul ist. Desweiteren ist f"ur
    $\bimod{B}{E}{A} \in \sMod[\alg{A},H](\alg{B})$
    $\bimodcross{B}{E}{A} \in
    \sMod[\cross{\alg{A}}{H}](\cross{\alg{B}}{H})$ und
    $\bimod{B}{E}{A} \in \sRep[\alg{A},H](\alg{B})$ $\bimodcross{B}{E}{A} \in
    \sRep[\cross{\alg{A}}{H}](\cross{\alg{B}}{H})$.
\end{satz}

\begin{proof}
Seien $\Phi^{\sss (1)},\cdots,\Phi^{\sss (n)} \in \alg{E} \otimes H$ gegeben und
sei $\Phi^{\sss (\alpha)}=\sum_{i=1}^{m} x_{i}^{\sss (\alpha)} \otimes
h_{i}^{\sss (\alpha)}$ mit $x_{i}^{\sss (\alpha)} \in \alg{E}$ und
  $h_{i}^{\sss (\alpha)} \in H$, und $m$ sei ohne Einschr"ankung der
  Allgemeinheit f"ur alle $\alpha = 1,...,n$ dasselbe. Dann gilt
  \begin{align*}
      \rSPcr[\sss \alg{E}\otimes H]{\Phi^{\sss
          (\alpha)}, \Phi^{\sss (\beta)}}{\cross{\alg{A}}{H}} =
      \sum_{i,\ell=1}^{m} \left( (h_{i}^{\sss (\alpha)})^{\ast}_{\sss
            (1)} \act \rSP[\alg{E}]{x_{i}^{\sss (\alpha)}, x_{\ell}^{\sss
              (\beta)}}{\alg{A}} \right) \otimes (h_{i}^{\sss
        (\alpha)})^{\ast}_{\sss (2)} h^{\sss (\beta)}_{\ell}. 
  \end{align*}
Wir definieren die folgende Abbildung 
\begin{align*}
f: M_{nm}(\alg{A}) \ni  A=(A_{i\ell}^{\sss \alpha \beta}) \mapsto \left(
    (h_{i}^{\sss (\alpha)})^{\ast}_{\sss (1)} \act A_{i\ell}^{\sss \alpha
      \beta} \otimes  (h_{i}^{\sss (\alpha)})^{\ast}_{\sss (2)}
    h_{\ell}^{\sss (\beta)} \right) \in M_{nm}(\cross{\alg{A}}{H}),
\end{align*}
die positiv ist, wie die folgende Rechnung zeigt:
\begin{align*}
    f(A^{\ast}A) &  = \sum_{\gamma,k} \left((h_{i}^{\sss
          (\alpha)})^{\ast}_{\sss (1)} \act \left((A_{ki}^{\sss \gamma
          \alpha})^{\ast} A_{k\ell}^{\sss \gamma \beta}\right) \otimes
        (h_{i}^{\sss (\alpha)})^{\ast}_{\sss (2)} h_{\ell}^{\sss (\beta)}
    \right) \\ & =   \sum_{\gamma,k} \left( \left((h_{i}^{\sss
          (\alpha)})^{\ast}_{\sss (1)} \act (A_{ki}^{\sss \gamma
          \alpha})^{\ast} \otimes  (h_{i}^{\sss
          (\alpha)})^{\ast}_{\sss (2)} \right) \left(A_{k\ell}^{\sss \gamma \beta} \otimes
        h_{\ell}^{\sss (\beta)} \right) \right) \\ & = \sum_{\gamma,k}
 \left(A_{ki}^{\sss \gamma \alpha} \otimes h_{i}^{\sss (\alpha)}
 \right)^{\ast} \left(A_{k\ell}^{\sss \gamma \beta} \otimes h_{i}^{\sss
      (\beta)} \right) \\ & = (A \otimes h)^{\ast} (A \otimes h). 
\end{align*}
Dabei ist $A \otimes h \in M_{nm}(\cross{\alg{A}}{H})$ gegeben durch
die Koeffizientenmatrix $(A \otimes h)^{\sss \alpha \beta}_{i \ell} =
A^{\sss \alpha \beta}_{i \ell} \otimes h^{\sss (\beta)}_{\ell}$. Also
ist $f(A^{\ast}A) \in M_{nm}(\cross{\alg{A}}{H})^{++}$, da $f$ positiv
ist. Weil $\big( \rSP[\sss \alg{E}]{x_{i}^{\sss (\alpha)},
x_{\ell}^{\sss (\beta)}}{\alg{A}} \big)$ eine positive Matrix in
$M_{nm}(\alg{A})$ ist, da das innere Produkt
$\rSP[\sss \alg{E}]{\cdot,\cdot}{\alg{A}}$ vollst"andig positiv ist, ist
auch die Abbildung $f\left( \big( \rSP[\sss \alg{E}]{x_{i}^{\sss (\alpha)},
x_{\ell}^{\sss (\beta)}}{\alg{A}} \big)\right)$ positiv. Die Summation
"uber $i, \ell$ ist eine positive Abbildung, wie in \citep[Example
2.1]{bursztyn.waldmann:2003a:pre} gezeigt wurde, und das Ergebnis ist
wieder positiv. Somit ist die vollst"andige Positivit"at des inneren Produkts $\rSPcr[\sss
\alg{E} \otimes H]{\cdot, \cdot}{\cross{\alg{A}}{H}}$ gezeigt, und die
vollst"andige Positivit"at von $\rSPcr[\sss \cross{\alg{E}}{H}]{\cdot, \cdot}{\cross{\alg{A}}{H}}$
folgt. Zu zeigen, da"s die inneren Produkte nicht ausgeartet sind, ist trivial.
\end{proof}

\begin{bemerkung}
Im Fall da"s die Algebra $\alg{A}$ mit einem Einselement ausgestattet
ist, vereinfacht sich der Beweis zu
Satz~\ref{Satz:VollstaendigePositivitaetundNichtausgeartetheitvonCrossProduktalgebra},
denn man sieht an folgendem Ausdruck leicht
\begin{align*}
    \rSPcr[\sss \alg{E} \otimes H]{x \otimes g, y \otimes
      h}{\cross{\alg{A}}{H}} = (1_{\sss \alg{A}} \otimes g)^{\ast}
    \left(\rSP[\sss \alg{E}]{x,y}{\alg{A}} \otimes 1_{\sss H} \right) (1_{\sss \alg{A}} \otimes h)
\end{align*}
die vollst"andige Positivit"at von $\rSPcr[\sss \alg{E} \otimes
H]{\cdot, \cdot}{\cross{\alg{A}}{H}}$. 
\end{bemerkung}

Wir wollen nun zeigen, da"s der "Ubergang von einem
$H$-"aquivarianten Bimodul mit inneren Produkten $\bimod{B}{E}{A}$ zu dem Bimodul
$\bimodcross{B}{E}{A}$ funktoriell geschieht.

\begin{lemma}[Verschr"ankungsoperator]
\label{Lemma:VerschraenkungsoperatorBimodulliefertVOaufCrossProdukten}
  Sei $T:\bimod{B}{E}{A} \to \bimod{B}{F}{A}$ ein
  Verschr"ankungsoperator zwischen $\bimod{B}{E}{A}$,
  $\bimod{B}{F}{A} \in \smod[\alg{A},H](\alg{B})$, dann ist die
  Abbildung $T \otimes \id_{\sss H}: \alg{E} \otimes H \to
  \alg{F} \otimes H$ ein Verschr"ankungsoperator zwischen
  $\cross{\alg{E}}{H}$ und $\cross{\alg{F}}{H} \in
  \smod[\cross{\alg{A}}{H}](\cross{\alg{B}}{H})$, dessen Adjungierte
  durch $T^{\ast} \otimes \id_{\sss H}$ gegeben ist.  
\end{lemma}

\begin{proof}
    Der Beweis nutzt die $H$-"Aquivarianz von $T$, sowie die Existenz
    des adjungierten Operators $T^{\ast}$. 
\end{proof}

Damit sind wir nun in der Lage die folgende Proposition zu formulieren.

\begin{proposition}[Die funktorielle Abbildung {$\cross{\cdot}{H}$}] 
 \label{Proposition:FunktorielleAbbildungEmapstoEH}    
 Die Abbildung 
    \begin{align}
        \label{eq:FunktorCrossHmodAHBtomodAHBH}
        \cross{\cdot}{H} : \smod[\alg{A},H](\alg{B}) \to
          \smod[\cross{\alg{A}}{H}](\cross{\alg{B}}{H}) 
    \end{align}
ist ein Funktor. Auf den Objekten bildet der Funktor Bimoduln auf
deren Cross-Produktalgebren ab $\alg{E} \mapsto \cross{\alg{E}}{H}$
und auf den Morphismen die Bimodulmorphismen auf das Tensorprodukt der
Morphismen mit der Identit"at auf der \Name{Hopf}-$^\ast$-Algebra $T \mapsto
T \otimes \id_{\sss H}$. Die Einschr"ankung auf die Unterkategorien von $\smod[\alg{A},H](\alg{B})$ liefert die folgenden Funktoren\index{Funktor!Cross-Produktalgebra}
\begin{align}
        \label{eq:FunktorCrossHModAHBtoModAHBH}
      \cross{\cdot}{H} & : \sMod[\alg{A},H](\alg{B}) \to
          \sMod[\cross{\alg{A}}{H}](\cross{\alg{B}}{H}), \\
\label{eq:FunktorCrossHrepAHBtorepAHBH}
        \cross{\cdot}{H} & : \srep[\alg{A},H](\alg{B}) \to
          \srep[\cross{\alg{A}}{H}](\cross{\alg{B}}{H}), \\ 
\label{eq:FunktorCrossHRepAHBtoRepAHBH}
        \cross{\cdot}{H} & : \sRep[\alg{A},H](\alg{B}) \to
          \sRep[\cross{\alg{A}}{H}](\cross{\alg{B}}{H}). 
    \end{align}
\end{proposition}

Da wir "uber die \Name{Morita}-"Aquivalenz sowie das
\Name{Picard}-Gruppoid\index{Picard-Gruppoid@\Name{Picard}-Gruppoid!Cross-Produktalgebren}
von Cross-Produktalgebren reden wollen, 
m"ussen wir die Vertr"aglichkeit des in Proposition
\ref{Proposition:FunktorielleAbbildungEmapstoEH} beschriebenen
Funktors mit der \Name{Rieffel}-Induktion, das hei"st mit dem
Tensorieren von Bimoduln, pr"ufen.

\begin{proposition}
\label{Proposition:IeinsundIzwei}
    Sei nun $\bimod{C}{F}{B} \in \smod[\alg{B},H](\alg{C})$ und
    $\bimod{B}{E}{A} \in \smod[\alg{A},H](\alg{B})$, dann existieren
    die beiden folgenden kanonischen Isomorphismen.
    \begin{compactenum}
    \item Die Abbildung
        \begin{equation}
        \begin{aligned}
            \label{eq:IsomorphismusI1}
            I_{1}: & \bimodcross{C}{F}{B}
            \tensorhat[\cross{\alg{B}}{H}] \bimodcross{B}{E}{A} \to \bimod{C}{(F
              \tensorhat[\alg{B}] E)}{A} \otimes H \\ & (x \otimes g)
            \tensorhat[\cross{\alg{B}}{H}] (y \otimes h) \mapsto (x
            \tensorhat[\alg{B}] (g_{\sss (1)} \act y)) \otimes g_{\sss
              (2)} h
        \end{aligned}
        \end{equation}
ist ein kanonischer Isomorphismus von $^\ast$-Darstellungen von
$\cross{\alg{C}}{H}$ auf $\cross{\alg{A}}{H}$-Rechtsmoduln mit
innerem Produkt.
    \item Die Abbildung
       \begin{equation}
        \begin{aligned}
            \label{eq:IsomorphismusI2}
         I_{2}: & \cc{\cross{\alg{E}}{H}} \to \cross{\cc{\alg{E}}}{H},
         \quad \cc{x \otimes h} \mapsto h_{\sss (1)}^{\ast} \ccact
         \cc{x} \otimes h_{\sss (2)}^{\ast},  
        \end{aligned}
        \end{equation}
ist ein kanonischer Isomorphismus von 
$\cross{\alg{B}}{H}$-Rechtsdarstellungen auf
$\cross{\alg{A}}{H}$-Linksmoduln mit inneren Produkten. Das
Inverse ist explizit gegeben durch 
\begin{align}
    \label{eq:IsomorphismusI2invers}
    I_{2}^{-1}(\cc{x} \otimes h) = \cc{h_{\sss (1)}^{\ast} \act x
      \otimes h_{\sss (2)}^{\ast}}.
\end{align}
    \end{compactenum}
\end{proposition}

\begin{proof}
Es ist leicht zu sehen, da"s die Abbildung $I_{1}$ eine
wohldefinierte Bimodulabbildung "uber dem Tensorprodukt
$\tensor[\cross{\alg{B}}{H}]$ ist. F"ur die Isometrie rechnen wir
nach
\begin{align*}
   & \rSPcr[\sss \cross{(\alg{F} \tensorhat \alg{E})}{H}]{ (x
      \tensorhat[\alg{B}] (g_{\sss (1)} \act y)) \otimes g_{\sss (2)}
      h, (x' \tensorhat[\alg{B}] (g'_{\sss (1)} \act y')) \otimes
      g'_{\sss (2)} h'}{\cross{\alg{A}}{H}} \\ & = \left( (h_{\sss
          (1)}^{\ast} g_{\sss (2)}^{\ast}) \act \rSPcr[\sss \alg{F}
        \otimes \alg{E}]{x \tensor[\alg{B}] (g_{\sss (1)} \act y),x'
          \tensor[\alg{B}] (g'_{\sss (1)} \act y')}{\alg{A}} \right)
    \otimes h_{\sss (2)}^{\ast} g_{\sss (3)}^{\ast} g'_{\sss (2)} h'
    \\ & = \left( (h_{\sss (1)}^{\ast} g_{\sss (2)}^{\ast}) \act
        \rSPcr[\sss \alg{E}]{g_{\sss (1)} \act y, \rSPcr[\sss
          \alg{F}]{x,x'}{\alg{B}} \cdot (g'_{\sss (1)} \act y')}{\alg{A}}
    \right) \otimes h_{\sss (2)}^{\ast} g_{\sss (3)}^{\ast} g'_{\sss
      (2)} h' \\  & = \left( h_{\sss (1)}^{\ast} \act \rSPcr[\sss
        \alg{E}]{(S(g_{\sss (2)}^{\ast})^{\ast} g_{\sss (1)}) \act y,
          g_{\sss (3)}^{\ast} \act \left( \rSPcr[\sss
          \alg{F}]{x,x'}{\alg{B}} \cdot (g'_{\sss (1)} \act y')\right)}{\alg{A}}
    \right) \otimes h_{\sss (2)}^{\ast} g_{\sss (4)}^{\ast} g'_{\sss
      (2)} h' \\ & = \left( h_{\sss (1)}^{\ast} \act \rSPcr[\sss
        \alg{E}]{y, (g_{\sss (1)}^{\ast} \act \rSPcr[\sss
          \alg{F}]{x,x'}{\alg{B}}) \cdot (g_{\sss (2)}^{\ast}g'_{\sss
            (1)})\act y'}{\alg{A}} \right) \otimes h_{\sss (2)}^{\ast}
    g_{\sss (3)}^{\ast} g'_{\sss (2)} h' \\ & = \rSPcr[\sss
    \cross{\alg{E}}{H}]{y \otimes h, \left( (g_{\sss (1)}^{\ast} \act
            \rSPcr[\sss \alg{F}]{x,x'}{\alg{B}} )\cdot (g_{\sss
              (2)}^{\ast} g')_{\sss (1)} \act y' \right) \otimes (g_{\sss
              (2)}^{\ast} g')_{\sss (2)} h' }{\cross{\alg{A}}{H}} \\
          & = \rSPcr[\sss \cross{\alg{E}}{H}]{y \otimes h, \left(
                (g_{\sss (1)}^{\ast} \act \rSPcr[\sss
                  \alg{F}]{x,x'}{\alg{B}}) \otimes g_{\sss
                    (2)}^{\ast}g' \right)\cdot (y' \otimes h')
            }{\cross{\alg{A}}{H}} \\ & = \rSPcr[\sss
            \cross{\alg{E}}{H}]{y \otimes h,
              \rSPcr[\sss \cross{\alg{F}}{H}]{x \otimes g, x' \otimes
                g'}{\cross{\alg{B}}{H}} \cdot (y' \otimes h')}
            {\cross{\alg{A}}{H}} \\ & = \rSPcr[\sss (\cross{\alg{F}}{H})
            \tensorhat (\cross{\alg{E}}{H})]{(x \otimes g)
              \tensor[\alg{B}] (y \otimes h),(x' \otimes g')
              \tensor[\alg{B}] (y' \otimes h')}{\cross{\alg{A}}{H}}
\end{align*}
wodurch die Abbildung $I_{1}$ schon auf der Ebene des Tensorprodukts
$\tensor[\alg{B}]$ isometrisch wird, und nicht erst nach der
Quotientenbildung zu $\tensorhat[\alg{B}]$. Injektivit"at folgt, da
die Quotienten beider inneren Produkte nichtausgeartet sind. Die Surjektivit"at ist
offensichtlich, da $(x \otimes 1_{\sss H}) \tensor[\alg{B}] (y \otimes
h) \mapsto (x \tensor[\alg{B}] y) \otimes h$. Somit ist $I_{1}$ ein Isomorphismus. 
Im folgenden Sinn kann man die Abbildung $I_{1}$ als kanonisch
ansehen:

 Seien $S: \bimod{C}{F}{B} \to \bimod{C}{F^{\prime}}{B}$ und
$T:\bimod{B}{E}{A} \to \bimod{B}{E^{\prime}}{A}$ Morphismen in den
Kategorien $\smod[\alg{B},H](\alg{C})$ und
$\smod[\alg{A},H](\alg{B})$, dann ist $S \tensorhat[\alg{B}] T$ ein
Morphismus in $\smod[\alg{A},H](\alg{C})$ und $S \otimes \id_{\sss H} $
bzw.~$T \otimes \id_{\sss H}$ die korrespondierenden Morphismen in
$\smod[\cross{\alg{B}}{H}](\cross{\alg{C}}{H})$
bzw.~$\smod[\cross{\alg{A}}{H}](\cross{\alg{B}}{H})$, nach Lemma
\ref{Lemma:VerschraenkungsoperatorBimodulliefertVOaufCrossProdukten}.
Die Abbildung $I_{1}$ ist kompatibel mit den Morphismen und es gilt
\begin{align*}
    I_{1} \circ \left((S \otimes \id_{\sss H})
        \tensorhat[\cross{\alg{B}}{H}] (T \otimes \id_{\sss H})
    \right) = \left( (S \tensorhat[\alg{B}] T) \otimes \id_{\sss H} \right) \circ I_{1}.
\end{align*}

Der zweite Teil folgt aus der einfachen, jedoch l"anglichen Rechnung,
da"s $I_{2}$ ein Bimodulmorphismus ist und die notwendige
$\ring{C}$-Linearit"at besitzt. Die Isometrie folgt aus der folgenden
Rechnung:
\begin{align*}
    & \lSPcr[\sss \cross{\cc{\alg{E}}}{H}]{\cross{\alg{A}}{H}}{I_{2}(\cc{x
      \otimes g}), I_{2}(\cc{y \otimes h})} \\ & = \lSPcr[\sss
  \cross{\cc{\alg{E}}}{H}]{\cross{\alg{A}}{H}}{(g_{\sss (1)}^{\ast}
    \ccact \cc{x}) \otimes g_{\sss (2)}^{\ast}, (h_{\sss (1)}^{\ast}
    \ccact \cc{y}) \otimes h_{\sss (2)}^{\ast}} \\ & =  \lSPcr[\sss
  \cross{\cc{\alg{E}}}{H}]{\cross{\alg{A}}{H}}{\cc{S(g_{\sss
        (1)}^{\ast})^{\ast} \act x} \otimes g_{\sss (2)}^{\ast},\cc{S(h_{\sss
        (1)}^{\ast})^{\ast} \act y} \otimes h_{\sss (2)}^{\ast}} \\ &
  = \left( g_{\sss (3)}^{\ast} \act \lSPcr[\sss
      \alg{E}]{\alg{A}}{S^{-1}(g_{\sss (2)}^{\ast}) \ccact
        (\cc{S^{-1}(g_{\sss (1)}) \act x}), S^{-1}(h_{\sss (2)}^{\ast}) \ccact
        (\cc{S^{-1}(h_{\sss (1)}) \act y})} \right) \otimes g_{\sss
    (4)}^{\ast} h_{\sss (3)} \\ & =  \left( g_{\sss (3)}^{\ast} \act \lSPcr[\sss
      \alg{E}]{\alg{A}}{\cc{(g_{\sss (2)}S^{-1}(g_{\sss (1)})) \act
          x}, \cc{(h_{\sss (2)}S^{-1}(h_{\sss (1)})) \act y}} \right) \otimes g_{\sss
    (4)}^{\ast} h_{\sss (3)} \\ & = \left(g_{\sss (1)}^{\ast} \act
      \rSPcr[\sss \alg{E}]{x,y}{\alg{A}} \right) \otimes g_{\sss
    (2)}^{\ast} h  \\ & = \lSPcr[\sss
  \cc{\cross{\alg{E}}{H}}]{\cross{\alg{A}}{H}}{\cc{x \otimes g}, \cc{y
      \otimes h}}. 
\end{align*}
Analog zu $I_{1}$ zeigt man die Vertr"aglichkeit von $I_{2}$ mit den
Verschr"ankungsoperatoren. 
\end{proof}

Die Vertr"aglichkeit der Funktoren k"onnen wir im folgenden Diagramm
graphisch verdeutlichen.

\begin{korollar}
Das folgende Diagramm 

\begin{equation}
\bfig

\Square/>`>`>`>/[{\smod[\alg{B},H](\alg{C})} \times
{\smod[\alg{A},H](\alg{B})}`{\smod[\alg{A},H](\alg{C})}`{\smod[\cross{\alg{B}}{H}](\cross{\alg{C}}{H})}
\times
{\smod[\cross{\alg{A}}{H}](\cross{\alg{B}}{H})}`{\smod[\cross{\alg{A}}{H}](\cross{\alg{C}}{H})};{\tensorhat[\alg{B}]}`{(\cross{\cdot}{H})}
\times {(\cross{\cdot}{H})}`{\cross{\cdot}{H}}`{\tensorhat[\cross{\alg{B}}{H}]}] 

\efig
\end{equation}    
kommutiert im funktoriellen Sinne, d.~h.~bis auf die nat"urliche
Transformation $I_{1}$.
\end{korollar}

\section{Das \Name{Picard}-Gruppoid von Cross-Produktalgebren}
\label{sec:DasPicardGruppoidvonCrossProdukten}
\index{Picard-Gruppoid@\Name{Picard}-Gruppoid!Cross-Produktalgebren|(}
Nach der Betrachtung von allgemeinen $^\ast$-Darstellungen, wollen wir
uns nun den Cross-Pro\-dukt\-al\-ge\-bren widmen, die wir mit Hilfe des Funktors
$\cross{\cdot}{H}$ aus \Name{Morita}-"Aquivalenzbimoduln
erhalten. Diese werden sich ebenfalls als "Aquivalenzbimoduln
herausstellen, so da"s wir darauf aufbauend das
\Name{Picard}-Gruppoid von Cross-Pro\-dukt\-al\-ge\-bren betrachten k"onnen.

\begin{lemma}
    Sei $\bimod{B}{E}{A}$ ein $H$-"aquivarianter
    $^\ast$-\Name{Morita}-"Aqui\-va\-lenz\-bi\-mo\-dul
    (bzw.~starker \Name{Morita}-"Aqui\-va\-lenz\-bi\-mo\-dul), so
    ist $\bimodcross{B}{E}{A}$ mit den induzierten inneren Produkten
    ein $^\ast$-\Name{Morita}-"Aqui\-va\-lenz\-bi\-mo\-dul
    (bzw.~starker \Name{Morita}-"Aqui\-va\-lenz\-bi\-mo\-dul) f"ur
    die Algebren $\cross{\alg{B}}{H}$ und $\cross{\alg{A}}{H}$.
\end{lemma}

\begin{proof}
    Der Bimodul $\alg{E} \otimes H$ hat die beiden geerbten inneren Produkt
    $\lSPcr[\sss \alg{E} \otimes H]{\cross{\alg{B}}{H}}{\cdot,
      \cdot}$ und  $\rSPcr[\sss \alg{E} \otimes H]{\cdot,
      \cdot}{\cross{\alg{A}}{H}}$, die wie folgt definiert sind
    \begin{align*}
        \lSPcr[\sss \alg{E} \otimes H]{\cross{\alg{B}}{H}}{x
          \otimes h, x' \otimes {h^{\prime}}} & = \left( h_{\sss (2)} \act \lSP[\sss
        \alg{E}]{\alg{B}}{S^{-1}(h_{\sss (1)}) \act x, S^{-1}({h^{\prime}}_{\sss
            (1)} \act x')} \right) \otimes h_{\sss (3)} {h^{\prime}}_{\sss (2)} \\ & =
        \lSP[\sss \alg{E}]{\alg{B}}{x, S(h_{\sss (1)})^{\ast}
          S^{-1}({h^{\prime}}_{\sss (1)}) \act x'} \otimes h_{\sss (2)} {h^{\prime}}_{\sss
          (2)}^{\ast},
    \intertext{und analog das $\cross{\alg{A}}{H}$-wertige Produkt aus
    Lemma~\ref{Lemma:StrukturenCrossProduktalgebra}}
         \rSPcr[\sss \alg{E} \otimes H]{x \otimes h, x' \otimes
           {h^{\prime}}}{\cross{\alg{A}}{H}} & = \left( h_{\sss (1)}^{\ast} \act
         \rSP[\sss \alg{E}]{x,x'}{\alg{A}} \right) \otimes h_{\sss
           (2)}^{\ast} {h^{\prime}} \\ & = \rSP[\sss
         \alg{E}]{S(h_{\sss (1)}^{\ast})^{\ast} \act x, h_{\sss
             (2)}^{\ast} \act x'}{\alg{A}} \otimes h_{\sss (3)}^{\ast}
         h'.
    \end{align*}
Beide inneren Produkte sind mit der $(\cross{\alg{B}}{H},
\cross{\alg{A}}{H})$-Bimodulstruktur vertr"aglich und haben die
richtigen Linearit"at. Wir m"ussen nun zeigen, da"s die beiden
inneren Produkte miteinander vertr"aglich sind. Dies ist eine
einfache Konsequenz aus der Vertr"aglichkeit der inneren Produkte des
"Aquivalenzbimoduls $\bimod{B}{E}{A}$, denn es gilt
\begin{align*}
    \lSPcr[\sss \alg{E}\otimes H]{\cross{\alg{B}}{H}}{x\otimes g,
      y\otimes h} \cdot (z \otimes k) & = \left( \lSP[\sss
        \alg{E}]{\alg{B}}{x, S(g_{\sss (1)})^{\ast} S^{-1}(h_{\sss
            (1)}) \act y} \otimes g_{\sss (2)} h_{\sss (2)}^{\ast}
    \right) \cdot (z \otimes k) \\ & = \lSP[\sss \alg{E}]{\alg{B}}{x,
      S(g_{\sss (1)})^{\ast} S^{-1}(h_{\sss (1)}) \act y} \otimes
    g_{\sss (2)} h_{\sss (2)}^{\ast} \left( (g_{\sss (2)} h_{\sss
          (2)}^{\ast}) \act z \right) \otimes g_{\sss (3)} h_{\sss
      (3)}^{\ast} k \\ & = x \cdot \rSP[\sss \alg{E}]{(S(g_{\sss
        (1)})^{\ast}S^{-1}(h_{\sss (1)}))\act y, (g_{\sss (2)} h_{\sss
          (2)}^{\ast}) \act z}{\alg{A}}  \otimes g_{\sss (3)} h_{\sss
      (3)}^{\ast} k \\ & = (x \otimes g) \cdot \left( \rSP[\sss
        \alg{E}]{S^{-1}(h_{\sss (1)}) \act y, h_{\sss (2)}^\ast \act
          z}{\alg{A}} \otimes h_{\sss (3)}^{\ast} k \right) \\ & = (x
    \otimes g) \cdot \left( h_{\sss (1)}^{\ast} \act \rSP[\sss
        \alg{E}]{y,z}{\alg{A}} \otimes h_{\sss (2)}^{\ast} k \right)
    \\ & = (x \otimes g) \cdot \rSPcr[\sss \alg{E} \otimes H]{y \otimes
      h, z \otimes k}{\cross{\alg{A}}{H}}.  
\end{align*}

Der Ausartungsraum der beiden inneren Produkte auf
$\bimodcrosso{B}{E}{A}$ ist der gleiche, daher kann man auf eindeutige
Weise $\cross{\alg{E}}{H}$ definieren. Falls $\alg{B} \cdot \alg{E} =
\bimod{B}{E}{A} = \bimod{B}{E}{A} \cdot \alg{A}$, so ist auch $\cross{\alg{E}}{H}$
stark nichtausgeartet f"ur beide Modulstrukturen. Aus der Vollheit der
inneren Produkte von $\bimod{B}{E}{A}$ folgt automatisch die Vollheit
der inneren Produkte auf $\bimodcross{B}{E}{A}$: Sei $a = \sum_{i}
\rSP[\sss \alg{E}]{x_{i},y_{i}}{\alg{A}}$, dann ist offensichtlich
\begin{align*}
    a \otimes h = \sum_{i} \rSPcr[\sss \cross{\alg{E}}{H}]{x_{i} \otimes
      1_{\sss H}, y_{i} \otimes h}{\cross{\alg{A}}{H}},
\end{align*}
wodurch wir die Vollheit von $\rSPcr[\sss
\cross{\alg{E}}{H}]{\cdot,\cdot}{\cross{\alg{A}}{H}}$ gezeigt
haben. Analog zeigt man die Vollheit von $\lSPcr[\sss
\cross{\alg{E}}{H}]{\cross{\alg{B}}{H}}{\cdot,\cdot}$. Damit wird
$\bimodcross{B}{E}{A}$ zu einem
$^\ast$-\Name{Morita}-"Aquivalenzbimodul. Da die vollst"andige
Positivit"at der inneren Produkte durch die Konstruktion der
Cross-Produktalgebra erhalten bleibt folgt, da"s
$\bimodcross{B}{E}{A}$ ein starker \Name{Morita}-"Aquivalenzbimodul
ist, wenn es $\bimod{B}{E}{A}$ war. 
\end{proof}

Wenden wir nun Lemma
\ref{Lemma:VerschraenkungsoperatorBimodulliefertVOaufCrossProdukten}
auf "Aquivalenzbimoduln an, so erhalten wir das folgende Korollar.

\begin{korollar}
\label{Korollar:VerschraenkungsoperatorBimodulliefertVOaufCrossProdukten}
  Seien $\bimod{B}{E}{A}$ und $\bimod{B}{F}{A}$ isomorphe,
  $H$-"aquivariante ($^\ast$-, starke) \Name{Morita}-"Aquivalenzbimoduln, und
  $T:\bimod{B}{E}{A} \to \bimod{B}{F}{A}$ ein Isomorphismus, dann ist die
  Abbildung 
\begin{align}
T \otimes \id_{\sss H}: \bimodcross{B}{E}{A} \to \bimodcross{B}{F}{A}
\end{align}
 ein Isomorphismus von ($^\ast$-, starken) \Name{Morita}-"Aquivalenzbimoduln.     
\end{korollar}

\begin{lemma}
 \label{Lemma:AbbildungIdrei}
    Die Abbildung 
    \begin{align}
        \label{eq:Idrei}
        I_{3}: \cross{\bimod{A}{A}{A}}{H} \ni x \otimes h \mapsto x
        \otimes h \in \bimodcross{A}{A}{A} 
    \end{align}
ist ein Isomorphismus zwischen starken \Name{Morita}-"Aquivalenzbimoduln.
\end{lemma}

\begin{proof}
    Da wir davon ausgehen, da"s $\cross{\alg{A}}{H}$ nichtausgeartet
    ist, sind die inneren Produkte ohne Quotientenbildung
    nichtausgeartet. Es ist leicht zu zeigen, da"s die
    Bimodulstrukturen sowie die inneren Produkte auf $\alg{A} \otimes
    H$ in beiden Interpretationen zusammenfallen. 
\end{proof}

Wir sind nun in der Lage folgendes wichtiges Ergebnis zu formulieren.

\begin{satz}[Cross-Produkt induziert Gruppoidmorphismus]
  Die durch die \Name{Hopf}-$^\ast$-Algebra $H$ erhaltenen
  Cross-Produkte induzieren die Gruppoidmorphismen
    \begin{align}
        \label{eq:CrossProduktGruppenmorphismusstar}
        \cross{\cdot}{H}: \starPicH \to \starPic 
\intertext{und}
        \label{eq:CrossProduktGruppenmorphismusstr}
        \cross{\cdot}{H}: \strPicH \to \strPic.
    \end{align}
Dabei werden die Einselemente $\bimod{A}{A}{A} \in \Obj(\starPicH)$ auf die Cross-Produktalgebren
$\cross{\alg{A}}{H}$ abgebildet und Pfeile $[\bimod{B}{E}{A}]
\in \Morph(\starPicH)$ auf die Pfeile $[\bimodcross{B}{E}{A}] \in
\starPic(\cross{\alg{B}}{H},\cross{\alg{A}}{H})$. Analoges gilt f"ur
den starken Fall. 
\end{satz}

\begin{proof}
    Der Beweis ist eine Folge der zuvor gezeigten Lemmata. Aufgrund
    des Korollars
    \ref{Korollar:VerschraenkungsoperatorBimodulliefertVOaufCrossProdukten} ist
 $\cross{\cdot}{H}$ auf den Isomorphieklassen wohldefiniert. Die
 Abbildung $I_{3}$ aus Lemma~\ref{Lemma:AbbildungIdrei} garantiert,
 da"s Einselemente auf Einselemente abgebildet werden.   
 Desweiteren sichert eine Erweiterung von Proposition
 \ref{Proposition:IeinsundIzwei}, da"s Tensorprodukte auf
 Tensorprodukte abgebildet werden -- inklusive der dadurch induzierten inneren
 Produkte. Im Fall von "Aquivalenzbimoduln 
 legt ein inneres Produkt das andere aufgrund der Kompatibilit"at
 fest. Weiter stellt $I_{2}$ aus Proposition~\ref{Proposition:IeinsundIzwei}
 sicher, da"s konjugiert komplexe Bimoduln auf konjugiert komplexe
 Bimoduln abgebildet werden. Demnach werden Produkte und Inverse aus
 $\starPicH$ auf Produkte und Inverse in $\starPic$ abgebildet und
 analog f"ur $\strPicH$ nach $\strPic$. 
\end{proof}

Aus dem Satz ergeben sich nun zwei einfache Korollare.

\begin{korollar}[\Name{Morita}-"Aquivalenz von Cross-Produktalgebren]
 \label{Korollar:MoritaAequivalenzCrossProduktalgebren}
\index{Morita-Aequivalenz@\Name{Morita}-\"Aquivalenz!Cross-Produktalgebren@von Cross-Produktalgebren}
    Seien $\alg{A}$ und $\alg{B}$ zwei $H$-"aquivariante stark
    \Name{Morita}-"aquivalente Algebren, dann sind
    $\cross{\alg{A}}{H}$ und $\cross{\alg{B}}{H}$ stark
    \Name{Morita}-"aquivalente Cross-Produktalgebren. Analoges gilt,
    wenn man stark durch $^\ast$- ersetzt.
\end{korollar}

\begin{korollar}[Gruppenhomomorphismus]
 \label{Korollar:GruppenhomomorphismusPicHAnachPicAxH}
Die Abbildungen
\begin{align}
    \label{eq:GruppenhomomorphismusPicHstarAnachPicstarAxH}
   \starPicH(\alg{A}) \to \starPic(\cross{\alg{A}}{H}) \quad
   \text{und} \quad \strPicH(\alg{A}) \to \strPic(\cross{\alg{A}}{H}) 
\end{align}
sind Gruppenhomomorphismen.
\end{korollar}

\section{Das Beispiel $\strPicH(\ring{C}) \to \strPic(H)$}

Wir wollen die in diesem Kapitel erarbeiteten Techniken an dem
einfachen, jedoch interessanten Beispiel $\strPicH(\ring{C}) \to
\strPic(H)$ veranschaulichen. Das hei"st $\alg{A}=\ring{C}$
ausgestattet mit der trivialen Wirkung der \Name{Hopf}-$^\ast$-Algebra
$H$. In diesem Fall vereinfacht sich die Cross-Produktalgebra, da
$\cross{\ring{C}}{H} = H$ ist.

\begin{lemma}
\label{Lemma:CharakterHopfAlgebra}
Sei nun $\chi \in \GR{GL}{H,\ring{C}}$, dann gilt
\begin{compactenum}
\item $\chi^{-1}(h)= \chi(S^{-1}(h)) = \chi(S(h))$.
\item $\chi \in \GR{U}{H,\ring{C}}$ genau dann wenn $\chi(h^{\ast}) =
    \cc{\chi(h)}$.
\item Die Abbildung $\Phi^{\sss \chi}(h) :=\chi(S(h_{\sss (1)}))
    h_{\sss(2)}$ definiert einen Automorphismus $\Phi^{\sss \chi} \in
    \Aut(H)$ und $\Phi^{\sss \chi} \in \starAut(H)$ falls $\chi \in
    \GR{U}{H,\ring{C}}$. 
\item Die Abbildung 
\begin{align}
\label{eq:GruppenhomomorphismusGLHCtoAutH} 
\GR{GL}{H,\ring{C}} \ni \chi \mapsto \Phi^{\sss \chi} \in \Aut(H)
\end{align}
 ist ein injektiver Gruppenhomomorphismus.
\item Die Abbildung $\Phi^{\sss \chi}$ ist genau dann ein innerer
    Automorphimus, wenn $\chi=\msf{e}$ ist. 
\end{compactenum}
\end{lemma}

\begin{proof}
Offensichtlich ist $\chi^{-1}(h) = \chi(S^{-1}(h))$ durch Gleichung
\eqref{eq:InversesGLHA} wohldefiniert, wenn die Wirkung trivial ist. Desweiteren
definiert $\chi(S(h))$ ein Inverses bez"uglich des
Konvolutionsprodukts und aufgrund der Eindeutigkeit ist $\chi(S(h)) =
\chi^{-1}(h)$. Den zweiten Teil zeigen wir, indem wir $\cc{\chi}(h):=
\cc{\chi(h^{\ast})}$ definieren. Nutzen wir {\it iv.)} und die
Unitarit"atsbedingung (vgl.~Definition
\ref{Definition:GLHAundUHA} {\it iv.)}) f"ur
$\chi$, dann sehen wir, da"s 
$\cc{\chi}$ ein Inverses zu $\chi^{-1}$ bez"uglich des
Konvolutionsprodukts ist und damit gleich $\chi$ sein mu"s. Die
umgekehrte Richtung ist trivial. F"ur Teil {\it iii.)} und {\it iv.)}
zeigen wir schnell, da"s $\Phi^{\sss\chi}$ ein Homomorphismus ist,
d.~h.~es gilt $\Phi^{\sss \msf{e}}=\id$ und $\Phi^{\sss \chi}(gh) = \Phi^{\sss
  \chi}(g) \Phi^{\sss \chi}(h)$. Da $\chi \in \GR{U}{H, \ring{C}}$ ist
$\Phi^{\sss \chi}(h^{\ast}) = \Phi^{\sss \chi}(h)^{\ast}$. Desweiteren
rechnen wir nach, da"s $\Phi^{\sss \chi}\circ \Phi^{\sss \tilde{\chi}} = \Phi^{\sss \chi \ast
\tilde{\chi}}$, woraus die Bijektivit"at von $\Phi^{\sss \chi}$
folgt, und Gleichung \eqref{eq:GruppenhomomorphismusGLHCtoAutH} ist ein
Gruppenhomomorphismus. F"ur die Injektivit"at sei $\Phi^{\sss
  \chi}(h)=h$. Wendet man nun $\varepsilon$ darauf an, so ergibt sich
sofort $\chi(S(h)) = \varepsilon(h)$ woraus $\chi= \msf{e}$ folgt.
\end{proof}

    Lemma~\ref{Lemma:CharakterHopfAlgebra} ist eine Verallgemeinerung
    der bekannten Konstruktion von Automorphismen der Gruppenalgebra
    $\ring{C}[G]$ aus den Charakteren der Gruppe $G$. Die Gruppe
    $\GR{U}{H,\ring{C}}$ gibt immer einen nichttrivialen Beitrag zur
    \Name{Picard}-Gruppe $\strPic(H)$. So folgt aus \citep[Gleichung
    (2.4)]{bursztyn.waldmann:2004a}, da"s es den injektiven
    Gruppenhomomorphismus gibt:
    \begin{align}
        \label{eq:InjektiverGruppenhomomorphismusXtolPhiChi}
        \GR{U}{H,\ring{C}} \ni \chi \mapsto \ell(\Phi^{\sss \chi}) \in
        \strPic(H).
    \end{align}

Desweiteren ist klar, da"s die Cross-Produktalgebra
$\cross{\ring{C}}{H}$ aufgrund der kanonischen Identifikation 
\begin{align*}
    \cross{\ring{C}}{H} \ni z \otimes h \mapsto zh \in H
\end{align*}

nur die \Name{Hopf}-Algebra selbst $H$ ist. Der Gruppoidmorphismus
$\cross{\cdot}{H}$ liefert damit einen Gruppenhomomorphismus 
\begin{align}
    \label{eq:GruppenhomomorphismusstrPicHCtostrPicH}
    \strPicH(\ring{C}) \to \strPic(H).
\end{align}

Da das Zentrum von $\ring{C}$ trivial ist, gilt $\GR{U}{H,\ring{C}} =
\GRn{U}{H,\ring{C}}$ (Vergleiche hierzu Bemerkung
\ref{Bemerkung:TrivialesZentrumGleichheitderGruppen}) und ferner ist
$\GRn{U}{H,\ring{C}}$ eine Untergruppe von $\strPicH(\ring{C})$, so
da"s wir folgende Proposition formulieren k"onnen.

\begin{proposition}
\label{Proposition:StrPicHCUHCundsoweiter}
    Das folgende Diagramm von Gruppenhomomorphismen kommutiert
\begin{equation}
\label{eq:StrPicHCUHCundsoweiter}
        \bfig
        \Square/^{ (}->`^{ (}->`>`>/%
        [\GR{U}{H, \ring{C}}%
        `\strPicH(\ring{C})%
        `\starAut(H)%
        `\strPic(H)%
        ;``\cdot\rtimes\!H`\ell]
        \efig
\end{equation}
wobei $\GR{U}{H, \ring{C}}$ als Untergruppe von $\strPic(H)$
verstanden wird. 
\end{proposition}

\begin{proof}
Sei nun $\chi \in \GR{U}{H,\ring{C}} = \GRn{U}{H,\ring{C}}$. Das Bild
von $\chi$ in $\strPicH(\ring{C})$ ist gegeben durch die
Isomorphieklassen der trivialen Bimoduln $\ring{C}$ mit den kanonischen
inneren Produkten und $H$-Wirkungen $h \neactt[\chi] z =
\chi(h_{\sss (1)}) h_{\sss (2)} \act z = \chi (h) z$. Wir bezeichnen
diesen Bimodul mit $\ring{C}^{\sss \chi}$. Dann bilden wir
$[\ring{C}^{\sss \chi}]$ auf $[\cross{\ring{C}^{\sss \chi}}{H}]$ ab,
wobei $\cross{\ring{C}^{\sss \chi}}{H} \cong H$ als
$\ring{C}$-Modul. Die $H$-Modulstruktur ist gegeben durch $g \cdot h 
= \chi(g_{\sss (1)})g_{\sss (2)} h = \Phi^{\sss {\chi^{-1}}} (g) h = g
\cdot_{\sss {\Phi^{\chi}}} h$, und die kanonische
$H$-Rechtsmodulstruktur. Das linkslineare innere Produkt ist gegeben
durch $\Phi^{\sss \chi}(gh^{\ast})$, das rechtslineare ist das
kanonische. Damit ist $\cross{\ring{C}^{\sss \chi}}{H}$ isomorph zu
$\bimodo{\Phi^{\sss \chi}(H)}{H}{H}$, dessen Klasse in $\strPic(H)$
ist nun $\ell(\Phi^{\sss \chi})$. Damit ist die Kommutativit"at von
Diagramm \eqref{eq:StrPicHCUHCundsoweiter} gezeigt.
\end{proof}

Sei nun $\ring{C}$ ein algebraisch abgeschlossener
K"orper\index{Koerper@K\"orper!algebraisch abgeschlossener}, was genau
dann der Fall ist, wenn $\ring{R}$ ein reell abgeschlossener K"orper
ist \citep[Section 5.1]{jacobson:1985a}, so k"onnen wir das Diagramm
\eqref{eq:StrPicHCUHCundsoweiter} genauer untersuchen.

\begin{korollar}
\label{Korollar:strPicCfallsCalgebraischabgeschlossen}
    Sei $\ring{C}$ ein algebraisch abgeschlossener K"orper, dann ist
    \begin{compactenum}
    \item $\strPic(\ring{C}) = \{\id\}$,
    \item $\strPicH(\ring{C}) = \GR{U}{H,\ring{C}} =
        \GRn{U}{H,\ring{C}}$,
    \item die Abbildung $\strPicH(\ring{C}) \to \strPic(H)$ injektiv
        und ihr Bild ist durch das Diagramm \eqref{eq:StrPicHCUHCundsoweiter} gegeben.
    \end{compactenum}
\end{korollar}

\begin{proof}
    Der einzige "Aquivalenzbimodul bis auf Isomorphie ist der
    eindimensionale Vektorraum $\ring{C}$ mit dem kanonischen, positiv
    definiten inneren Produkt $\rSP{z,w}{\ring{C}} = \cc{z} w$. In
    diesem Fall ist $\starPic(\ring{C}) = \field{Z}_{2}$. Der zweite
    Teil folgt unmittelbar aus Proposition
    \ref{Proposition:EinsToGRnUHAtoostrPicHAtoostrPicAexakt}. Der
    dritte Teil ist eine Folge aus 
    Proposition~\ref{Proposition:StrPicHCUHCundsoweiter}.
\end{proof}
\index{Cross-Produktalgebra|)}
\index{Algebra!Cross-Produktalgebra|)}
\index{Picard-Gruppoid@\Name{Picard}-Gruppoid!Cross-Produktalgebren|)}


%% file: sternprodukte.tex
\chapter{Sternprodukte}
\label{chapter:Sternprodukte}
\fancyhead[CO]{\slshape \nouppercase{\rightmark}} 
\fancyhead[CE]{\slshape \nouppercase{\leftmark}} 

\section{Das Quantisierungsproblem}
\label{sec:DasQuantisierunsproblem}

In diesem Abschnitt soll eine kurze Einf"uhrung in das Problem des
Quantisierens gegeben werden. Eine sehr ausf"uhrliche Beschreibung
findet sich in \citep{waldmann:2004a:script}.  
Unter Quantisieren verstehen wir den Proze"s von einer gegebenen
klassischen Mechanik\index{Mechanik!klassische} auf eine
Quantenmechanik\index{Quantenmechanik} zu schlie"sen. Dieses 
Unterfangen erweist sich im allgemeinen als beliebig schwierig und
keineswegs eindeutig.  
W"ahrend die Konzepte (und deren Formulierung) in der klassischen
Mechanik klar sind und prinzipiell keine Probleme bereiten, stellt es
sich als "au"serst schwierig -- wenn nicht sogar als unm"oglich
-- heraus, eine {\it a priori-Quantenmechanik} zu formulieren. Eine
Ausnahme bildet hier die {\it axiomatische
  Quantenfeldtheorie}\index{Quantenfeldtheorie!axiomatische} 
\citep{haag.kastler:1964a,streater.wightman:2000a}, mit
der man bisher allerdings kein Beispiel au"ser dem freien Teilchen
beschreiben kann. Daher versucht man ausgehend von einem bekannten
klassischen System eine Quantisierung anzugeben. Leider existiert kein
\glqq Quantisierungsfunktor\grqq, mit dessen Hilfe man zu der
klassischen Beschreibung eines physikalischen Systems eine {\em eindeutige}
quantenmechanische angeben k"onnte.  

Man beachte, da"s die Quantisierung deshalb notwendig ist, weil die
klassische Beschreibung diverser Ph"anomene in der Physik nur
unzureichend ist. Quantisierung ist somit kein Proze"s, den man dem
System aufzwingt, sondern die zu beschreibenden Ph"anomene sind
bereits quantenmechanischer Natur, und es ist die klassische
Beschreibung, die mangelhaft ist. Dies bedeutet allerdings nicht,
da"s die klassische Mechanik keine Berechtigung hat, vielmehr ist
es davon abh"angig auf welcher Skala wir messen.   
M"ochten wir die Bewegung eines Planeten um die Sonne beschreiben
oder die Bewegung eines Elektrons um einen Atomkern, so haben wir es (im
wesentlichen) in beiden F"allen mit einer Zentralkraft (dem
\Name{Kepler}-Problem)\index{Kepler-Problem@\Name{Kepler}-Problem} zu
tun, und die Probleme scheinen "ahnlicher 
Natur zu sein. Im Falle der Planeten liefert uns die klassische Mechanik,
d.~h.~die \Name{Newton}sche
Mechanik\index{Mechanik!Newtonsche@\Name{Newton}sche} sowie die
(Allgemeine) 
Relativit"atstheorie ein korrektes Ergebnis, das mit all unseren
Messungen (im Rahmen der verf"ugbaren Me"sgenauigkeit)
"ubereinstimmt, so da"s wir nie auf die Idee k"amen, das Modell
anzuzweifeln. Die klassische Mechanik ist in der Astrophysik eine sehr
gute N"aherung. Im Falle des Atomkerns hingegen versagt die
klassische Theorie, und wir brauchen eine andere (wie sich
herausstellen wird allgemeinere) Beschreibung, um z.~B.~das
Energiespektrum eines Wasserstoffatoms zu bestimmen. Hier spielt die
Naturkonstante $\hbar$, das \Name{Planck}sche
Wirkungsquantum\index{Plancksches Wirkungsquantum@\Name{Planck}sches Wirkungsquantum}, eine
wichtige Rolle. Es hat die Einheit einer {\it Wirkung}\index{Wirkung} und w"ahrend
im ersten Beispiel $\hbar$ verschwindend klein im Vergleich zur
Wirkung der Planten ist, macht sich die quantenmechanische Natur im
zweiten Beispiel deutlich bemerkbar.\\ 

Wir wollen somit eine quantenmechanische Theorie haben, die eine
Verallgemeinerung der klassischen Theorie darstellt. Daraus resultiert
die Forderung, da"s man von jeder 
quantenmechanischen Theorie auf eine eindeutige Weise zur"uck zu
einer klassischen Theorie kommen mu"s, d.~h.~den {\it klassischen
  Limes}\index{klassischer Limes} \glqq $\hbar \to 0$\grqq{} bilden k"onnen
mu"s. In der Physik bezeichnet man diesen 
Proze"s auch als {\em Quantenkorrekturen vernachl"assigen}. Diese
Sprechweise ist jedoch mi"sverst"andlich, da sie aus den obigen
Gr"unden ein falsches Bild von der Physik vermittelt.

Was es letztendlich bedeutet, von der 
quantenmechanischen Beschreibung zur klassischen zu gelangen, ist im
Rahmen der kanonischen Quantisierung\index{Quantisierung!kanonische}
ein nichttrivialer Proze"s und 
daher mit Vorsicht zu behandeln. Im Rahmen der
Deformationsquantisierung soll motiviert werden wie dieser
Grenz"ubergang zu verstehen ist. 

\section{Kanonische Quantisierung (auf dem $\field{R}^{2n}$)}
\label{sec:KanonischeQuantisierung}

\subsection{Axiomatische Betrachtung einer Quantenmechanik}

Es liegt nahe, uns die Gemeinsamkeiten und Unterschiede der
(nichtrelativistischen) klassischen \Name{Hamilton}schen Mechanik
\index{Mechanik!Hamiltonsche@\Name{Hamilton}sche}und
der Quantenmechanik nach den Ideen von \Name{Heisenberg} vor Augen zu
f"uhren.  
 
Das zentrale Objekt bei unserer Betrachtung wird die {\it Observable}
sein. Eine {Observable}\index{Observable} ist eine durch physikalische Experimente
me"sbare Kenngr"o"se, wie zum Beispiel Ort, Energie, Impuls
oder Drehimpuls. In der klassischen Theorie ist die {\it
  Observablenalgebra}\index{Observablenalgebra}
\index{Algebra!Observablenalgebra} eine assoziative
\Name{Poisson}-$^\ast$-Algebra\index{Algebra!Poisson-Stern-Algebra@\Name{Poisson}-$^\ast$-Algebra}
$\alg{A}_{\sss \text{cl}} \subseteq \Cinf{M}$, wobei $(M,\Lambda)$
eine \Name{Poisson}-Mannigfaltigkeit\index{Mannigfaltigkeit!Poisson-Mannigfaltigkeit@\Name{Poisson}-Mannigfaltigkeit} ist. Die Motivation, eine
\Name{Poisson}-Mannigfaltigkeit\index{Poisson-Mannigfaltigkeit@\Name{Poisson}-Mannigfaltigkeit} zu betrachten, liegt darin begr"undet,
da"s die einzige \glqq "uberlebende\grqq{} Struktur die
\Name{Poisson}-Klammer\index{Poisson-Klammer@\Name{Poisson}-Klammer}
sein wird. Durch die Wahl der glatten 
Funktionen auf der \Name{Poisson}-Mannigfaltigkeit w"ahlen wir zum
einen eine m"oglichst \glqq angenehme\grqq{} Funktionenklasse, zum
anderen ist jedes reale physikalische System beliebig gut durch glatte
Funktionen zu approximieren. 

Die {\it reinen Zust"ande}\index{Zustand!reiner} sind Punkte im Phasenraum $M$, die {\it
  gemischten Zust"ande}\index{Zustand!gemischter} positive 
\Name{Borel}-Ma"se\index{Borel-Mass@\Name{Borel}-Ma{\ss}} $\mu$ auf
$M$. Ein Erwartungswert ist die   
Integrationen eines Elementes $f\in \alg{A}_{\sss \text{cl}}$ "uber das
\Name{Borel}-Ma"s. Die Zeitentwicklung einer Observablen $f$ ist
durch die
\Name{Hamilton}-Funktion\index{Hamilton-Funktion@\Name{Hamilton}-Funktion}
$H$ gegeben und stellt sich infinitesimal dar als 
\begin{align}
\frac{\de}{\de t} f(t) = \{f(t),H\}.
\end{align} 

Dem gegen"uber steht die {\em Quantenmechanik}\index{Quantenmechanik}. Die
Observablenalgebra $\alg{A}_{\sss \text{QM}}$ ist eine im allgemeinen
nichtkommutative, assoziative $^\ast$-Unteralgebra von (beschr"ankten
sowie unbeschr"ankten)
Operatoren auf $\hilbert{D} \subseteq \hilbert{H}$, wobei
$\hilbert{H}$ ein
\Name{Hilbert}-Raum\index{Hilbert-Raum@\Name{Hilbert}-Raum} ist und
$\hilbert{D}$ ein gemeinsamer, dichter Definitionsbereich. Die
Nichtkommutativit"at 
spiegelt sich insbesondere in der \Name{Heisenberg}schen
Unsch"arferelation\index{Heisenbergsche Unschaerferelation@\Name{Heisenberg}sche Unsch\"arferelation} wider 

\begin{align}
\label{eq:Vertauschungsrelation}
[P_{j}, Q^{i}] = -\im \hbar \delta_{j}^{i} \id_{\sss \hilbert{H}}.
\end{align} 

Die reinen Zust"ande werden durch "Aquivalenzklassen von
(nichtverschwindenden) Vektoren $\psi \in \hilbert{D}$
beschrieben. Dabei sind zwei Vektoren genau dann
"aquivalent,\index{Aequivalenz@\"Aquivalenz!Vektoren in
  Hilbert-Raum@von Vektoren in \Name{Hilbert}-Raum} wenn
sie auf dem gleichen Strahl liegen, d.~h.~$\psi \sim \psi^\prime
\Leftrightarrow \psi=c \psi^\prime$ mit $c\in \field{C}\backslash
\{0\}$. Somit sind reine Zust"ande eine Untermenge des projektiven
\Name{Hilbert}-Raums
$\field{P}\hilbert{H}$.\index{Hilbert-Raum@\Name{Hilbert}-Raum!projektiver}
Man beachte, da"s in 
der Natur nicht alle denkbaren Zust"ande realisiert werden, da die
Vektoren, die den Strahl $[\psi] \in \field{P}\hilbert{H}$
repr"asentieren, im Definitionsbereich $\hilbert{D}$ aller
physikalisch relevanten Observablen liegen m"ussen. Die gemischten
Zust"ande werden durch Dichtematrizen $\varrho$ \index{Dichtematrix} realisiert. Der
Erwartungswert eines Operators $A\in \alg{A}_{\sss \text{QM}}$ ist
definiert als $\tr(\varrho A)$, der Spur "uber die Dichtematrix multipliziert mit $A$.    
Sowohl klassisch wie quantenmechanisch fordern wir von einer
Observablen, da"s sie ein reelles Spektrum\index{Spektrum} besitzt. Daher beschr"anken
wir uns in der klassischen Mechanik auf reelle Observablen $f =
\cc{f}$ und in der quantenmechanischen Beschreibung auf
selbstadjungierte Operatoren $A =
A^{\ast}$.\index{Operator!selbstadjungierter} "Aquivalent k"onnten
wir Observable auch mittels {\em normaler} Operatoren beschreiben, da wir durch
Linearkombination aus zwei normalen Operatoren\index{Operator!normaler} einen
selbstadjungierten gewinnen k"onnen und auf die gleiche Art wieder
zur"uckkommen. 

Die Zeitentwicklung einer Observablen ist durch den {\it
  Kommutator} mit dem {\it \Name{Hamilton}-Operator}
\index{Hamilton-Operator@\Name{Hamilton}-Operator} $H$ gegeben
 
\begin{align}
\frac{\de}{\de t} A(t) = \frac{1}{\im \hbar} [A(t),H].
\end{align}

Da Symmetrien\index{Symmetrie} im weiteren dieser Arbeit eine wichtige Rolle spielen
werden, sei erw"ahnt, da"s diese in der klassischen Mechanik
mittels einer \Name{Lie}-Algebrawirkung von $\LieAlg{g}$ durch
Derivationen oder
\Name{Lie}-Gruppenwirkung durch Algebraautomorphismen auf von $G$ auf
$M$ realisiert werden, in der Quantenmechanik mittels $\LieAlg{g}$- bzw.~$G$-Darstellungen auf dem
\Name{Hilbert}-Raum $\hilbert{H}$.  

\subsection{Umsetzung der kanonischen Quantisierung im flachen Phasenraum} 

Wie bereits erw"ahnt, wollen wir uns insbesondere mit der
Observablenalgebra auseinandersetzten und funktional-analytische
Aspekte vorerst au"ser acht lassen. Eine offensichtliche Frage ist,
wie die Umsetzung der Gleichung~\eqref{eq:Vertauschungsrelation}, der
\Name{Heisenberg}schen Unsch"arferelation,
geschieht. Dies wollen wir f"ur den einfachsten und
sehr gut verstandenen Fall diskutieren. 

Dazu betrachten wir ein physikalisches System, dessen Phasenraum der
$\field{R}^{2n}\cong T^{\ast}\field{R}^{n}$ mit der kanonischen
symplektischen Form $\omega_{\sss 0}=-\sum_{i} \de p^{i} \wedge \de q_{i}$ und
der daraus resultierenden
\Name{Poisson}-Klammer\index{Poisson-Klammer@\Name{Poisson}-Klammer}
ist.   

Die kanonische Quantisierung geschieht, indem wir den
Koordinatenfunktionen Differentialoperatoren\index{Differentialoperatoren} auf dem
Pr"a-\Name{Hilbert}-Raum $(\Cinfc{\field{R}^{n}}, \SP{\cdot,\cdot})$ zuordnen.
Dieser sind die Funktionen mit kompaktem Tr"ager auf $\field{R}^{n}$
mit dem \Name{Hermite}sche Produkt, dem $L^{2}$-Skalarprodukt
bez"uglich des \Name{Lebesgue}-Ma"ses, das wie folgt gegeben ist:\index{Lebesgue-Mass@\Name{Lebesgue}-Ma{\ss}}

\begin{align}
\label{eq:SkalarproduktR2n}
\SP{\phi,\psi} = \int_{\field{R}^{n}} \cc{\phi(q)} \psi(q)\, \de^{n} q.
\end{align}

 Sei nun $\psi \in \Cinfc{\field{R}^{n}}$, dann definieren wir folgende Abbildungen

\begin{equation}
\label{eq:OperatorzuordnungimR2n}
\begin{aligned}
& q^{i} \mapsto Q^{i}: \psi \mapsto \left( q \mapsto
    (Q^{i}\psi)(q)= q^{i} \psi(q) \right), \\ 
& p_{i} \mapsto P_{i}: \psi \mapsto \left( q \mapsto
    (P_{i}\psi)(q) = \tfrac{\hbar}{\im} (\tfrac{\del}{\del q^{i}})
    \psi(q) \right).  
\end{aligned}
\end{equation}   

Mit dieser Zuordnung werden wir Gleichung~\eqref{eq:Vertauschungsrelation}
gerecht. Allerdings sieht man an dieser Stelle zweierlei. Zum
einen ist aufgrund der Nichtkommutativit"at der Operatoralgebra
keine eindeutige Zuordnung von polynomialen Funktionen von der
klassischen Seite auf Operatoren m"oglich. Dies liefert uns sp"ater
den Begriff der {\it Ordnungsvorschrift}\index{Ordnungsvorschrift}. Zum anderen sehen wir
bereits, da"s die Definition eines klassischen Limes \glqq$\hbar
\to 0$\grqq{} auf der Operatoralgebra zu naiv w"are, da jeder
Impulsoperator auf die Null abgebildet w"urde.

Die in Gleichung~\eqref {eq:OperatorzuordnungimR2n} definierten Operatoren sind
symmetrisch bez"uglich des Skalarproduktes in Gleichung~\eqref{eq:SkalarproduktR2n},
d.~h.~$\SP{\phi,Q^{i}\psi}=\SP{Q^{i}\phi,\psi}$ und
$\SP{\phi,P_{j}\psi}=\SP{P_{j}\phi,\psi}$, ferner bilden sie
$\Cinfc{\field{R}^{n}}$ auf sich selbst ab.  

\subsection*{\Name{Weyl}-, (Anti-)Standard- und $\kappa$-Ordnung}

Nun m"ochten wir nicht nur eine Zuordnung f"ur die
Koordinatenfunktionen haben, sondern f"ur alle Polynome auf dem
$\field{R}^{2n}$, die wir mit $\Pol(\field{R}^{2n})$ bezeichnen. Da
die Differentialoperatoren auf dem Pr"a-\Name{Hilbert}-Raum nicht
kommutieren, m"ussen wir zu\-s"atz\-lich eine {\it Ordnungsvorschrift}
angeben. Es ist beispielsweise v"ollig unklar, welcher Operator dem
Produkt $q^{i}p_{j}$ entspricht, schlie"slich w"aren alle
Linearkombinationen $\tau P_{j}Q^{i} + (1-\tau) Q^{i}P_{j}$ mit $\tau
\in \field{R}$ denkbar. Im folgenden werden wir eine sehr einfache
Ordnungsvorschrift\index{Ordnungsvorschrift} \index{Ordnungsvorschrift!Standardordnung} angeben, die {\it Standardordnung}, und aus dieser
weitere entwickeln, sowie erste {\it Sternprodukte} erhalten. Bei der
Standardordnung schreiben wir die Impulsoperatoren nach rechts,
d.~h.~wir definieren f"ur Polynome die $\field{C}$-lineare und
injektive Abbildung $\ordstd$, bei der wir zuerst alle Impulse nach
rechts schreiben und dann gem"a"s Gleichung
\eqref{eq:OperatorzuordnungimR2n} ersetzen. Im weiteren sei
$\DiffOp(\field{R}^{n})$ die {\em Algebra der Differentialoperatoren mit
glatten Koeffizientenfunktionen} auf $\field{R}^{n}$ und
$\DiffOpPol(\field{R}^{n})$ sei die {\em Algebra der Differentialoperatoren mit
polynomialen Koeffizientenfunktionen} auf $\field{R}^{n}$. 

\begin{dlemma}[Die Standardordnung $\ordstd$ und die Symbolabbildung $\symbstd$] 
Die {\em Standardordnung} ist die Abbildung  
\begin{equation}
\begin{aligned}
& \ordstd: \Pol(\field{R}^{2n})  \to \DiffOp(\field{R}^{n}), \\
& q^{i_{1}} \cdots q^{i_{r}}p_{i_{1}} \cdots p_{i_{s}}  \mapsto
\left(\frac{\hbar}{\im}\right)^n q^{i_{1}} \cdots q^{i_{r}}
\frac{\del^n}{\del q^{i_{1}} \cdots \del q^{i_{s} } } \quad \text{und}
\quad 1 \mapsto 1. 
\end{aligned}
\end{equation}

Die Umkehrabbildung $\symbstd$ ist gegeben durch 

\begin{equation}
\begin{aligned}
& \symbstd:  \DiffOpPol(\field{R}^{n}) \ni D \mapsto \eu^{-\frac{\im}{\hbar}pq} D \left(
    \eu^{\frac{\im}{\hbar}pq} \right) \in  \Pol(\field{R}^{2n}), 
\end{aligned}
\end{equation}
und wird als {\em standardgeordnete Symbolabbildung}\index{Symbolabbildung!standardgeordnete} bezeichnet. 
\end{dlemma}

Wir k"onnen nun f"ur eine Funktion $f \in \Pol (\field{R}^{2n}) \subseteq
\Cinf{\field{R}^{2n}}$ schreiben 

\begin{align}
\label{eq:StandardOrdnungOperator}
\ordstd(f) = \sum_{n=0}^\infty \frac{1}{n!}
\left(\frac{\hbar}{\im}\right)^n \sum_{i_1...i_n} \frac{\del^n f}{\del
  p_1 \cdots \del p_n}\bigg|_{p=0}  \frac{\del^n}{\del q^{i_1} \cdots
  \del q^{i_n}}. 
\end{align}
 
Da wir uns bisher auf $\Pol (\field{R}^{2n})$ beschr"ankt haben,
liefert die Summe in Gleichung~\eqref{eq:StandardOrdnungOperator} nur endlich
viele Terme. Fassen wir den $\field{R}^{2n}$ als Kotangentialb"undel
$T^{\ast}\field{R}^{n}$ des Konfigurationraums $\field{R}^{n}$ auf, so
bricht die Reihe auch dann ab, wenn wir uns auf Polynome in den
Impulsen beschr"anken und {\it glatte} Funktionen auf der Basis
$\field{R}^{n}$, d.~h.~in den Orten zulassen. Man kann leicht
zeigen, da"s $\ordstd \circ \symbstd = \id_{\sss 
  \DiffOp(\field{R}^n)}$ und $\symbstd \circ \ordstd = \id_{\sss
  \Pol(T^{\ast}\field{R}^{n})}$ ist. Dabei bezeichnet
$\Pol(T^{\ast}\field{R}^{n})$ die in den Impulsen polynomialen
Funktionen und wir h"atten gezeigt, da"s $\ordstd$ eine Bijektion ist.

Die Standardordnung erweist sich physikalisch als nicht sehr
befriedigend, da reelle Polynome, d.~h.~observable Gr"o"sen, im
allgemeinen nicht auf symmetrische Operatoren abgebildet werden. Daher wollen
wir weitere Ordnungsvorschriften angeben, unter anderem auch eine, bei
der ebendies der Fall sein wird.    

Der bijektive \Name{Neumaier}-Operator $N_{\kappa}$, der im allgemeinen
f"ur beliebige Kotangentialb"undel $T^{\ast}Q$ definiert ist
\citep{neumaier:2001a}, stellt sich als sehr n"utzlich heraus. Mit seiner Hilfe werden wir
Ordnungen parametrisieren und so eine Familie von Ordnungen angeben k"onnen. Da
wir im weiteren noch des "ofteren auf den Operator zur"uckgreifen
werden, seien im folgenden einige Eigenschaften zusammengetragen.

\begin{dlemma}[\Name{Neumaier}-Operator f"ur $T^{\ast}\field{R}^{n}$]
\index{Neumaier-Operator@\Name{Neumaier}-Operator|textbf}
Der \Name{Neumaier}-Operator $N_{\kappa}: \Pol(T^{\ast}\field{R}^{n})
\to \Pol(T^{\ast}\field{R}^{n})$ mit  
\begin{align}
N_{\kappa} = \eu^{-\im \hbar \kappa \Delta} \qquad \mbox{und} \qquad
\Delta=\sum_k \frac{\del^{2}}{\del q^{k} \del p_{k}} 
\end{align}
ist f"ur alle $\kappa \in \field{R}$ eine lineare, bijektive
Abbildung mit den folgenden Eigenschaften: 

\begin{compactenum}
\item $N_{\kappa}^{-1}=N_{-\kappa}$ (Inverses), 
\item $(N_{\kappa})^{\alpha} = N_{\alpha\kappa}$ f"ur alle $\alpha
    \in \field{R}$ insbesondere ist $N_{0} = \id$, 
\item $N_{\kappa} N_{\kappa^\prime} = N_{\kappa + \kappa^\prime}$,
\item $\cc{N_{\kappa} f} = N_{-\kappa} \cc{f}$ f"ur alle $f \in
    \Pol(T^{\ast}\field{R}^{n})$. 
\end{compactenum} 
\end{dlemma} 

Nun kann man ausgehend von der Standardordnung eine Familie von
Ordnungsvorschriften, die $\kappa$-Ordnungen, angeben.  

\begin{definition}[Die $\kappa$-Ordnungsvorschrift]
Sei $\ordstd: \Pol(T^{\ast}\field{R}^{n}) \to \DiffOp(\field{R}^{n})$
die Standardordnungsvorschrift und  $N_{\kappa}$ der
\Name{Neumaier}-Operator f"ur $T^{\ast}\field{R}^{n}$. Man definiert die
  {\em $\kappa$-Ordnung} mittels
\index{Ordnungsvorschrift!kappa-Ordnung@$\kappa$-Ordnung}
\index{kappa-Ordnung@$\kappa$-Ordnung}
\begin{align}
\label{eq:KappaOrdnungimR2n}
\ordkappa: =\ordstd \circ N_{\kappa} : \Pol(T^{\ast} \field{R}^{n}) \to \DiffOp(\field{R}^{n}).
\end{align}
Desweiteren nennen wir 
\begin{align}
\label{eq:WeylOrdnungOperator}
\ordweyl :=\ordzahlweyl = \ordstd \circ \Nweyl
\end{align}
die {\em \Name{Weyl}-Ordnung}. 
\index{Ordnungsvorschrift!Weyl-Ordnung@\Name{Weyl}-Ordnung}
\index{Weyl-Ordnung@\Name{Weyl}-Ordnung}
\end{definition}
 
Offensichtlich ist $\varrho_{\sss 0}=\ordstd$ und f"ur alle
$\kappa$-Ordnungen gilt $\ordkappa(p_{i})=P_{i}$ und
$\ordkappa(q^{i})=Q^{i}$. Die \Name{Weyl}-Ordnung wird sich als besonders wichtig 
herausstellen, da diese reelle Funktionen auf symmetrische Operatoren
abbildet. Wir rechnen dazu f"ur den adjungierten Operator zu $\ordstd$ nach 
\begin{align}
\ordstd(f)^{\dagger} = \ordstd(\Nweyl[2] \cc{f}), 
\end{align}
der im allgemeinen von $\ordstd(\cc{f})$ verschieden ist. Im Gegensatz
dazu gilt f"ur die \Name{Weyl}-Ordnung 
\begin{align}
\ordweyl(f)^{\dagger}=\ordstd(\Nweyl f)^{\dagger}= \ordstd(\Nweyl[2]
\cc{\Nweyl f}) = \ordstd (\Nweyl \cc{f}) = \ordweyl (\cc{f}). 
\end{align}
Dazu haben wir ausgenutzt, da"s $\cc{\Nweyl f} = \Nweyl[-1] \cc{f}$
ist. Ferner sind wir nun in der Lage, eine explizite Form f"ur
$\ordweyl$ anzugeben: 
\begin{align}
\label{eq:WeylOrdnungOperator2}
\ordweyl(f)= \sum_{n=0}^\infty \frac{1}{n!}
\left(\frac{\hbar}{\im}\right)^n \sum_{i_1...i_n} \frac{\del^n (\Nweyl
  f)}{\del p_1 \cdots \del p_n}\bigg|_{p=0}  \frac{\del^n}{\del
  q^{i_1} \cdots \del q^{i_n}}. 
\end{align}
Dies entspricht der \Name{Weyl}schen Symmetrisierungsvorschrift, die
einem Polynom in Orts- und Impulsvariablen $q^k$ und $p_i$ das
vollst"andig symmetrisierte Polynom in den zugeh"origen Operatoren
$Q^k$ und $P_i$ zuordnet. 

\begin{beispiel}[Standardordnung und \Name{Weyl}-Ordnung]
Ein einfaches nichttriviales Beispiel ergibt sich f"ur das Polynom
$q^{2} p$.
\begin{align}
\ordweyl(q^{2}p) =\frac{1}{3} (Q^2P + QPQ + PQ^2) = - \im \hbar q^2
\frac{\del}{\del q} - \im \hbar q,  
\end{align}
wohingegen wir in der Standardordnung 
\begin{align}
\ordstd(q^{2}p)= Q^2 P = -\im \hbar q^2 \frac{\del}{\del q}
\end{align}
erhalten. 
\end{beispiel}

\subsection*{\Name{Wick}- und $\tilde{\kappa}$-Ordnung}

Neben den bisher genannten Ordnungsvorschriften wollen wir noch die z.~B.~in
der Quantenfeldtheorie wichtige \Name{Wick}-Ordnung $\ordwick$,\index{Ordnungsvorschrift!Wick-Ordnung@\Name{Wick}-Ordnung} sowie
die verallgemeinerte $\ordtkappa$-Ordnung diskutieren. Wir fassen
dazu den $\field{R}^{2n}$ als $\field{C}^{n}$ mit der kanonischen
komplexen Struktur auf und f"uhren komplexe Koordinaten
\begin{align}
\label{eq:KoordinatenaufCn}
z^{i} = q^{i} + \im p_{i} \quad \mbox{und} \quad \cc{z}^{i} = q^{i} -
\im p_{i}
\end{align}

ein, wobei darauf zu achten ist, da"s man Orte und Impulse auf eine
gemeinsame physikalische Einheit $[\text{Wirkung}]^{\frac{1}{2}}$
skalieren mu"s. Dies geschieht beispielsweise mit einer Massen- und
einer Frequenzskala. Die Quantisierung wird wie folgt realisiert. Wir
betrachten den \Name{Hilbert}-Raum der quadratintegrablen
Funktionen\index{Hilbert-Raum@\Name{Hilbert}-Raum@!quadratintegrabler Funktionen@der quadratintegrablen Funktionen} 
$L^{2}(\field{C}^{n}, \de \mu)$, wobei $\de \mu$ das
\Name{Gauss}-Ma"s  
\begin{align}
\de \mu (z,\cc{z}) = \frac{1}{(2\pi \hbar)^{n}} \eu^{-\frac{z\cc{z}}{2
    \hbar}} \de z \de \cc{z} 
\end{align}
ist. Das zugeh"orige
$L^{2}$-Skalarprodukt\index{L2-Skalarprodukt@$L^{2}$-Skalarprodukt}
ist gegeben durch 

\begin{align}
\SP{\phi, \psi} = \frac{1}{(2\pi \hbar)^{n}} \int_{\field{C}^{n}}
\cc{\phi(z,\cc{z})} \psi(z,\cc{z})  \eu^{-\frac{z\cc{z}}{2 \hbar}} \de
z \de \cc{z}. 
\end{align}

Bemerkenswert ist, da"s der Raum der quadratintegrablen,
antiholomorphen Funktionen einen abgeschlossen Untervektorraum
$\hilbert{H}$ von $L^{2}(\field{C}^{n}, \de \mu)$ bildet, also selbst
ein \Name{Hilbert}-Raum ist. Die Vektoren 
\begin{align}
e_{\sss k_1 \ldots k_n}(\cc{z}) = (2 \hbar)^{k_1 +\cdots + k_n}
(k_{1}! \cdots k_{n}!)^{-\frac{1}{2}} (\cc{z}^1)^{k_1} \cdots
(\cc{z}^n)^{k_n} 
\end{align}

bilden ein vollst"andiges Orthonormalsystem, d.~h.~eine
\Name{Hilbert}-Basis in $\hilbert{H}$.\index{Hilbert-Basis@\Name{Hilbert}-Basis} Diesen Raum
nennt man den  
{\it\Name{Bargmann}-\Name{Fock}-Raum}\index{Bargmann-Fock-Raum@\Name{Bargmann}-\Name{Fock}-Raum}. Das Ziel ist nun eine 
Quantisierung f"ur Polynome $\Pol (\field{C}^{n})$ durch Operatoren
auf $\hilbert{H}$ anzugeben. Dabei bezeichnen wir mit $\Pol
(\field{C}^{n})$ die Polynome in $z$ {\em und} $\cc{z}$ auf
$\field{C}^{n}$. Zu diesem Zweck definieren wir {\it Erzeuger}\index{Erzeuger} und
{\it Vernichter} durch \index{Vernichter}
\begin{align}
a_{i}^{\dagger} = \cc{z}^{i} \quad \mbox{sowie} \quad a_{j}=2\hbar
\frac{\del}{\del \cc{z}^{j}}  
\end{align}
f"ur alle $i,j = \{1,\ldots, n\}$. Auf einem geeigneten
Definitionsbereich $\hilbert{D}$ ergibt sich  

\begin{align}
\SP{\phi,a_{k} \psi} = \SP{a_{k}^{\dagger} \phi, \psi} \quad
\mbox{und} \quad \left[a_i, a_{j}^{\dagger} \right] = 2\hbar
\delta_{i\cc{j}} \id_{\sss \hilbert{D}}.  
\end{align}
  
Letztere Gleichung ist das Analogon zu Gleichung
\eqref{eq:Vertauschungsrelation}. 

\begin{definition}[Die \Name{Wick}-Ordnung]
\label{Definition:WickOrdnung}
\index{Ordnungsvorschrift!Wick-Ordnung@\Name{Wick}-Ordnung}
\index{Wick-Ordnung@\Name{Wick}-Ordnung}
Man definiert die \Name{Wick}-Ordnung, indem man die Erzeuger nach
links schreibt. F"ur Monome ergibt sich
\begin{equation}
\begin{aligned}
\ordwick: & \Pol(\field{C}^{n}) \to \DiffOp(\field{C}^{n}) \\  &
z^{i_1} \cdots z^{i_r} \cc{z}^{j_1} \cdots \cc{z}^{j_s} \mapsto
a_{j_1}^{\dagger} \cdots a_{j_s}^{\dagger} a_{i_1} \cdots a_{i_r} = (2\hbar)^r \cc{z}^{j_1} \cdots \cc{z}^{j_s}
\frac{\del^r}{\del \cc{z}^{i_1} \cdots \del \cc{z}^{i_r}}. \\
.
\end{aligned}
\end{equation}
\end{definition}

Diese Abbildung kann man auf $\Pol(\field{C}^{n})$ fortsetzen und
f"ur alle $f\in \Pol (\field{C}^{n})$ gilt  

\begin{align}
\ordwick(f) = \sum_{r,s=0}^{\infty} \frac{(2\hbar)^{r}}{r!s!}
\frac{\del^{r+s}f }{\del z^{i_1} \cdots \del z^{i_r} \del \cc{z}^{j_1}
  \cdots \del \cc{z}^{j_s}} (0)  \cc{z}^{j_1} \cdots \cc{z}^{j_s}
\frac{\del^{r}}{\del \cc{z}^{i_1} \cdots \del \cc{z}^{i_r}}. 
\end{align}

Analog zu den $\kappa$-Ordnungen, bei denen wir ausgehend von $\ordstd$
mittels $N_{\kappa}$ zu $\ordkappa$ kamen, wollen wir auch hier einen
Operator $S_{\tilde{\kappa}}$ einf"uhren, der uns verschiedene Ordnungen erzeugt:

\begin{align}
S_{\tilde{\kappa}} = \eu^{2 \hbar \tilde{\kappa} \tilde{\Delta}} \quad
\mbox{mit} \quad \tilde{\Delta}= \frac{\del^{2}}{\del z^{k} \del
  \cc{z}^{k}}. 
\end{align}  

\begin{definition}[Die $\tilde{\kappa}$-Ordnung]
\label{Definition:tkappaOrdnung}
\index{Ordnungsvorschrift!Kappatilde-Ordnung@$\tilde{\kappa}$-Ordnung}
\index{Kappatilde-Ordnung@$\tilde{\kappa}$-Ordnung}
Die {\em $\tilde{\kappa}$-Ordnung} definiert man durch
\begin{align}
\ordtkappa =\ordwick \circ S_{1-\tilde{\kappa}}: \Pol(\field{C}^{n})
\to \DiffOp(\field{C}^{n}),
\end{align}
so da"s man f"ur $\tilde{\kappa}=1$ die \Name{Wick}-Ordnung
erh"alt. Desweiteren liefert $\tilde{\kappa}=0$ eine \Name{Weyl}-Ordnung in
Erzeugern und Vernichtern und $\tilde{\kappa}=-1$ nennt man {\em
Anti-\Name{Wick}-Ordnung}. 
\end{definition}

Die Operatoren $\ordtkappa$ sind f"ur
alle $f = \cc{f}$ symmetrisch\index{Operator!symmetrisch}, und allgemein gilt 
\begin{align}
\ordtkappa(f)^{\dagger} = \ordtkappa(\cc{f}).
\end{align}

Eine Motivation sowie ein ausf"uhrlicher Beweis findet sich in
\citep{waldmann:2004a:script}. 

\section{Sternprodukte f"ur $\field{R}^{2n}$ und $\field{C}^n$}
\label{sec:SternproduktefuerdenR2n}

Nach den Vor"uberlegungen in Kapitel
\ref{sec:KanonischeQuantisierung} ist der Weg zu {\em
 Sternprodukten} nicht  mehr weit. Die Idee besteht darin, auf den
Funktionen des
Phasenraums ein assoziatives, nichtkommutatives Produkt einzuf"uhren,
bei dem der Kommutator in erster Ordnung das $\im \hbar$-fache einer
\Name{Poisson}-Klammer ist. Die ersten Sternprodukte gehen auf
\citet{moyal:1949a} und \citet{berezin:1975b} zur"uck und wurden von
\citet{bayen.et.al:1977a} systematisch betrachtet. Wir nutzen dazu
die oben definierten Ordnungsvorschriften
\eqref{eq:StandardOrdnungOperator} sowie
\eqref{eq:WeylOrdnungOperator} bzw.~\eqref{eq:WeylOrdnungOperator2}
und ziehen das Operatorprodukt, mittels der Umkehrabbildungen von
$\ordstd$ und $\ordweyl$, zweier glatter, in den Impulsen  polynomialer
Funktionen, wieder zur"uck auf die Funktionen auf dem Phasenraum.    
  
\subsection{Standardprodukt, \Name{Weyl}-Produkt und
  $\kappa$-geordnete Produkte} 
\label{sec:StdWeylundKappaProdukt}
Wir wollen nun explizite Formeln f"ur das standardgeordnete Sternprodukt sowie das
\Name{Weyl}-Produkt f"ur den $\field{R}^{2n}$ angeben. 
\index{Sternprodukt|(}
\begin{equation}
\label{eq:ExplizitStandardR2n}
\index{Sternprodukt!standardgeordneten Typ@vom standardgeordneten Typ}
\begin{aligned}
f \spstd g & := \symbstd (\ordstd(f) \ordstd(g)) \\ &=
\sum_{j=0}^{\infty} \frac{1}{j!} \left( \frac{\hbar}{\im} \right)^j
\sum_{i_1,\ldots, i_j} \frac{\del^j f}{\del p_{i_1} \cdots \del
  p_{i_j}} \frac{\del^j g}{\del q^{i_1} \cdots \del q^{i_j}},  
\end{aligned}
\end{equation}

\begin{equation}
\label{eq:ExplizitWeylR2n}
\index{Sternprodukt!Weyl-Typ@vom \Name{Weyl}-Typ}
\begin{aligned}
f \spweyl g & := \symbweyl (\ordweyl(f) \ordweyl(g)) \\ & =
\sum_{r=0}^{\infty} \frac{1}{r!} \left( \frac{\im \hbar}{2}
\right)^{r} \sum_{s=0}^{r} (-1)^{r-s} \sum_{i_1,\ldots,i_r}
\frac{\del^r f}{\del q^{i_1} \cdots \del q^{i_s} \del p_{i_{s+1}}
  \cdots \del p_{i_{r}}} \frac{\del^r g}{\del p_{i_1} \cdots \del
  p_{i_s} \del q^{i_{s+1}} \cdots \del q_{i_{r}}}.  
\end{aligned}
\end{equation}

Damit haben wir die ersten expliziten {\it Sternprodukte} f"ur den
flachen Phasenraum $\field{R}^{2n} \cong T^{\ast}\field{R}^{n}$
konstruiert. Die so definierten Produkte $\spweyl$ und $\spstd$ sind
wohldefiniert, da $\ordweyl$ und $\ordstd$ injektiv sind, und das Bild
unter Operatorverkn"upfung abgeschlossen ist. Die Assoziativit"at
ist durch die Konstruktion klar. Der {\em Sternkommutator} zweier
Funktionen zu einem Sternprodukt $\star$ ist definiert durch

\begin{align}
    \label{eq:SternkommutatorDefinition}
\index{Sternkommutator}
[f,g]_{\sss \star} := f\star g - g \star f.
\end{align}

und f"ur das standardgeordneten Produkt, als
auch das \Name{Weyl}-Produkt ergibt sich das $\im \hbar$-fache der
\Name{Poisson}-Klammer plus Terme in h"oheren Ordnungen, die
verschwinden, falls $f$ und $g$ linear in den Variablen sind:

\begin{gather}
\begin{aligned}
f \spstd g - g \spstd f  = \im \hbar \{f,g\} + \alg{O}(\hbar^2) \quad
\text{und} \quad f \spweyl g - g \spweyl f  = \im \hbar \{f,g\} + \alg{O}(\hbar^2). 
\end{aligned}
\end{gather}

Ferner ergibt sich f"ur das \Name{Weyl}-Produkt, da"s die komplexe
Konjugation ein antilinearer
Anti\-au\-to\-mor\-phis\-mus\index{Antiautomorphismus!antilinearer} des Sternproduktes
$\spweyl$ ist  

\begin{align}
\cc{f \spweyl g} = \cc{g} \spweyl \cc{f}.
\end{align}

Diese Eigenschaft bezeichnen wir als
{\em \Name{Hermite}zit"at des
  Sternproduktes}\index{Sternprodukt!Hermitesches@\Name{Hermite}sches},
und wir werden diese im folgenden explizit beweisen. Das standardgeordnete
Sternprodukt erf"ullt diese Eigenschaft erwartungsgem"a"s
nicht. Wir werden nun einen recht eleganten Formalismus einf"uhren
und sehen, da"s es kein anderes $\kappa$-geordnetes Produkt
au"ser dem \Name{Weyl}-Produkt gibt, das diese Eigenschaft
besitzt. Dazu definieren wir 
 
\begin{equation}
\begin{aligned}
& P,P^{\ast} : \Pol(T^{\ast} \field{R}^{n}) \otimes \Pol(T^{\ast}
\field{R}^{n}) \to \Pol(T^{\ast} \field{R}^{n})\otimes \Pol(T^{\ast}
\field{R}^{n}), \\ 
& P= \frac{\del}{\del q^{k}} \otimes \frac{\del}{\del p_{k}} \quad
\mbox{und} \quad P^{\ast}= \frac{\del}{\del p_{k}} \otimes
\frac{\del}{\del q^{k}}. 
\end{aligned}
\end{equation}

Ferner sei $\mu$ die punktweise, kommutative Multiplikation von
Funktionen, d.~h.~$\mu (f \otimes g) = fg$ (vergleiche hierzu die
Definition~\ref{Definition:Algebra}). So schreibt sich das
standardgeordnete Produkt als 

\begin{align}
f \spstd g = \mu \circ \eu^{\frac{\hbar}{\im} P^{\ast}} (f \otimes g),
\end{align}
und das \Name{Weyl}-Produkt kann man in der Form
\begin{align}
f \spweyl g = \mu \circ \eu^{\frac{\hbar}{2\im} (P-P^{\ast})} (f
\otimes g) 
\end{align}

angeben. Wir wollen nun eine explizite Formel f"ur
$\kappa$-geordnete Sternprodukte zeigen. Dazu ben"otigen wir jedoch
vorher noch einige Identit"aten. 

\begin{lemma}[Zusammenhang $\kappa$-Ordnung und Standardordnung]
F"ur $\kappa$-geordnete Sternprodukte gilt
\begin{align}
    \label{eq:ZusammenhangKappaOrdnungStandardordnung}
    f \spkappa g = N_{\kappa}^{-1} (N_{\kappa} f \spstd N_{\kappa} g).  
\end{align}
\end{lemma}
\begin{proof}
Wegen der Bijektivit"at von $N_{\kappa}$ und der Tatsache, da"s $\ordkappa=\ordstd \circ
N_{\kappa}$ folgt 
\begin{align*}
f \spkappa g & = \ordkappa^{-1} (\ordkappa (f) \ordkappa (g)) \\ & =
(\ordstd \circ N_{\kappa} )^{-1} ((\ordstd \circ N_{\kappa})(f)
(\ordstd \circ N_{\kappa}) (g)) \\ & = N_{\kappa} \ordstd^{-1}
(\ordstd (N_{\kappa} f) \ordstd (N_{\kappa} g)) \\ & = N_{\kappa}^{-1}
(N_{\kappa} f \spstd N_{\kappa} g).  
\end{align*}
\end{proof}  

Dies bedeutet, da"s wir mit Hilfe des \Name{Neumaier}-Operators\index{Neumaier-Operator@\Name{Neumaier}-Operator}
alle $\kappa$-Ordnungen aus der Standardordnung gewinnen
k"onnen. Da der \Name{Neumaier}-Operator $N_{\kappa} = \eu^{-\im \hbar \kappa
  \Delta}$ im wesentlichen ein exponenzierter \Name{Laplace}-Operator
ist, brauchen wir  
\begin{align}
\Delta \circ \mu =  \mu \circ (\Delta \otimes \id + P + P^{\ast} + \id
\otimes \Delta), 
\end{align}

was direkt aus der Definition f"ur Derivationen $D$\index{Derivation} folgt: $D\circ
\mu = (D \otimes \id + \id \otimes D)$. Damit ist 

\begin{align}
N_{\kappa} \circ \mu = \mu \circ \eu^{-\im \kappa \hbar (\Delta
  \otimes \id + P + P^{\ast} + \id \otimes \Delta)} 
\end{align}

und wir k"onnen folgende Proposition formulieren.

\begin{proposition}[Explizite Formel f"ur $\kappa$-geordnete
    Sternprodukte $\spkappa$] 
\index{Sternprodukt!kappa-geordnetes@$\kappa$-geordnetes}
Seien $f,g \in \Pol (T^{\ast}\field{R}^{n})$ und $\kappa \in \field{R}$, dann
ist das $\kappa$-geordnete Sternprodukt $\spkappa$ durch  
\begin{align}
f \spkappa g = \mu \circ \eu^{\im \hbar (\kappa P - (1-\kappa)
  P^{\ast})} (f\otimes g) 
\end{align}
gegeben. Man bezeichnet die Produkte f"ur $\kappa=0$, $\kappa=
\tfrac{1}{2}$ und $\kappa =1$ als standardgeordnetes Sternprodukt, \Name{Weyl}-Produkt
und antistandardgeordnetes Sternprodukt. 
\end{proposition}

\begin{lemma}[Komplexe Konjugation $\kappa$-geordneter Sternprodukte]
F"ur die komplexe Konjugation $\kappa$-geordneter Sternprodukte gilt
\begin{align}
\cc{f \spkappa g} = \cc{g} \star_{\sss 1-\kappa} \cc{f}. 
\end{align}
Insbesondere ist die komplexe Konjugation nur f"ur
$\kappa=\tfrac{1}{2}$, d.~h.~f"ur das \Name{Weyl}-Produkt, ein
Antiautomorphismus.  
\end{lemma}

\begin{proof}
Um obiges Lemma zu zeigen, f"uhren wir die Abbildung (komplexe
Konjugation) $C : \Pol(T^{\ast}\field{R}^{n}) \ni f \mapsto \cc{f} \in
\Pol(T^{\ast}\field{R}^{n})$ ein, die offensichtlich mit dem Produkt
$\mu$ vertauscht, d.~h.~$\mu\circ(C \otimes C) = C \circ
\mu$. F"ur den kanonischen {\em Vertauschungsoperator}\index{Vertauschungsoperator} (oder {\em
  Flip}) $\tau:  \Pol(T^{\ast}\field{R}^{n}) 
\otimes   \Pol(T^{\ast}\field{R}^{n}) \ni f \otimes g \mapsto g
\otimes f \in \Pol(T^{\ast}\field{R}^{n}) \otimes
\Pol(T^{\ast}\field{R}^{n})$ gilt $\mu \circ \tau = \mu$, da $\mu$
kommutativ ist und $\tau^2 = \id$. Wir k"onnen nun $P$ mittels
$P^{\ast}$ ausdr"ucken und umgekehrt: 
\begin{align*}
P^{\ast}= \tau \circ P \circ \tau \quad \mbox{und} \quad P = \tau
\circ P^{\ast} \circ \tau. 
\end{align*}
Ferner gilt $(C \otimes C) \circ P = P \circ (C\otimes C)$ und $(C
\otimes C) \circ P^{\ast} = P^{\ast} \circ (C\otimes C)$, so da"s
wir nun zeigen k"onnen 
\begin{equation}
\begin{aligned}
\cc{(f \spkappa g)}&  = C \circ \mu \circ \eu^{\im \hbar (\kappa P -
  (1-\kappa) P^{\ast})} (f\otimes g) \\ & = \mu \circ \tau \circ (C
\otimes C) \circ \eu^{\im \hbar (\kappa P - (1-\kappa) P^{\ast})}
(f\otimes g) \\ & = \mu \circ \tau \circ \eu^{-\im \hbar (\kappa P +
  (1-\kappa) P^{\ast})} \circ (C \otimes C)(f\otimes g) \\ &= \mu
\circ \eu^{-\im \hbar (\kappa P^{\ast} + (1-\kappa) P)} \circ \tau
(\cc{f} \otimes \cc{g}) \\ &= \mu \circ \eu^{\im \hbar ((1-\kappa) P -
  (1-(1-\kappa)) P^{\ast})} (\cc{g} \otimes \cc{f}) \\ &= \cc{g}
\star_{\sss 1-\kappa} \cc{f}. 
\end{aligned}
\end{equation}
\end{proof}

Zu guter Letzt wollen wir eine explizite Summenformel f"ur die
$\kappa$-geordneten Sternprodukte angeben, die sich zuweilen als
"au"serst n"utzlich erweist

\begin{align}
    \label{eq:ExplizipeFormelKappaGeordnetesSternprodukt}
f \spkappa g = \sum_{r=0}^{\infty} \sum_{j=0}^{r}
\left(\frac{\hbar}{\im} \right)^{r} \frac{\left(-\kappa \right)^{j}
  \left(1-\kappa\right)^{r-j}}{j!(r-j)!} 
\sum_{i_{1},\cdots,i_{r}} \frac{\del^{r}f}{\del q^{i_1} \cdots \del
  q^{i_j} \del 
      p_{i_{j+1}} \cdots p_{i_{r}}} 
    \frac{\del^{r}g}{\del p_{i_1} \cdots \del p_{i_j} \del q^{i_{j+1}}
      \cdots q^{i_{r}}}.  
\end{align}

\subsection{\Name{Wick}-Produkt und $\tilde{\kappa}$-geordnete Produkte}
\label{sec:WickProduktundKappatGeordneteProdukte}

Aus der $\tilde{\kappa}$-Ordnung bekommen wir analog zu den
$\kappa$-geordneten Sternprodukten ein zur"uckgezogenes Produkt.  
Dazu f"uhren wir die Symbolabbildung $\symbtkappa$ als
Umkehrabbildung zu $\ordtkappa$ ein, und wir k"onnen eine Formel f"ur
das $\tilde{\kappa}$-geordnete Produkt angeben. 

\begin{definition}[$\tilde{\kappa}$-geordnetes Sternprodukt]
\index{Sternprodukt!kappatilde-geordnetes@$\tilde{\kappa}$-geordnetes}
Das mittels 
\begin{align}
f \sptkappa g & = \symbtkappa (\ordtkappa(f)\circ \ordtkappa(g)) 
\end{align}
definierte Produkt auf $\Pol(\field{C}^{n})$ mit $\tilde{\kappa} \in
\field{R}$ hei"st $\tilde{\kappa}$-geordnetes Sternprodukt. Die
Produkte f"ur $\tilde{\kappa}=1$ und $\tilde{\kappa}=-1$ nennt man
{\em \Name{Wick}-Produkt} bzw.~{\em Anti-\Name{Wick}-Produkt}.  
\index{Sternprodukt!Wick-Produkt@\Name{Wick}-Produkt}
\index{Sternprodukt!Anti-Wick-Produkt@Anti-\Name{Wick}-Produkt}
\end{definition}

Wie bereits in Kapitel~\ref{sec:StdWeylundKappaProdukt} wollen wir
eine explizite und zudem elegantere Darstellung angeben. Dazu
definieren wir  

\begin{align}
Z,\cc{Z}: \Pol(\field{C}^{n}) \otimes \Pol(\field{C}^{n}) \to
\Pol(\field{C}^{n}) \otimes \Pol(\field{C}^{n}) 
\end{align}
explizit gegeben durch 

\begin{align}
Z = \frac{\del}{\del z^{k}} \otimes \frac{\del}{\del \cc{z}^{k}} \quad
\text{und} \quad \cc{Z}= \frac{\del}{\del \cc{z}^{k}} \otimes
\frac{\del}{\del {z}^{k}}. 
\end{align}

Analog zu den Ergebnissen f"ur $\kappa$-geordnete Sternprodukte
k"onnen wir folgende Proposition formulieren. 

\begin{proposition}[\Name{Wick}-Produkt und $\tilde{\kappa}$-geordnete
    Sternprodukte] 
Die $\tilde{\kappa}$-geordneten Sternprodukte auf
$\Pol(\field{C}^{n})$ haben die Form 
\begin{align}
f \sptkappa g = \mu \circ \eu^{(\tilde{\kappa}+1) \hbar Z +
  (\tilde{\kappa} -1) \hbar \cc{Z}} (f \otimes g), 
\end{align}
womit sich f"ur das \Name{Wick}-Produkt\index{Wick-Produkt@\Name{Wick}-Produkt} folgender Ausdruck ergibt
\begin{equation}
\begin{aligned}
f \spwick g  &= \mu \circ \eu^{2 \hbar Z} (f \otimes g)\\ &=
\sum_{r=0}^{\infty} \frac{(2 \hbar)^{r}}{r!} \sum_{k_1, \cdots, k_r}
\frac{\del^{r} f}{\del z^{k_1} \cdots \del z^{k_r}} \frac{\del^{r}
  g}{\del \cc{z}^{k_1} \cdots \del \cc{z}^{k_r}}. 
\end{aligned}
\end{equation}

F"ur alle $\tilde{\kappa}$-geordneten Produkte ist die komplexe
Konjugation ein antilinearer Antiautomorphismus, d.~h.~alle
$\tilde{\kappa}$-geordneten Produkte sind \Name{Hermite}sch 
\begin{align}
\cc{f\sptkappa g} = \cc{g} \sptkappa \cc{f}.
\end{align}
Da $P-P^{\ast} = \tfrac{2}{\im} (Z-\cc{Z})$ ist, erh"alt man f"ur
$\tilde{\kappa}=0$ das \Name{Weyl}-Produkt $\star_{\tilde{\kappa}=0} = \spweyl$. 
\end{proposition}

\begin{bemerkung}
Es ist interessant, die letzte Gleichung in Beziehung zur
kanonischen \Name{Poisson}-Klammer auf $\field{R}^{2n}$ zu setzen 
\begin{align}
\{ f,g\} = \mu \circ (P-P^{\ast}) (f\otimes g) = \frac{2}{\im} \mu
\circ (Z- \cc{Z}) (f\otimes g).  
\end{align}
Damit k"onnen wir das \Name{Weyl}-Produkt als eine \glqq
exponenzierte\grqq{} \Name{Poisson}-Klammer deuten.  
\end{bemerkung}

Wir wollen die Eigenschaften der bisher betrachteten Sternprodukte
zusammenfassen. Im n"achsten Kapitel werden wir diese
Eigenschaften nutzen, um Sternprodukte zu axiomatisieren.

\begin{bemerkung}[Eigenschaften der Sternprodukte $\spkappa$ und
    $\sptkappa$]
\label{Bemerkung:EiegnschaftenDerSternprodukteimFlachenFall}
~\vspace{-5mm}
\begin{compactenum}
\item Jedes der bisher behandelten Sternprodukte l"a"st sich in der Form
\begin{align}
\label{eq:FormaleEntwicklungErsteSternprodukte}
 f\star g= fg + \sum_{n=1}^{\infty} \hbar^{n} B_{n}(f,g) 
\end{align}    
schreiben, dabei sind alle $B_{n}$ Bidifferentialoperatoren\index{Bidifferentialoperator} der
Differentiationsordnung $n$. Sind die Funktionen $f,g$ polynomial in
den Koordinatenfunktionen, so bricht die unendliche Summe in Gleichung
\eqref{eq:FormaleEntwicklungErsteSternprodukte} nach endlich vielen
Termen ab. F"ur den Grenzfall $\hbar \to 0$ wird die klassische, kommutative
Algebra reproduziert.  

\item Die Sternprodukte $\spkappa$ und $\sptkappa$ sind assoziative
    Verkn"upfungen.  

\item $f \star g - g\star f = \im \hbar \{f,g\} + \alg{O}(\hbar^{2})$. Der
    schiefsymmetrische Teil des Kommutators modulo Termen der Ordnung
    $\alg{O}(\hbar^{2})$ ist proportional zur \Name{Poisson}-Klammer.
\item $f \star 1 = 1 \star f = f$. Das Einselement der undeformierten Algebra
    ist auch das Einselement bez"uglich des Sternprodukts. 
\item $ \cc{f \sptkappa g} = \cc{g} \sptkappa \cc{f}$ und $\cc{f
      \spkappa g} = \cc{g} \spemkappa \cc{f}$. Das hei"st f"ur
    alle $\tilde{\kappa}$-Sternprodukte ist die komplexe Konjugation
    ein antilinearer
    Antiautomorphismus\index{Antiautomorphismus!antilinearer} und
    ebenso f"ur das \Name{Weyl}-Produkt, das der Ordnung $\kappa=\tfrac{1}{2}$ entspricht.  
\end{compactenum} 
    
\end{bemerkung}

\section{Formale Deformationstheorie}
\label{sec:FormaleSternprodukte}

Aus dem im letzten Kapitel Erarbeiteten wollen wir nun eine Verallgemeinerung
f"ur Sternprodukte angeben. Dazu definieren wir zuerst
{\em Deformationen} einer assoziativen
Algebra\index{Algebra!assoziative} wie sie von 
\citet{gerstenhaber:1964a} vorgeschlagen wurden. In Kapitel
\ref{sec:DeformationProjektiverModul} werden wir im Rahmen der
\Name{Morita}-"Aquivalenz von deformierten Algebren\index{Algebra!deformierte} auch die
Deformation von Moduln ben"otigen, diese jedoch erst dort einf"uhren.
In Kapitel~\ref{sec:FormalesSternprodukt} werden wir darauf aufbauend eine
Definition f"ur {\em formale Sternprodukte} angeben. Ein formales
Sternprodukt wird eine assoziative Algebra sein, allerdings mit
weiteren Strukturen, die als Axiomatisierung der in Bemerkung
\ref{Bemerkung:EiegnschaftenDerSternprodukteimFlachenFall}
herausgestellten Ergebnisse zu
verstehen sind. Da von hier an {\em formale Potenzreihen}\index{Potenzreihe!formale} eine
zentrale Rolle spielen werden, verweisen wir auf Anhang
\ref{sec:FormalePotenzreihen}, in dem wir dieses algebraische Konzept
n"aher beschreiben. 

\subsection{Deformationen von Algebren}
\label{sec:DeformationenvonAlgebren}

\begin{definition}[Deformation einer assoziativen Algebra]
   \label{Definition:FormaleAlgebraDeformation} 
\index{Deformation!formale|textbf}
   Sei $\ring{C}=\ring{R}(\im)$ die komplexe Erweiterung eines
   geordneten Rings $\ring{R}$ und $\alg{A}$ eine assoziative Algebra
   "uber $\ring{C}$. Eine
   {\em formale Deformation} der Algebra $\alg{A}$ ist 
    ein $\ringf{C}$-lineares assoziatives Produkt $\star$ auf $\algf{A}$, so
    da"s $(\algf{A},\star)=\defalg{A}$ eine assoziative Algebra
    wird, und das Produkt von der Form
    \begin{align}
        \label{eq:FormaleAlgebraDeformation}
        a \star a' = aa' + \sum_{n=1}^{\infty} \lambda^{n} C_{n}(a,a') 
    \end{align}
ist. Dabei sind die $C_{n}:\alg{A} \times \alg{A} \to \alg{A}$
$\ring{C}$-bilineare Abbildungen, $a,a' \in \alg{A}$ und $\lambda$ ist
ein reeller, formaler Parameter. Ist die Algebra
$\alg{A}$ mit einem Einselement $1_{\sss \alg{A}}$ ausgestattet, so soll
$1_{\sss \alg{A}}$ auch das Einselement der deformierten Algebra
$\defalg{A}$ sein, d.~h.~$1_{\sss \alg{A}} \star a = a \star 1_{\sss
  \alg{A}} = a$ f"ur alle $a \in \defalg{A}$.    
\end{definition}

\begin{bemerkungen}[Assoziativit"at]
\label{Bemerkungen:AssoziativitaetDeformierteAlgebren}
~\vspace{-5mm}
\begin{compactenum}
 \item Die Assoziativit"at des deformierten Produktes $\star$
     liefert offensichtlich eine Bedingung an die
     $C_n$. Insbesondere ist der schiefsymmetrische Teil von $C_{1}$
     proportional zu einer \Name{Poisson}-Klammer, bzw.~definiert
     eine \Name{Poisson}-Klammer
     \begin{align*}
     \{a,a^{\prime} \}:= C_{1}(a,a^{\prime}) - C_{1}(a^{\prime},a).    
     \end{align*}
\item Aufgrund der Assoziativit"at des Produktes $\star$, ist es
    n"otig, die Abbildungen $C_n$ auf formale
    Potenzreihen zu erweitern, so da"s alle $C_n: \algf{A} \times
    \algf{A} \to \algf{A}$ zu $\ringf{C}$-bilinearen Abbildungen werden.
\end{compactenum}
\end{bemerkungen}

\begin{definition}["Aquivalenz zweier deformierter assoziativer Algebren]
 \label{Definition:AEquivalenzZweierDeformationen} 
\index{Aequivalenz@\"Aquivalenz!deformierten Algebren@von deformierten Algebren}  
 Man nennt zwei Deformationen
    $\defalg{A}_{1}=(\algf{A},\star_{\sss 1})$ und
    $\defalg{A}_{2}=(\algf{A},\star_{\sss 2})$ der Algebra $\alg{A}$
    {\em "aquivalent}, falls es $\ring{C}$-lineare Abbildungen $T_{r}:
    \alg{A} \to \alg{A}$ gibt, so da"s 
    \begin{align}
        T=\id + \sum_{r=1}^{\infty} \lambda^{r} T_{r} \,: \quad \defalg{A}_{1}
        \to \defalg{A}_{2} 
     \end{align}
ein Algebraisomorphismus ist. Die "Aquivalenzklasse einer
Deformation $\star$ bezeichnet man mit $[\star]$ und die Menge der
"Aquivalenzklassen von Deformationen einer Algebra $\alg{A}$ mit
$\Def(\alg{A})$. Desweiteren wollen wir die "Aquivalenzklassen der
Deformationen\index{Aequivalenz@\"Aquivalenz!Deformationen@von Deformationen} von $\alg{A}$ zu einer {\em festen} \Name{Poisson}-Struktur (vgl.~Bemerkung
\ref{Bemerkungen:AssoziativitaetDeformierteAlgebren} {\it i.)}) mit
$\Def(\alg{A}, \{\cdot, \cdot \})$ bezeichnen.     
\end{definition}

\begin{proposition}[Isomorphe Deformationen]
\label{Proposition:IsomorpheDeformationen}
\index{Deformation!isomorphe}
    Seien $\defalg{A}_{1} = (\algf{A},\star_{\sss 1})$ und
    $\defalg{A}_{2}= (\algf{A},\star_{\sss 2})$ zwei
    formale Deformationen der Algebra $\alg{A}$, dann sind diese genau
    dann isomorph, wenn es einen Automorphismus $\phi \in
    \Aut(\alg{A})$ gibt, so da"s $[\phi^{\ast}(\star_{\sss 2})] =
    [\star_{\sss 1}]$ ist. 
\end{proposition}

\begin{definition}[\Name{Hermite}sche Deformation]
\label{Definition:HermitescheDeformationeinerformalenAlgebra}
\index{Deformation!Hermitesche@\Name{Hermite}sche}
Sei $\alg{A}$ eine $^\ast$-Algebra\index{Algebra!Stern-Algebra@$^\ast$-Algebra}
(vgl.~Definition~\ref{Definition:SternAlgebra}). Man nennt eine formale Deformation $\defalg{A}=(\algf{A},\star)$ {\em \Name{Hermite}sch}, falls $(a \star
  a^{\prime})^{\ast} = {a^{\prime}}^{\ast} \star {a}^{\ast}$ f"ur alle $a,a^{\prime}
\in \defalg{A}$.
\end{definition}

\begin{lemma}[\Name{Hermite}sche "Aquivalenztransformation, {\citep{bursztyn.waldmann:2000b}}]
\label{Lemma:HermitescheAequivalenztransformation}
\index{Aequivalenztransformation@\"Aquivalenztransformation!Hermitesche@\Name{Hermite}sche}
    Sei $\field{Q} \subseteq \ring{R}$ und seien $\defalg{A}_{1}=(\algf{A},\star_{\sss 1})$ und
    $\defalg{A}_{2}=(\algf{A},\star_{\sss 2})$  zwei "aquivalente \Name{Hermite}sche
    Deformationen der Algebra $\alg{A}$, so existiert eine
    "Aquivalenztransformation $T: \defalg{A}_{1} \to \defalg{A}_{2}$,
    so da"s $T(a^{\ast}) = T(a)^{\ast}$ f"ur alle $a \in
    \defalg{A}_{1}$.  
\end{lemma}

Bevor wir uns in Kapitel~\ref{sec:FormalesSternprodukt} der
Definition von formalen Sternprodukten widmen, wollen wir den
{\em klassischen Limes} einer deformierten Algebra definieren. Dieser
ist im Bezug auf das Quantisieren von gro"ser Wichtigkeit.  
Anders als in der kanonischen Quantisierung, bei der die Definition eines
klassischen Limes nichttrivial ist, kann man im Rahmen von (formal)
deformierten Algebren sehr einfach den klassischen Limes definieren.
Insbesondere bei der im f"unften Kapitel betrachteten
\Name{Morita}-Theorie von deformierten Algebren wird der klassischen 
Limes-Abbildung eine wichtige Rolle zukommen.

\begin{definition}[Die klassische Limes-Abbildung $\cl$]
\label{Definition:KlassischeLimesAbbildung}
\index{klassischer Limes}
\index{Abbildung!klassische Limes-Abbildung}
   Sei $\defalg{A}$ eine deformierte Algebra nach Definition
  ~\ref{Definition:FormaleAlgebraDeformation}. Wir bezeichnen die
   Abbildung 
   \begin{align}
       \label{eq:KlassischerLimesAbbildung}
       \cl: \defalg{A} \ni  \sum_{n=0}\lambda^{n} a_{n} \mapsto a_{0}
       \in \alg{A},
   \end{align}
als die {\em klassische Limes-Abbildung}. Die Algebra $\alg{A}$
bezeichnet man dann als den {\em klassischen Limes}\index{klassischer Limes} der
  Algebra $\defalg{A}$.
\end{definition}

Die klassische Limes-Abbildung ist offensichtlich physikalisch
motiviert, da wir von der deformierten Theorie zur"uck zur
klassischen Theorie gelangen m"ochten. Dies entspricht in der Sprache des
Physikers einem "Ubergang \glqq$\hbar \to 0$\grqq{}. Dieser "Ubergang ist im
Rahmen formaler Potenzreihen offensichtlich wohldefiniert.

\subsection{Formale Sternprodukte}
\label{sec:FormalesSternprodukt}
\index{Sternprodukt!formales|(}
In der klassischen Physik kommt, wie wir in Kapitel
\ref{sec:KanonischeQuantisierung} gesehen haben, der
\Name{Poisson}-Struktur eine wichtige Rolle zu. Dies in der
Formulierung von Sternprodukten zu implementieren, war bereits bei
der Konstruktion der Sternprodukte auf einem flachen Phasenraum in Kapitel
\ref{sec:SternproduktefuerdenR2n} von fundamentaler Wichtigkeit. Die
Assoziativit"at einer Deformation bringt automatisch eine
\Name{Poisson}-Struktur ins Spiel (vgl.~Bemerkungen
\ref{Bemerkungen:AssoziativitaetDeformierteAlgebren}), die wir uns im
weiteren zu Nutze machen werden. 

Wir wollen nun \glqq allgemeinere\grqq{} 
Sternprodukte definieren und die wichtigen Strukturen, die wir in
Bemerkung~\ref{Bemerkung:EiegnschaftenDerSternprodukteimFlachenFall}
zusammengetragen haben, "ubernehmen, um
eine m"oglichst allgemeine Klasse von assoziativen, deformierten Algebren zu
erhalten. 

Wir werden in diesem Kapitel eine allgemeine
Definition angeben und Sternprodukte "uber diese
Eigenschaften axiomatisieren, wie es von \citet{bayen.et.al:1977a}
vorgeschlagen wurde. Dies bedeutet insbesondere, da"s
wir uns nun komplett von der Algebra der Operatoren trennen, mit
deren Hilfe wir in Kapitel~\ref{sec:SternproduktefuerdenR2n} die ersten Sternprodukte
eingef"uhrt haben. Da wir desweiteren eine differentielle
Struktur ben"otigen, stellt sich der Raum der glatten Funktionen auf
\Name{Poisson}-Man\-nig\-fal\-tig\-keiten mit formalen
Potenzreihen\index{formale Potenzreihe}
als die richtige Kategorie heraus. Diese so erhaltene deformierte
Algebra werden wir dann mit $(\Cinff{M}, \star)$ kennzeichnen und sie
als ein {\em formales Sternprodukt}\index{Sternprodukt!formales} bezeichnen. Formal deswegen, weil
formale Potenzreihen erstmal ein rein algebraisches Konzept
darstellen. Der Parameter 
$\lambda$ ist keine reelle Zahl, und erst in einem konvergenten Rahmen
k"onnen wir dazu "ubergehen, $\lambda \leadsto \hbar$ anzunehmen.
Den Preis, den man f"ur diese formale Definition eines Sternprodukts zahlt
ist damit offensichtlich: man kann keinerlei Aussagen "uber
Konvergenz machen! Einen konvergenten Rahmen zu schaffen, verlangt viel Arbeit
und ist im allgemeinen extrem schwierig zu konstruieren. Die Konstruktion einer konvergenten
Unteralgebra $\alg{A}_{\sss \hbar}$ der formalen Potenzreihen der reell-analytischen Funktionen
$\Comegaf{\field{C}^{n}}$ f"ur das \Name{Wick}-Produkt
$(\alg{A}_{\sss \hbar}, \spwick)$ mit $\lambda=\hbar$ findet man in
\citep{beiser.roemer.waldmann:2005a:pre,beiser:2005a}.    

\begin{definition}[Formale Sternprodukte]
\label{Definition:FormaleSternprodukte}

Sei $(M,\Lambda)$ eine
\Name{Poisson}-Mannigfaltigkeit.\index{Mannigfaltigkeit!Poisson-Mannigfaltigkeit@\Name{Poisson}-Mannigfaltigkeit}
\index{Poisson-Mannigfaltigkeit@\Name{Poisson}-Mannigfaltigkeit} Ein formales
Sternprodukt ist eine $\fieldf{C}$-bilineare 
Abbildung $\star:\Cinff{M} \times \Cinff{M} \to  \Cinff{M}$ der Form  
\begin{align}
f\star g = \sum_{n=0}^{\infty} \lambda^{n} C_n(f,g),
\label{eq:FormalesSternprodukt}
\end{align}
so da"s f"ur alle $f,g,h \in \Cinff{M}$ gilt:
\begin{compactenum}
\item $(f\star g) \star h = f \star (g\star h) $ (Assoziativit"at des
    Sternprodukts),  
\item $ f\star g = fg + \sum_{n=1}^{\infty} \lambda^{n} C_n(f,g)$
    (Deformation des punktweisen Produkts),  
\item $f \star g - g\star f = \im \lambda \{f,g\} + \alg{O}(\lambda^2)$
    (Deformation in Richtung der \Name{Poisson}-Klammer), 
\item $f \star1 = 1\star f =f$ (Die konstante Funktion $1$ ist
    Einselement). 
\end{compactenum}
Dabei ist $\lambda$ ein formaler Parameter. Man bezeichnet die
deformierte, assoziative Sternproduktalgebra mit $(\Cinff{M},
\star)$. Alternativ spricht man von $(M,\Lambda,\star)$, einer {\em
  \Name{Poisson}-Mannigfaltigkeit $M$ mit Sternprodukt}.
\end{definition}

\begin{bemerkung}[Symplektischer Fall]
Oft ben"otigen wir symplektische Mannigfaltigkeiten mit
Sternprodukt, die wir mit $(M,\omega,\star)$ bezeichnen werden.    
\end{bemerkung}

\begin{bemerkung}[Alternative Formulierung]
Mit Hilfe der bilinearen $C_{n}:
\Cinff{M} \times \Cinff{M} \to  \Cinff{M}$ k"onnen wir eine dazu "aquivalente Definition von
Sternprodukten angeben. F"ur manche Anwendung
wird sich diese Schreibweise als n"utzlich erweisen. Sei im weiteren
$n \in \field{N}$, dann sind die Punkte {\it i.)} bis {\it iv.)} aus Definition
\ref{Definition:FormaleSternprodukte} "aquivalent zu
\begin{compactenum}
\item $\sum_{s=0}^{n} \left[C_{s}(C_{n-s}(f,g),h) -
        C_{s}(f,C_{n-s}(g,h)) \right] =0$ f"ur alle $n\in \field{N}$,
\item $C_{0}(f,g) = fg$,
\item $C_{1}(f,g) - C_{1}(g,f) = \im \{f,g\}$,
\item $C_{n}(f,1) = C_{n}(1,f) = 0$ f"ur $n\ge 1$.
\end{compactenum} 

Man sieht mit etwas Arbeit beispielsweise an {\it i.)},
da"s die Assoziativit"at ein kohomologisches Problem ist und eine Obstruktion in der
\Name{Hochschild}-Ko\-ho\-mo\-lo\-gie\index{Hochschild-Kohomologie@\Name{Hochschild}-Kohomologie}
darstellt; siehe hierzu beispielsweise \citep{gutt.rawnsley:1999a}. 
\end{bemerkung}

Die Forderung der \Name{Hermite}zit"at, d.~h.~die Existenz einer
$^\ast$-Involution auf der deformierten Algebra $(\Cinff{M},\star)$,
auf die wir in den 
einf"uhrenden Kapitel relativ viel Wert gelegt haben, ist nicht Teil
der Definition von Sternprodukten. Statt dessen werden wir es als eine
Struktur ansehen, die ein Sternprodukt haben {\em kann}. Dies "andert
allerdings  
nichts an der Tatsache, da"s f"ur die Physik die Sternprodukte, bei
denen die komplexe Konjugation ein antilinearer
Antiautomorphismus\index{Antiautomorphismus!antilinearer} ist, eine wichtige Rolle spielen.

\begin{definition}[Typen von Sternprodukten]
\label{Definition:TypenVonSternprodukten}
Gegeben sei eine \Name{Poisson}-Mannigfaltigkeit $(M,\Lambda,\star)$ mit
einem Sternprodukt nach Definition
\ref{Definition:FormaleSternprodukte}. Man nennt ein Sternprodukt    
\begin{compactenum}
\item {\em \Name{Hermite}sch}\index{Sternprodukt!Hermitesches@\Name{Hermite}sches|textbf}, falls $\cc{f \star g} = \cc{g} \star
    \cc{f}$ f"ur alle $f,g \in \Cinff{M}$,
\item {\em lokal}\index{Sternprodukt!lokales|textbf} falls $\supp(f\star g) \subseteq \supp(f) \cap
    \supp(g)$ f"ur alle $f,g \in \Cinff{M}$, 
\item {\em differentiell}\index{Sternprodukt!differentielles|textbf}, falls alle $C_n$
    Multidifferentialoperatoren\index{Multidifferentialoperator} sind. 
\end{compactenum}
\end{definition}

Wir werden uns auch in Zukunft auf {\em differentielle} Sternprodukte
beschr"anken, und daher einige Definitionen dazu angeben.  

\begin{definition}[Typen differentieller Sternprodukte]
\label{Definition:TypenDifferentiellerSternprodukte}
Gegeben sei eine \Name{Poisson}-Mannigfaltigkeit $(M,\Lambda,\star)$ mit
einem Sternprodukt. 
\begin{compactenum}
\item Man nennt ein Sternprodukt {\em
      nat"urlich}\index{Sternprodukt!natuerliches@nat\"urliches} oder vom
    {\em \Name{Vey}-Typ}\index{Sternprodukt!Vey-Typ@vom
      \Name{Vey}-Typ}, falls alle $C_r$
    Bidifferentialoperatoren\index{Bidifferentialoperator} der Ordnung
    $r$ in jedem Argument sind.  
\item Man nennt ein Sternprodukt $f \star g = \sum_{r=0}^{\infty}
    \left(\frac{\im \lambda}{2} \right)^{r} C_{r}(f,g)$ vom {\em
      \Name{Weyl}-Typ}\index{Sternprodukt!Weyl-Typ@vom \Name{Weyl}-Typ|textbf}, falls alle $C_r$ reell sind und  
    $C_{r}(f,g)=(-1)^{r} C_{r}(g,f)$ 
    gilt. Insbesondere bedeutet dies, da"s $C_{1} (f,g)$
    gleich der \Name{Poisson}-Klammer ist: $C_1(f,g)=\{f,g\}$. 
\item Auf dem Kotangentialb"undel $(M=T^{\ast}Q, \omega_{\sss 0})$
    nennt man ein Sternprodukt vom {\em standardgeordneten Typ},\index{Sternprodukt!standardgeordneten Typ@vom standardgeordneten Typ|textbf} falls
  mit der kanonischen symplektischen Form $\omega_{\sss 0}$ die erste
    Funktion nur in Faserrichtung differenziert wird. Als Phasenraum
    entspricht dies der Differentiation in
    Impulsrichtung. Analog definiert man den {\em
      antistandardgeordneten Typ}, d.~h.~falls die zweite Funktion nur
    in Faserrichtung differenziert wird. 
\item Sei $(M,\omega,I,g,\star)$ eine
    \Name{K"ahler}-Mannigfaltigkeit\index{Mannigfaltigkeit!Kaehler-Mannigfaltigkeit@\Name{K\"ahler-Mannigfaltigkeit}} \index{Kaehler-Mannigfaltigkeit@\Name{K\"ahler}-Mannigfaltigkeit}
    mit einem Sternprodukt $\star$, so
    nennt man ein Sternprodukt vom {\em
      \Name{Wick}-Typ}\index{Sternprodukt!Wick-Typ@vom \Name{Wick}-Typ|textbf}, falls die
    erste Funktion nur in holomorphe, die zweite nur in antiholomorphe
    Richtung differenziert wird, und vom
    Anti-\Name{Wick}-Typ\index{Sternprodukt!Anti-Wick-Typ@vom Anti-\Name{Wick}-Typ|textbf}, falls
    die erste Funktion nur in antiholomorphe Richtung und die zweite nur
    in holomorphe Richtung differenziert wird.           
\end{compactenum}
\end{definition}

Nachdem wir Sternprodukte definiert
haben, liegt die Frage nach der {\em Existenz} solcher Objekte auf
Mannigfaltigkeiten nahe. In den Artikeln \citep{dewilde.lecomte:1983a}
beziehungsweise \citep{dewilde.lecomte:1984a} wurde die Existenz von
Sternprodukten auf Kotangentialb"undeln $T^{\ast}Q$ einer
Basismannigfaltigkeit $Q$ gezeigt. Sp"ater gelang
\citet{dewilde.lecomte:1983b,dewilde.lecomte:1983c} der Beweis der
Existenz von Sternprodukten auf beliebigen symplektischen
Mannigfaltigkeiten $(M,\omega)$. Dieser Beweis ist ein reiner
Existenzbeweis und gibt keine Vorgehensweise f"ur die Konstruktion
von Sternprodukten an. Einen Durchbruch gelang
\citet{fedosov:1994a,fedosov:1996a}, da er ein (rekursives) Rezept
f"ur die Konstruktion von Sternprodukten auf symplektischen
Mannigfaltigkeiten mit einem symplektischen Zusammenhang
$(M,\omega,\nabla)$ angeben konnte. Diese Konstruktion werden wir in
Kapitel~\ref{sec:FedosovKonstruktion} genauer beschreiben. Mit
Hilfe einer leicht abge"anderten Variante der
\Name{Fedosov}-Konstruktion konnten
\citet{bordemann.neumaier.waldmann:1998a,bordemann.neumaier.waldmann:1999a}
zeigen, da"s es standardgeordnete Sternprodukte auf
Kotagentialb"undeln gibt. Dies gelang unabh"angig davon
\citet{pflaum:1998b,pflaum:2000a}. F"ur
\Name{K"ahler}-Mannigfaltigkeiten zeigte
\citet{bordemann.waldmann:1997a} mit einer ge"anderten
\Name{Fedosov}-Konstruktion, die Existenz von
\Name{Wick}-Produkten. Ein Jahr zuvor ver"offentlichte
\citet{karabegov:1996a}, 
da"s es auf \Name{K"ahler}-Mannigfaltigkeiten immer Sternprodukte
mit Trennung der Variablen gibt. Weitere Ausf"uhrungen dazu findet man bei
\citet{karabegov:1999a,neumaier:2003a}. 

Im Jahre 1997, "uber ein Jahrzehnt nach dem allgemeinen
Existenzbeweis f"ur symplektische Mannigfaltigkeiten von
\Name{Lecomte} und \Name{de Wilde}, gelang \citet{kontsevich:2003a}
der (konstruktive) Beweis f"ur die Existenz von Sternprodukten auf
\Name{Poisson}-Mannigfaltigkeiten.\index{Kontsevich-Konstruktion} Zusammenfassend
k"onnen wir folgenden Satz formulieren. 

\begin{satz}[Existenz verschiedener Typen von Sternprodukten]
\label{Satz:ExistenzTypenSternprodukte}
~\vspace{-5mm}
\begin{compactenum}
\item Auf jeder symplektischen Mannigfaltigkeit $(M,\omega)$ existieren
    (differentielle, nat"urliche, \Name{Hermite}sche) Sternprodukte
    nach Definition~\ref{Definition:FormaleSternprodukte}. 
\item Auf jeder symplektischen Mannigfaltigkeit gibt es Sternprodukte
    vom \Name{Weyl}-Typ. 
\item Auf jeder \Name{K"ahler}-Mannigfaltigkeit $(M,\omega,I,g)$
    existieren nat"urliche Sternprodukte vom (An\-ti-)
    \Name{Wick}-Typ. 
\item Auf jedem Kotangentialb"undel $(T^{\ast}Q,\omega_{\sss 0})$
    existieren Sternprodukte vom (anti-)stan\-dard\-ge\-ord\-ne\-ten
    Typ.  
\item Auf jeder \Name{Poisson}-Mannigfaltigkeit $(M,\Lambda)$
    existieren (nat"urliche, \Name{Hermite}sche) Sternprodukte. 
\end{compactenum}
\end{satz}

\section{"Aquivalenz und Klassifikation von Sternprodukten}
\label{sec:AeuqivalenzundKlassifikationvonSternprodukten}
\index{Aequivalenz@\"Aquivalenz!Sternprodukten@von Sternprodukten|(}
Nachdem wir in Kapitel~\ref{sec:FormalesSternprodukt} die
Existenz von Sternprodukten gekl"art haben, wollen wir kurz auf die
Eindeutigkeit eingehen. Schon im flachen Fall des $\field{R}^{2n}$
haben wir gesehen, da"s Sternprodukte auf einer festen
Mannigfaltigkeit nicht eindeutig sind. Statt
dessen konnten wir eine {\em Familie} von 
Sternprodukten angeben. Ebenso werden wir auf symplektischen
oder \Name{Poisson}-Mannigfaltigkeiten Familien von Sternprodukten identifizieren
k"onnen. Dazu werden wir im folgenden einen {"Aqui\-va\-lenz\-be\-griff}
definieren.  

\begin{definition}["Aquivalenz von Sternprodukten]
\label{Definition:AequivalenzSternprodukte}
Man nennt zwei (differentielle) Sternprodukte $\star$ und $\star'$ auf
der \Name{Poisson}-Mannigfaltigkeit $(M,\Lambda)$ genau dann
{\em "aquivalent}, falls es eine formale Potenzreihe 
\begin{align}
\label{eq:AequivalenzSternprodukte}
T=\id + \sum_{n=1}^{\infty} \lambda^{n} T_{n}
\end{align}
mit $\field{C}$-linearen (Differential-) Operatoren $T_{n}: \Cinf{M}
\to \Cinf{M}$ mit $T_{n}(1)=0$ gibt, so da"s f"ur alle $f,g \in
\Cinff{M}$  
\begin{align}
T(f\star g) = (Tf) \star' (Tg)
\end{align} 
ist.  
\end{definition}

Beispiele f"ur solche "Aquivalenzen sind die $\kappa$- und
$\tilde{\kappa}$-Ordnungen, die man aus dem Standard- oder dem
\Name{Wick}-Produkt auf $\field{R}^{2n}$
bzw.~$\field{C}^{n}$ erh"alt. Die "Aquivalenz-Operatoren in Definition
\ref{Definition:AequivalenzSternprodukte} sind damit eine
Verallgemeinerung von $N_{\sss \kappa}$ und $S_{\sss
  \tilde{\kappa}}$. Wir k"onnen nun zeigen, da"s
\glqq"aquivalent sein\grqq{} wirklich eine "Aquivalenzrelation ist, und
Sternprodukte einer Klasse eine "Aquivalenzklasse bilden. Die
Klassifizierung von Sternprodukten wird "uber die "Aquivalenzklassen
geschehen.

\begin{dlemma}["Aquivalenzrelation]
Die "Aquivalenz von Sternprodukten in Definition
\ref{Definition:AequivalenzSternprodukte} ist eine
"Aquivalenzrelation. Zwei Sternprodukte $\star$ und $\star'$ sind somit
"aquivalent, und wir schreiben $\star \sim \star'$, falls ein $T$
existiert, so da"s Gleichung~\eqref{eq:AequivalenzSternprodukte}
erf"ullt ist. Man schreibt  
\begin{equation}
\begin{aligned}
{}[\star]{} & =\left\{ \star' | \star \sim \star' \right\} \\ & = \left\{
    \star'| \; \exists T \;\text{so da"s} \; \forall f,g \in \Cinff{M} \;\,
    T(f \star g) = (Tf) \star' (Tg) \right\}. 
\end{aligned}
\end{equation}  
Wir bezeichnen $\star' \in [\star]$ als einen {\em Repr"asentanten} der Klasse.
\end{dlemma}
\begin{proof}
Man pr"uft {Reflexivit"at}, {Symmetrie} und {Transitivit"at}.

{\it Reflexivit"at:} Klar. Man w"ahlt $T=\id$, d.~h.~$T_{n}=0$ f"ur $n\ge 1$.

{\it Symmetrie:} Klar. $T$ ist ein invertierbarer Operator, da er
    in unterster Ordnung mit der Identit"at beginnt und daher gilt $T^{-1}(f
    \star g) = (T^{-1}f)\star (T^{-1}g)$. 

{\it Transitivit"at:} Seien $T(\tilde{f}\star
    \tilde{g})=(T\tilde{f}) \star' (T\tilde{g})$ und $S(f\star g) = (Sf)
    \star'' (Sg)$, dann ist $f=T^{-1}\tilde{f}$, etc.~und damit
    $(T^{-1}f) \star (T^{-1} g) = T^{-1}(f \star' g)$. Folglich gilt 
\begin{align*}
ST^{-1}(f\star' g) = S((T^{-1}f) \star (T^{-1}) g) = (ST^{-1}) f \star''
(ST^{-1}) g, 
\end{align*}  
womit wir die "Aquivalenz von $\star'$ und $\star''$ via $R=ST^{-1}$
gezeigt haben. 
\end{proof}

\begin{lemma}["Aquivalenz \Name{Hermite}scher Sternprodukte]
    \label{Lemma:AequivalenzHermitescherSternprodukte}
Sind $\star$ und $\star'$ "aquivalente, \Name{Hermite}sche Sternprodukte auf
$(M,\Lambda)$, so gibt es eine "Aqui\-va\-lenz\-trans\-for\-ma\-tion $T$, so
da"s $\cc{T(f)}=T(\cc{f})$ f"ur alle $f\in \Cinff{M}$ ist. 
\end{lemma}
\begin{proof}
    Der Beweis ist in \citep[Prop.~5.6]{neumaier:2002a} nachzulesen. 
\end{proof}

\begin{satz}[Klassifikation von symplektischen Sternprodukten I]
\label{Satz:KlassifikationSymplektischeSternprodukteI}
Gegeben sei eine symplektische Mannigfaltigkeit $(M,\omega)$. Die
Menge der "Aquivalenzklassen ist in Bijektion zu $\HdeRham[2](M)\FP$.  
\end{satz}

\begin{proof} Der Beweis dieses Satzes findet sich in
    \citep{bertelson.cahen.gutt:1997a,nest.tsygan:1995a,nest.tsygan:1995b,deligne:1995a,gutt.rawnsley:1999a}.    
\end{proof}

Genauer gesagt definiert jedes Sternprodukt $\star$ eine {\em
  charakteristische Klasse} $c(\star)$.
\begin{equation}
\begin{aligned}
\label{eq:CharakteristischeKlasseSternprodukt}
 \frac{1}{\im \lambda}[\omega] + \HdeRham[2](M)\FP \ni c(\star) =
\frac{1}{\im \lambda}[\omega] + \sum_{n=0}^{\infty} \lambda^{n} [\omega_{n}]. 
\end{aligned}
\end{equation}

So kann man Satz~\ref{Satz:KlassifikationSymplektischeSternprodukteI}
wie folgt auffassen.

\begin{satz}[Klassifikation von symplektischen Sternprodukten II]
\label{Satz:KlassifikationSymplektischeSternprodukteII}
Gegeben sei eine symplektische Mannigfaltigkeit $(M,\omega)$. Zwei
Sternprodukte $\star$ und $\star'$ sind genau dann "aquivalent, wenn
$c(\star) = c(\star')$. 
\end{satz}

\begin{korollar}["Aquivalenz auf symplektischen Mannigfaltigkeiten]
Sei $(M,\omega)$ eine symplektische Mannigfaltigkeit bei der die
zweite \Name{de Rham} Kohomologie verschwindet,
d.~h.~$\HdeRham[2](M)=\{0\}$, so sind je zwei Sternprodukte $\star$
und $\star^{\prime}$ "aquivalent. 
\end{korollar}

\begin{korollar}["Aquivalenz von Sternprodukten auf dem
    $\field{R}^{2n}$] 
Alle Sternprodukte auf dem $\field{R}^{2n}$ sind "aquivalent,
d.~h.~bis auf die Wahl einer (verallgemeinerten)
Ordnungsvorschrift\index{Ordnungsvorschrift!verallgemeinerte} ist die Quantisierung auf dem $\field{R}^{2n}$ eindeutig. 
\end{korollar}

\begin{beispiel}[Nicht"aquivalente Sternprodukte]
Ein explizites Beispiel f"ur nicht"aquivalente Sternprodukte findet sich in
\citep{cahen.flato.gutt.sternheimer:1985a}. Hier wird gezeigt, da"s
auf einem $2n$-dimensionalen Torus
$\field{T}^{2n}=\field{R}^{2n}/\field{Z}^{2n}$ Sternprodukte
existieren, die nicht durch eine "Aquivalenztransformation, d.~h.~einen Isomorphismus,
ineinander "uberf"uhrt werden k"onnen. Dieses Beispiel war lange
vor der Klassifikation bekannt und daher von Interesse. Nach
Satz~\ref{Satz:KlassifikationSymplektischeSternprodukteI} ist dieses
Beispiel trivial, da $\HdeRham[2](\field{T}^{2n}) \neq \{0\}$ ist.
\end{beispiel}

Nun wollen wir noch die Klassifikation von Sternprodukten auf
\Name{Poisson}-Man\-nig\-fal\-tig\-kei\-ten angeben. Zusammenfassend
l"a"st sich der folgende Satz (vgl.~\citep{kontsevich:2003a}) formulieren. 

\begin{satz}[Klassifikation von Sternprodukten auf
    \Name{Poisson}-Mannigfaltigkeiten] 
Die "Aquivalenzklassen $[\star]$ von Sternprodukten auf einer
\Name{Poisson}-Mannigfaltigkeit $(M,\Lambda_{0})$ sind in Bijektion zu
den "Aquivalenzklassen $[\Lambda]$ von formalen Deformationen
$\Lambda= \Lambda_{0} + \sum_{n\ge 1}\lambda^{n} \Lambda_{n} \in
\schnittf{\Lambda^2 TM}$ des \Name{Poisson}-Tensors $\Lambda_{0}$
modulo formalen Diffeomorphismen. 
\end{satz}
\index{Aequivalenz@\"Aquivalenz!Sternprodukten@von Sternprodukten|)}

\section{Konstruktion von Sternprodukten: \Name{Fedosov}-Konstruktion}

\subsection{Die \Name{Fedosov}-Konstruktion}
\label{sec:FedosovKonstruktion}
\index{Fedosov-Konstruktion@\Name{Fedosov}-Konstruktion|(}
Die \Name{Fedosov}-Konstruktion ist nicht nur ein konstruktiver Beweis der
Existenz von Sternprodukten auf symplektischen Mannigfaltigkeiten,
sondern stellt desweiteren ein m"achtiges Werkzeug zur Verf"ugung,
um verschiedene Typen von Sternprodukten konstruieren und damit ihre
Existenz beweisen zu k"onnen. 

Wir werden in mehreren Schritten vorgehen. Zuerst zeigen wir recht
ausf"uhrlich eine leichte Verallgemeinerung der von \Name{Fedosov}
gemachten Konstruktion. Diese erweitern wir in einem weiteren Schritt
auf Vektorb"undel\index{Vektorbuendel@Vektorb\"undel}. Die Schnitte
im Vektorb"undel stellen einen 
Bimodul f"ur ein Sternprodukt von 
rechts und das deformierte Endomorphismenb"undel des Vektorb"undels von links
dar. Das bedeutet, wir sind mit Hilfe einer Verallgemeinerung der
Konstruktion in der Lage, nicht nur ein Sternprodukt f"ur Funktionen
auf einer symplektischen Mannigfaltigkeit anzugeben, sondern es gelingt
uns ferner, einen deformierten Bimodul zu konstruieren, bei dem alle
Strukturen deformiert sind. In Kapitel
\ref{sec:HInvarianteFedosovKonstruktion} werden wir dann letztendlich
eine unter einer Gruppenwirkung {\em $G$-invariante
  \Name{Fedosov}-Konstruktion} pr"asentieren, die auf Vektorb"undel
erweitert wird. Zuvor 
werden wir aber die gew"ohnliche Konstruktion nachvollziehen,
wie man sie im wesentlichen schon in \citep{fedosov:1994a,fedosov:1996a} findet. Eine
"uberarbeitete, moderne und ausf"uhrlichere Formulierung findet man
in \citep{neumaier:2001a,waldmann:2004a:script}. Auf Beweise werden
wir in diesem Kapitel so gut wie vollst"andig verzichten, da sie
beispielsweise in den zuvor erw"ahnten Quellen ausf"uhrlich
dargelegt sind.

\subsubsection{Das formale \Name{Weyl}-Algebrab"undel} 
Sei $(M,\omega)$ eine symplektische Mannigfaltigkeit. Wir betrachten
einen Punkt $p \in U \subseteq M$ in einer Umgebung $U$. Es existiert eine Karte $(U,x)$ mit
den lokalen Koordinaten $(x^{1},\cdots,x^{n})$,
so da"s wir die symplektische Form sowie den \Name{Poisson}-Tensor
in der Form 

\begin{align}
\omega\big|_{U} = \frac{1}{2} \omega_{ij} \de x^{i} \wedge \de x^{j}
\quad \text{und} \quad \Lambda\big|_{U}= \frac{1}{2} \Lambda^{ij}
\frac{\del}{\del x^{i}} \wedge \frac{\del}{\del x^{j}} 
\end{align}

schreiben k"onnen, und es gilt $\Lambda^{ij}=-\omega^{ij}$ und
$\omega^{ij} \omega_{jk} = \delta^{i}_{k}$.

\begin{definition}[Die formale \Name{Weyl}-Algebra]
\index{formale Weyl-Algebra@formale \Name{Weyl}-Algebra}
Sei $(M,\omega)$ eine symplektische Mannigfaltigkeit. Die {\em formale
 \Name{Weyl}-Algebra} "uber $p$ ist die formal in Vervollst"andigung
der $\lambda$ gradierten, symmetrische Tensoralgebra "uber $T^{\ast}_{p} M$.   

\begin{align}
    W_{p}=\prod_{k=0}^{\infty} W_{p}^{k} = \left(
        \prod_{k=0}^{\infty} S^{k}T_{p}^{\ast}M\right)\FP \quad
    \text{mit} \quad W_{p}^{k} = S^{k}T_{p}^{\ast}M \FP.
\end{align}
\end{definition}

Die Elemente in $W_{p}$ sind somit sowohl formale Potenzreihen in $\lambda$
als auch formale Reihen im symmetrischen Grad der symmetrischen
Tensoren "uber $T^{\ast}_{p}M$. Wir wollen die Algebra als die
Unteralgebra $W_{p}:= (W \otimes \Lambda^{0})_{p}$ von

\begin{align}
    (W \otimes \Lambda ^{\bullet})_{p} = \left( \prod_{k=0}^{\infty}
        S^{k}T_{p}^{\ast}M \otimes \Lambda^{\bullet}T_{p}^{\ast}M
    \right) \FP
\end{align}

verstehen. In $(W \otimes \Lambda^{\bullet})_{p}$ haben Elemente die
Form $a = \sum_{i} f_{i} \otimes \alpha_{i}$, sind also
Linearkombinationen, wobei der hintere Teil die
schiefsymmetrischen Tensoren "uber $T^{\ast}_{p}M$ sind. Die
Einsformen $\theta$ kann man sowohl als symmetrisch wie auch als
schiefsymmetrisch interpretieren, was wir durch die Schreibweise
$\theta \otimes 1$ bzw.~$1 \otimes \theta$ kennzeichnen werden. Auf
der Algebra $(W \otimes \Lambda^{\bullet})_{p}$ k"onnen wir nun
mittels des symmetrischen Produkts $\vee$ und des schiefsymmetrischen
Produkts $\wedge$ ein assoziatives, im schiefsymmetrischen Teil
$\field{Z}_{2}$-gradiertes\footnote{$\field{Z}_{2}$-gradierte
  Produkte, Klammern, Derivationen etc.~werden oft auch mit dem Pr"afix {\it
    Super-} versehen, so da"s man von Superprodukt, Superklammer,
  Superderivation, Superkommutativit"at, etc.~spricht. Wir verwenden
  beide Bezeichnungen gleicherma"sen.} Produkt $\mu:(W \otimes
\Lambda^{\bullet})_{p} \times (W \otimes \Lambda^{\bullet})_{p} \to (W \otimes
\Lambda^{\bullet})_{p}$ einf"uhren. F"ur $a=f
\otimes \alpha$ und $b=g \otimes \beta$ ist dann
\begin{equation}
\begin{aligned}
  \mu(a,b) = ab & = (f \otimes \alpha)(g \otimes \beta) = (f \vee g)
  \otimes (\alpha \wedge \beta) \\ & = (-1)^{k\ell} (g \vee f) \otimes
  (\beta \wedge \alpha) \\ &= (-1)^{k\ell} ba. 
\end{aligned}
\end{equation}

Dabei ist $\alpha \in \Lambda^k T_{p}^{\ast}M$ und $\beta \in
\Lambda^{\ell} T_{p}^{\ast}M$. 

\begin{definition}[Gradabbildungen auf $(W \otimes
    \Lambda^{\bullet})_{p}$] 
\index{Gradabbildungen}
Man definiert die folgenden {\em Gradabbildungen} auf der Algebra $(W
\otimes \Lambda^{\bullet})_{p}$.  
\begin{align}
\dega, \degs, \degl : (W \otimes \Lambda^{\bullet})_{p} \to (W \otimes
\Lambda^{\bullet})_{p}  
\end{align}

\begin{compactenum}
\item $\degs (f\otimes \alpha) = m f \otimes \alpha$ f"ur $f\in W_p^{m}$,
\item $\dega (f\otimes \alpha) = n f \otimes \alpha$ f"ur $\alpha \in
    \Lambda^{n}T_{p}^{\ast}M$,
\item $\degl(f \otimes \alpha) = \lambda\frac{\del}{\del \lambda} (f
    \otimes \alpha)$,
\item $\Deg=\degs + 2 \degl$,  
\end{compactenum}
Die ersten beiden Gradabbildungen messen den {\em symmetrischen} bzw.~{\em
  schiefsymmetrischen Grad}, die dritte den {\em $\lambda$-Grad}
und die letzte den {\em totalen Grad}.
\end{definition}

\begin{bemerkungen}[Gradabbildungen]
~\vspace{-5mm}
\begin{compactenum}
\item Der totale Grad ist derzeit noch unmotiviert, wird sich aber als
    n"utzlich herausstellen, wenn wir 
sp"ater die Struktur der \Name{Fedosov}-Derivation oder des
faserweisen Produkts verstehen wollen. 
\item Die Gradabbildungen sind Derivationen der Multiplikation $\mu$. Um
dies zu zeigen, und f"ur sp"atere Anwendungen, ist es n"utzlich, die
Gradabbildung $\degs$ und $\dega$ in lokalen Koordinaten zu schreiben. Wir
k"onnen zeigen, da"s  
\begin{equation}
\begin{aligned}
    \degs (f \otimes \alpha) &= (\de x^{i} \otimes 1)
    i_{s}\left(\tfrac{\del}{\del x^{i}} \right) (f \otimes \alpha) \\
    & = \de x^{i} \vee i_{s}\left(\tfrac{\del}{\del x^{i}} \right) f
    \otimes \alpha
\end{aligned}
\end{equation}
und 
\begin{equation}
\begin{aligned}
    \dega (f \otimes \alpha) &= (1 \otimes \de x^{i})
    i_{a}\left(\tfrac{\del}{\del x^{i}} \right) (f \otimes \alpha) \\
    & = f \otimes \de x^{i} \wedge i_{a} \left(\tfrac{\del}{\del
          x^{i}} \right) \alpha
\end{aligned}
\end{equation}
ist. Dabei bezeichnen wir mit $i_{s}$ und $i_{a}$ den symmetrischen
bzw.~schiefsymmetrischen Einsetzer. Der Beweis erfolgt, indem man die Behauptung auf Tensoren in
$S^{\bullet} T_{p}^{\ast}M \otimes \Lambda^{\bullet}T_{p}^{\ast}M$
nachrechnet.  
\end{compactenum}
\end{bemerkungen}

\begin{definition}[Die Differentiale $\delta$, $\delta^{\ast}$ und $\delta^{-1}$]
\index{Differential!Fedosov@\Name{Fedosov}}
Sei $f \otimes \alpha \in (W \otimes \Lambda)_{p}$. Wir
definieren die Differentiale $\delta$ und $\delta^{\ast}$ mittels 
\begin{align}
 \label{eq:DefinitionDeltas}
    \delta(f\otimes \alpha) :& = i_{s}\left(\tfrac{\del}{\del x^{i}} \right) f
    \otimes \de x^{i} \wedge \alpha, \\
    \delta^{\ast}(f\otimes \alpha) :& = \de x^{i} \vee f \otimes i_{a} \left(
        \tfrac{\del}{\del x^{i}} \right) \wedge \alpha.  
\end{align}
F"ur homogene Elemente\index{Elemente!homogene} $a\in (W^n \otimes \Lambda^{m})_{p}$ sei weiter 
\begin{align}
\delta^{-1} a := \begin{cases} 0 \quad &\text{falls} \quad n+m=0, \\
    \tfrac{1}{n+m} \delta^{\ast}a  \quad & \text{falls} \quad n+m \neq 0. \end{cases}
\end{align}
\end{definition}

Die Abbildung $\delta^{-1}$ ist linear auf alle Elemente in $(W \otimes \Lambda)_{p}$
fortsetzbar. Es sei zu bemerken, da"s $\delta^{-1}$ nicht
das Inverse von $\delta$ ist, da $\delta$ keinen invertierbaren
Operator darstellt. 

\begin{definition}[Der Projektionsoperator $\sigma$]
\index{Projektionsoperator}
Der Operator
\begin{align}
    \sigma: (W \otimes \Lambda)_{p} \to \fieldf{C}
\end{align}
ist die Projektion auf den symmetrischen und schiefsymmetrischen Grad
$0$. 
\end{definition}

Die Eigenschaften der Operatoren $\delta,\delta^{\ast},
\delta^{-1}$ und $\sigma$ werden im folgenden Lemma zusammengefa"st.

\begin{lemma}[Endomorphismen $\delta$, $\delta^{\ast}$, $\delta^{-1}$ und $\sigma$]
    Es gilt
    \begin{compactenum}
        \item Die lokal definierten Abbildungen $\delta$,
             $\delta^{\ast}$, $\delta^{-1}$ sind global definierte Objekte
             und nicht von einer Kartenwahl abh"angig.
        \item Der Endomorphismus $\delta$ verringert den symmetrischen
            Grad und erh"oht den schiefsymmetrischen Grad um jeweils
            eins, was man formal als \begin{align}[\dega,\delta]=-
                \delta \quad 
            \text{und} \quad [\degs, \delta]=\delta \end{align}
            schreiben kann. Analog verringert $\delta^{\ast}$ den
            schiefsymmetrischen Grad und erh"oht den symmetrischen
            Grad \begin{align}[\dega,\delta^{\ast}]= \delta^{\ast} \quad
                \text{und} \quad [\degs, \delta^{\ast}]= -\delta^{\ast}. \end{align} 
         \item Die Endomorphismen $\delta$, $\delta^{\ast}$ und
             $\delta^{-1}$ sind nilpotent vom Grad $2$, d.~h.~es gilt
            $\delta^2 = (\delta^{\ast})^2 = (\delta^{-1})^2=0$.
         \item Es gilt das \Name{Poincar\'{e}}-Lemma\index{Poincare-Lemma@\Name{Poincar\'{e}}-Lemma} f"ur Differentialformen auf $T_{p}M$:  

\begin{align}
\label{eq:PoicareLemmaaufTstarM}
 \delta \delta^{-1} + \delta^{-1} \delta + \sigma =
             \id.
\end{align}
        
    \end{compactenum}
\end{lemma}

Insbesondere der letzten Gleichung wird eine wichtige Rolle
zukommen. Zuerst wollen wir jedoch eine Deformation f"ur die Algebra
$((W \otimes \Lambda^{\bullet}), \mu)$ angeben.  

\begin{definition}[Das Faserweise \Name{Weyl-Moyal}-Produkt]
\label{Definition:FaserweisesWeylMoyalProdukt}
\index{Weyl-Moyal-Produkt@\Name{Weyl-Moyal}-Produkt!faserweises}
Das {\em faserweise \Name{Weyl-Moyal}-Produkt} $\fpweyl$ ist f"ur
$a,b \in (W \otimes \Lambda^{\bullet})_{p}$ durch        
\begin{align}
        a \fpweyl b = \mu \circ \eu^{\tfrac{\im \lambda}{2}
          \Lambda^{k\ell}_{p} i_s\left(\tfrac{\del}{\del x^k} \right)
          \otimes i_s\left(\tfrac{\del}{\del x^\ell} \right)} (a\otimes b)  
    \end{align}
definiert. 
\end{definition}

\begin{lemma}[Eigenschaften des faserweisen \Name{Weyl-Moyal}-Produkts]
\label{Lemma:EigenschaftenFaserweisesWeylMoyalProdukt}
Das in Definition~\ref{Definition:FaserweisesWeylMoyalProdukt}
definierte Produkt hat folgende Eigenschaften.
    \begin{compactenum}
        \item Das faserweise \Name{Weyl-Moyal}-Produkt ist global definiert und
            eine (koordinatenunabh"angige) Deformation von
            $((W\otimes \Lambda)_{p}, \mu)$. 
        \item Das faserweise \Name{Weyl-Moyal}-Produkt ist bez"uglich
            des schiefsymmetrischen Teils und des totalen Grads
            $\Deg$ gradiert, allerdings nicht f"ur
            $\degs$ und $\degl$. Die Abbildungen
            $\delta$, $\dega$ und $\Deg$ sind -- anders formuliert --
            $\field{Z}_{2}$-gradierte Derivationen von $\fpweyl$ von
            schiefsymmetrischen Grad $+1$ f"ur $\delta$ bzw.~$0$
            f"ur $\dega$ und $\Deg$.
    \end{compactenum} 
\end{lemma}

Aufgrund der Gradierung des faserweisen Produktes $\fpweyl$ wollen wir
festlegen, was {\em homogene Elemente} bez"uglich des $\Deg$-Grades sind.

\begin{dlemma}[Homogene Elemente bez"uglich des totalen Grades
    $\Deg$] 
\index{Elemente!homogene}
    Wir definieren die homogenen Elemente bez"uglich des totalen
    Grades vom Grad $k$ als
    \begin{align}
        \label{eq:HomogeneElementeDeg}
W^{(k)}_{p} = \{a \in W_{p}| \Deg a = ka\}
    \end{align}
und weiter $(W^{(k)} \otimes \Lambda)_{p}$. Wir k"onnen $W_{p}$ als
kartesisches Produkt aller $W_p^{(k)}$ auffassen und damit ist jedes
  Element $a\in W_{p}$ von der Form $a=\sum_{k=0}^{\infty}
  a^{(k)}$, und jedes homogene Element in $W^{(k)}_{p}$ ist von der
  Form

  \begin{align}
      a^{(k)}= \sum_{r=0}^{\lfloor k/2 \rfloor} \lambda^{r} a_{k-2r}^{(k)},
  \end{align}

\end{dlemma}
wobei die Klammer $\lfloor \cdot \rfloor: \field{R} \to \field{Z}$
immer auf die n"achste 
ganze Zahl abrundet. $W_{p}$ ist damit {\em formal
  $\Deg$-gradiert}, und Elemente vom $\Deg$-Grad $\ge k$ bezeichnen wir mit 
\begin{align}
(W_{k} \otimes \Lambda)_{p} = \bigcup_{n=k}^{\infty} (W^{(n)} \otimes
\Lambda).    
\end{align}

Man definiert mittels der schiefsymmetrischen Gradierung $\dega$ einen
$\field{Z}_2$-gradierten Kommutator bez"uglich des deformierten
Produktes $\fpweyl$, der gegeben ist durch 
\begin{align}
   \ad(a)(b) = [a,b]_{\sss \fpweyl} = a \fpweyl b - (-1)^{mn} b \fpweyl a
\end{align}
mit $a\in (W \otimes \Lambda^{m})_{p}$ und $b\in (W \otimes
\Lambda^{n})_{p}$. Das deformierte Produkt $\fpweyl$ ist im Gegensatz
zu $\mu$ nicht mehr superkommutativ. Als {\em
  Zentrum}\index{Zentrum} bezeichnen wir alle Elemente $a$ f"ur die $\ad(a) =
0$ ist. Dies ist immer genau dann der Fall, wenn $\degs a=0$. 

\begin{bemerkung}[Quasiinnere Derivation]
\index{Derivation!quasiinnere}
    Wir werden $\field{Z}_2$-gradierte Derivationen von $\fpweyl$ in der
    Form $\tfrac{\im}{\lambda} \ad(a)$ schreiben. Dies ist
    wohldefiniert, da die unterste Ordnung von $\lambda$ im Kommutator
    verschwindet. Man nennt eine solche Derivation eine {\em
      quasiinnere Derivation}. 
\end{bemerkung}

\begin{definition}[B"undel aller formalen \Name{Weyl}-Algebren]
\label{Definition:BuendelFormaleWeylAlgebren}
\index{formales Weyl-Algebrabuendel@formales \Name{Weyl}-Algebrab\"undel}
    Wir definieren das B"undel aller formalen \Name{Weyl}-Algebren
    f"ur $p \in M$ als
    \begin{align}
        W=\bigcup_{p\in M} W_{p} \qquad \text und \qquad W \otimes
        \Lambda^{\bullet} = \bigcup_{p \in M} (W \otimes
        \Lambda^{\bullet})_{p}.
    \end{align}
\end{definition}

Das so definierte Vektorb"undel hat unendlichdimensionale
Fasern. Wir k"onnen eine glatte Struktur f"ur $W$ und $(W \otimes
\Lambda)$ angeben und daher glatte Schnitte definieren, und schreiben 
\begin{align}
\alg{W}& =\left( \prod_{k=0}^{\infty}
    \schnitt{S^{k}T^{\ast}M}\right)\FP, \\
\label{eq:TensoralgebraFedosov1} 
\alg{W} \otimes \mit{\Lambda}^{\bullet} & = \left(\prod_{k=0}^{\infty}
    \schnitt{S^{k}T^{\ast}M \otimes \Lambda^{\bullet}T^{\ast}M} \right)\FP.  
\end{align}

Wir setzen alle bisher eingef"uhrten punktweisen Abbildungen, d.~h.~$\dega$,
$\degs$, $\degl$, $\Deg$, $\delta$, $\delta^{\ast}$, $\delta^{-1}$,
$\sigma$, $\mu$ und $\fpweyl$ auf $\alg{W}$ bzw.~$\alg{W}\otimes
\mit{\Lambda}^{\bullet}$ fort. Damit ist $(\alg{W},\fpweyl)$ eine assoziative
Algebra "uber $\fieldf{C}$, und $\Deg$ wird zu einer Derivation
dieser. Es gilt weiterhin die wichtige Gleichung 
\begin{align}
    \delta \delta^{-1} + \delta^{-1} \delta + \sigma = \id \quad
    \text{mit} \quad \sigma: \alg{W} \otimes \mit{\Lambda}^{\bullet} \to \Cinff{M}.
\end{align}

\subsubsection{Die \Name{Fedosov}-Derivation $\alg{D}$}
\index{Fedosov-Derivation@\Name{Fedosov}-Derivation}
Nachdem wir nun die zugrundeliegende Algebra definiert und deren
wichtige Eigenschaften er\-l"au\-tert haben, wollen wir uns in einem
zweiten Schritt der Konstruktion eines Sternproduktes widmen. Dies
bedeutet konkret, da"s wir eine Unteralgebra von $\alg{W} \otimes
\mit{\Lambda}^{\bullet}$ suchen, die mittels $\sigma$ in Bijektion zu
$\Cinff{M}$ ist. Das mit $\sigma$ zur"uckgezogene faserweise Produkt
$\fpweyl$ liefert dann ein $\fieldf{C}$-lineares Produkt, das sich bei
geeigneter Wahl der Unteralgebra von $\alg{W} \otimes \mit{\Lambda}^{\bullet}$
als ein Sternprodukt herausstellt. Diese Algebra erh"alt man als
Kern der $\field{Z}_2$-gradierten {\em \Name{Fedosov}-Derivation}.

Wir ben"otigen f"ur die weitere Konstruktion einen torsionsfreien,
symplektischen Zusammenhang\index{Zusammenhang!symplektischer} $\nabla$, den es nach den
Lemmata~\ref{Lemma:TorsionsfreierZusammenhang} und
\ref{Lemma:HessTrick} auf 
jeder symplektischen Mannigfaltigkeit gibt. Die Kr"ummung $\hat{R}$
ist in Kapitel~\ref{sec:SymplektischeZusammenhaenge}
beschrieben. Die {\em symplektische
  Kr"ummung}\index{Kruemmung@Kr\"ummung!symplektische} (vergleiche 
Lemma~\ref{Lemma:SymplektischeKruemmung})   
\begin{align}
    \label{eq:SymplektischeKruemmungFedosovKonstruktion}
  R(Z,U,X,Y) = \omega(Z,\hat{R}(X,Y)U)   
\end{align}
ist ein Element in
 $R\in \schnitt{S^{2} T^{\ast}M \otimes \Lambda^2 T^{\ast} M}$ und
 damit als ein Element der Algebra $(\alg{W} \otimes \mit{\Lambda}^{\bullet})$ zu
 interpretieren. Es gilt 

 \begin{align}
     \label{eq:GradSymplektischeKruemmung}
     \degs R = 2R, \qquad \dega R = 2R, \qquad \degl R = 0.
 \end{align}

Wir verwenden nun den Zusammenhang, um ein kovariantes Differential zu
definieren.

\begin{lemma}[Kovariantes Differential $D$]
\index{Differential!kovariantes}
Sei $\nabla$ ein torsionsfreier, symplektischer Zusammenhang. Das {\em
  kovariante Differential} 
\begin{align}
    \label{eq:KovariantesDifferential}
    D: \alg{W} \otimes \mit{\Lambda}^{\bullet} \to  \alg{W} \otimes
    \mit{\Lambda}^{\bullet+1} 
\end{align}
ist in lokalen Koordinaten mittels 
\begin{equation}
\begin{aligned}
    \label{eq:KovariantesDifferentialFormel}
    D(f \otimes \alpha) & = \de x^{i} \wedge \lconn[\frac{\del}{\del x^{i}}] (f\otimes \alpha) \\ & = \lconn[\frac{\del}{\del x^{i}}] f
    \otimes \de x^{i} \wedge \alpha + f \otimes \de x^{i} \wedge
    \lconn[\frac{\del}{\del x^{i}}] \alpha \\ & =
    \lconn[\frac{\del}{\del x^{i}}] f \otimes \de x^{i} \wedge \alpha
    +f \otimes \de \alpha.
\end{aligned}
\end{equation}
kartenunabh"angig und global definiert. 
\end{lemma} 

Wir wollen dessen Eigenschaften in der folgenden Proposition zusammenfassen. 

\begin{proposition}[Eigenschaften des kovarianten Differentials $D$]
    Das kovariante Differential $D: \alg{W} \otimes
    \mit{\Lambda}^{\bullet} \to  \alg{W} \otimes
    \mit{\Lambda}^{\bullet+1}$ hat die folgenden Eigenschaften:
    \begin{compactenum}
    \item $[\degs,D]=[\degl,D]=[\Deg,D]=0$ und $[\dega,D]=D$.
    \item $D$ ist eine Superderivation des undeformierten Produkts $\mu$ und
        des faserweisen \Name{Weyl\--Mo\-yal}-Pro\-dukts $\fpweyl$
        \begin{align}
            \label{eq:DSuperderivationWeylMoyal}
            D(a \fpweyl b) = (Da) \fpweyl b + (-1)^{k} a \fpweyl (Db), 
        \end{align}
mit $a \in \alg{W} \otimes \mit{\Lambda}^{k}$ und $b \in \alg{W}
\otimes \mit{\Lambda}^{\bullet}$. 
    \item $DR=0$, $[\delta,D] = \delta D + D \delta = 0$ und
        \begin{align}
            \label{eq:DKommutator}
         D^{2}= [D,D] = \frac{\im}{\lambda} \ad(R). 
        \end{align}
     \end{compactenum}
\end{proposition}

Wie bereits am Anfang dieses Kapitels angedeutet, wollen wir eine
Superderivation $\alg{D}$
konstruieren\index{Fedosov-Derivation@\Name{Fedosov}-Derivation},
deren Kern via $\sigma$ in 
Bijektion zu $\Cinff{M}$ ist. Die Superderivation $\delta$ w"are ein
solcher Kandidat, allerdings ist das so geerbte Produkt auf
$\Cinff{M}$ nur das
punktweise Produkt. Ebenso stellt sich $-\delta +D$ als unzureichend
heraus, obwohl es sich immer noch um eine Superderivation handelt. Der
Kern ist jedoch zu klein, da $D^2 = \tfrac{\im}{\lambda}
\ad(R) \neq 0$. Eine Verbesserung dieser Formel erhalten wir durch 

\begin{align}
    \label{eq:FedosovDerivation}
    \alg{D}=-\delta+D+\frac{\im}{\lambda}\ad(r),
\end{align}

wobei $r \in \alg{W} \otimes \mit{\Lambda}^{1}$. Damit ist $\alg{D}:
\alg{W} \otimes \mit{\Lambda}^{\bullet}  \to \alg{W} \otimes
\mit{\Lambda}^{\bullet+1}$. Man kann Gleichung
\eqref{eq:FedosovDerivation} als eine nach dem totalen Grad
entwickelte Summe auffassen, da $\delta$ den totalen Grad um $1$
verringert, $D$ den totalen Grad konstant h"alt, und $\ad(r)$ den
Grad um mehr als $0$ erh"ohen soll. Damit ist $r$ vom
Deg-Grad $\ge +3$, d.~h.
\begin{align}
    \label{eq:r}
    r=\sum_{k=3}^{\infty} r^{k} \in \alg{W}_{3} \otimes
    \mit{\Lambda}^{1} \qquad \text{mit}\quad r^{k} \in \alg{W}^{(k)}
    \otimes \mit{\Lambda}^{1}.
\end{align}

Die Kr"ummung von $\nabla$ stellt eine Obstruktion f"ur die Gr"o"se
des Kerns dar, daher wollen wir eine Formel f"ur $\alg{D}^2$
berechnen. Wir erhalten f"ur ein beliebiges $r \in \alg{W}_{3}
\otimes \mit{\Lambda}^{1}$ 
\begin{align}
    \label{eq:FedosovDerivationeigenschaft1}
    \alg{D}^{2}=\frac{1}{2}[\alg{D}, \alg{D}] = \frac{\im}{\lambda}
    \ad( -\delta r +R + Dr + \tfrac{\im}{\lambda} r \fpweyl
          r )  
\end{align}
sowie 

\begin{align}
    \label{eq:FedosovDerivationeigenschaft2}
    \alg{D}\left( -\delta r +R + Dr + \tfrac{\im}{\lambda} r
\fpweyl r \right) =0. 
\end{align}

Wir werden nun den entscheidenden Satz zur \Name{Fedosov}-Derivation
formulieren. 

\begin{satz}[\Name{Fedosov}-Derivation $\alg{D}$]
    \label{Satz:FedosovKonstruktion}
Sei $\Omega=\sum_{k=1}^{\infty} \lambda^{k} \Omega_{k} \in \lambda
\schnitt{\Lambda^{2}T^{\ast}M}\FP$ eine geschlossene Zweiform, $\de
\Omega=0$, und sei $s\in \alg{W}_{3} \otimes \Lambda^{0}$ mit
$\sigma(s) = 0$. Es gibt ein eindeutiges Element $r \in \alg{W}_{3}
\otimes \Lambda^{1}$ mit 
\begin{align}
    \label{eq:FedosovRekursion}
    \delta r = R + Dr + \frac{\im}{\lambda} r \fpweyl r + (1\otimes
    \Omega) \qquad \text{und} \qquad \delta^{-1}r=s.
\end{align}
in diesem Fall erf"ullt die \Name{Fedosov}-Derivation $\alg{D} =
-\delta +D +\frac{\im}{\lambda} \ad(R)$ die Gleichung $\alg{D}^2=0$.
\end{satz}

\begin{bemerkungen}
Die \Name{Fedosov}-Konstruktion ist somit von der Wahl dreier
Parameter abh"angig, dem Zusammenhang $\nabla$, der
Reihe $\Omega$ geschlossener Zweiformen und dem Element $s \in
\alg{W}_{3} \otimes \mit{\Lambda}^{0}$. Diese Gr"o"sen sind
allerdings nicht unabh"angig voneinander. Die Wahl der untersten
Ordnung von $s$ kann als Wahl
eines anderen Zusammenhangs $\nabla'$ interpretiert werden,
beziehungsweise umgekehrt kann man durch Wahl eines anderen
Zusammenhangs das gleiche Sternprodukt erhalten, wenn man $s$ in der
untersten Ordnung
anpa"st. Da die Zweiform auch frei w"ahlbar ist, kann man
insbesondere die Wahl $(1 \otimes \Omega)=0$ und $s=0$ treffen. Diese
Wahl traf auch \Name{Fedosov} in seiner urspr"unglichen Konstruktion.
Ferner kann man bei dieser Wahl eine Rekursionsformel f"ur Gleichung
\eqref{eq:FedosovRekursion} angeben
\begin{align}
    \label{eq:FedosovRekursionLoesung}
    r^{(k+3)}=\delta^{-1} \left( Dr^{(k+2)}+ \frac{\im}{\lambda}
          \sum_{m=1}^{k-1} r^{(m+2)}\fpweyl r^{(k+2-m)} \right).
\end{align}
In der Arbeit \citep{neumaier:2001a} wird gezeigt, da"s man durch
die Wahl verschiedener $\Omega$ \Name{Fedosov}-Sternprodukte in jeder
m"oglichen Klasse konstruieren kann. Dies bedeutet insbesondere, da"s jedes
Sternprodukt auf einer symplektischen Mannigfaltigkeit $(M,\omega)$
"aquivalent (im Sinne von Definition
\ref{Definition:AequivalenzSternprodukte}) zu einem
\Name{Fe\-do\-sov}-Stern\-pro\-dukt ist.  
\end{bemerkungen}

\subsubsection{Die Konstruktion des Sternprodukts}

Wir k"onnen nun den Kern von $\alg{D}$ im schiefsymmetrischen Grad
$0$ bestimmen. Aufgrund der Eigenschaft $\alg{D}^2 = 0$ ist
$\alg{D}$ als eine deformierte Version von $\delta$ aufzufassen. Im
folgenden werden wir uns die Gleichung 

\begin{align}
    \label{eq:DeformiertesPoincareLemma}
    \alg{D}\alg{D}^{-1}  + \alg{D}^{-1} \alg{D} + \frac{1}{\id -A}
    \sigma  = \id \qquad \text{mit} \quad A=[\delta^{-1}, D + \tfrac{\im}{\lambda} \ad(r)]
\end{align}
 
ansehen. Sie ist eine deformierte Version von $\delta \delta^{-1} +
\delta^{-1} \delta + \sigma = \id$ und $\alg{D}^{-1}$ ist gegeben durch
\begin{align}
    \label{eq:InversesFedosovDerivation}
    \alg{D}^{-1}=-\delta^{-1}\frac{1}{\id - \left[ \delta^{-1}, D +
          \tfrac{\im}{\lambda} \ad(r) \right]},
\end{align}
einem wohldefinierten Endomorphismus von $\alg{W} \otimes
\mit{\Lambda}^{\bullet}$.

\begin{definition}[Die \Name{Fedosov}-\Name{Taylor}-Reihe]
\label{Definition:DieFedosovTaylorReihe}
F"ur die \Name{Fedosov}-Derivation $\alg{D}= -\delta + D +
\tfrac{\im}{\lambda} \ad(r)$ ist die
\Name{Fedosov}-\Name{Taylor}-Reihe  definiert durch
\begin{align}
    \label{eq:FedosovTaylorReihe}
    \Cinff{M} \ni f \mapsto \tau(f) = \frac{1}{\id - [\delta^{-1}, D +
    \tfrac{\im}{\lambda} \ad(r)]} f \in \alg{W}
\end{align}
\end{definition}

\begin{lemma}
   Sei $f \in \Cinff{M}$, dann gilt f"ur die
   \Name{Fedosov}-\Name{Taylor}-Reihe $\tau(f)$ 
   \begin{align}
       \sigma(\tau(f)) = f
   \end{align}
und 
   \begin{align}
     \tau(f) = \sum_{n=0}^{\infty} \left[ \delta^{-1}, D +
    \frac{\im}{\lambda} \ad(r) \right]^{n} f = f + \de f \otimes 1 + \cdots. 
   \end{align}
\end{lemma}

Damit haben wir die gesuchte Bijektion zwischen $\Cinff{M}$ und $\ker
\alg{D} \cap \alg{W}$ gefunden, die wir ben"otigen, um das
\Name{Fedosov}-Sternprodukt zu konstruieren. 

\begin{satz}[Das \Name{Fedosov}-Sternprodukt]
\label{Satz:FedosovSternprodukt}
\index{Fedosov-Sternprodukt@\Name{Fedosov}-Sternprodukt}
    Sei $a \in \alg{W}$. Es gilt $\alg{D}a=0$ genau dann, falls
    $a=\tau(\sigma(a))$, womit $\tau:\Cinff{M} \to \ker \alg{D} \cap
    \alg{W}$ eine $\fieldf{C}$-lineare Bijektion mit dem Inversen $\sigma$
    ist. Das Produkt
\begin{equation}
    \begin{aligned}
        \label{eq:FedosovSternprodukt1}
        f\star g & = \sigma (\tau (f) \fpweyl \tau (g) ) \\ & = fg+
        \frac{\im \lambda}{2} \{f, g\} + \alg{O}(\lambda^{2})
    \end{aligned}
\end{equation}
ist ein differentielles und nat"urliches Sternprodukt, das genau dann
\Name{Hermite}sch ist, falls $\cc{\Omega} = \Omega$ und $\cc{s}=s$
reell gew"ahlt wurden. 
\end{satz}

\subsection{Die \Name{Fedosov}-Konstruktion auf Vektorb"undeln}
\label{sec:FedosovKonstruktionVektorbuendel}
\index{Fedosov-Konstruktion@\Name{Fedosov}-Konstruktion!auf Vektorbuendeln@auf Vektorb\"undeln|(}
Im n"achsten Schritt wollen wir die \Name{Fedosov}-Konstruktion auf
eine bestimmte Art von Bimoduln fortsetzen, bei dem der Bimodul ein
(ggf.~\Name{Hermite}sches) Vektorb"undel ist. Diese Konstruktion geht auf
die Arbeiten \citep{bursztyn.waldmann:2000b} und
\citep{waldmann:2002b} zur"uck. Eine Verallgemeinerung dieser
Konstruktion findet man in der Diplomarbeit von \citet{weiss:2006a}.

Neben dem in Definition~\ref{Definition:BuendelFormaleWeylAlgebren} definierten
B"undel aller formalen \Name{Weyl}-Algebren, wollen wir noch
zu\-s"atz\-lich die folgenden B"undel definieren. Dazu sei
$(\bundle{E}{\pi}{(M,\omega,\nabla)})$ ein Vektorb"undel "uber einer
symplektischen Mannigfaltigkeit mit Zusammenhang, das
ebenfalls mit einem linearen Zusammenhang\index{Zusammenhang!linearer}
$\lconnE$ ausgestattet ist. Dieser Zusammenhang $\lconnE$ induziert
einen Zusammenhang $\lconnEnd{E}$ auf dem
Endomorphismenb"undel\index{Zusammenhang!Endomorphismenbuendel@auf Endomorphismenb\"undel} "uber $E$, vergleiche Lemma 
\ref{Lemma:ZusammenhangEndomorphismen}.   

Wir f"uhren im folgenden zwei neue B"undel ein
\begin{align}
\alg{W} \otimes \mit{\Lambda}^{\bullet} \otimes  \alg{E} & = \left(\prod_{k=0}^{\infty}
    \schnitt{S^{k}T^{\ast}M \otimes \Lambda^{\bullet}T^{\ast}M \otimes E} \right)\FP, \\ \alg{W} \otimes \mit{\Lambda}^{\bullet} \otimes \END{E}  & = \left(\prod_{k=0}^{\infty}
    \schnitt{S^{k}T^{\ast}M \otimes \Lambda^{\bullet}T^{\ast}M \otimes \End{E}} \right)\FP. 
\end{align}

\begin{lemma}[Eigenschaften der Tensoralgebren]
~\vspace{-5mm}
\begin{compactenum}
\item Die Algebra $\alg{W} \otimes \mit{\Lambda}^{\bullet} \otimes
    \END{E}$ ist assoziativ. In den ersten beiden
    Faktoren f"allt die Struktur mit der von $\alg{W} \otimes
    \mit{\Lambda}^{\bullet}$ zusammen, im letzten ist es die
    Komposition von Endomorphismen. 

\item Das B"undel $\alg{W} \otimes \mit{\Lambda}^{\bullet} \otimes
    \alg{E}$ ist ein Bimodul f"ur $\alg{W} \otimes
    \mit{\Lambda}^{\bullet} \otimes \END{E}$ von links und f"ur
    $\alg{W} \otimes \mit{\Lambda}^{\bullet}$ von rechts. Die
    Modulstrukturen sind gegeben durch 
\begin{equation}
\begin{aligned}
\label{eq:ModulMultiplikationWeylBuendel}
(f \otimes \alpha \otimes A) \cdot' (g \otimes \beta \otimes s) & := f
\vee g \otimes \alpha \wedge \beta \otimes As, \\ (g \otimes \beta
\otimes s)\cdot (h \otimes \gamma) &:= g \vee h \otimes \beta \wedge
\gamma \otimes s. 
\end{aligned}
\end{equation}
f"ur alle $(f \otimes \alpha \otimes A) \in \alg{W} \otimes
\mit{\Lambda}^{\bullet} \otimes \END{E}$, $(g \otimes \beta \otimes s)
\in \alg{W} \otimes \mit{\Lambda}^{\bullet} \otimes  \alg{E}$ und $(h
\otimes \gamma) \in \alg{W} \otimes \mit{\Lambda}^{\bullet}$. 
\end{compactenum}
\end{lemma}

Die Endomorphismen $\delta$, $\delta^{\ast}$ und $\delta^{-1}$ werden auf
triviale Weise auf die beiden B"undel $\alg{W} \otimes
\mit{\Lambda}^{\bullet} \otimes  \alg{E}$ und $\alg{W} \otimes
\mit{\Lambda}^{\bullet} \otimes \END{E}$ erweitert, indem sie auf
Elementen das letzte Argument unangetastet lassen.
Andersherum fassen wir die Algebra $\alg{W} \otimes
\mit{\Lambda}^{\bullet}$ als $\alg{W} \otimes {\mit{\Lambda}}^{\bullet}
\otimes 1$ auf und schreiben 
 
\begin{align}
\label{eq:EndomorphismendeltaUnddeltastrichVerallgemeinert}
\delta = (1 \otimes \de x^{i} \otimes 1) i_{s}\left(\tfrac{\del}{\del
      x^{i}} \right) \quad \text{und} \quad \delta^{\ast} = (\de x^{i}
\otimes 1 \otimes 1) i_{a}\left(\tfrac{\del}{\del x^{i}} \right), 
\end{align}

so da"s wir symbolisch einen Operator f"ur alle drei B"undel
haben. 

Wir wollen nun eine Deformation der Modulstruktur angeben.  
Dazu bedienen wir uns des auf $\alg{W} \otimes
\mit{\Lambda}^{\bullet}$ eingef"uhrten \Name{Weyl-Moyal}-Produkts,
und wir definieren die deformierten Mo\-dul\-ver\-kn"up\-fun\-gen
$\circ'$ und $\circ$ der Gleichungen 
\eqref{eq:ModulMultiplikationWeylBuendel} via 

\begin{equation}
\begin{aligned}
    \label{eq:DeformierteModulMultiplikationWeylBuendel}
(f \otimes \alpha \otimes A) \circ' (g \otimes \beta \otimes s) & := f
\fpweyl g \otimes \alpha \wedge \beta \otimes As \\ (g \otimes \beta
\otimes s) \circ (h \otimes \gamma) & := g \fpweyl h \otimes \beta
\wedge \gamma \otimes s.
\end{aligned}
\end{equation}

Damit und den deformierten Produkten der Algebren $(\alg{W} \otimes
\mit{\Lambda}^{\bullet} \otimes  \END{E}$ und $\alg{W} \otimes \mit{\Lambda}^{\bullet})$ wird $\alg{W} \otimes \mit{\Lambda}^{\bullet} \otimes  \alg{E}$
zu einem $(\alg{W} \otimes \mit{\Lambda}^{\bullet} \otimes  \END{E}$, $\alg{W} \otimes
\mit{\Lambda}^{\bullet})$-Bimodul bez"uglich der deformierten
Bimodulstrukturen $\circ'$ und $\circ$.

\subsubsection{Kovariante Ableitungen, die \Name{Fedosov}-Derivation
  $\alg{D}'$ und die Deformation $\star'$}

Die weitere Vorgehensweise ist analog zur "ublichen
\Name{Fedosov}-Konstruktion: mit Hilfe der Derivationen $\delta$,
$\delta^{-1}$  und eines linearen Zusammenhangs $\lconnEnd{E}$ auf dem Endomorphismenb"undel
bzw.~$\lconnE$ auf dem Vektorb"undel konstruieren wir die
\Name{Fedosov}-Derivationen $\alg{D}'$ und $\alg{D}^{\sss E}$, so
da"s deren Quadrat verschwindet. Der Kern der Superderivationen
bildet dann eine Unteralgebra bzw.~einen Untermodul, und der Pullback der faserweisen
\Name{Weyl-Moyal}-Produkte wird dann zu einer Moduldeformation
f"uhren. In einem ersten Schritt werden wir die notwendigen
Strukturen kurz erl"autern, dann eine assoziative Deformation der
Algebra $\schnitt{\End{E}}$ betrachten und letztendlich die
Modulstrukturen deformieren.

\begin{definition}[Kovariante Ableitungen $D^{\sss E}$ und $D'$]
  Man definiert auf dem Raum $\alg{W} \otimes \mit{\Lambda}^{\bullet}
  \otimes \alg{E}$ und $\alg{W} \otimes \mit{\Lambda}^{\bullet}
  \otimes \END{E}$ kovariante Ableitungen, mittels 
 \begin{equation} 
 \begin{aligned}
      \label{eq:KovarianteAbleitungDE}
      D^{\sss E}: \alg{W} \otimes \mit{\Lambda}^{\bullet}
  \otimes \alg{E}& \to \alg{W} \otimes \mit{\Lambda}^{\bullet +1}
  \otimes \alg{E} \\  (f \otimes \alpha \otimes s) & \mapsto D(f\otimes
  \alpha) \otimes s + f \otimes \de x^{i} \wedge \alpha \otimes
  \lconnE[\frac{\del}{\del x^{i}}] s \\ &  = \lconn[\frac{\del}{\del
    x^{i}}] f \otimes \de x^{i} \wedge \alpha \otimes s +f \otimes \de \alpha
  \otimes s + f \otimes \de x^{i} \wedge \alpha \otimes \lconnE[\frac{\del}{\del
    x^{i}}] s.
  \end{aligned}
\end{equation}
und analog definiert man die kovariante Ableitung $D':=D^{\sss
  \End{E}}$ auf $\alg{W} \otimes \mit{\Lambda}^{\bullet} \otimes \END{E}$ 
\begin{equation}
    \begin{aligned}
        \label{eq:KovarianteAlbleitungDEndE}
      D' : \alg{W} \otimes \mit{\Lambda}^{\bullet} \otimes \END{E} & \to
      \alg{W} \otimes \mit{\Lambda}^{\bullet + 1} \otimes \END{E} \\   
       (f \otimes \alpha \otimes A)  & \mapsto  D(f\otimes \alpha) \otimes A +
        f \otimes \de x^{i} \wedge \alpha \otimes \lconnEnd[\frac{\del}{\del
    x^{i}}]{E} A \\ & = \lconn[\frac{\del}{\del x^{i}}] f \otimes \de
  x^{i} \wedge \alpha \otimes A + f \otimes \de \alpha  \otimes
  A + f \otimes \de x^{i} \wedge \alpha \otimes \lconnEnd[\frac{\del}{\del x^{i}}]{E} A.  
    \end{aligned}
\end{equation}
\end{definition}

\begin{lemma}[Eigenschaften von $D^{\sss E}$ und $D'$]
\label{Lemma:EigenschaftenvonDStrichundD}
    Sei $R^{\sss E}$ die Kr"ummung des Vektorb"undels
    $\bundle{E}{\pi}{M}$. Die kovarianten Ableitungen $D^{\sss E}$ und
    $ D'$ haben die folgenden Eigenschaften: 
  \begin{compactenum}
  \item $D^{\sss E}$ ist superderivativ bez"uglich der
      deformierten Modulstrukturen. Sei nun  $a \in \alg{W} \otimes
      \mit{\Lambda}^{\bullet} \otimes \END{E}$, $\Psi \in  \alg{W}
      \otimes \mit{\Lambda}^{\bullet} \otimes \alg{E}$ und $b \in
      \alg{W} \otimes \mit{\Lambda}^{\bullet}$, dann gilt 
      \begin{align*}
          D^{\sss E}(a \circ' \Psi) & = (D' a) \circ' \Psi + (-1)^{\dega a}
          a \circ' (D^{\sss E} \Psi), \\ D^{\sss E} (\Psi \circ b) & =
          (D^{\sss E} \Psi) \circ b + (-1)^{\dega \Psi} \Psi \circ (Db), 
      \end{align*}
  \item $({D'})^{2} = \tfrac{\im}{\lambda} \ad(R-\im \lambda
        R^{\sss E})$ und $(D^{\sss E})^{2} = \tfrac{\im}{\lambda} \ad(R) +
        R^{\sss E}$,    
  \item $\delta R^{\sss E} = 0$, $D'R^{\sss E}=0$.

  \end{compactenum}
\end{lemma}

Die Konstruktion der \Name{Fedosov}-Derivation $\alg{D}':
\alg{W} \otimes\mit{\Lambda}^{\bullet} \otimes \END{E} \to \alg{W}
\otimes\mit{\Lambda}^{\bullet} \otimes \END{E}$ verl"auft
komplett analog zu der Konstruktion von $\alg{D}$, und der Ansatz ist  

\begin{align}
    \label{eq:FedosovDerivationDprime}
    \alg{D}' = - \delta + D' + \frac{\im}{\lambda}\ad(r'),
\end{align}

wobei $r'$ ein Element in $\alg{W} \otimes \mit{\Lambda}^{\bullet}
\otimes \END{E}$ ist.

\begin{satz}[{\citep[Sect.~5.3]{fedosov:1996a}}]
Sei $\Omega = \sum_{n=1}^{\infty} \lambda^{n} \Omega_{n}$ eine
geschlossene Zweiform, dann existiert genau ein $r'\in \alg{W} \otimes
\mit{\Lambda}^{\bullet} \otimes \END{E}$ mit dem schiefsymmetrischen
Grad $+1$ und einem totalen Grad $\ge 3$ so da"s
\begin{align*}
\delta r' = R - \im \lambda R^{\sss E} + D'r' +
\tfrac{\im}{\lambda} r' \fpweyl r' + (1 \otimes \Omega \otimes 1) \quad
\text{und} \quad \delta^{-1} r' = 0
\end{align*}
erf"ullt ist. Damit ist $(\alg{D}')^2 = 0$.

\end{satz}

Der Unterschied zu $\alg{D}$ besteht darin, da"s sowohl die
Kr"ummung des Vektorb"undels $E$ als auch die der Mannigfaltigkeit $M$
kompensiert werden mu"s, was den zu\-s"atz\-li\-chen Term erkl"art
(vgl.~hierzu Satz~\ref{Satz:FedosovKonstruktion}). 

Mit der Projektion $\sigma': \ker \alg{D}' \cap \ker \dega \to
\schnittf{\End{E}}$ und deren Umkehrabbildung $\tau'$ , die wir auch
als \Name{Fedosov}-\Name{Taylor}-Reihe bezeichnen, sind wir in der Lage,
eine assoziative Deformation auf $\schnittf{\End{E}}$ anzugeben.

\begin{satz}
Die Abbildung $\sigma': \ker \alg{D}' \cap \ker \dega \to \schnittf{\End{E}}$
ist eine $\fieldf{C}$-lineare Bijektion.    
\end{satz}

Aus der Tatsache, da"s der Kern einer Superderivation eine
Unteralgebra ist, folgt nun, da"s wir das faserweise
\Name{Weyl-Moyal}-Produkt $\circ'$ auf $\alg{W} \otimes
\mit{\Lambda}^{\bullet} \otimes \END{E}$ mittels $\sigma'$ und $\tau'$
auf $\schnittf{\End{E}}$ zur"uckziehen k"onnen:
\begin{align}
    \label{eq:SternproduktAufSchnittenInEndE}
    A \star' B = \sigma' (\tau' (A) \circ \tau' (B)),
\end{align}

f"ur $A,B \in \schnittf{\End{E}}$. Damit wird $(\schnittf{\End{E}},
\star')$ zu einer assoziativen, deformierten Algebra. 

\begin{bemerkung}
A priori ist nicht klar, da"s $r'$ von der Form $r'= \sum_{n =
  0}^{\infty} \lambda^{n} r_{n}'$ ist, da das undeformierte Produkt
auf $\alg{W} \otimes \mit{\Lambda}^{\bullet} \otimes \END{E}$ schon
nichtkommutativ ist. Es  k"onnten in der Rekursionsformel f"ur
$r'$ durch den $\tfrac{\im}{\lambda}$-Term beliebig viele 
negative Potenzen von $\lambda$ generiert werden. Dies kann man a
posteriori allerdings ausschlie"sen. Eine genaue Betrachtung findet
man in \citep{bursztyn.waldmann:2000b}, die in der Kategorie der formalen
\Name{Laurent}-Reihen in $\lambda$ rechnen, jedoch a posteriori zeigen, da"s
keine negativen Potenzen von $\lambda$ auftauchen.    
\end{bemerkung}

\subsubsection{Die \Name{Fedosov}-Derivation $\alg{D}^{\sss E}$ und
  die deformierten Bimodulstrukturen $\bullet'$ und $\bullet$}

Wir haben bislang zwei Algebren deformiert: zum einen die
Funktionenalgebra $(\Cinff{M}, \star)$ und zum anderen die Schnitte
im Endomorphismenb"undel $(\schnittf{\End{E}}, \star')$. Dabei
m"ussen wir f"ur die folgenden Betrachtungen beachten, da"s wir
in beiden F"allen dieselbe
geschlossene Zweiform $\Omega$ gew"ahlt haben m"ussen, da sonst die
folgende Konstruktion nicht m"oglich w"are.

Es liegt nun nahe, auch die Bimodulstrukturen des Bimoduls
$\bimodo{\schnitt{\End{E}}}{\schnitt{E}}{\Cinf{M}}$ zu
deformieren. Die wesentliche Arbeit hierzu haben wir bereits
geleistet, es bleibt, eine geeignete \Name{Fedosov}-Derivation
$\alg{D}^{\sss E}$ sowie eine \Name{Fedosov}-\Name{Taylor}-Reihe
$\tau^{\sss E}$ anzugeben. 

Wir definieren $\alg{D}^{\sss E}:  \alg{W} \otimes \mit{\Lambda}^{\bullet}
\otimes \alg{E} \to \alg{W} \otimes \mit{\Lambda}^{\bullet +1}
\otimes \alg{E}$ mittels 

\begin{align} 
   \label{eq:FedosovDerivationDE}
 \alg{D}^{\sss E} = -\delta + D^{\sss E} + \frac{\im}{\lambda} \ad(r) +
 r^{\sss E}, 
\end{align}
wobei $r^{\sss E}= \tfrac{\im}{\lambda}(r' -r)$ ist.

\begin{satz}[\Name{Fedosov}-Derivation $\alg{D}^{\sss E}$ und die
    \Name{Fedosov-Taylor}-Reihe $\tau^{\sss E}$] 
    Die \Name{Fedosov}-Derivation $\alg{D}^{\sss E}$ hat die folgenden
    Eigenschaften
    \begin{compactenum}
    \item $\alg{D}^{\sss E}$ ist gradiert derivativ bez"uglich der
        deformierten Bimodulstrukturen $\circ'$ und $\circ$
        \begin{equation}
            \begin{aligned}
                \alg{D}^{\sss E}(a \circ' \Psi) & = (\alg{D}' a) \circ' \Psi +
                (-1)^{\dega a} a \circ' (\alg{D}^{\sss E} \Psi), \\
                \alg{D}^{\sss 
                  E} (\Psi \circ b) & = (\alg{D}^{\sss E} \Psi) \circ b +
                (-1)^{\dega \Psi} \Psi \circ (\alg{D}b),  
            \end{aligned}
        \end{equation}
f"ur $a \in \alg{W} \otimes \mit{\Lambda}^{\bullet} \otimes \END{E}$,
$\Psi \in  \alg{W} \otimes \mit{\Lambda}^{\bullet} \otimes \alg{E}$
und $b \in  \alg{W} \otimes \mit{\Lambda}^{\bullet}$. 
    
\item Nilpotent von der Stufe $2$, d.~h.~$(\alg{D}^{\sss E})^2 = 0$.
\end{compactenum}
Desweiteren ist
\begin{align}
    \sigma^{\sss E}: \ker \alg{D}^{\sss E} \cap \ker \dega \to
    \schnittf{E} 
\end{align}
eine $\fieldf{C}$-lineare Bijektion, deren Inverses die
\Name{Fedosov-Taylor}-Reihe $\tau^{\sss E}$ ist. Diese kann explizit
rekursiv f"ur ein Element $s\ \in \schnittf{E}$ berechnet werden. 
\end{satz}

\begin{proof}
    Der Beweis hierzu findet sich in \citep{waldmann:2002b}.
\end{proof}

Damit sind wir in der Lage, auch die Bimodulstrukturen zu
deformieren, und es gilt folgendes Korollar.

\begin{korollar}[Der deformierte Bimodul {$\bimodo{(\schnittf{\End{E}},
        \star')}{(\schnittf{E},\bullet',\bullet)}{(\Cinff{M}, \star)}$}]
    Die Schnitte $\schnittf{E}$ werden ein Bimodul f"ur $\star'$ und
    $\star$ mittels
    \begin{align}
        A \bullet' s := \sigma^{\sss E}(\tau'(A) \circ' \tau^{\sss
          E}(s)) \quad \text{und} \quad s \bullet f := \sigma^{\sss E}
        (\tau^{\sss E}(s) \circ \tau(f)), 
     \end{align}
f"ur $A \in \schnittf{\End{E}}$, $s\in \schnittf{E}$ und $f\in
\Cinff{M}$. 
\end{korollar}

\begin{definition}[\Name{Fedosov}-Bimodul]
    \label{Definition:FedosovBimodul}
Man nennt einen deformierten Bimodul $\bimodo{(\schnittf{\End{E}},
  \star')}{(\schnittf{E},\bullet',\bullet)}{(\Cinff{M}, \star)}$, der
mittels der \Name{Fedosov}-Konstruktion f"ur Vektorb"undel 
erhalten werden kann, einen {\em \Name{Fedosov}-Bimodul}.
\index{Fedosov-Bimodul@\Name{Fedosov}-Bimodul}  
\end{definition} 

Ein interessanter Fall liegt vor, wenn das Vektorb"undel $E$ ein
komplexes Geradenb"undel\index{Geradenbuendel@Geradenb\"undel!komplexes}
$E=\bundle{L}{\pi}{M}$ ist. Dann ist die von links wirkende 
Algebra $(\schnittf{\End{L}},\star')$ auch eine Sternproduktalgebra\index{Sternproduktalgebra}, da
$\schnittf{\End{L}} = \Cinff{M}$. 

\begin{lemma}[Charakteristische Klassen von $\star'$ und $\star$]
\index{charakteristische Klasse}
Sei $\bundle{L}{\pi}{M}$ ein Geradenb"undel. Die charakteristischen
Klassen der Sternprodukte $\star'$ und $\star$ h"angen "uber die
folgende Beziehung zusammen
\begin{align*}
    c(\star') = c(\star) + 2 \pi \im c_{1}(L),
\end{align*}
wobei $c_{1}(L)$ die erste \Name{Chern}-Klasse von $L$ ist.
\end{lemma} 
\index{Fedosov-Konstruktion@\Name{Fedosov}-Konstruktion|)}
\index{Fedosov-Konstruktion@\Name{Fedosov}-Konstruktion!auf Vektorbuendeln@auf Vektorb\"undeln}

\subsection{$G$-Invariante \Name{Fedosov}-Konstruktion}
\label{sec:HInvarianteFedosovKonstruktion}
\index{Fedosov-Konstruktion@\Name{Fedosov}-Konstruktion!G-invariante@$G$-invariante|(}
Nachdem wir in Kapitel~\ref{sec:FedosovKonstruktion} die
\Name{Fedosov}-Kon\-struk\-tion detailliert besprochen und in 
\ref{sec:FedosovKonstruktionVektorbuendel} eine Verallgemeinerung
f"ur Vektorb"undel angegeben haben, werden wir nun durch die Wahl
von speziellen Zusammenh"angen (bzw.~der Wahl von Klassen spezieller Zusammenh"ange) und einer
invarianten Zweiform in der Lage sein, unter einer \Name{Lie}-Gruppe
$G$\index{Lie-Gruppe@\Name{Lie}-Gruppe} oder unter einer \Name{Lie}-Algebra
$\LieAlg{g}$\index{Lie-Algebra@\Name{Lie}-Algebra}
\index{Algebra!Lie-Algebra@\Name{Lie}-Algebra} invariante 
Sternprodukte und Bimoduldeformationen zu konstruieren. Dabei
bezeichnen wir ein Sternprodukt als {\em $G$-invariant}, wenn es
eine Wirkung einer Gruppe $G$ auf $M$ gibt, so da"s 
\begin{align*}
    g.(f \star h) = (g.f) \star (g.h) 
\end{align*}
f"ur alle $f,h \in \Cinff{M}$ und $g\in G$. F"ur eine detaillierte Beschreibung
sowie eine ausf"uhrliche Motivation, solche Sternprodukte zu betrachten,
verweisen wir auf Kapitel~\ref{sec:HInvarianteSternprodukte}.

Die M"oglichkeit der Konstruktion $G$-invarianter Sternprodukte\index{Sternprodukt!G-invariantes@$G$-invariantes} wurde schon in \citep{fedosov:1996a} erw"ahnt. Eine $\LieAlg{g}$-in\-va\-rian\-te \Name{Fedosov}-Konstruktion geht auf die Arbeit von
\citet{mueller-bahns.neumaier:2004a} zur"uck. Die hier vorgestellte
Konstruktion ist eine Verallgemeinerung, da sie eine
$G$-invariante Formulierung der Deformation von Vektorb"undeln
liefert. Dazu ben"otigen wir auf dem Vektorb"undel $\bundle{E}{\pi}{M}$ eine Wirkung,
die von der $G$-Wirkung auf der Mannigfaltigkeit $M$ kommt, also eine
Hebung darstellt. Ist das Vektorb"undel mit einer solchen $G$-Wirkung
versehen, so erbt das Endomorphismenb"undel\index{Endomorphismenbuendel@Endomorphismenb\"undel}
auf kanonische Weise die Wirkung.

Die Idee der Konstruktion ist denkbar einfach: durch die Wahl $G$-
bzw.~$\LieAlg{g}$-invarianter Eingangsdaten, sprich einem invarianten, torsionsfreien
symplektischen Zusammenhang $\nabla$, sowie einer invarianten symplektische
Zweiform $\omega$, erh"alt man auch invariante Sternprodukte
bzw.~invariante Moduldeformationen. Die geometrischen Grundlagen zur
$\LieAlg{g}$- und $G$-Invarianz sind in Anhang~\ref{sec:GundgInvarianz} zusammengetragen.  

Wir werden nun in aller K"urze die $G$-Invarianz der entscheidenden
Strukturen aufzeigen.

\begin{lemma}[Invarianz der Endomorphismen $\delta, \delta^{\ast},
    \delta^{-1}$] 
Die Differentialoperatoren $\delta$ und $\delta^{\ast}$ sind unter einer
$G$-Wirkung invariant. Da die $G$-Wirkung graderhaltend, so ist auch
$\delta^{-1}$ $G$-invariant
   \begin{align} 
   g.\delta(a) = \delta(g.a), \quad  g.\delta^{\ast}(a) =
   \delta^{\ast} (g.a), \quad  g.\delta^{-1}(a) = \delta^{-1}(g.a).
    \end{align}
\end{lemma}

\begin{proof}
 Es sei nun $a$ von der Form $a= f \otimes \alpha \otimes u$,
     wobei $u$ ein Platzhalter je nach B"undel ist, d.~h.~wir setzen
     
\begin{align*}
a = f \otimes \alpha \otimes u :=\begin{cases} f \otimes \alpha \otimes 1
    \qquad & \text{f"ur} \qquad a\in \alg{W} \otimes \mit{\Lambda}^{\bullet}
    \\ f \otimes \alpha \otimes s & \text{f"ur} \qquad a\in \alg{W}
    \otimes \mit{\Lambda}^{\bullet} \otimes \alg{E} \\ f \otimes
    \alpha \otimes A & \text{f"ur} \qquad a\in \alg{W} \otimes
    \mit{\Lambda}^{\bullet} \otimes \END{E}. \end{cases} 
\end{align*}

\begin{align*}
        g.(\delta a) &= g.\left( i_s\left( \tfrac{\del}{\del^{x^i}}
            \right) f \otimes \de x^i \wedge 
        \alpha \otimes u \right) \\ &= g.i_s \left( \tfrac{\del}{\del^{x^i}}
    \right) f \otimes g.(\de x^i \wedge \alpha) \otimes g.u \\ & = \phi_{g^{-1}}^\ast
        i_s \left(\tfrac{\del}{\del_{x^i}} \right) f \otimes \phi_{g^{-1}}^\ast (\de x^i \wedge
        \alpha) \otimes g.u \\ & = i_s\Big( \underbrace{\phi_{g^{-1}}^\ast
        i_s\left( \tfrac{\del}{\del_{x^i}}
        \right)}_{\tfrac{\del}{\del_{\tilde{x}^i}}}\Big) \underbrace{ 
        \phi^\ast_{g^{-1}}f}_{g.f} \otimes \de \underbrace{ 
        \phi^\ast_{g^{-1}}x^i}_{\tilde{x}^i}  \wedge
      \underbrace{\phi^\ast_{g^{-1}} \alpha}_{g.\alpha}  \otimes g.u \\
      &= i_s\left( \tfrac{\del}{\del_{\tilde{x}^i}} \right) g.f \otimes \de \tilde{x}^i \wedge
      g.\alpha \otimes g.u \\ & = \delta (g.a)  
\end{align*} 
Analog dazu zeigt man, da"s $g.(\delta^{\ast} a)= \delta^{\ast}
(g.a)$ ist. Damit ist ebenfalls der Beweis f"ur $\delta^{-1}$
erbracht, falls die $G$-Wirkung die Grade erh"alt.
Wie die Wirkung auf den Endomorphismen und den Schnitten wirkt haben
wir in Anhang~\ref{sec:GundgInvarianz} zusammengestellt. 

\end{proof}

\begin{lemma}["Aquivarianz des faserweisen \Name{Weyl-Moyal}-Produkts $\fpweyl$]
Gegeben sei ein faserweises \Name{Weyl-Moyal}-Produkt nach Definition
\ref{Definition:FaserweisesWeylMoyalProdukt} mit einer $G$-invarianten
\Name{Poisson}-Struktur. Das faserweise \Name{Weyl-Moyal}-Produkt ist dann
$G$-invariant, d.~h.~es gilt     
   \begin{align} 
   g.(a \fpweyl b) = (g.a) \fpweyl (g.b).  
    \end{align}
\end{lemma}

\begin{proof}
    Der Beweis ist eine einfache Rechnung
    \begin{align*}
            g.(a\fpweyl b) &=  \mu \fpweyl \eu^{\tfrac{\im \hbar}{2}
          g.\Lambda^{k\ell}_{p} i_s\left(g. \tfrac{\del}{\del x^k} \right)
          \otimes i_s\left(g. \tfrac{\del}{\del x^\ell} \right)} 
        (g.a) \otimes (g.b) \\ & = \mu \fpweyl \eu^{\tfrac{\im \hbar}{2}
          \tilde{\Lambda}^{k\ell}_{p} i_s\left(\tfrac{\del}{\del \tilde{x}^k} \right)
          \otimes i_s\left(\tfrac{\del}{\del \tilde{x}^\ell} \right)} 
        (g.a) \otimes (g.b) \\ & = (g.a) \fpweyl (g.b).
    \end{align*}
\end{proof}

\begin{bemerkung}
Nat"urlich ist nicht nur das faserweise \Name{Weyl-Moyal}-Produkt
    $G$-invariant, sondern aufgrund der Tensorstruktur auch die Fortsetzung
    auf die B"undel $\alg{W} \otimes \mit{\Lambda}^{\bullet}$,
    $\alg{W} \otimes \mit{\Lambda}^{\bullet} \otimes \alg{E}$ und $\alg{W} \otimes
    \mit{\Lambda}^{\bullet} \otimes \END{E}$.
\end{bemerkung}

\begin{lemma}[$G$-Invarianz der kovarianten Differentiale $D$ , $D'$
    und $D^{\sss E}$]
Seien $\nabla$ und $\lconnE$ unter der Gruppe $G$-invariante
Zusammenh"ange auf $TM$ bzw.~$\bundle{E}{\pi}{M}$. Die kovarianten
Differentiale $D$, $D'$ und $D^{\sss E}$ sind unter der $G$-Wirkung
invariant.  
\end{lemma}

\begin{proof}
Wir rechnen nach
\begin{align*} 
    g.(D^{\sss E}\Psi) &= g.( D(f \otimes \alpha) \otimes s + f
    \otimes \de x^{i} \wedge \alpha \otimes \lconnE[\frac{\del}{\del
      x^{i}}] s ) \\ & = g.(\lconn[\frac{\del}{\del{x^i}}] f
\otimes \de x^i \wedge \alpha \otimes s) + g.(f \otimes \de \alpha \otimes
s) + g.(f \otimes \de x^{i} \wedge \alpha \otimes
\lconnE[\frac{\del}{\del x^{i}}] s ) \\ & = g.(\lconn[\frac{\del}{\del{x^i}}] f)
\otimes g.(\de x^i \wedge \alpha) \otimes g.s + g.f \otimes
g.\de \alpha \otimes g.s + g.f \otimes g.(\de x^{i} \wedge \alpha)
\otimes g. (\lconnE[\frac{\del}{\del x^{i}}] s)  \\ & =
\lconn[g.\frac{\del}{\del{x^i}}] g.f \otimes g.\de x^{i} \wedge g.\alpha
\otimes g.s + g.f \otimes \de g.\alpha \otimes g.s + g.f \otimes g.\de x^{i}
\wedge g. \alpha \otimes \lconnE[g. \frac{\del}{\del x^{i}}] g.s \\ & = 
  \lconnE[\frac{\del}{\del{\tilde{x}^i}}] g.f \otimes 
  \de \tilde{x}^i \wedge g.\alpha \otimes g.s + g.f \otimes \de 
  g.\alpha \otimes g.s + g.f \otimes \de \tilde{x}^{i} \wedge \alpha \otimes
  \lconnE[\frac{\del}{\del \tilde{x}^{i}}] g.s \\ & = D^{\sss E} (g.\Psi). 
  \end{align*}
Damit haben wir gleichzeitig gezeigt, da"s auch $D$ mit der
Wirkung $g.$ vertauscht, indem wir den zweiten Term der ersten Zeile
\glqq vergessen\grqq{}. Da der Zusammenhang $\lconnEnd{E}$ auf dem
Endomorphismenb"undel auch invariant ist, wenn
$\lconnE$ invariant ist (siehe Lemma
\ref{Lemma:GInvarianterZusammenhangAufEndomorphismenbuendel}), folgt
die Invarianz von $D'$ automagisch.    
\end{proof}

\begin{lemma}[$G$-Invarianz der Kr"ummungen $R$ und $R^{\sss E}$]
Sei $(M,\omega, \nabla)$ eine symplektische Mannigfaltigkeit mit
einem $G$-invarianten Zusammenhang, und $(\bundle{E}{\pi}{M},
\lconnE)$ sei ein Vektorb"undel "uber
der Mannigfaltigkeit $M$ mit einem $G$-invarianten Zusammenhang $\lconnE$. Die
sym\-plek\-ti\-sche Kr"um\-mung $R$ nach Definition 
und Lemma~\ref{Lemma:SymplektischeKruemmung} und die Kr"ummung
$R^{\sss E}$ sind $G$-invariant, d.~h.~es gilt  
\begin{align}
 g.R=R \quad \text{und} \quad g.R^{\sss E} = R^{\sss E}.   
\end{align}
\end{lemma} 

\begin{proof}
F"ur den Beweis rechnen wir nach
\begin{gather}
\begin{aligned}
\label{eq:GInvarianzvonDRechnung}
    g.(Da) = D(g.a) \quad \text{woraus dann folgt, da"s} \quad Da =
    g.(D(g^{-1}.a)) \\ \Rightarrow D^{2}a  = 
    \frac{\im}{\lambda} \ad(R) (a) = g.(D^2 (g^{-1}.a)) =
    \frac{\im}{\lambda} g.(\ad(R)(g^{-1}.a)) = \frac{\im}{\lambda}
    \ad(g.R)(a)
\end{aligned}
\end{gather}
Daraus folgt unmittelbar, da"s $R-g.R$ zentral ist. Da andererseits
f"ur den symmetrischen Grad $\degs (R-g.R) = 2(R-g.R)$ gilt, mu"s
$R-g.R=0$ sein. Analog verh"alt es sich f"ur $R^{\sss E}$. 
Nach Lemma~\ref{Lemma:EigenschaftenvonDStrichundD} gilt
\begin{align}
    (D')^{2} = \frac{\im}{\lambda} \ad(R-\im \lambda R^{\sss E}) \quad
    \text{und} \quad (D^{\sss E})^2 = \frac{\im}{\lambda} \ad(R) +
    R^{\sss E}.
\end{align}
Mit einer zu Rechnung~\eqref{eq:GInvarianzvonDRechnung} analogen
Betrachtung ergibt sich die Behauptung.  
\end{proof}

Wir gehen nun davon aus, da"s sowohl das Element $s \in
\alg{W}_{3} \otimes \mit{\Lambda}^{0}$ als auch
die geschlossene Zweiform $\Omega \in  \lambda
\schnitt{\Lambda^2T^{\ast}M}\FP$ invariant unter der $G$-Wirkung ist, d.~h.~es gilt 
\begin{align}
    \label{eq:GInvarianzvonOmegaUnds}
    g.s=s \quad \text{und} \quad  g.\Omega=\Omega
\end{align}

Solche finden wir nat"urlich immer, da wir insbesondere $s=0$ und
$\Omega=0$ w"ahlen k"onnen. Wir w"ahlen $s=0$, dann gilt
$0=g.(\delta^{-1}r) = \delta^{-1}(g.r)$. F"ur das Element $r$ ist somit 

\begin{align*}
   \delta(g.r) = g.(\delta r) &= g.(Dr) + \frac{\im}{\lambda} g.(r \fpweyl
   r) + g.R + g.(1 \otimes \Omega \otimes 1) \\ &= D(g.r) + \frac{\im}{\lambda}
   (g.r \fpweyl g.r) + g.R + 
   (1 \otimes g.\Omega \otimes 1) \\ &= D(g.r) + \frac{\im}{\lambda} (g.r
   \fpweyl g.r) + R + (1 \otimes \Omega \otimes 1), 
\end{align*}
woraus aufgrund der Eindeutigkeit $g.r=r$ folgt. Analog verfahren wir mit $r'$, und damit haben
wir auch die Invarianz von $r^{\sss E}$ geschenkt bekommen. 

Nun sind wir in der Lage die $G$-Invarianz der
\Name{Fedosov}-Derivationen $\alg{D}$, $\alg{D}'$ und $\alg{D}^{\sss E}$ zu zeigen.

\begin{lemma}[$G$-Invarianz der \Name{Fedosov}-Derivationen
    $\alg{D}$, $\alg{D}'$ und $\alg{D}^{\sss E}$]
Die \Name{Fedosov}-Derivationen $\alg{D}$, $\alg{D}'$ und
$\alg{D}^{\sss E}$ sind $G$-invariant.
\end{lemma}

\begin{proof}
\begin{align*}
     g.(\alg{D}b) &= -g.(\delta b) + g.(Db) +\frac{\im}{\lambda}
     g.(\ad(r) b) \\ & = - \delta(g.b) + D(g.b) + \frac{\im}{\lambda}
         \ad(g.r)(g.b) \\ & = - \delta(g.b) + D(g.b) + \frac{\im}{\lambda}
         \ad(r)(g.b) \\ &= \alg{D}(g.b) 
\end{align*}
Analog dazu zeigt man $g.(\alg{D}' a) = \alg{D}' (g.a)$ und
$g.(\alg{D}^{\sss E} \Psi) = \alg{D}^{\sss E} (g.\Psi)$. 
\end{proof}

\begin{satz}[$G$-invariantes \Name{Fedosov}-Sternprodukt $\star$]
\label{Satz:GInvariantesFedosovSternprodukt}
\index{Fedosov-Sternprodukt@\Name{Fedosov}-Sternprodukt!G-invariantes@$G$-invariantes}
Gegeben eine symplektische
$G$-Mannigfaltigkeit\index{G-Mannigfaltigkeit@$G$-Mannigfaltigkeit}
mit Zusammenhang 
$(M,\omega,G,\nabla)$. Das 
resultierende \Name{Fe\-do\-sov}-Stern\-pro\-dukt $\star_{\sss
  (\nabla,\Omega, s)}$ ist genau dann $G$-invariant, falls die
symplektische Zweiform $\omega$, der Zusammenhang $\nabla$, sowie 
die geschlossene Zweiform $\Omega$ und das Element $s$ auch
$G$-invariant gew"ahlt wurden, und es gilt 
 
\begin{align}
    \label{eq:InvariantesFedosovSternprodukt}
   g.( a \star_{\sss (\nabla,\Omega, s)} b) = (g.a) 
   \star_{\sss (\nabla,\Omega, s)}(g.b).
\end{align}
\end{satz}

\begin{proof}
Der Beweis ist nach der gemachten Vorarbeit nicht weiter schwierig.    
Wie bereits gezeigt vertauscht $\alg{D}$ mit allen $g \in
G$. Daraus folgt nun

    \begin{align*}
        \alg{D}(g.(\tau(f))) & = g.(\alg{D} \tau (f)) =0 \\
        \Rightarrow g.\tau(f) & = \tau(g.f), \quad \mbox{da} \quad
        \sigma(g.\tau (f)) = g.f.
    \end{align*}

Und damit ist
\begin{equation} 
\begin{aligned}
    \label{eq:Beweis2:GWirkungAutomorphismusSternprodukt}
    g.(f \star_{\sss (\nabla,\Omega, s)} h) & = g.(\sigma(\tau(f)
    \fpweyl \tau(h))) \\ & =\sigma  
    (g.(\tau(f) \fpweyl \tau(h))) \\ &= \sigma(g.\tau(f) \fpweyl 
    g.\tau(h)) \\ &= \sigma (\tau(g.f) \fpweyl \tau (g.h)) \\ &= (g.f)
    \star_{\sss (\nabla,\Omega, s)} (g.h). 
\end{aligned}
\end{equation}
\end{proof}

\begin{satz}[$G$-"aquivarianter deformierter Bimodul]
\index{Bimodul!deformierter!G-aequivarianter@$G$-\"aquivarianter}
    Sei $G$ eine Gruppe, $(M,\omega,G,\nabla)$ eine symplektische
    $G$-Mannigfaltigkeit mit einem $G$-invarianten Zusammenhang
    $\nabla$. Desweiteren sei $(\bundle{E}{\pi}{M},G,\lconnE)$ ein
    $G$-B"undel mit einem $G$-invarianten Zusammenhang $\lconnE$. Ist nun
    $\Omega$ eine $G$-invariante geschlossene Zweiform, so ist
    der daraus konstruierte \Name{Fedosov}-Bimodul
    $(\alg{E},\bullet',\bullet)$ $G$-"aquivariant. 
\end{satz} 

\begin{proof}
    Der Beweis geschieht vollst"andig analog zu dem von
    Satz~\ref{Satz:GInvariantesFedosovSternprodukt}. Aufgrund der
    Invarianz 
    von $\alg{D}'$ und $\alg{D}^{\sss E}$ sowie der
    B"undeldeformationen $\circ'$ und $\circ$ folgt analog zur
    Rechnung~\eqref{eq:Beweis2:GWirkungAutomorphismusSternprodukt}  
    \begin{align}
       g.(A\star' B) = (g.A)\star' (g.B), \quad  g.(A \bullet' s) =
       (g.A)\bullet' (g.s), \quad g.(s\bullet f) = (g.s) \bullet (g.f)    
    \end{align}
f"ur $A,B \in \schnittf{\End{E}}$, $s \in \schnittf{E}$ und $f \in
\Cinff{M}$. 
\end{proof}

\begin{bemerkungen}
~\vspace{-5mm}
\begin{compactenum}
\item Analog zu der hier behandelten $G$-Invarianz kann man auch die unter
einer \Name{Lie}-Algebra $\LieAlg{g}$-"aquivariante
\Name{Fedosov}-Konstruktion betrachten
\citep{mueller-bahns.neumaier:2004a}.
\index{Fedosov-Konstruktion@\Name{Fedosov}-Konstruktion!g-invariante@$\LieAlg{g}$-\"aquivariante}   
\item Die hier gemachte Konstruktion f"ur Sternprodukte oder die Bimoduln ist sehr
    kanonisch, und scheinbar ist es einfach, einen $\LieAlg{g}$- oder
    $G$-"aquivarianten deformierten \Name{Fedosov}-Bimodul zu konstruieren. Das Problem
    besteht darin, invariante Zusammenh"ange auf der Mannigfaltigkeit
    bzw.~auf dem Vektorb"undel zu finden. Im allgemeinen ist dies sehr
    schwierig und nicht immer m"oglich. Wir werden dieses Problem in
    Anhang~\ref{sec:GundgInvarianz} aufgreifen und einige
    F"alle angeben, in denen es immer invariante Zusammenh"ange (und
    damit invariante Sternprodukte und Bimoduln) gibt.
\end{compactenum} 
\end{bemerkungen} 
\index{Fedosov-Konstruktion@\Name{Fedosov}-Konstruktion!G-invariante@$G$-invariante|)}
\index{Sternprodukt!formales|)}
\index{Sternprodukt|)}


%% file: defmorita.tex
\chapter{$H$-"Aquivarianz und \Name{Morita}-"Aquivalenz deformierter Algebren}
\label{sec:MoritaAequivalenzdeformierterAlgebren}

\fancyhead[CO]{\slshape \nouppercase{\rightmark}} 
\fancyhead[CE]{\slshape \nouppercase{\leftmark}} 

\section{Invariante und $H$-"aquivariante Sternprodukte}
\label{sec:HInvarianteSternprodukte}

\subsection{Motivation und erste Beispiele}

In diesem Kapitel wollen wir einige Motivationen und Grundlagen
geben, Sternprodukte anzugeben, die unter einer sp"ater n"aher
beschriebenen Wirkung invariant sind. 

Die Motivation liegt auf der Hand, denn in der Physik spielen
Invarianzen unter (\Name{Lie}-) Gruppen-
bzw.~\Name{Lie}-Algebrenwirkungen eine wichtige Rolle
\citep{marsden.ratiu:2000a}.  

\begin{compactenum}
   \item {\em Symmetrien}\index{Symmetrie} spielen in der klassischen wie in der
       Quantenmechanik eine wichtige Rolle. Symmetrien und
       Erhaltungsgr"o"sen\index{Erhaltungsgroesse@Erhaltungsgr\"o{\ss}e}
       sind eng miteinander verkn"upft, da (unter \glqq
       g"unstigen\grqq{} Voraussetzungen) jede 
       Symmetrie eines Systems einer Erhaltungsgr"o"se entspricht
       \citep{noether:1918a}. Die naheliegenden Beispiele sind
       {\it Translationsinvarianz}, {\it Rotationsinvarianz} und {\it
         Homogenit"at der Zeit}, die {\it Impulserhaltung}, {\it
         Drehimpulserhaltung} und {\it Energieerhaltung} zur Folge
       haben.    
   \item {\em Zwangsbedingungen}\index{Zwangsbedingung} bzw.~{\em Constraints}\index{Constraints}\footnote{Der
         Unterschied zwischen 
       Zwangsbedingungen und Constraints ist, da"s erstere sich
       insbesondere auf den Ortsraum beziehen, also Bedingungen an
       Orte stellen, letztere jedoch auf 
       Einschr"ankungen im Phasenraum. Diese subtile Unterscheidung
       wird oft dadurch hervorgehoben, da"s man auch von {\em
       Phasenraum-Constraints} spricht.} sind "au"sere
       Zw"ange auf ein physikalisches System. Diese sind h"aufig
       geometrischer Natur und schr"anken den Konfigurationsraum (und
       damit auch den Phasenraum) ein.  Durch diese Einschr"ankungen
     ist es in der klassischen wie in der Quantenmechanik oft
     m"oglich den anf"anglich \glqq gro"sen\grqq{} Phasenraum zu
     verkleinern. Dies geschieht mittels {\em 
       Phasenraumreduktion}, was mathematisch dem "Ubergang von einem
     symplektischen Phasenraum\index{Phasenraum!reduzierter} $(M,\omega)$ zu einem reduzierten,
     d.~h.~niedrigerdimensionalen $(M_{\sss \mathrm{red}}, \omega_{\sss
       \mathrm{red}})$ entspricht. Siehe hierzu
     \citep{abraham.marsden:1985a,marsden.ratiu:1999a} und f"ur den
     Fall von Mannigfaltigkeiten mit Sternprodukten
     \citep{bordemann.brischle.emmrich.waldmann:1996a,bordemann.herbig.waldmann:2000a,kowalski.neumaier.pflaum:2004a:pre,bordemann:2004a:pre,bordemann.et.al:2005a:pre}.       
    Wir wollen einige Beispiele f"ur die Phasenraumreduktion\index{Phasenraumreduktion} angeben.
       \begin{compactenum}
       \item Ein starrer K"orper im $\field{R}^{3}$. Dessen
           Konfigurationsraum ist $\field{R}^{3} \times 
       \GR{SO}{3}$. Der Phasenraum ist damit $T^{\ast}(\field{R}^{3}
       \times \GR{SO}{3}) = T^{\ast} \field{R}^{3} \times T^{\ast}
       \GR{SO}{3}$. Fixiert man den Schwerpunkt und spaltet die
       Translationen ab, so ist der Phasenraum $\field{R}^{3} \times \GR{SO}{3}$. 
       \item Ein (starres) Pendel im Schwerefeld. Der urspr"ungliche
           Konfigurationsraum $\field{R}^3$ wird durch die
           Zwangsbedingung zu $S^{2}$ reduziert, der Phasenraum
           ist damit $T^{\ast}S^2$.   
        \end{compactenum}
  \item {\em Eichfreiheit} ist eine spezielle (da physikalisch
          unbeobachtbare) Art der Constraints, die 
    insbesondere in der Feldtheorie auftritt. Auch hier spielen
    Wirkungen von \Name{Lie}-Gruppen und \Name{Lie}-Algebren eine wesentliche
    Rolle.     
\end{compactenum} 

 Wir wollen uns nun ansehen, wie man Symmetrien in den
Sternproduktformalismus implementiert. Dazu werden wir einige
Definitionen und bekannte
Beispiele angeben, sowie das verallgemeinernde Konzept der
\Name{Hopf}-Algebrawirkung implementieren.

\begin{definition}[Invariante Sternprodukte]
    \label{Definition:InvarianteSternprodukte}
Sei $\Phi: M \to M$ ein Diffeomorphismus der symplektischen
Mannigfaltigkeit $(M,\omega,\star)$. Wir
nennen das Sternprodukt 
\begin{compactenum}

\item {\em invariant unter $\Phi$}, falls $\Phi^{\ast} (f\star g)
    = (\Phi^{\ast} f) \star (\Phi^{\ast} g)$ f"ur alle $f,g \in
    \Cinff{M}$. 
 
\item {\em quanteninvariant}, falls es eine formale
    Reihe $T=\id+ \sum_{r=1}^{\infty} \lambda^{r} T_{r}$ von
    Dif\-fe\-ren\-tial\-ope\-ra\-to\-ren\footnote{Im allgemeinen
      m"ussen die $T_{r}$ keine Differentialoperatoren sein, da wir
      uns aber auf differentielle Sternprodukte beschr"anken, ist
      dies automatisch gegeben.} $T_{r}$ gibt, so da"s $\Phi^{\ast}
    \circ T$ ein Automorphismus von $\star$ ist. 
\end{compactenum}
\end{definition}

\begin{beispiel}[Invarianz des \Name{Wick}-Produkts, \citep{waldmann:1995a}]
\index{Sternprodukt!Wick-Typ@vom \Name{Wick}-Typ}
Das \Name{Wick}-Produkt im $\field{C}^{n}$ ist unter der kanonischen
$\GR{U}{n}$-Wirkung invariant, da die Gruppe $\GR{U}{n}$ symplektisch durch
Matrixmultiplikation auf $\field{C}^{n}$ operiert: $\Phi_{\sss U} z
= Uz$, wobei $U \in \GR{U}{n}$ ist.
\end{beispiel}

Die Begriffe der Invarianz und der Quanteninvarianz sind allerdings noch
zu wenig restriktiv. Durch das Einf"uhren bzw.~der Suche nach {\em Impulsabbildungen}
k"onnen wir die Freiheitsgrade reduzieren. Dazu orientieren wir uns an
Arbeiten von 
\citet{arnal.cortet.molin.pinczon:1983a},
\citet{bordemann.herbig.waldmann:2000a} und \citet{gutt.rawnsley:2003a}.
Eine ausf"uhrliche Beschreibung zu klassischen Impulsabbildungen
findet man in \citep{abraham.marsden:1985a}. 

\begin{definition}[Invariante Sternprodukte und Impulsabbildung]
\label{Definition:InvarianteSternprodukteImpulsabbildung}
\index{Sternprodukt!invariantes} \index{Sternprodukt!stark invariantes} \index{Sternprodukt!aequivariantes@\"aquivariantes}
\index{Sternprodukt!star aequivariantes@stark \"aquivariantes}
Gegeben eine symplektische Mannigfaltigkeit $(M, \omega, \star)$ mit
Sternprodukt. Desweiteren sei eine $\ad^{\ast}$-invariante
Impulsabbildung $J:M \to \LieAlgd{g}$ gegeben. Man nennt das Sternprodukt
$\star$ 
\begin{compactenum}
\item {\em invariant}, falls f"ur alle $f,g \in \Cinff{M}$ und $\xi
    \in \LieAlg{g}$ gilt
    \begin{align}
        \label{eq:invariantesSternprodukt}
        \{\SP{J,\xi}, f\star g \} = \{\SP{J,\xi}, f\} \star g + f
        \star \{\SP{J,\xi}, g \}
    \end{align}
\item {\em "aquivariant}, falls f"ur alle $\xi, \eta \in \LieAlg{g}$
    gilt
    \begin{align}
        \label{eq:aequivariantesSternprodukt}
        \SP{J,\xi} \star \SP{J, \eta} - \SP{J,\eta} \star \SP{J,\xi} =
        \im \lambda \{\SP{J,\xi}, \SP{J,\eta} \} =  \im \lambda
        \SP{J,[\xi, \eta]}. 
    \end{align}
\item {\em stark invariant}, falls f"ur alle $\xi \in \LieAlg{g}$ und $f
    \in \Cinff{M}$ gilt
   \begin{align}
        \label{eq:starkinvariantesSternprodukt}
        \SP{J,\xi} \star f - f \star \SP{J,\xi} =
        \im \lambda \{\SP{J,\xi},f \}.  
    \end{align}
\item {\em quanten"aquivariant}, falls es glatte Abbildungen
    $\mbf{J}_{n}: M \to \LieAlgd{g} \otimes \field{C}$ mit $\mbf{J}_{0} = J$ gibt, so
    da"s f"ur die {\em Quantenimpulsabbildung} $\mbf{J}=\sum_{n=0}^{\infty} \lambda^{n}
    \mbf{J}_{n}$ und f"ur alle $\xi, \eta \in \LieAlg{g}$ gilt
    \begin{align}
        \label{eq:quantenaequivarianteSternprodukte}
        \SP{\mbf{J},\xi} \star \SP{\mbf{J}, \eta} - \SP{\mbf{J},\eta} \star \SP{\mbf{J},\xi} =
        \im \lambda  \SP{\mbf{J},[\xi, \eta]}.
    \end{align}
\end{compactenum}
Dabei ist $\SP{J, \xi}$ (bzw.~$\SP{\mbf{J}, \xi}$) die nat"urliche Paarung mit Werten in
$\Cinf{M}$ (bzw.~$\Cinff{M}$), und wir schreiben auch $\SP{J, \xi} =:
J_{\xi}$ (bzw.~$\SP{\mbf{J}, \xi} =: \mbf{J}_{\xi}$).
\end{definition}

\begin{bemerkung}["Aquivarianz und Kovarianz]
\label{Bemerkung:AequivariantKovariant}
Die Begriffe {\em "aquivariant} und {\em kovariant} meinen im
Zusammenhang mit Gruppenwirkungen dasselbe und werden in der Literatur
synonym verwendet. Wir haben uns bem"uht in dieser Arbeit
ausschlie"slich ausschlie"slich den Ausdruck
{\em "aquivariant} zu verwenden, auch wenn in der jeweils angegebenen
Literatur der Begriff {\em kovariant} verwendet wurde.
\end{bemerkung}

Im Rahmen der Deformationsquantisierung ist die Quanten"aquivarianz in
gewisser Hinsicht die \glqq nat"urlichste\grqq{} Art, da sie die
kategoriell besten Eigenschaften hat. Da bei der Quanten"aquivarianz,
aufgrund der h"oheren Ordnungen in $\lambda$, mehr Freiheitsgrade zur
Verf"ugung stehen, ist sie erstmal weniger restriktiv als die "Aquivarianz oder die starke
Invarianz. Jedoch gibt es f"ur die Existenz von
Quantenimpulsabbildungen keine Beweise. 
Nichtsdestotrotz werden wir uns im Rahmen der 
\Name{Hopf}-$^\ast$-Algebren im wesentlichen mit der Quanten"aquivarianz
auseinandersetzen. Dazu werden wir im folgenden Kapitel auf den Begriff der Impulsabbildung
zur"uckgreifen, den wir bereits in
Kapitel~\ref{sec:BildGruppoidmorphismusstrPicHtostrPic} eingef"uhrt
haben.

\subsection{Formulierung mittels \Name{Hopf}-Algebra Techniken}
\label{sec:HopfAlgebraTechniken}
\index{Lie-Gruppe@\Name{Lie}-Gruppe} \index{Lie-Algebra@\Name{Lie}-Algebra}
In diesem Kapitel werden wir das algebraische Konzept
der \Name{Hopf}-Algebren nutzen, um einen geeigneten Begriff der
Invarianz im Rahmen von deformierten Algebren und insbesondere
Sternprodukten zu formulieren. Als Beispiele werden die Gruppenalgebra
einer Gruppe sowie die universell Einh"ullende einer
\Name{Lie}-Algebra dienen, die wir durch diese Formulierung gemeinsam
behandeln k"onnen. Desweiteren sind wir in der Lage eine weitere Klasse
von Symmetrien, die als {\em Quantengruppen} oder
{\em $q$-deformierte Gruppen} in der Literatur zu finden sind
\citep{majid:1995a,kassel:1995a}, zu behandeln. Die $q$-deformierten
\Name{Hopf}-Algebren sind dann 
im Gegensatz zur universell Einh"ullenden einer \Name{Lie}-Algebra
oder der Gruppenalgebra nicht mehr {\em kokommutativ}
(vgl.~Beispiel~\ref{Beispiel:NichtKoKommutativeHopfAlgebra}). 

Die Notation und die wichtigsten Definitionen und S"atze
bez"uglich \Name{Hopf}-Algebren sind in Anhang~\ref{sec:HopfAlgebren}
zusammengetragen. 

\subsection{Erste Definitionen und Beispiele}

\subsubsection{$H$-invariante Sternprodukte}

\begin{definition}[$H$-Invariantes Sternprodukt]
\label{Definition:HInvariantesSternprodukt}
\index{Sternprodukt!H-invariantes@$H$-invariantes}
Gegeben eine \Name{Poisson}-Man\-nig\-fal\-tig\-keit
$(M,\Lambda, \star)$ mit einem Sternprodukt nach Definition
\ref{Definition:FormaleSternprodukte} 
und einer \Name{Hopf}-Algebra $H$ mit einer Wirkung $\neact$ auf der
Funktionenalgebra $\Cinf{M}$ (beziehungsweise auf $\Cinff{M}$ durch
summandenweises Fortsetzen). Das Sternprodukt
$\star$ ist {\em $H$-invariant unter der Wirkung $\neact$} (oder auch {\em invariant unter einer
  Wirkung $\neact$ von $H$}), falls $(H, (\Cinff{M}, \star), \act)$
eine $H$-Linksmodulalgebra nach Definition
\ref{Definition:AxiomeWirkungHopf} ist. Insbesondere gilt f"ur alle
$h\in H$ und $f,g \in \Cinff{M}$ 
\begin{align}
\label{eq:HInvariantesSternprodukt}
h\act(f \star g) = (h_{\sss (1)} \act f) \star (h_{\sss (2)} \act g).
\end{align}   
\end{definition}

Die beiden wichtigen Beispiele sind die Wirkung von \Name{Lie}-Gruppen und \Name{Lie}-Algebren.

\begin{beispiel}[Wirkung einer Gruppe $G$ via Automorphismen]
\label{Beispiel:WirkungGruppeGAlsAutomorphismus}
Sei $G$ nun eine \Name{Lie}-Gruppe. Die Gruppe
$G$ wirkt via Automorphismen  $\phi_{g}^{\ast}:\Cinf{M} \to \Cinf{M}$ und $g\in G$ auf die
Algebra $(\Cinff{M},\star)$. 
Um dies in Termen einer \Name{Hopf}-Algebra zu formulieren, m"ussen wir zur
Gruppenalgebra $\field{C}(G)$ "ubergehen (siehe
Beispiel~\ref{Beispiel:GruppenalgebraHopf}). Da $\phi_{g}$ ein
Diffeomorphismus von $M$ ist, entspricht dies 
nat"urlich genau der Definition eines invarianten Sternprodukts
aus Definition~\ref{Definition:InvarianteSternprodukte}.  
\begin{align}
\label{eq:GruppenwirkungAufSternprodukt}
\phi^{\ast}_{g} (a \star b)  = (\phi^{\ast}_{g}a) \star
(\phi^{\ast}_{g}b) \quad \text{f"ur $a,b \in \Cinff{M}$}. 
\end{align}
\end{beispiel}

\begin{beispiel}[Wirkung einer \Name{Lie}-Algebra $\LieAlg{g}$ via
    Derivationen] 
\label{Beispiel:WirkungLiAlgebraGAlsDerivation}
Analog zu Beispiel~\ref{Beispiel:WirkungGruppeGAlsAutomorphismus}
kann man die \Name{Lie}-Algebrawirkung $\Lie[\xi]:\Cinf{M} \to
\Cinf{M}$, die nicht notwendigerweise von einer Gruppe $G$ kommen
mu"s, mit $\xi \in \LieAlg{g}$ als
\Name{Hopf}-Al\-ge\-bra\-wir\-kung 
ansehen. Die \Name{Lie}-Algebra stellt dabei selbst noch keine
\Name{Hopf}-Algebra dar, allerdings deren universell
Einh"ullende $\universell{\LieAlg{g}}$ (Beispiel
\ref{Beispiel:UniverselleEinhuellendeHopf}). Die Wirkung wird dann zu   
\begin{align}
    \Lie[\xi](a \star b) =  (\Lie[\xi]a) \star b +  a \star (\Lie[\xi]
    b) \quad \text{f"ur $a,b \in \Cinff{M}$}.
\end{align}
Kommt die \Name{Lie}-Algebra von einer \Name{Lie}-Gruppe $G$, ist also
$\LieAlg{g}=T_{e}G$, so ist die Wirkung offensichtlich die
infinitesimale Version der Gruppenwirkung. Ferner sieht man, da"s
die Definition~\ref{Definition:HInvariantesSternprodukt} im
\Name{Lie}-Algebrafall
mit $\Lie[\xi] ( \cdot)= \{ \SP{J,\xi}, \cdot \}$ der urspr"unglichen
in Definition~\ref{Definition:InvarianteSternprodukteImpulsabbildung}
entspricht, sofern ein $J$ existiert.   
\end{beispiel}

Bei den Beispielen~\ref{Beispiel:WirkungGruppeGAlsAutomorphismus} und
\ref{Beispiel:WirkungLiAlgebraGAlsDerivation} sind die beiden
\Name{Hopf}-Algebren jeweils kokommutativ.

\begin{lemma}[Triviale Wirkung einer \Name{Hopf}-Algebra]
\label{Lemma:TrivialeWirkungEinerHopfAlgebra}
Gegeben sei eine \Name{Poisson}-Mannigfaltigkeit mit Sternprodukt
$(M,\Lambda,\star)$ und eine beliebige \Name{Hopf}-Al\-ge\-bra 
$H$. Wirkt die \Name{Hopf}-\-Al\-ge\-bra "uber $\field{C}$  trivial 
(nach Beispiel~\ref{Beispiel:TrivialeHWirkung}), so ist das Sternprodukt
unter dieser Wirkung invariant.  
\end{lemma}

\begin{proof} 
Es gilt offensichtlich
\begin{align*}
    h \act (f \star g) & = \varepsilon(h) (f \star g) \\ & =
    \varepsilon(h_{\sss (1)} \varepsilon (h_{\sss (2)})) (f \star  g)
    \\ & = (\varepsilon(h_{\sss (1)}) f) \star (\varepsilon(h_{\sss
      (2)}) g) \\ & = (h_{\sss (1)} \act f) \star (h_{\sss (2)} \act
    g), 
\end{align*}
so da"s jede \Name{Hopf}-Algebra $H$ durch die triviale
Wirkung auf ein gegebenes Sternprodukt wirkt.
\end{proof}

\subsubsection{Impulsabbildung und $H$-"aquivariante Sternprodukte}

Nun wollen wir Definition \ref{Definition:InvarianteSternprodukteImpulsabbildung} auf den Fall
der Wirkung einer \Name{Hopf}-$^\ast$-Algebra auf \Name{Hermite}sche
Sternprodukte erweitern\footnote{Im Fall von nicht-\Name{Hermite}schen
Sternprodukten kann auf die $^\ast$-Struktur verzichten werden, da
diese nicht essentiell f"ur die Definition von $H$-"Aquivarianz
ist. Desweiteren kann man in diesem Fall auch \Name{Hopf}-Algebren
ohne $^\ast$-Struktur verwenden}.

Sei nun $\alg{A}$ eine kommutative $^\ast$-Algebra "uber dem Ring
$\ring{C}$ (wir denken nat"urlich insbesondere an
$\alg{A}=\Cinf{M}$) und $\defalg{A}=(\algf{A}, \star)$ sei eine
\Name{Hermite}sche Deformation von $\alg{A}$. Sei weiter $H$ eine
\Name{Hopf}-$^\ast$-Algebra "uber dem Ring $\ring{C}$.
 
Wir wollen nun untersuchen wann wir eine {\em innere Wirkung} der
\Name{Hopf}-$^\ast$-Algebra auf die deformierte Algebra $\defalg{A}$
finden k"onnen. Das besondere Interesse an innere Wirkungen ist durch die
Quantenmechanik motiviert, wie das folgende Beispiel zeigt.

\begin{beispiele}[Innere Wirkungen einer \Name{Lie}-Gruppe]
Sei $G$ eine \Name{Lie}-Gruppe. Wir betrachten die \Name{Hopf}-$^\ast$-Algebra
    $H=\field{C}[G]$ mit der "ublichen Koalgebrastruktur
    (vgl.~\ref{Beispiel:GruppenalgebraHopf}). Diese wirkt
    mittels innerer Automorphismen auf die $^\ast$-Algebra
    $\alg{B}(\hilbert{H})$ der beschr"ankten Operatoren auf einem
    \Name{Hilbert}-Raum $\hilbert{H}$
    \begin{align*}
        U: \field{C}[G] \ni g \mapsto U_{g} \in \alg{B}(\hilbert{H}).
    \end{align*}
Die innere Wirkung wird dann zu $g \neact: \alg{B}(\hilbert{H}) \ni A
\mapsto U_{g} A U_{g}^{-1} \in \alg{B}(\hilbert{H})$ und
$(\field{C}[G], \alg{B}(\hilbert{H}), \neact)$ ist eine
$\field{C}[G]$-Linksmodulalgebra. 
\end{beispiele}

Bei der Formulierung wollen wir uns einer Impulsabbildung bedienen,
wie wir sie bereits in Kapitel~\ref{sec:BildGruppoidmorphismusstrPicHtostrPic} eingef"uhrt
haben. Mit Hilfe der Impulsabbildung $J:H \to  \alg{A}$ k"onnen wir eine innere $^\ast$-Wirkung
auf einer $^\ast$-Algebra definieren. Diese ist, wie wir in
Lemma~\ref{Lemma:InnereWirkungSternprodukteQIA} gezeigt haben, gegeben durch 
\begin{align*}
    h \act a := J(h_{\sss(1)})\cdot a \cdot J(S(h_{\sss (2)})).
\end{align*}
Auf der undeformierten, kommutativen Algebra $\alg{A}$ ist diese Wirkung trivial,
wie wir in Korollar~\ref{Korollar:InnereWirkungaufKommutativerAlgebra}
gezeigt haben. Genau dies wird uns bei der Formulierung einer inneren
Wirkung auf der deformierten Algebra $(\algf{A},\star)$ noch
Probleme einbringen.

Sei im weiteren 
\begin{align}
\label{Definition:ImpulsabbildungDeformierteAlgebra}
J= \sum_{n=0}^{\infty} \lambda^{n} J_{n}: H \to (\algf{A},\star)
\end{align}

eine Impulsabbildung. Nach Definition~\ref{Definition:InvarianteSternprodukteImpulsabbildung}
  handelt es sich bei dieser Definition um eine {\em
    Quantenimpulsabbildung}, jedoch ist diese Definition mit unserer
  Definition~\ref{Definition:ImpulsabbildungHtoA} konsistent, da es
  sich bei einer \Name{Hermite}schen Deformation $(\algf{A},\star)$
  um eine $^\ast$-Algebra handelt. F"ur die in der klassischen
  Mechanik gebr"auchliche Impulsabbildung wollen wir auf
  \citep{abraham.marsden:1985a} verweisen, f"ur einen tieferen
  Einblick in die Quantenimpulsabbildungen im Sinne von
  Definition~\ref{Definition:InvarianteSternprodukteImpulsabbildung} auf die Arbeiten
  \citep{xu:1998a,bordemann.herbig.waldmann:2000a,mueller-bahns.neumaier:2004a}.  

Die Frage, die sich uns nun stellt, ist: 
{\em Welche \Name{Hopf}-Algebren k"onnen eine (nichttriviale) innere
  Wirkung f"ur eine deformierte kommutative Algebra liefern?} Das
hei"st, k"onnen wir f"ur eine gegebene \Name{Hopf}-Algebra eine
Impulsabbildung $J: H \to (\algf{A}, \star)$ finden, so da"s wir
mit Hilfe dieser eine innere Wirkung auf der deformierten Algebra
bekommen? Wie bereits erw"ahnt wird uns die Kommutativit"at der
untersten Ordnung dabei Probleme bereiten. Die naive Idee, da"s die
Gruppenalgebra $\field{C}[G]$ oder die 
universell Einh"ullende einer \Name{Lie}-Algebra
$\universell{\LieAlg{g}}$ eine Symmetrie eines Sternprodukts sein
k"onnten, scheitert bereits im Ansatz, wie wir in den folgenden Gegenbeispielen zeigen.

\begin{beispiel}[Gegenbeispiel 1: Gruppenalgebra {$\field{C}[G] \FP$}]
\label{Beispiel:GegenbeispielGruppenalgebra}
Sei nun $\alg{A}$ eine kommutative Algebra und $\Phi: G \to
\Aut(\alg{A})$. Da $\alg{A}$ kommutativ ist, ist $\Aut(\alg{A}) =
\OutAut(\alg{A})$. Desweiteren sei $\defalg{A}=(\algf{A}, \star)$ ein
Sternprodukt, und $\mbf{\Phi}: G \to \Aut(\algf{A}, \star)$ eine
deformierte Wirkung, so da"s $g \mapsto \sum_{n=0}^{\infty}
\lambda^{n} \mbf{\Phi}^{(n)}_{g}$ mit 
$\mbf{\Phi}^{(0)}_{g} = \Phi_{g}$. Es ist nun einfach zu zeigen, da"s
es kein $U_{g}\in \defalg{A}$ geben kann, so da"s sich die Wirkung mit
$\mbf{\Phi}$ als eine innere Wirkung der Form $\mbf{\Phi}_{g}=\Ad_{\sss
  \star} (U_{g})$ schreiben lassen kann. Da die unterste Ordnung der Algebra
$\defalg{A}$ kommutativ ist, mu"s $\mbf{\Phi}^{(0)}_{g} = \id$
sein. Da die Wirkung auf die undeformierte Algebra $\alg{A}$ ein
beliebiger Automorphismus der Algebra ist, kann die Wirkung auf die
deformierte Algebra kein innerer Automorphismus sein.

Damit ist klar, da"s die \Name{Hopf}-Algebra  $H=\field{C}[G]\FP$
keine innere (wohl jedoch eine "au"sere) Symmetrie f"ur ein Sternprodukt sein kann.
\end{beispiel}

\begin{beispiel}[Gegenbeispiel 2: Universell einh"ullende Algebra $\universell{\LieAlg{g}}\FP$]
\label{Beispiel:GegenbeispielUniverselleinhuellendeLieAlgebra}
\index{Algebra!universell Einh\"ullende}
Analog zu Beispiel \ref{Beispiel:GegenbeispielGruppenalgebra}
verl"auft die "Uberlegung im Fall einer formalen Potenzreihe der universell einh"ullenden
Algebra $\universell{\LieAlg{g}}\FP$ einer \Name{Lie}-Algebra
$\LieAlg{g}$. Wieder sei $\defalg{A}= (\algf{A}, \star)$ ein Sternprodukt. Mit der
"ublichen Koalgebrastruktur $\Delta(\xi)= 1 \otimes \xi + \xi \otimes
1$ und mit $S(\xi) = -\xi$ ergibt sich f"ur eine innere Wirkung   
\begin{align*}
    J(\xi_{\sss (1)}) \star J(\eta) \star J(S(\xi_{\sss (2)})) & =
    J(\xi) \star J(\eta) - J(\eta) \star J(\xi).  
\end{align*}
Da die unterste Ordnung in der Algebra $(\algf{A}, \star)$ kommutativ
ist, ist der gesamte Ausdruck von der Ordnung $\alg{O}(\lambda)$. Andererseits
mu"s $J$ ein $^\ast$-Homomorphismus sein, so da"s 
\begin{align*}
    J(\xi) \star J(\eta) - J(\eta) \star J(\xi) & = J([\xi, \eta]).
\end{align*}
Dies f"uhrt aber zu einem Widerspruch, denn $J([\xi, \eta])$ ist im
allgemeinen von der Ordnung $\alg{O}(1)$. Somit kann auch die formale Potenzreihe
einer universell einh"ullenden Algebra $\universell{\LieAlg{g}}\FP$ einer
\Name{Lie}-Algebra im allgemeinen keine Symmetrie eines Sternprodukts sein. 
\end{beispiel}

Wir wollen nun eine \Name{Hopf}-Algebra angeben, die eine Symmetrie
eines Sternprodukts sein kann. Dazu betrachten wir die folgende
Konstruktion nach \citet{gutt:1983a}: 
\begin{align}
    \label{eq:LambdaUniverselleinhuellendeAlgebra}
    \universelll{\LieAlg{g}}:= T_{\sss
      \field{C}}^{\bullet}(\LieAlg{g})\FP / \SP{\xi \otimes \eta - \eta
    \otimes \xi - \lambda [\xi, \eta]}.
\end{align}
Dabei ist $T_{\sss \field{C}}^{\bullet}(\LieAlg{g})\FP$ die formale
Potenzreihe der komplexifizierten Tensoralgebra einer
\Name{Lie}-Algebra $\LieAlg{g}$. Elemente in
$\universelll{\LieAlg{g}}$ k"onnen wir aufgrund des Quotienten nicht
explizit als formale Potenzreihen schreiben. Als \Name{Hopf}-Algebra
interpretiert ist $\universelll{\LieAlg{g}}$ mit den
gew"ohnlichen \Name{Hopf}-Strukturen $(\universelll{\LieAlg{g}}, \cdot, \eta, \Delta,
\varepsilon, S, ^\ast)$ ausgestattet. Von \citet{gutt:1983a} wurde
mit dem Satz von
\Name{Poincar\'{e}-Birk\-hoff-Witt}\index{Satz!Poincare-Birkhoff-Witt@von \Name{Poincar\'{e}-Birkhoff-Witt}} gezeigt, da"s 
diese Algebra kanonisch isomorph zu $(\Pol^{\bullet}(\LieAlgd{g})\FP,
\spgutt)$ ist und wir bezeichnen mit $\spgutt$ das
\Name{Gutt}-Sternprodukt\index{Sternprodukt!Gutt@von \Name{Gutt}} \index{Gutt-Sternprodukt@\Name{Gutt}-Sternprodukt}. Auf kanonische Weise wird $\Hgutt=(\Pol^{\bullet}(\LieAlgd{g})\FP,
\spgutt, \eta, \Delta, \varepsilon, S, ^\ast)$ zu einer
\Name{Hopf}-$^\ast$-Algebra mit $\Delta (\xi)  = \xi \otimes 1 + 1
\otimes \xi$, $\varepsilon(\xi) =0$ und $S(\xi)=-\xi$. 

Damit haben wir nun eine \Name{Hopf}-$^\ast$-Algebra gefunden, so
da"s die Impulsabbildung 
\begin{align}
J: \Hgutt \to (\algf{A},\star)
\end{align}
 ein -- im \Name{Hopf}-$^\ast$-Algebra Sinne -- innerer $^\ast$-Homomorphismus
sein kann. Denn f"ur $\xi, \eta \in \Pol^{\bullet}(\LieAlgd{g})\FP$ gilt:

\begin{align*}
    J(\xi_{\sss (1)}) \star  J(\eta)  \star J(S(\xi_{\sss
      (2)})) &  = J(\xi) \star J(\eta) - J(\eta) \star J(\xi) \\ & =
    \im \lambda J([\xi, \eta]).  
\end{align*}
Dies entspricht der Definition eines quanten"aquivarianten
Sternprodukts in Definition~\ref{Definition:InvarianteSternprodukteImpulsabbildung}.

\section{Deformation von projektiven Moduln}
\label{sec:DeformationProjektiverModul}

In diesem Abschnitt wollen wir nur die wichtigsten Ergebnisse kurz
aufgreifen und keine tiefgehenden Beweise f"uhren. Statt dessen
verweisen wir an den betreffenden Stellen auf geeignete Artikel.

Nun wollen wir im Rahmen von deformierten Algebren "uber
starke \Name{Morita}-"Aquivalenz, $^\ast$-\Name{Morita}-"Aquivalenz
und \Name{Morita}-"Aquivalenz reden. Dies bedeutet, da"s wir nun \Name{Hermite}sche, formal
deformierte $^\ast$-Al\-ge\-bren nach Definition
\ref{Definition:FormaleAlgebraDeformation} betrachten. Im weiteren sei
$\defalg{A}=(\algf{A},\star)$ eine \Name{Hermite}sche, 
deformierte Algebra mit Einselement "uber einem Ring $\ring{C}=\ring{R}(\im)$,
wobei $\ring{R}$ geordnet sei (wir denken insbesondere nat"urlich an
$\fieldf{R}$ und $\fieldf{C}$). Die folgenden bekannten
Ergebnisse findet man beispielsweise in
\citep{bursztyn.waldmann:2000b,bursztyn:2001a:phd}.

\begin{lemma}[Deformierter Projektor]
\label{Lemma:DeformierterProjektor}
 Sei $M_{n}(\alg{A}) \ni P_{0} = P_{0}^{2}$ ein idempotentes Element,
 dann existiert ein idempotentes $M_{n}(\defalg{A}) \ni \defpro{P} = P_{0} +
 \alg{O}(\lambda)$. Ist desweiteren $P_{0} = P_{0}^{2} = P_{0}^{\ast}$
 ein Projektor, so kann man ein $\defpro{P} \in M_{n}(\defalg{A})$
 w"ahlen, das ebenso ein Projektor ist, d.~h.~es gilt $\defpro{P}=
 \defpro{P} \star \defpro{P} = \defpro{P}^{\ast}$.
\end{lemma}

\begin{proof}
 Der Beweis findet sich in
 \citep{gerstenhaber.schack:1990a,fedosov:1996a}. Insbesondere kann
 man eine explizite Formel f"ur einen deformierten Projektor angeben
 \citep[Eq. (6.1.4)]{fedosov:1996a}
    \begin{align}
        \label{eq:DeformierterProjektor}
        \defpro{P} = \frac{1}{2} + \left( P_{0} - \frac{1}{2} \right)
        \star \frac{1}{\sqrt[\star]{1+4(P_{0} \star P_{0} - P_{0})}},
    \end{align}
der genau die geforderten Eigenschaften hat.
\end{proof}

\begin{bemerkung}[Deformierter Projektor]
Der Nenner in des deformierten Projektors $\defpro{P}$ in
Gleichung~\eqref{eq:DeformierterProjektor} ist wohldefiniert.
Die unterste Ordnung des Sternprodukts $P_{0} \star P_{0} - P_{0}$
ist offensichtlich von der Ordnung $\alg{O}(\lambda)$. Die $^\star$-Wurzel ist
als eine formale \Name{Taylor}-Reihe zu verstehen. 
\end{bemerkung}

\begin{definition}[Stark voller Projektor {\citep{bursztyn.waldmann:2005a}}]
Sei $\alg{A}$ eine $^\ast$-Algebra. Man nennt einen Projektor $P_{0}
\in M_{n}(\alg{A})$ {\em stark voll}, falls es ein invertierbares
  Element $\tau \in \alg{A}$ gibt, so da"s $\tr(P_{0}) = (\tau \tau^{\ast})^{-1}$.   
\end{definition}

\begin{lemma}[(Stark) voller deformierter Projektor]
Sei $\alg{A}$ eine $^\ast$-Algebra, und $P_{0} \in
M_{n}(\alg{A})$ sei idempotent. Sei weiter
$\defpro{P} \in M_{n}(\defalg{A})$ eine Deformation von $P_{0}$, dann ist $P_{0}$
genau dann voll, wenn $\defpro{P}$ voll ist. Ist $P_{0}$ eine Projektion,
dann ist $\defpro{P}$ genau dann stark voll, wenn $P_{0}$ stark voll ist.
\end{lemma}

\begin{proof}
    Siehe beispielsweise \citep[Lemma 4.4, 4.5]{bursztyn:2001a:phd}.
\end{proof}

\begin{lemma}[Deformation von {$P_{0} M_{n}(\alg{A}) P_{0}$}]
Sei $P_{0}$ ein idempotentes Elemente in $M_{n}(\alg{A})$ und $\alg{A}
\in M_{n}(\alg{A}) \FP$. Die
Abbildung $P_{0} M_{n}(\alg{A}) P_{0} \to \defpro{P} \star 
M_{n}(\defalg{A}) \star \defpro{P}$, $P_{0}AP_{0} \mapsto \defpro{P} \star (P_{0}AP_{0})
\star \defpro{P}$ ist ein $\ringf{C}$-Isomorphismus und eine formale
Deformation von $P_{0} M_{n}(\alg{A}) P_{0}$. Ist desweiteren $P_{0}$
projektiv und $\defalg{A}$ eine \Name{Hermite}sche Deformation, dann
induziert dies eine \Name{Hermite}sche Deformation von $P_{0}
M_{n}(\alg{A}) P_{0}$.       
\end{lemma}

\begin{proof}
    Siehe \citep[Lemma 4.6, Corollary 4.7]{bursztyn:2001a:phd}.
\end{proof}

Wir wollen uns nun im weiteren der formalen Deformation von {\em projektiven
  Moduln} widmen. Diese werden bei der
\Name{Morita}-"Aquivalenz von deformierten Algebren eine wichtige
Rolle spielen.

\begin{definition}[Deformation eines $\alg{A}$-Rechtsmoduls $\rmod{E}{A}$]
 \label{Definition:DeformationEinesModuls}
    Sei $\alg{A}$ eine assoziative Algebra und $\rmod{E}{A}$ ein
    $\alg{A}$-Rechtsmodul. Ferner sei $\defalg{A}=(\algf{A},\star)$ die formale
    Deformation der Algebra $\alg{A}$ nach Definition  
    \ref{Definition:FormaleAlgebraDeformation}. Eine {\em Deformation des
    $\alg{A}$-Rechtsmoduls} $\rmod{E}{A}$ ist gegeben durch    
   \begin{equation} 
   \begin{aligned}
        \label{eq:FormaleDeformationEinesModuls}
 \rmodf{E}{A} \times \defalg{A} \ni    (x,a)  \mapsto  x \bullet a =
 \sum_{n=0}^\infty \lambda^{n} R_{n} (x,a) \in \rmodf{E}{A}.
    \end{aligned}
    \end{equation}
Dabei seien $R_{n}: \rmodf{E}{A} \times \defalg{A} \to \rmodf{E}{A}$ $\ringf{C}$-bilineare
Abbildungen, und $R_{0}(x,a)=x \cdot a$ ist die gew"ohnliche
Modulstruktur. Die Kompatibilit"atsbedingung f"ur die deformierte
Algebra-Rechts\-mul\-ti\-pli\-ka\-tion lautet   
\begin{align}
\label{eq:DeformierteModulmultiplikationUndSternprodukt}
(x \bullet a) \bullet a' = x \bullet ( a \star a').  
\end{align}
Wir bezeichnen den deformierten $\defalg{A}$-Rechtsmodul (kompatibel
mit der deformierten Algebra $\defalg{A}$) mit
$\defrmod{E}{A}=(\rmodf{E}{A},\bullet)$.  
\end{definition}

Schreiben wir die Deformation der Algebra $\alg{A}$ als eine formale
Potenzreihe (vgl.~Gleichung~\eqref{eq:FormaleAlgebraDeformation}), so
k"onnen wir auf dem  Niveau der 
bilinearen Abbildungen $R_{i}$ und $C_{j}$ die Kompatibilit"atsbedingung
aus Gleichung~\eqref{eq:DeformierteModulmultiplikationUndSternprodukt} in jeder
$\lambda$-Ordnung $n$ als  
\begin{align}
\sum_{i=0}^{n} \left[ R_{n}(R_{n-i}(x,a),a') - R_{n}(x,C_{n-i}(a,a'))
\right] =0
\end{align}

schreiben. Analog zu Definition \ref{Definition:AEquivalenzZweierDeformationen}
f"uhren wir die "Aquivalenz zweier deformierter Moduln ein.

\begin{definition}["Aquivalenz zweier deformierter Moduln]
\label{Definition:DeformierterModulisomorphismus}
Zwei $\defalg{A}$-Rechtsmoduln $(\defrmod{E}{A})_{1}=
(\rmodf{E}{A},\bullet_{\sss 1})$ und $(\defrmod{E}{A})_{2} =
(\rmodf{E}{A},\bullet_{\sss 2})$ 
nennt man genau dann {\em "aquivalent}, falls es $\ring{C}$-lineare
Abbildungen $T_{r}: \rmod{E}{A} \to \rmod{E}{A}$ gibt, so da"s
\begin{align}
     T=\id + \sum_{r=1}^{\infty} \lambda^{r} T_{r} \,: \quad
     (\defrmod{E}{A})_{1} \to (\defrmod{E}{A})_{2} 
\end{align}
ein $\defalg{A}$-Modulisomorphismus ist.
\end{definition}

Analog zu den $\defalg{A}$-Rechtmoduln lassen sich nat"urlich auch
$\defalg{B}$-Linksmoduln $\deflmod{B}{E} = (\lmodf{B}{E}, \bullet')$
definieren.  

\begin{proposition}[Endlich erzeugte und projektive Deformation]
    Sei $\alg{A}$ eine Algebra mit Einselement "uber dem Ring $\ring{C}$,
    und $\defalg{A}=(\algf{A}, \star)$ sei eine Deformation von
    $\alg{A}$. Sei nun $\rmod{E}{A}$ ein endlich erzeugter, projektiver
    $\alg{A}$-Rechtsmodul, dann existiert ein deformierter
    $\defrmod{E}{A}$ $\defalg{A}$-Rechtsmodul, der endlich
    erzeugt und projektiv "uber $\defalg{A}$ ist. Diese Deformation ist bis auf
    "Aquivalenz eindeutig.
\end{proposition}

\begin{proof}
    Siehe \citep[Proposition 4.10]{bursztyn:2001a:phd}.
\end{proof}

\begin{proposition}[Ringtheoretische \Name{Morita}-"Aquivalenz]
Sei $P_{0} \in M_{n}(\alg{A})$ voll idempotent, dann existiert eine
Bijektion $\Phi: \Def(\alg{A}) \to \Def(\alg{B})$ mit $\alg{B} = P_{0}
M_{n}(\alg{A}) P_{0}$. Dies bedeutet, da"s die formalen Deformationen
im ringtheoretischen Sinne \Name{Morita}-"aquivalent sind.
\end{proposition}

\begin{proof}
Ein ausf"uhrlicher Beweis findet sich in \citep[Proposition 4.11]{bursztyn:2001a:phd}.
\end{proof}

Im weiteren wollen wir nun innere Produktmoduln deformieren, so
da"s wir die Struktur des $\alg{A}$-wertigen inneren Produktes von
$\rmodplus{E}{A}$ auf den deformierten Modul retten. Damit werden wir
in der Lage sein, den Begriff der $^\ast$-\Name{Morita}-"Aquivalenz
und der starken \Name{Morita}-"Aquivalenz auf deformierte Algebren zu "ubertragen. 
Sei im weiteren nun $\alg{A}$ eine $^\ast$-Algebra und $\rmodplus{E}{A}$ sei ein
"uber $\alg{A}$ endlich erzeugter projektiver innerer Produktmodul. 

\begin{definition}[$\defalg{A}$-wertiges inneres Produkt auf
    {$\defrmod{E}{A}$}]
Sei $\defalg{A} = (\algf{A},\star)$ eine \Name{Hermite}sche
Deformation und $\defrmod{E}{A} = (\rmodf{E}{A}, \bullet)$ eine zugeh"orige
Deformation von $\rmod{E}{A}$. Man nennt $\defrmod{E}{A} =
(\rmodf{E}{A}, \bullet, \rSPf{\cdot, \cdot}{\defalg{A}})$ eine {\em \Name{Hermite}sche
  Deformation des inneren Produktmoduls} $\rmodplus{E}{A}$, falls
sesquilineare Abbildungen $\rSPcr{\cdot, \cdot}{\sss (n), \alg{A}}:
\rmod{E}{A} \times \rmod{E}{A} \to \alg{A}$ existieren, so da"s 
\begin{align}
    \label{eq:PositiveDeformationInneresProdukt}
    \rSPf{\cdot, \cdot}{\defalg{A}} = \sum_{n=0}^{\infty} \lambda^{n}
    \rSPcr{\cdot, \cdot}{\sss (n), \alg{A}},
\end{align}
ein vollst"andig positiv definites $\defalg{A}$-wertiges inneres Produkt auf
$\rmodf{E}{A}$ ist. 
\end{definition}

\begin{definition}["Aquivalenz zweier deformierter innerer Produktmoduln]
    Man nennt zwei \Name{Hermite}sche Deformation $(\rmodf{E}{A},
    \bullet, \rSPf{\cdot, \cdot}{A})$ und $(\rmodf{E}{A}, \bullet^{\prime},
    \rSPf[\prime]{\cdot, \cdot}{A})$ des inneren Produktmoduls $\rmodplus{E}{A}$
    {\em "aquivalent}, falls es einen $\defalg{A}$-Modulisomorphismus
    nach Definition \ref{Definition:DeformierterModulisomorphismus}
    der Form $T: (\rmodf{E}{A}, \bullet) \to (\rmodf{E}{A}, \bullet^{\prime})$ 
    gibt, so da"s 
\begin{align}
 \rSPf[\prime]{T(x), T(y)}{A} = \rSPf{x,y}{A}, \quad \text{f"ur alle} \quad
 x,y \in \rmod{E}{A}.
\end{align}
 \end{definition}

\begin{proposition}
    Sei $\alg{A}$ eine $^\ast$-Algebra mit Einselement und $\defalg{A}$
    eine \Name{Hermite}sche Deformation von $\alg{A}$. Desweiteren sei
    $P_{0} \in M_{n}(\alg{A})$ eine Projektion. Es existiert nun f"ur
    den 
    $\alg{A}$-Rechtsmodul $\rmod{E}{A} = P_{0} \alg{A}^n$, mit dem
    kanonischen inneren Produkt 
    \begin{align*}
        \rSP{x,y}{\alg{A}} = \sum_{i=1}^{n} x_{i}^{\ast} y_{i}, \quad
        \text{f"ur} \quad x,y\in P_{0}\alg{A}^{n},
    \end{align*}
 eine (bis auf "Aquivalenz eindeutige) \Name{Hermite}sche Deformation
 von $\rmod{E}{A}$ bez"uglich der deformierten Algebra $\defalg{A}$.
\end{proposition}

\begin{proof}
Wir wollen uns auf den ersten Teil der Behauptung konzentrieren, da
man hier die wesentliche Konstruktion erkennen kann. Den
zweiten Teil findet man detailliert in \citep[Proposition
4.13]{bursztyn:2001a:phd}. Aus $P_{0}$ erh"alt man einen deformierten Projektor $\defpro{P}
  \in M_{n}(\defalg{A})$, und $\defpro{P} \star \defalg{A}^{n}$ wird
  zu einen $\defalg{A}$-Rechtsmodul, so da"s wir eine
  Deformation $\defrmod{E}{A} = (\rmodf{E}{A}, \bullet)$ gefunden
  haben. Desweiteren ist $\rSPf{\cdot,\cdot}{A}=\rSP{\cdot, \cdot}{A} +
  \sum_{n=1}^{\infty} \lambda^{n} \rSPcr{\cdot,\cdot}{(n),\alg{A}}$ eine Deformation von
  $\rSP{\cdot, \cdot}{A}$, so da"s $\defrmod{E}{A} = (\defpro{P}
  \star \defalg{A}^{n}, \bullet, \rSPf{\cdot,\cdot}{A})$ eine \Name{Hermite}sche
  Deformation von $\rmodplus{E}{A}$ wird.
\end{proof}

Komplett analog zur \Name{Morita}-Theorie, die wir in den vorherigen
Kapitel behandelt haben, k"onnen wir \Name{Morita}-"Aquivalenz von
deformierten Algebren formulieren, und damit zum 
\Name{Picard}-Gruppoid gelangen. 

\begin{definition}[Deformierter Bimodul $\defbimod{B}{E}{A}$]
    \label{Definition:DeformierterBimodul}
Seien $\defalg{A}=(\algf{A},\star)$ und
$\defalg{B}=(\algf{B},\star')$ zwei deformierte Algebren nach
Definition \ref{Definition:FormaleAlgebraDeformation}. Ein
deformierter Bimodul $\defbimod{B}{E}{A} =
(\bimodf{B}{E}{A},\bullet',\bullet)$ ist ein $\defalg{B}$-Linksmodul
und ein $\defalg{A}$-Rechtsmodul, so da"s die
$\defalg{B}$-Linksmodulstruktur analog zu Gleichung
\eqref{eq:FormaleModulDeformation} eine formale Potenzreihe der Form 
\begin{align}
        \label{eq:FormaleModulDeformation}
        b \bullet' x  = \sum_{n=0}^\infty \lambda^{n} R'_{n} (b,x), \qquad
          x \in \bimod{B}{E}{A}, \quad b \in \alg{B}
\end{align} 
ist. Die Modulbedingungen schreiben sich dann als
\begin{compactenum}
\item $(b \star' b') \bullet' x = b \bullet' (b' \bullet' x)$,
\item $(b \bullet' x) \bullet a = b \bullet' (x \bullet a)$,
\item $(x \bullet a) \bullet a' = x \bullet (a \star a')$
\end{compactenum}
f"ur alle $x \in \bimodf{B}{E}{A}$, $a, a' \in \algf{A}$ und $b,b'
\in \algf{B}$.    
\end{definition}

\begin{bemerkung}
Wir bezeichnen die Deformation $\defbimod{B}{E}{A}$ eines $(\alg{B},\alg{A})$-Bimoduls
$\bimod{B}{E}{A}$ bez"uglich der beiden deformierten Algebren
$\defalg{A} = (\algf{A}, \star)$ und $\defalg{B} = (\algf{B},\star^{\prime})$ als
$((\algf{B},\star^{\prime}),(\algf{A}, \star))$-Bimodul bzw.~kurz als
$(\star^{\prime},\star)$-Bimodul.  
\end{bemerkung}

\begin{lemma}[Isomorphe Bimoduldeformationen]
Zwei $(\star^{\prime},\star)$-Bimoduldeformationen des Bimoduls
$\bimod{B}{E}{A}$ sind genau dann isomorph, wenn es einen
Bimodulisomorphismus $T= \id + \sum_{n=1}^{\infty}
\lambda^{n} T_{n}: \defbimod{B}{E}{A} \to 
\defbimod{B}{E^{\prime}}{A}$ gibt, so da"s 
\begin{align*}
    b \mathbin{\hat{\bullet}^{\prime}} x = T^{-1}(b \bullet^{\prime} T(x)) \quad
    \text{und} \quad  x \mathbin{\hat{\bullet}} a = T^{-1}(T(x) \bullet a). 
\end{align*}
Damit wirkt $T$ auf den $(\star^{\prime},\star)$-Bimoduldeformationen
durch $(\bimodf{B}{E}{A}, \bullet^{\prime}, \bullet) \mapsto
(\bimodf{B}{E}{A}, \hat{\bullet}^{\prime}, \hat{\bullet})$. 
\end{lemma}

Nach diesen Definitionen ist klar, was es auf dem deformierten Niveau
bedeutet, da"s zwei Algebren \Name{Morita}-"aquivalent sind. Die
Definition ist komplett analog zu der in Kapitel
\ref{sec:MoritaAequivalenz}.
 
Seien im weiteren $\defalg{A}=(\algf{A},\star)$ und
$\defalg{B}=(\algf{B},\star^{\prime})$ \Name{Morita}-"aquivalent.

\begin{proposition}[{\citep[Prop.~3.7]{bursztyn.waldmann:2004a}}]
Jeder Bimodul $\defbimod{B}{E}{A} \in \KPic(\defalg{B},\defalg{A})$
ist isomorph zu einer $(\star^{\prime}, \star)$-Bimoduldeformation
eines Elements $\bimod{B}{E}{A} \in \KPic(\alg{B},\alg{A})$, und 
\begin{align}
    \label{eq:clGruppoidhomomorphismusKPictoKPic}
    \cl[\ast]: \Pic(\defalg{B},\defalg{A}) \ni [\defbimod{B}{E}{A}]
    \mapsto [\bimod{B}{E}{A}] \in \Pic(\alg{B},\alg{A}) 
\end{align}
ist eine wohldefinierte Abbildung.
\end{proposition}

\begin{korollar}[{\citep[Lem.~3.8]{bursztyn.waldmann:2004a}}]
\label{Korollar:clPictoPicGruppenhomomorphismus}
Die klassische Limes-Abbildung 
\begin{align}
    \label{eq:clPictoPicGrupenhomomorphismus}
    \cl[\ast]:\Pic(\defalg{A}) \to \Pic(\alg{A})
\end{align}
ist ein Gruppenhomomorphismus.
\end{korollar}

\begin{proposition}[\Name{Morita}-"Aquivalenz von deformierten
    Algebren, {\citep{bursztyn:2001a:phd}}]
  Die deformierten $^\ast$-Algebren $\defalg{A}=(\Cinff{M},\star)$ und
  $\defalg{A}^{\prime}=(\Cinff{M},\star^{\prime})$ sind genau dann
  stark \Name{Morita}-"aquivalent wenn sie auch im ringtheoretischen
  Sinne \Name{Morita}-"aquivalent sind. Dies bedeutet weiter, da"s der
  Gruppoidmorphismus
  \begin{align*}
      \strPic(\defalg{A},\defalg{A^{\prime}}) \to \Pic(\defalg{A},\defalg{A^{\prime}}) 
  \end{align*}
injektiv ist. 
\end{proposition}

\section{Die \Name{Picard}-Gruppe von Sternprodukten}
\label{sec:PicardGruppevonSternprodukten}

Wir wollen uns nun auf Sternprodukte auf symplektischen
Mannigfaltigkeiten einschr"anken\footnote{In
  {\citep{bursztyn.waldmann:2004a}} wird auch der
  \Name{Poisson}-Fall behandelt.}. Sei dazu $(M,\omega)$ eine
symplektische Mannigfaltigkeit und $\defalg{A}=(\Cinff{M}, \star)$
eine Deformation der kommutativen Algebra $\alg{A}=\Cinf{M}$. Es ist
klar, da"s die \Name{Morita}-"Aquivalenzbimoduln Schnitte in einem
(komplexen) Geradenb"undel $\bundle{L}{\pi}{M}$ "uber der Mannigfaltigkeit sind. 

Wie bereits in Kapitel
\ref{sec:AeuqivalenzundKlassifikationvonSternprodukten} beschrieben
und beispielsweise in den Arbeiten
\citep{bertelson.cahen.gutt:1997a,deligne:1995a,nest.tsygan:1995b} gezeigt,
existiert eine Bijektion
\begin{align*}
    c: \Def(M,\omega) \to \frac{1}{\im \lambda}[\omega]+ \HdeRham[2](M)\FP,
\end{align*}
was im Sternproduktfall der charakteristischen Klasse $c(\star)$
(vgl.~Gleichung \eqref{eq:CharakteristischeKlasseSternprodukt}) entspricht.

\begin{proposition}[\Name{Morita}-"Aquivalenz von Sternprodukten]
Zwei Sternprodukte $(\Cinff{M}, \star)$ und $(\Cinff{M},
\star^{\prime})$ sind genau dann \Name{Morita}-"aquivalent, wenn es
ein komplexes Geradenb"undel $\bundle{L}{\pi}{M}$ und $\psi \in
\Sympl(M,\omega)$ gibt, so da"s 
\begin{align}
    \label{eq:MoritaAequivalenzSternprodukteChernKlasse}
    \psi^{\ast}c(\star^{\prime}) = c(\star) + 2 \pi \im c_{1}(L), 
\end{align}
wobei $c_{1}(L) \in \HdeRham[2](M,\field{C})$ das Bild von $e:
\HChern[2](M, \field{Z}) \to \HdeRham[2](M, \field{C})$ der ersten
\Name{Chern}-Klasse von $L$ ist. Dies ist
"aquivalent dazu, da"s f"ur die Wirkung $\Phi: \SPic(\Cinf{M})
\times \Def(\Cinf{M}) \to \Def(\Cinf{M})$ einen
Symplektomorphismus $T$ von $\Cinf{M}$ gibt, so da"s
$[T^{\ast}(\star)]$ und $[\star^{\prime}]$ im gleichen $\Phi$-Orbit liegen.  
\end{proposition}

\begin{satz}[Kern von {$\cl[\ast]$} im symplektischen Fall]
\label{Satz:KernvonclPicAtoPicA}
Sei $(M,\omega,\star)$ eine symplektische Mannigfaltigkeit mit
Sternprodukt, dann ist der Kern der klassischen Limes-Abbildung
\begin{align*}
    \cl[\ast]: \Pic((\Cinff{M},\star)) \to \Pic((\Cinf{M},\cdot))
\end{align*}
gegeben durch 
\begin{align}
    \label{eq:KernvonclPicAtoPicA}
    \ker(\cl[\ast]) \cong \frac{\HPoisson[1](M, \field{C})}{2 \pi \im
      \HPoisson[1](M, \field{Z})} + \lambda \HPoisson[1](M,\field{C})\FP.
\end{align}
Die Abbildung $\cl[\ast]$ ist ein Gruppenmorphismus\footnote{Im
  \Name{Poisson}-Fall ist diese Aussage im allgemeinen falsch!} und genau dann injektiv, wenn
$\HPoisson[1](M,\field{C}) = 0$. 
\end{satz}

\begin{proof}
Der Beweis findet sich in \citep{bursztyn.waldmann:2004a}.  
\end{proof}

Wie wir bereits in Beispiel \ref{Beispiel:PicardGruppevonCinfM}
gesehen haben, ist die \Name{Picard}-Gruppe der glatten Funktionen auf
einer Mannigfaltigkeit gegeben durch
\begin{align*}
    \Pic(\Cinf{M}) = \lcross{\Diff(M)}{\Pic(M)}
\end{align*}

Wir definieren die folgende Abbildung
\begin{align*}
\clr[\ast] = \pr \circ \cl[\ast]: \Pic(\defalg{A}) \to \SPic(\alg{A}),
\end{align*}
dabei ist $\pr: \Pic(\alg{A}) \to \SPic(\alg{A})$ die nat"urliche
Projektion der \Name{Picard}-Gruppe auf die statische
\Name{Picard}-Gruppe, die uns den \glqq interessanten\grqq{} Teil der
\Name{Picard}-Gruppe liefert. Desweiteren definieren wir folgende
Untermenge der Symplektomorphismen von $M$. Sei $\ell \in \HChern[2](M,\field{Z})$.
\begin{align}
    \label{eq:Pell}
   P_{\ell}=\{T \in \Sympl(M) | e(\ell) = T^{\ast} [\omega_{0}]-[\omega_{0}] \quad \text{und}
    \quad  T^{\ast} [\omega_{n}]=[\omega_{n}] \, \text{f"ur $n\ge 0$}
    \}. 
\end{align}

\begin{lemma}
Sei $(M,\omega, \star)$ eine symplektische Mannigfaltigkeit mit
Sternprodukt mit $c(\star) = \tfrac{1}{\im \lambda} [\omega] +
\sum_{n=0}^{\infty} \lambda^{n} [\omega_{n}]$, dann gilt 
\begin{equation}
\begin{aligned}
    \label{eq:Bildvonclr}
    \bild ( \clr[\ast]) = & \{ \ell \in \HChern[2](M,\field{Z}) |
    \exists T \in P_{\ell} \}.
\end{aligned}
\end{equation}
\end{lemma}

\begin{satz}[Bild von {$\cl[\ast]$} im symplektischen Fall]
 \label{Satz:BildvonclPicAtoPicA}   
Sei $(M, \omega, \star)$ eine symplektische Mannigfaltigkeit mit
Sternprodukt. Das Bild der klassischen Limes-Abbildung $\cl[\ast]$ in $\Pic(\Cinf{M}) =
\lcross{\Diff(M)}{\Pic(M)}$ ist gegeben durch
\begin{align}
    \label{eq:Bildvoncl}
    \bild(\cl[\ast]) = \{ (T,\ell) \in \lcross{\Diff(M)}{\Pic(M)} | T
    \in P_{\ell} \quad \text{und} \quad \ell \in \bild(\clr[\ast]) \}.
\end{align}
\end{satz}

Konkretere Ausf"uhrungen sowie einige Beispiele findet man in
\citep[Chapter 7]{bursztyn.waldmann:2004a}.


%% file: ausblick.tex
\chapter*{Ausblick}
\fancyhead[CE]{\slshape \nouppercase{Ausblick}} 
\fancyhead[CO]{\slshape \nouppercase{Ausblick}} 
\addcontentsline{toc}{chapter}{Ausblick}

Wir wollen noch einen kurzen Blick "uber den Tellerrand wagen und uns
interessanten zuk"unftigen Projekte zuwenden. 
Es liegt nahe sich, aufbauend auf dieser Arbeit, den $H$-"aquivarianten Fall deformierter
Algebren anzusehen. Dabei fallen mehrere interessante Aufgaben an, die
wir nun kurz umrei"sen wollen, und deren L"osung wir gerne dem
interessierten Leser "uberlassen.

\begin{compactenum}
\item Besonders interessant w"are es den Kern und das Bild der klassischen
Limes-Abbildung 
\begin{align*}
\cl[\ast]: & \strPicH(\defalg{B},\defalg{A}) \to
\strPicH(\alg{B},\alg{A}) 
\end{align*}
zu verstehen. Wie von \citet{bursztyn.waldmann:2004a} gezeigt wurde
(vgl.~Kapitel~\ref{sec:PicardGruppevonSternprodukten}), ist der
gew"ohnliche Fall eines Sternproduktes auf einer Mannigfaltigkeit,
d.~h.~ohne die Wirkung durch eine \Name{Hopf}-Algebra, schon
schwierig zu behandeln. Im allgemeinen wird es daher nicht einfach sein ein
nichttriviales Ergebnis zu erlangen. 

\item Aufbauend auf den Ergebnissen aus
Kapitel~\ref{sec:DerGruppoismorphismusstrPicHtostrPic} liegt es nahe
den Morphismus 
\begin{align*}
\strPicH \to \strPic
\end{align*}
f"ur einige konkrete Beispiele zu betrachten, und so die abstrakten
Ergebnisse dieser Arbeit mit Leben zu f"ullen. Wir denken hier
insbesondere an die Deformationsquantisierung, so da"s der Morphismus 
\begin{align*}
    \strPicH((\Cinff{M},\star)) \to \strPic((\Cinff{M},\star))
\end{align*}
f"ur ein konkretes $H$ von besonderem Interesse w"are. Wie wir in
Kapitel~\ref{sec:DerGruppoismorphismusstrPicHtostrPic} auch gezeigt
haben w"are daf"ur das Verstehen der Gruppe
$\GR{U}{H,(\Cinff{M},\star)}$ vonn"oten. Das \glqq einfache\grqq{}
Beispiel~\ref{Beispiel:GruppeUUCgB} von \citet{waldmann:2005a:misc}
zeigt aber, da"s auch dies ein eher kompliziertes Unterfangen
ist. 
\item Ein deutlich einfacheres und ebenso interessantes Beispiel w"are
die Gruppe  $\GR{U}{\ring{C}(G),\alg{A}}$ zu verstehen, im
speziellen f"ur $\alg{A}=\Cinf{M}$. Das Problem hierbei besteht darin die
auftretende Gruppenkohomologie zu verstehen. 
\end{compactenum}

%% file: algebra.tex
\chapter{Algebraische Grundlagen}
\label{chapter:AlgebraischeGrundlagen}

\fancyhead[CO]{\slshape \nouppercase{\rightmark}}
\fancyhead[CE]{\slshape \nouppercase{\leftmark}}

\section{Geordnete Ringe und formale Potenzreihen}
\label{sec:GeordneteRingeFormalePotenzreihen}

\subsection{Ringe und geordnete Ringe}
\label{sec:GeordeneteRinge}

F"ur einen tieferen Einblick in die hier behandelten Themen verweisen
wir auf die B"ucher \citep{rowen:1991a,lang:1997a} und \citep{jacobson:1985a}.  

\begin{definition}[Gruppe, Untergruppe]
\index{Gruppe|textbf} \index{Gruppe!Untergruppe}
Eine Gruppe $(G,\cdot)$ ist eine nichtleere Menge $G$ mit der Verkn\"upfung
$\cdot$, so da{\ss} f"ur alle Elemente $g, g_{1},g_{2},g_{3} \in G$
die folgenden Bedingungen gelten.
\begin{compactenum}
\item Die Verkn"upfung ist assoziativ, d.~h.~es gilt $g_{1} \cdot (g_{2} \cdot
    g_{3}) = (g_{1} \cdot g_{2}) \cdot g_{3}$.
\item Es existiert ein {\em Einselement $e \in G$}, so da"s $g\cdot e
    = e \cdot g = g$ ist.
\item Zu jedem Element $g \in G$ existiert ein Element $g^{-1} \in G$,
    so da"s $g^{-1} \cdot g = g \cdot g^{-1} =e$ ist.
\end{compactenum}
Desweiteren bezeichnet man eine Gruppe als {\em kommutativ}, wenn die
Verkn"upfung kommutativ ist. Eine {\em Untergruppe $H$ von $G$} ist
eine Menge $H \subseteq G$, so da"s f"ur alle $h_{1}, h_{2} \in H$
gilt $h_{1} \cdot h_{2}^{-1} \in H$.   
\end{definition}

\begin{definition}[Ring]
\label{Definition:Ring}
\index{Ring}
Ein Ring $(\ring{R},\cdot,+)$ ist eine nichtleere Menge $\ring{R}$ mit zwei
Verkn"upfungen $\cdot : \ring{R} \times \ring{R} \to \ring{R}$, der {\em
Multiplikation} und $+: \ring{R} \times \ring{R} \to \ring{R}$, der
{\em Addition}, so da"s die folgenden Eigenschaften erf"ullt sind:
\begin{compactenum}
\item $(\ring{R},+)$ ist eine kommutative Gruppe, und das neutrale
    Element bezeichnet man mit $0 \in \ring{R}$.
\item $(\ring{R},\cdot)$ ist {\em assoziativ} und besitzt ein {\em Einselement}
    $1_{\sss \ring{R}} \in \ring{R}$.
\item F"ur alle $r,s,t \in \ring{R}$ gelten die {\em
      Distributivgesetze} $(r+s)\cdot t = r\cdot t + s\cdot t$ und $ t
    \cdot (r+s) = t \cdot r + t\cdot s$.
\end{compactenum}
\end{definition}

\begin{definition}[Modul, Untermodul]
\index{Modul|textbf}
    Ein {\em $\ring{R}$-Rechtsmodul} (oder {\em Rechtsmodul})
    $(\rmod{M}{\ring{R}},+,\cdot)$ f"ur den Ring $\ring{R}$ ist eine
    nichtleere Menge $\rmod{M}{\ring{R}}$ mit den folgenden
    Eigenschaften f"ur alle $m,n \in \rmod{M}{\ring{R}}$ und
    $r,s \in \ring{R}$:
    \begin{compactenum}
    \item $(\rmod{M}{\ring{R}},+)$ ist eine \Name{Abel}sche Gruppe.
    \item Die Verkn"upfungen des Moduls $\cdot :\rmod{M}{\ring{R}} \times \ring{R}
        \to \rmod{M}{\ring{R}}$ und $(\rmod{M}{\ring{R}},+)$ sind mit
        den Ringstrukturen vertr"aglich
 \begin{align*}
  (m \cdot r) \cdot s = m\cdot (r \cdot s), \quad m\cdot (r+s) =
  m\cdot r + m\cdot s \quad \text{und} \quad (m+n)\cdot r = m.r + n\cdot r.
 \end{align*}
  \end{compactenum} 
Ein {\em Untermodul} ist ein Modul $\rmod{N}{\ring{R}} \subseteq \rmod{M}{\ring{R}}$ 
f"ur den  $\rmod{N}{\ring{R}}\cdot\ring{R} \subseteq \rmod{N}{\ring{R}}$ und
 $\rmod{N}{\ring{R}} + \rmod{N}{\ring{R}} \subseteq
 \rmod{N}{\ring{R}}$ gilt. 

Analog definiert man einen $\ring{R}$-Linksmodul.
\end{definition}

\begin{definition}[Ideal eines Rings]
\index{Ideal!Ring@eines Rings}
    Sei $(\ring{R},+,\cdot)$ ein Ring. Wir bezeichnen $\ring{I}_{\sss L}$
    als {\em Linksideal}, falls  
    \begin{compactenum}
    \item $(\ring{I}_{\sss L},+)$ eine Untergruppe von $(\ring{R},+)$ ist,
    \item $x\cdot r \in \ring{I}_{\sss L}$ f"ur alle $x\in
        \ring{I}_{\sss L}$ und alle $r \in \ring{R}$ ist.
    \end{compactenum}
Analog ist $\ring{I}_{\sss R}$
    als {\em Rechtsideal} definiert, falls  
    \begin{compactenum}
    \item $(\ring{I}_{\sss R},+)$ eine Untergruppe von $(\ring{R},+)$ ist,
    \item $r\cdot x  \in \ring{I}_{\sss R}$ f"ur alle $x\in
        \ring{I}_{\sss R}$ und alle $r \in \ring{R}$ ist.
    \end{compactenum} 
Ein {\em Ideal} $\ring{I}$ des Rings $\ring{R}$ ist sowohl Rechts- als
auch Linksideal. Ferner bezeichnet man ein Ideal als {\em echt}, falls
$\ring{I} \subset \ring{R}$. Die {\em trivialen Ideale} sind der Ring
$\ring{R}$ selbst und das Nullideal $0$. 
\end{definition}

\begin{lemma}[Ideale des Rings $\ring{R}$]
    Jeder Ring $\ring{R}$ ist ein Modul "uber sich selbst. Die Ideale
    des Rings $\ring{R}$ sind genau die Untermoduln des 
    Rings $\ring{R}$ als Modul. 
\end{lemma}

\begin{definition}[Eigenschaften von Ringen]
\label{Definition:EigenschaftenVonRingen}
 Sei $(\ring{R}, \cdot,+)$ ein Ring. Man nennt den Ring $\ring{R}$
 \begin{compactenum}

\item {\em kommutativ} oder {\em \Name{Abel}sch}, falls die
    Verkn"upfung $\cdot$ kommutativ ist.   

\item {\em einfach}, falls er nur triviale Ideale hat.

\item {\em (links-,rechts-) primitiv}, falls es einen treuen
    irreduziblen $\ring{R}$-(links-,rechts-) Modul gibt.

\item {\em \Name{Noether}sch}, wenn der Ring $\ring{R}$ als
    $\ring{R}$-Modul \Name{Noether}sch ist. Ein $\ring{R}$-Modul ist \Name{Noether}sch, wenn jeder
    Untermodul endlich erzeugt ist. "Aquivalent dazu sind die {\em
      aufsteigend Kettenbedingung}, d.~h.~es existiert ein $N \in
    \field{N}$, so da"s jede aufsteigende Kette von Untermoduln 
    station"ar wird
  \begin{align*}
        \modul{N}_{1} \subseteq \modul{N}_{2} \subseteq \modul{N}_{3} \subseteq \ldots
        \subseteq \modul{N}_{N} = \modul{N}_{N+1} = \ldots,
    \end{align*}
und die {\em Maximalbedingung f"ur Untermoduln}, d.~h.~jede nichtleere
Menge von $\ring{R}$-Untermoduln von $\modul{M}$ hat ein maximales Element bez"uglich
    der Inklusion.
\item {\em \Name{Artin}sch}, wenn der Ring 
    $\ring{R}$ als $\ring{R}$-Modul \Name{Artin}sch
    ist, d.~h.~man nennt einen Modul $\modul{M}$ \Name{Artin}sch, wenn
    jede absteigende Folge von Untermoduln station"ar wird: es
    existiert ein $N\in \field{N}$, so da"s 
    \begin{align*}
        \modul{M}_{1} \supseteq \modul{M}_{2} \supseteq \modul{M}_{3} \supseteq \ldots
        \supseteq \modul{M}_{N} = \modul{M}_{N+1} = \ldots.
    \end{align*}
    Dies bezeichnet man als {\em absteigende Kettenbedingung}.
    "Aquivalent ist die {\em Minimalbedingung f"ur Untermoduln},
    da"s jede nichtleere Menge von 
    $\ring{R}$-Untermoduln von $\modul{M}$ ein minimales Element bez"uglich
    der Inklusion hat.
\item {\em halbeinfach}, falls der Ring $\ring{R}$ als
    ($\ring{R}$-Modul betrachtet) eine Summe von
    einfachen $\ring{R}$-Un\-ter\-mo\-duln ist.

 \end{compactenum}
\end{definition}

\begin{definition}[\Name{Jacobson}-Radikal, \citep{jacobson:1989a}]
\label{Definition:JacobasonRadikal}
\index{Jacobson-Radikal@\Name{Jacobson}-Radikal}
    Das \Name{Jacobson}-Radikal eines Rings ist die Schnittmenge aller
    Kerne der irreduziblen Darstellungen des Rings. Wir Bezeichnen
    das \Name{Jacobson}-Radikal eines Rings $\ring{R}$ mit $J(\ring{R})$.
\end{definition}

\begin{definition}[Geordneter Ring]
    \label{Definition:GeordneterRing}
\index{Ring!geordneter}
Sei $\ring{R}$ ein kommutativer Ring mit $1_{\sss \ring{R}} \neq 0$ und
sei $\ring{P} \subset \ring{R}$. Man nennt $(\ring{R},\ring{P})$ einen
{\em geordneten Ring}, falls f"ur alle $r \in  \ring{R}$ entweder $r
\in \ring{P}$ oder $r \in -\ring{P}$ oder $r = 0$, und f"ur alle $r,s
\in \ring{P}$ auch $rs \in  \ring{P}$ und $r+s \in \ring{P}$
liegt. $\ring{R}$  ist somit die disjunkte Vereinigung
$\ring{R}=\{-\ring{P}\} \dot{\cup} \{0\} \dot{\cup} \{\ring{P}\}$ 
\end{definition}

\begin{definition}[Relationen und Betrag auf einem geordneten Ring]
Sei $(\ring{R}, \ring{P})$ ein geordneter Ring nach Definition
\ref{Definition:GeordneterRing}.     
\begin{compactenum}
        \item Wir nennen die Elemente $r \in \ring{P}$ {\em positive
              Elemente} und Elemente $r \in -\ring{P}$ {\em negative
              Elemente}.
        \item Wir bezeichnen $r > s$ ($r$ gr"o"ser $s$), falls $r-s
            \in \ring{P}$, und demnach $r < s$ ($r$ kleiner $s$), falls
            $r-s \in -\ring{P}$. Sei $r-s = 0$ so ist $r=s$ ($r$
            gleich $s$). 
        \item Den {\em Betrag} $|r|$ definieren wir als $|r|=r$ falls
            $r \in \ring{P}$ und $|r|=-r$ falls $r \in
            -\ring{P}$.
   \end{compactenum}  
\end{definition}

\begin{korollar}[Eigenschaften eines geordneten Rings $\ring{R}$]
   \label{Korollar:GeordneterRing} 
   Sei $(\ring{R}, \ring{P})$ ein geordneter Ring nach Definition
    \ref{Definition:GeordneterRing}.
       \begin{compactenum}
        \item Es gilt $r^2 =rr>0$ oder $r^2 =0 \; \Leftrightarrow \;
            r=0$ insbesondere ist $-1<0<1$.
        \item $1+\cdots+1 =n1 > 0$, was insbesondere hei"st, da"s
            der Ring $(\ring{R},\ring{P})$ die {\em Charakteristik
              $0$} hat und damit $\field{Z} \subseteq \ring{R}$ ist.
        \item Ein geordneter Ring ist immer {\em nullteilerfrei},
            d.~h.~aus $rs=0$ folgt immer $r=0$ oder $s=0$. 
    \end{compactenum}
\end{korollar}

\begin{definition}[Archimedische Ordnung]
   \label{Definition:ArchimedischeOrdnung}
\index{Ordnung!archimedische}
    Man nennt einen geordneten Ring $(\ring{R},\ring{P})$ {\em archimedisch
      geordnet}, falls f"ur alle Elemente $r,s \in \ring{P}$ ein $n
    \in \field{N}$ existiert, so da"s $r < ns$.
\end{definition}

\begin{bemerkungen}
~\vspace{-5mm}
\begin{compactenum}
\item Die komplexe Erweiterung $\ring{C}=\ring{R}(\im)$ eines geordneten
Rings $\ring{R}$ ist die Menge $\ring{C} = \ring{R} \times \ring{R}$
mit $\im=(0,1)$ und der komponentenweise Addition, sowie der
Multiplikation mit $\im^{2}=-1$. Die Elemente $z=(u,v) \in \ring{C}$
schreiben wir als $z=u+\im v$ und bezeichnen $u$ als den {\em
  Realteil} und $v$ als den {\em Imagin"arteil} von $z$. Die komplexe
Konjugation bezeichnen wir mit $z=u+\im v \mapsto \cc{z}= u-\im
v$. Offensichtlich gilt $\ring{R} \ni z\cc{z} >0$ f"ur $z\neq 0$ und
$z\cc{z}=0$ f"ur $z=0$. 

\item Aufgrund der Nullteilerfreiheit k"onnen wir bei einem Ring immer zu
seinem {\em Quotientenk"orper} $\qring{R}$ von $\ring{R}$
"ubergehen. Dabei ist $\qring{R}$ die "Aquivalenzklasse
formaler Br"uche $\tfrac{a}{b}$ mit $a,b \in \ring{R}$ und $b\neq 0$,
so da"s $\tfrac{a}{b} \cong \tfrac{a'}{b'}$ falls es $c,c' \in
\ring{R}$ gibt mit $ca=c'a'$ und $cb=c'b'$. Mittels der gew"ohnlichen
Addition und Multiplikation f"ur Br"uche wird $\qring{R}$ zu einem
geordneten K"orper. Die Inklusion $\ring{R} \hookrightarrow
\qring{R}$ geht via $a \mapsto \tfrac{a}{1}$ und ist ordnungserhaltend.
\end{compactenum}
\end{bemerkungen}

\begin{beispiele}[Beispiele f"ur (geordnete) Ringe]
~\vspace{-5mm}
\begin{compactenum}
   \item Die ganzen Zahlen $\field{Z}$ bilden einen geordneten und
       sogar archimedischen Ring. Er zeichnet sich dadurch aus,
       da"s er der kleinste geordnete Ring und Teilmenge jedes
       geordneten Rings $\field{Z} \subseteq \ring{R}$ ist.
   \item Die rationalen Zahlen $\field{Q}$ und die reellen Zahlen
       $\field{R}$ bilden geordnete und archimedische
       Ringe. $\field{Q}$ und $\field{R}$ sind zudem sogar
       K"orper. Offensichtlich ist der Quotientenk"orper von
       $\field{Z}$ der K"orper der rationalen Zahlen $\field{Q}$. 
   \item Sei $\ring{R}$ ein beliebiger Ring, so sind die $n \times
       n$-Matrizen $M_{n}(\ring{R})$ "uber diesem Ring $\ring{R}$ wieder ein
       Ring. F"ur $n \ge 2$ ist $M_{n}(\ring{R})$ nicht geordnet,
       selbst wenn $\ring{R}$ geordnet (und archimedisch) ist.    
\end{compactenum}
\end{beispiele}

\subsection{Formale Potenzreihen}
\label{sec:FormalePotenzreihen}
Der richtige Rahmen, um "uber {\em formale Potenzreihen} nachzudenken,
ist ein {\em Modul} $\modul{M}$ "uber einem geordneten Ring
$\ring{R}$. Wir werden uns bei der Darstellung an \cite[Seite
422f]{jacobson:1989a} und \citep[Seite 139f]{waldmann:1999a}
orientieren.  

\begin{definition}[Formale Potenzreihen]
    \label{Definition:FormalePotenzreihe}
\index{formale Potenzreihe|textbf}
Sei $\modul{M}$ ein Modul "uber einem geordneten Ring $\ring{R}$, dann
definieren wir den Modul $\modulf{M}$ der formalen Potenzreihen als
kartesisches Produkt   

\begin{align}
\modulf{M}:= \bigX_{j=0}^{\infty} \modul{M}_j \quad
\text{mit} \quad \modul{M}_{j}=\modul{M}. 
\end{align}
Elemente in $\modulf{M}$ sind dann Folgen
$m=(m_0,m_1,m_2,\cdots)$, die wir symbolisch in der folgenden Form
\begin{align}
 \label{eq:FormaleReihe}
m= \sum_{j=0}^{\infty} \lambda^j m_{j}, \quad \text{mit} \quad m_{j}
\in \modul{M},
\end{align}
schreiben werden, dabei bezeichnen wir $\lambda$ als den {\em formalen
  Parameter} der formalen Reihe.  
\end{definition}

\begin{bemerkungen}[Formale Potenzreihen]
~\vspace{-5mm}
\begin{compactenum}
\item Formale Reihen sind ein rein algebraisches Konzept, d.~h.~wir machen
uns bei der Notation \eqref{eq:FormaleReihe} keinerlei Gedanken "uber
Konvergenz im herk"ommlichen Sinne\footnote{Ist man, wie bei manchen
  Sternprodukten, in einem konvergenten Rahmen, so \"ubernimmt $\hbar$
  die Stelle von $\lambda$. Vergleiche hierzu beispielsweise
  \citep{agarwal.wolf:1970a,agarwal.wolf:1970b,agarwal.wolf:1970c,beiser:2005a,beiser.roemer.waldmann:2005a:pre}.}.    
\item Das Objekt $\modulf{M}$ ist offensichtlich wieder ein Modul "uber
dem Ring $\ring{R}$, wobei die
Modulstruktur die Multiplikation mit Elementen aus $\ring{R}$ sowie die gliedweise Addition ist.
\end{compactenum}
\end{bemerkungen}

Sei $\alg{A}$ eine Algebra, so kann man analog $\algf{A}$, eine
formale Potenzreihe in $\lambda$ mit Werten in $\alg{A}$,
definieren, wobei die $\ring{R}$-Algebrastruktur durch 

\begin{align}
    \label{eq:MultiplikationSummeFormaleReihe}
    ab & := \sum_{n=0}^{\infty} \lambda^{n} \sum_{j=0}^{n} a_{j}b_{n-j}
  \end{align}

gegeben ist. Die Multiplikation ist $\ring{R}$-bilinear und
$\algf{A}$ ist genau dann kommutativ oder assoziativ, wenn die
Algebra $\alg{A}$ kommutativ oder assoziativ ist. Ist $1_{\sss
  \alg{A}} \in \alg{A}$
das Einselement, so ist $1_{\sss \alg{A}} :=(1_{\sss \alg{A}},0,0,\cdots) \ni \algf{A}$ das
Einselement von $\algf{A}$.

\begin{definition}[$\ringf{R}$-lineare Abbildungen]
    \label{Definition:RlambdaLineareAbbildungen}
Seien $\modulf{M}$ und $\modulf{N}$ zwei $\ring{R}$-Moduln. Wir
definieren eine $\ringf{R}$-lineare Abbildung
\begin{align*}
\phi = \sum_{k=0}^{\infty} \lambda^{k} \phi_{k}:\modulf{M}
\to\modulf{N},
\end{align*}
wobei $\phi_{n}:\modul{M} \to \modul{N}$ $\ring{R}$-lineare Abbildungen sind, "uber 
\begin{align}
    \phi(m)=\sum_{n=0}^{\infty} \lambda^{n} \sum_{k=0}^{n}
    \phi_k(m_{n-k}), \quad \text{f"ur} \quad m\in \modulf{M}. 
\end{align}
\end{definition}

\begin{lemma}
    \label{}
 Seien $\modulf{M}$ und $\modulf{N}$ zwei $\ring{R}$-Moduln. Zu jeder
 $\ringf{R}$-linearen Abbildung $\phi: \modulf{M} \to\modulf{N}$ gibt
 es eindeutig bestimmte $\ring{R}$-lineare Abbildungen $\phi_0, \phi_1,
 \ldots : \modul{M} \to\modul{N}$, so da"s es ein
 $\phi=\sum_{n} \lambda^{n} \phi_{n}$ nach
 Definition  \ref{Definition:RlambdaLineareAbbildungen} gibt. Es gilt also
 \begin{align}
     \Hom[\sss \ring{R}]\modulf{(\modul{M},\modul{N})} \cong
     \Hom[\sss \ringf{R}](\modulf{M},\modulf{N}). 
 \end{align}
\end{lemma}

Die {\em Ordnung $o$} ist eine Struktur mit der formale Potenzreihen
auf nat"urliche Weise ausgestattet sind. Mit ihrer Hilfe werden wir
eine {\em Ultrametrik $d$} auf dem Raum der formalen Potenzreihen
definieren. 

\begin{dlemma}[Ordnung auf $\modulf{M}$]
    \label{Definition:OrdnungMlambda}
\index{Ordnung!formale Potenzreihe@auf formaler Potenzreihe}
Sei $\modulf{M} \ni m=\sum_{n=0}^{\infty} \lambda^{n} m_{n}$ eine formale Potenzreihe
. Wir definieren die {\em Ordnung $o$} durch
\begin{equation}  
\begin{aligned} 
& o :\modulf{M} \to \field{N} \cup \{+\infty\} \\
 & o(m) = \begin{cases} k \quad & \text{falls} \quad m_{0},m_{1}, \cdots
          m_{k-1} = 0 \quad \text{und} \quad m_k \neq 0, \\ +\infty
          \quad & \text{falls} \quad m=0.\end{cases} 
  \end{aligned}
\end{equation}
 Sie hat die folgenden Eigenschaften
\begin{compactenum}
\item $o(m+n)\ge \min (o(m),o(n))$,
\item $o(m)=o(-m).$
\end{compactenum}
\end{dlemma}

\begin{definition}[$\lambda$-adische Bewertung auf $\modulf{M}$ und
    Ultrametrik $d_{\sss \varphi}$]
\index{Lambda-adische Bewertung@$\lambda$-adische Bewertung}
\index{Ultrametrik}
\label{Definition:LambdaBewertungUltraMetrik} 
Gegeben $\modulf{M}$. Wir definieren die {\em $\lambda$-adische
  Bewertung} $\varphi: \modulf{M} \to\field{Q}$ durch 
\begin{align} 
    \label{eq:LambdaBewertung}
    \varphi(m)=2^{-o(m)},  
\end{align} 
wobei wir festlegen, da"s $2^{-\infty}=0$ ist. Durch die
$\lambda$-adische Bewertung $\varphi$ wird eine
{\em Ultrametrik} 
\begin{align}
    \label{eq:UltrametrikDef}
d_{\sss \varphi}: \modulf{M} \times \modulf{M} \to \field{Q} 
\end{align}
auf $\modulf{M}$ definiert mittels
\begin{align}
 \label{eq:Ultrametrik} 
d_{\sss \varphi}(m,n):=\varphi(m-n)= 2^{-o(m-n)}. 
\end{align}
\end{definition} 
 
\begin{lemma}[Eigenschaften der Ultrametrik {$d_{\sss \varphi}$}]
\label{Lemma:EigenschaftenUltrametrik}
Die Ultrametrik $d_{\sss \varphi}$ nach Definition
\ref{Definition:LambdaBewertungUltraMetrik} hat f"ur alle Elemente $m,n,\ell \in \modulf{M}$ die folgenden Eigenschaften:
\begin{compactenum}
    \item $d_{\sss \varphi}(m,n) = d_{\sss \varphi}(n,m) \ge 0$,
    \item $d_{\sss \varphi}(m,n) = 0 \quad \Longleftrightarrow \quad m = n$,
    \item $d_{\sss \varphi}(m,n) \le \min (d_{\sss \varphi}(m,\ell), d_{\sss \varphi}(n,\ell))$.
\end{compactenum}
\end{lemma}

\begin{lemma}[$\lambda$-adische Topologie]
\label{Lemma:LambdaadischeTopologie}
\index{Lambda-adische Topologie@$\lambda$-adische Topologie}
Sei $\modul{M}$ ein $\ring{R}$-Modul, $\alg{A}$ eine
$\ring{R}$-Algebra und $\lmod{A}{E}$ ein $\alg{A}$-Linksmodul. Es
gilt
\begin{compactenum}
\item Die Abbildung $d_{\sss \varphi}$ nach Definition
    \ref{Definition:LambdaBewertungUltraMetrik} definiert eine
    Ultrametrik f"ur $\modulf{M}$, so da"s der Modul mit der
    Metrik $(\modulf{M},d_{\sss \varphi})$ ein
    vollst"andiger metrischer Raum wird. Man nennt die durch $d_{\sss \varphi}$
    induzierte Topologie die $\lambda$-adische Topologie.
\item Mittels der $\lambda$-adischen Topologie wird $\ringf{R}$ zu
    einem topologischen Ring, $\modulf{M}$ zu einem topologischen
    $\ringf{R}$-Modul, $\algf{A}$ eine topologische
    $\ringf{R}$-Algebra und $\lmodf{A}{E}$ ein topologischer
    $\algf{A}$-Linksmodul. Die auf $\modul{M} \subseteq
    \modulf{M}$ induzierte Topologie ist die diskrete.
\item Seien $m_0,m_1,\cdots \in \modul{M}$, so gilt im Rahmen der
    $\lambda$-adischen Topologie
    \begin{align}
        \lim_{N\to \infty} \sum_{n=0}^{N} \lambda^{n} m_{n} =
        \sum_{n=0}^{\infty} \lambda^{n} m_{n}.  
    \end{align}
\item Der $\ring{R}$-Modul $\modul{M}[\lambda]$ der Polynome in
    $\lambda$ mit Koeffizienten in $\modul{M}$ liegt dicht in
    $\modulf{M}$ bez"uglich der $\lambda$-adischen Topologie.
\end{compactenum}  
\end{lemma}

\begin{proof}
Der Beweis wird in \citep[Kapitel XVI.1-4]{kassel:1995a} gef"uhrt.
\end{proof}

\section{\Name{Hopf}-$^\ast$-Algebren} 
\label{sec:HopfAlgebren}
Zuerst wollen wir uns mit Algebren und Koalgebren auseinandersetzen,
und den g"angigen Formalismus einf"uhren. Eine \Name{Hopf}-Algebra
wird dann beide Strukturen haben, sowie eine weitere, die der
Antipode.  
In einem weiteren Schritt werden wir "uber
\Name{Hopf}-$^\ast$-Algebren sprechen, die zus"atzlich zu der
\Name{Hopf}-Struktur auch die einer $^\ast$-Algebra besitzen, also mit
einem antilinearen Antiautomorphismus ausgestattet sind, der mit allen
\Name{Hopf}-Algebrastrukturen auf eine nat"urliche Weise
vertr"aglich sein wird. Die Notationen in diesem Kapitel orientiert
sich im wesentlichen an \citep{majid:1995a} und
\citep{kassel:1995a}. Ferner legen wir dem geneigten Leser
\citep{sweedler:1969a,chari.pressley:1994a,klimyk.schmuedgen:1997a,etingof.schiffmann:1998a},
sowie  \citep{jansen.waldmann:2005a} nah. In diesem Kapitel sei
$\ring{C}=\ring{R}(\im)$, wobei $\ring{R}$ ein geordneter Ring sei, und es gelte $\im^2=-1$.  

\subsection{Algebren, Koalgebren und Bialgebren}

\subsubsection{Algebra}

\begin{definition}[Algebra]
    \label{Definition:Algebra}
\index{Algebra!assoziative}
     Eine {\em assoziative  Algebra $(\alg{A},\mu,\eta)$} "uber einem
     kommutativen Ring  $\ring{C}$ ist ein $\ring{C}$-Modul $\alg{A}$
     mit zwei $\ring{C}$-Modul Abbildungen: 
    \begin{compactenum}
    \item $\mu: \alg{A} \otimes \alg{A} \to \alg{A}$, der
        Multiplikation, 
    \item $\eta: \ring{C} \to \alg{A}$, der Einsabbildung,
    \end{compactenum}
mit den Eigenschaften
\begin{align}
    \label{eq:AssoziativeAlgebraEigenschaften}
    \mu \circ (\id \otimes \mu) = \mu \circ (\mu \otimes \id) \quad
    \mbox{und} \quad \mu \circ (\eta \otimes \id) = \mu \circ (\id \otimes \eta)=\id,
\end{align}
was der Kommutativit"at der folgenden Diagramme entspricht:

\begin{equation}
\label{eq:KommutativeDiagrammeAlgebra}
\bfig\square/>`>`>`>/<900,500>[\alg{A}\otimes\alg{A}\otimes
\alg{A}`\alg{A}\otimes\alg{A}`\alg{A}\otimes\alg{A}`\alg{A};\mu 
\otimes\id`\id\otimes\mu`\mu`\mu]\efig \quad  
\bfig\Vtrianglepair/>`<-`=`>`=/<600,500>[\ring{C}\otimes\alg{A}`\alg{A}\otimes\alg{A}`\alg{A}\otimes\ring{C}`\alg{A};\eta\otimes\id`\id\otimes\eta``\mu`]\efig
\end{equation}
Die erste Eigenschaft ist die Assoziativit"at.
\end{definition}

\begin{bemerkung}[Einsabbildung]
Die Einsabbildung ist so zu verstehen, da"s jedem Element aus
$\ring{C}$ ein vielfaches des Einselements in der Algebra zugeordnet
wird:
\begin{align*}
    \ring{C} \ni z \mapsto \eta(z)=z \cdot 1_{\sss \alg{A}} \in  \alg{A}.
\end{align*}
\end{bemerkung}

\begin{definition}[Algebramorphismus]
Ein {\em Algebramorphismus $f:(\alg{A}_{\sss 1}, \mu_{\sss 1},
  \eta_{\sss 1}) \to (\alg{A}_{\sss 2}, \mu_{\sss 2}, \eta_{\sss 2})$}
ist eine lineare Abbildung, so da"s  
\begin{align} 
\mu_{\sss 2}\circ(f\otimes f) = f\circ \mu_{\sss 1} \quad \mbox{und}
\quad f \circ \eta_{\sss 1} = \eta_{\sss 2},
\end{align}
was "aquivalent zu der Kommutativit"at der folgenden Diagramme ist:

\begin{equation}
\label{eq:KommutativeDiagrammeAlgebramorphismus}
\bfig\square/>`>`>`>/<900,500>[{\alg{A}_{\sss 1} \otimes\alg{A}_{\sss
  1}}`{\alg{A}_{\sss 2} \otimes\alg{A}_{\sss 2}}`{\alg{A}_{\sss
1}}`{\alg{A}_{\sss 2}};{f \otimes f}`{\mu_{\sss 1}}`{\mu_{\sss 2}}`f]\efig \quad  
\bfig\Atriangle<400,500>[\ring{C}`{\alg{A}_{\sss 1}}`{\alg{A}_{\sss
      2}};{\eta_{\sss 1}}`{\eta_{\sss 2}}`f]\efig
\end{equation}
\end{definition}

\begin{definition}[Kommutativit"at einer Algebra $\alg{A}$]
\index{Algebra!kommutative}
Eine Algebra hei"st {\em kommutativ} oder {\em \Name{Abel}sch} wenn
das folgende Diagramm kommutiert: 
\begin{equation}
\label{eq:KommutativesDiagrammKommutativitaetAlgebra}
\bfig\Vtriangle<400,500>[\alg{A}\otimes\alg{A}`\alg{A}\otimes\alg{A}`\alg{A};\tau`\mu`\mu]\efig
\end{equation}

wobei $\tau: \alg{A} \otimes \alg{A} \to \alg{A}
\otimes \alg{A}$, der {\em
  Vertauschungsoperator}\index{Vertauschungsoperator|textbf} oder {\em
  Flip}, zwei 
Argumente vertauscht $\tau(a \otimes b) =b 
\otimes a$. Die Kommutativit"at von $\alg{A}$ bedeutet also $\mu
\circ \tau = \mu$ f"ur alle Elemente der Algebra.  
\end{definition}
Wir f"uhren daher hier eine neue Multiplikation ein, die sich im Hinblick auf
Koalgebren und \Name{Hopf}-Algebren als n"utzlich erweisen wird
$\muop := \mu \circ \tau$.

\begin{definition}[Zentrum $\zentrum{\alg{A}}$ einer Algebra $\alg{A}$]
\index{Zentrum|textbf}
Die Elemente $z\in \alg{A}$ f"ur die $\mu(a \otimes z) = \muop (a
\otimes z)$ f"ur alle $a\in \alg{A}$ ist, bezeichnen wir als
{\em Zentrum} $\zentrum{\alg{A}}$ der Algebra $\alg{A}$. 
\end{definition}

\begin{bemerkung}
Offensichtlich ist eine Algebra $\alg{A}$ genau dann kommutativ, falls
$\zentrum{\alg{A}}= \alg{A}$. 
\end{bemerkung}

\subsubsection{Koalgebra}

Unter einer Koalgebra versteht man nun das Duale einer Algebra. Salopp
gesprochen bedeutet dies, da"s man alle Pfeile in den Diagrammen
\eqref{eq:KommutativeDiagrammeAlgebra} umdreht. Im folgenden werden wir eine mathematische
Definition angeben. 

\begin{definition}[Koalgebra]
    \label{Definition:Koalgebra}
\index{Koalgebra}
\index{Komultiplikation}
\index{Koeins}
    Eine koassoziative {\em Koalgebra} mit einer Koeins
    $(\alg{K},\Delta,\varepsilon)$ "uber einem kommutativen Ring
    $\ring{C}$ ist ein $\ring{C}$-Modul $\alg{K}$ mit zwei
    $\ring{C}$-Modul Abbildungen:  
    \begin{compactenum}
    \item $\Delta: \alg{K} \to \alg{K} \otimes \alg{K}$ der
        Komultiplikation und 
    \item $\varepsilon: \alg{K} \to \ring{C}$ der Koeins,
    \end{compactenum}
    
mit den Eigenschaften
\begin{align}
    \label{eq:KoassoziativeKoalgebraEigenschaften}
(\Delta \otimes \id)\circ \Delta = (\id \otimes \Delta)\circ \Delta
\quad \mbox{und} \quad (\varepsilon \otimes \id)\circ \Delta = (\id \otimes \varepsilon)
 \circ \Delta =\id,   
\end{align}

was der Kommutativit"at der folgenden Diagramme entspricht:

$$\bfig\square/<-`<-`<-`<-/<900,500>[\alg{K}\otimes\alg{K}\otimes\alg{K}`\alg{K}\otimes\alg{K}`\alg{K}\otimes\alg{K}`\alg{K};\Delta
\otimes\id`\id\otimes\Delta`\Delta`\Delta]\efig \quad
\bfig\Vtrianglepair/<-`>`=`<-`=/<600,500>[\ring{C}\otimes\alg{K}`\alg{K}\otimes\alg{K}`\alg{K}\otimes\ring{C}`\alg{K};\varepsilon\otimes\id`\id\otimes\varepsilon``\Delta`]\efig$$.

\end{definition}
F"ur die Komultiplikation verwenden wir die Notation von
\citet{sweedler:1969a}\index{Sweedler-Notation@\Name{Sweedler}-Notation}.
Dies bedeutet, da"s wir f"ur alle $k\in 
\alg{K}$ die Komultiplikation als 
\begin{align}
    \label{eq:SweedlersNotation}
    \Delta(k)=\sum_{i} k_{\sss i(1)} \otimes k_{\sss i(2)}
\end{align}
schreiben werden. Um die Schreibweise noch ein wenig zu vereinfachen,
beziehungsweise "uber\-sicht\-li\-cher zu gestalten, werden wir auf
den Summationsindex $i$ und das Summenzeichen 
verzichten. Gelegentlich brauchen wir den Ausdruck $\tau\circ \Delta$,
der so wichtig ist, da"s er von uns ein eigenes Symbol bekommt: 
 \begin{align}
    \label{eq:OppositeComultiplication}
    \Deltaop(k):= \tau \circ \Delta (k) = \sum_{i} k_{\sss i(2)}
    \otimes k_{\sss i(1)}. 
\end{align}

Die Koassoziativit"at\index{Koassoziativitaet@Koassoziativit\"at}
schreibt sich dann explizit f"ur $k \in \alg{K}$, unter Verwendung
von \Name{Sweedler}s Notation:   
 
\begin{align}  
     \label{eq:CoassociativityArguments}
    k_{\sss (1)} \otimes  k_{\sss (2)(1)} \otimes  k_{\sss (2)(2)} =
    k_{\sss (1)(1)} \otimes  k_{\sss (1)(2)} \otimes  k_{\sss (2)} =
    k_{\sss (1)} \otimes  k_{\sss (2)} \otimes  k_{\sss (3)}.
\end{align}

\begin{definition}[Kokommutativit"at einer Koalgebra $\alg{K}$]
\index{Koalgebra!kokommutative}
Man nennt eine Koalgebra $\alg{K}$ {\em kokommutativ}, falls 
$\Deltaop = \Delta$ f"ur alle $k\in \alg{K}$ ist. Diese Aussage ist
"aquivalent dazu, da"s das folgende Diagramm kommutiert:  
\begin{equation}
\label{eq:KommutativesDiagrammKokommutativitaet}
\bfig\Atriangle<400,500>[\alg{K}`\alg{K}\otimes\alg{K}`{\alg{K}\otimes\alg{K}.};\Delta`\Delta`\tau]\efig
\end{equation}
\end{definition}

 

\begin{definition}[Koalgebramorphismus]
\index{Koalgebramorphismus}
Ein {\em Koalgebramorphismus} $f:(\alg{K}_{\sss 1}, \Delta_{\sss 1},
\varepsilon_{\sss 1}) \to (\alg{K}_{\sss 2}, \Delta_{\sss 2},
\varepsilon_{\sss 2})$ ist eine lineare Abbildung, so da"s 
\begin{align}
(f\otimes f)\circ \Delta_{\sss 1} = \Delta_{\sss 2} \circ f \quad
\mbox{und} \quad \varepsilon_{\sss 1}=\varepsilon_{\sss 2} \circ f. 
\end{align}
Dies entspricht der Kommutativit"at der folgenden Diagramme
\begin{equation}
    \label{eq:KommutativesDiagrammKoalgebramorphismus}
\bfig\square/>`>`>`>/<900,500>[{\alg{K}_{\sss 1}}`{\alg{K}_{\sss
    2}}`{\alg{K}_{\sss 1} \otimes\alg{K}_{\sss 1}}`{\alg{K}_{\sss 2}
  \otimes\alg{K}_{\sss 2}};{f}`{\Delta_{\sss
    1}}`{\Delta_{\sss 2}}`{f \otimes f}]\efig \quad   
\bfig\Vtriangle<400,500>[{\alg{K}_{\sss 1}}`{\alg{K}_{\sss
    2}}`{\ring{C}};{f}`{\varepsilon_{\sss 1}}`{\varepsilon_{\sss 2}}]\efig    
\end{equation}
\end{definition}

\begin{definition}[Primitive und gruppenartige Elemente in einer
    Koalgebra]
\index{Elemente!primitive}
\index{Elemente!gruppenartige}
Sei $(\alg{K}, \Delta,\varepsilon)$ eine Koalgebra. Ein Element $k\in
\alg{K}$ nennt man  
\begin{compactenum}
\item {\em primitiv}, falls $\Delta(k) = 1\otimes k + k \otimes 1$
    ist. Die Menge aller primitiven Elemente der Koalgebra $\alg{K}$
    bezeichnen wir mit $\Prim{\alg{K}}=\{k\in \alg{K} | \Delta(k)=
    1\otimes k + k \otimes 1\}$.
\item {\em gruppenartig}, falls $\Delta(k) = k \otimes k$ ist. 
\end{compactenum}  
\end{definition}

\begin{beispiel}[Der Ring $\ring{C}$ ist Koalgebra]
\label{Beispiel:RingCistKoalgebra}
Der Ring $\ring{C}$ ist eine kokommutative Koalgebra "uber sich
selbst mit $\varepsilon(\ring{C})=\id_{\sss \ring{C}}$ und $\Delta:
\ring{C} \to \ring{C} \otimes \ring{C}$ als nat"urlichen
Isomorphismus. 
\end{beispiel}

\subsubsection{Bialgebra}

Wie schon zu Beginn des Kapitels erw"ahnt, versteht man unter einer
Bialgebra eine Menge, die sowohl Algebra als auch Koalgebra ist. Beide
Strukturen m"ussen miteinander vertr"aglich sein. 

\begin{definition}[Bialgebra]
\index{Bialgebra}
Eine {\em Bialgebra $(H,\mu,\eta,\Delta,\varepsilon)$} "uber einem
Ring $\ring{C}$ ist sowohl eine Algebra als auch eine Koalgebra. Die
Kompatibilit"atsbedingungen lauten
\begin{compactenum}
\item $\Delta \circ \mu =(\mu \otimes \mu) \circ (\id \otimes \tau
    \otimes \id) \circ (\Delta \otimes \Delta)$ und $\Delta(1_{\sss
      H})=1_{\sss H} \otimes 1_{\sss H}$,
\item $\varepsilon \circ \mu = \mu_{\sss \ring{C}} \circ(\varepsilon
   \otimes \varepsilon)$ und $\varepsilon(1_{\sss H})=1_{\sss \ring{C}}$. 
\end{compactenum}

\end{definition}

Dabei bezeichnen wir mit $\mu_{\sss \ring{C}} : \ring{C} \otimes \ring{C}
\to\ring{C}$ die gew"ohnliche Multiplikation im Ring
$\ring{C}$. Damit sind die Abbildungen $\Delta: H \to H \otimes H$ und
$\varepsilon: H \to \ring{C}$ Algebrahomomorphismen. Dies ist "aquivalent dazu, da"s $\mu:
H \otimes H \to H$ und $\eta: \ring{C} \to H$ Koalgebrahomomorphismen sind.

\subsection{\Name{Hopf}-$^\ast$-Algebren}

\subsubsection{Einige Definitionen}

\begin{definition}[Antipodenabbildung {$S$} und \Name{Hopf}-Algebra]
\label{Definition:HopfAlgebra}
\index{Antipodenabbildung}
    Gegeben sei eine Bialgebra $(H,\mu,\varepsilon,\Delta,\eta)$
    "uber einem Ring $\ring{C}$. Sei $S: H \to H$ eine
    lineare Abbildung, die die folgenden Bedingungen erf"ullt 
    \begin{align}
        \label{eq:AntipodeDefinition}
        \mu \circ (S \otimes \id) \circ \Delta = \mu \circ (\id
        \otimes S) 
        \circ \Delta = \eta \circ \varepsilon. 
    \end{align}
    so nennen wir $S$ eine {\em Antipodenabbildung} oder kurz {\em Antipode}. Eine Bialgebra
    mit einer Antipode $S: H \to H$ nennt man {\em \Name{Hopf}-Algebra}. 
\index{Hopf-Algebra@\Name{Hopf}-Algebra|textbf}
\end{definition}

Die Eigenschaften der Antipodenabbildung einer \Name{Hopf}-Algebra
wollen wir in der folgenden Proposition zusammenfassen.

\begin{proposition}[Eigenschaften der Antipode]
    \label{Proposition:Antipode}
    Gegeben sei eine \Name{Hopf}-Algebra $(H,\mu, \eta, \Delta,
    \varepsilon,S)$, dann hat die Antipode $S$ folgende Eigenschaften
    f"ur alle Elemente in $H$. 
    \begin{compactenum}
    \item Die Antipode $S$ ist eindeutig.
    \item $S\circ \mu = \mu\circ (S\otimes S)\circ\tau$ und $S(1_{\sss
          H})=1_{\sss H}$, d.~h.~$S$ ist ein
        Algebraantiautomorphismus. 
    \item $(S\otimes S)\circ \Delta = \Deltaop \circ S$ und
        $\varepsilon \circ S  = \varepsilon $, d.~h.~$S$ ist ein
        Koalgebraantiautomorphismus.   
    \end{compactenum}
\end{proposition}
\begin{proof}
    Der Beweis findet sich zum Beispiel in \citep[Theorem III.3.4]{kassel:1995a}.
\end{proof}

\begin{bemerkungen}[Antipode]
    \label{Bemerkung:Antipode}
~\vspace{-5mm}
\begin{compactenum}
\item Die Antipode ist vergleichbar mit einer Inversen, jedoch ist durch
    die Existenz einer solchen Abbildung nicht gesichert, da"s
    $S\circ S=\id$. Allerdings gilt f"ur alle kommutativen und / oder
    kokommutativen \Name{Hopf}-Algebren $S\circ S=\id$.  
    Ferner ist nicht einmal gefordert, da"s $S$ als lineare Abbildung
    ein Inverses $S^{-1}$ besitzen mu"s, allerdings ist dies bei
    endlichdimensionalen Bialgebren mit einer Antipode immer der
    Fall. 
\item Jede endlichdimensionalen \Name{Hopf}-Algebra ist genau dann
    halbeinfach\index{Hopf-Algebra@\Name{Hopf}-Algebra!halbeinfache},
    wenn der Ring $\ring{C}$ die Charakteristik $0$ hat 
    und $S\circ S=\id$ gilt \citep[Seite 33f]{majid:1995a}. 
\end{compactenum}
\end{bemerkungen}

Eine weitere Struktur mit der wir \Name{Hopf}-Algebren ausstatten
wollen ist eine $^\ast$-Involution, das hei"st einen antilinearen
Antiautomorphismus. 

\begin{definition}[Die $^\ast$-Involution]
    \label{Definition:HopfSternAlgebra1}
\index{Involution}
\index{Stern-Involution@$^\ast$-Involution}
\index{Antiautomorphismus!antilinearer|textbf}
Eine {\em \Name{Hopf}-$^\ast$-Algebra} $(H,\mu, \eta, \Delta, \varepsilon,S,
I)$ ist eine \Name{Hopf}-Algebra mit einem antilinearen
Antiautomorphismus $I:H \to H$, der $^\ast$-Involution, so
da"s f"ur alle Elemente in $H$ gilt:  
\begin{compactenum}
\item $I \circ \mu = \mu \circ (I \otimes I)\circ \tau$,
\item $I\circ I=\id$, 
\item $\Delta \circ I = (I\otimes I) \circ \Delta$,
\item $\varepsilon \circ I = I_{\sss \ring{C}} \circ \varepsilon$,
\item $S\circ I \circ S \circ I = \id$.
\end{compactenum}
\end{definition}

Dabei bezeichnen wir mit $I_{\sss \ring{C}}$ die gew"ohnliche
komplexe Konjugation im Ring $\ring{C}$.  

\begin{lemma}[Invertierbarkeit der Antipode]
Sei $(H,\mu, \eta, \Delta, \varepsilon,S,I)$ eine
\Name{Hopf}-$^\ast$-Algebra, so ist die Antipode $S$ invertierbar.
\end{lemma}

\begin{proof}
Die Invertierbarkeit der Antipode ist eine einfache Konsequenz aus
Definition~\ref{Definition:HopfSternAlgebra1} {\it ii.) und v.)}.
\end{proof}

Auch wenn die Formulierung mittels (argumentfreien) Abbildungen sehr
"asthetisch ist, so zeigt sie sich als etwas umst"andlich, wenn man
konkrete Rechnungen angeht. Daher werden wir dazu "ubergehen eine
k"urzere und intuitivere Notation zu nutzen.

\begin{definition}[\Name{Hopf}-$^\ast$-Algebra]
\label{Definition:HopfSternAlgebra}
\index{Hopf-Stern-Algebra@\Name{Hopf}-$^\ast$-Algebra|textbf}
Unter einer {\em \Name{Hopf}-$^\ast$-Algebra $(H,\cdot,\eta, \Delta,
  \varepsilon, S, ^\ast)$} "uber einem Ring $\ring{C}$ versteht man
eine Menge $H$ mit einer assoziativen Multiplikation $\cdot:H \otimes
H \to H$, einer koassoziativen Komultiplikation $\Delta: H
\to H \otimes H$, einer Eins $\eta:\ring{C} \to H$,
einer Koeins $\varepsilon: H \to \ring{C}$, der linearen
Antipode $S: H \to H$ und einer antilinearen Involution
$^\ast: H \to H$, so da"s, die folgenden Bedingungen f"ur
alle $g,h \in H$ erf"ullt sind: 
\begin{compactenum}
\item $\Delta(gh)=\Delta(g)\Delta(h)$, $\Delta(1_{\sss H})=1_{\sss
      H}\otimes 1_{\sss H}$,
\item $\varepsilon(gh)=\varepsilon(g)\varepsilon(h)$,
    $\varepsilon(1_{\sss H})=1_{\sss \ring{C}}$,
\item $S(g_{\sss (1)})g_{\sss (2)}=g_{\sss (1)}S(g_{\sss (2)})=
    1_{\sss H}$, $S(1_{\sss H})=1_{\sss H}$, $S(gh)=S(h)S(g)$,
    $\varepsilon(S(g))=\varepsilon(g)$, 
\item $(gh)^{\ast}=h^{\ast}g^{\ast}$, $(g^{\ast})^{\ast}=g$,
    $\Delta(g^\ast)=\Delta(g)^{\ast \otimes \ast}$,
    $\varepsilon(g^{\ast}) = \cc{\varepsilon(g)}$,
    $S(S(g^{\ast})^\ast)=g$. 
\end{compactenum}
\end{definition}

Dabei haben wir bei der Multiplikation in der Algebra
$(H,\cdot,\eta, \Delta, \varepsilon, S, ^\ast)$ und im Ring $\ring{C}$ auf den Punkt
verzichtet\footnote{Wenn wir von 
  einer \Name{Hopf}-Algebra oder von einer 
  \Name{Hopf}-$^\ast$-Algebra $H$ reden, dann sparen wir uns meist die
  Arbeit alle Strukturen (Multiplikation, Komultiplikation, Eins,
  Koeins, Antipode und Involution) aufzuf"uhren, da dies klar sein sollte.}.    

\begin{definition}[\Name{Hopf}-$^\ast$-Algebra Homomorphismus]
\index{Hopf-Stern-Algebra@\Name{Hopf}-$^\ast$-Algebra!Homomorphismus}
Seien $H$ und $H^{\prime}$ \Name{Hopf}-$^\ast$-Algebren, so ist die
Abbildung $\varphi: H \to H^{\prime}$ ein {\em
  \Name{Hopf}-$^\ast$-Algebra Homomorphismus}, falls $\varphi$ ein
$^\ast$-Homomorphismus f"ur die Algebrastruktur und die
Koalgebrastruktur ist und zus"atzlich $S_{\sss H^{\prime}} \circ
\varphi = \varphi \circ S_{\sss H}$ gilt. 
\end{definition} 

\begin{bemerkung}[\Name{Hopf}-$^\ast$-Algebren "uber formalen Potenzreihen]
\label{Bemerkung:HopfAlgebraFormalePotenzreihen}
\index{Hopf-Stern-Algebra@\Name{Hopf}-$^\ast$-Algebra!formale Potenzreihe}
Sei der Ring $\ringf{C}$ eine formale Potenzreihe in einem formalen
Parameter $\lambda$, so ist die die \Name{Hopf}-$^\ast$-Algebra $H$
ein $\ringf{C}$-Modul mit den $\ringf{C}$-linearen Abbildungen $(\cdot,\eta, \Delta,
\varepsilon, S, ^\ast)$. Dabei ist das algebraische Tensorprodukt aus
Definition~\ref{Definition:HopfSternAlgebra1}
(bzw.Definition~\ref{Definition:HopfSternAlgebra}) durch das in der
$\lambda$-adischen Topologie vervollst"andigte zu ersetzen.
\end{bemerkung}

\begin{bemerkung}[Gradierte \Name{Hopf}-Algebren,
    {\citep{chari.pressley:1994a}}]
\index{Hopf-Stern-Algebra@\Name{Hopf}-$^\ast$-Algebra!gradierte}
Wir k"onnen auch eine gradierte Version von
\Name{Hopf}-$^\ast$-Algebren formulieren. Dann ist 
\begin{align}
    \label{eq:GradierteHopfAlgebra}
    H = \bigoplus_{n=0}^{\infty} H_{n}.
\end{align}
Die so definierte {\em gradierte \Name{Hopf}-$^\ast$-Algebra "uber
  dem Ring $\ring{C}$} erf"ullt die gew"ohnlichen Axiome, allerdings gilt 
\begin{align*}
    \tau(h \otimes g) = (-1)^{mn} g \otimes h, \quad \text{f"ur} \quad
    h \in H_{n}, g \in H_{m}.
\end{align*}
Desweiteren m"ussen die \Name{Hopf}-Abbildungen die Gradierung respektieren,
d.~h.~es gilt $\eta(c) \subset H_{0}$ f"ur $c \in \ring{C}$,  $\varepsilon (H_{n}) = 0$
f"ur $n > 0$, $S(H_{n}) = H_{n}$, $H_{n}^{\ast} = H_{n}$, $\mu (H_{n} \otimes H_{m}) \subset
H_{n+m}$ und $\Delta(H_{n}) \subset \otimes_{p+q=n} H_{p} \otimes H_{q}$.    
\end{bemerkung}

\begin{definition}[\Name{Hopf}-Ideal]
\index{Ideal!Hopf@\Name{Hopf}}
Ein {\Name{Hopf}-Ideal} einer \Name{Hopf}-Algebra $H$ "uber einem
Ring $\ring{C}$ ist ein zweiseitiges Ideal $\ideal{I}$ von $H$, so da"s  
\begin{align}
\Delta(\ideal{I})= \alg{I} \otimes H + H \otimes \ideal{I}, \quad
\varepsilon(\ideal{I})=0, \quad S(\ideal{I}) \subseteq \ideal{I}. 
\end{align}
\end{definition}

\subsubsection{Einige Beispiele}

Nun wollen wir dieses abstrakte Konzept ein wenig mit Leben f"ullen
und einige wichtige Beispiele anbringen. Insbesondere der
Gruppenalgebra $\ring{C}[G]$ sowie der universell Einh"ullenden der 
\Name{Lie}-Algebra $\universell{\LieAlg{g}}$ kommen eine wichtige
Bedeutung zu, da sie die in der Physik relevanten Beispiele f"ur
\Name{Hopf}-Algebren darstellen.

\begin{beispiel}[Tensorprodukt zweier Koalgebren]
Seien $\alg{K}_{\sss 1}$ und $\alg{K}_{\sss 2}$ Koalgebren, so ist
auch $\alg{K}_{\sss 1} \otimes \alg{K}_{\sss 2}$ eine Koalgebra. Es
gilt: $\Delta_{\sss \alg{K}_{\sss 1} \otimes \alg{K}_{\sss 2}}
(k_{\sss 1} \otimes k_{\sss 2}) = \Delta_{\sss 13} (k_{\sss 1})
\Delta_{\sss 24} (k_{\sss 2})$ und $\varepsilon_{\sss \alg{K}_{\sss 1}
  \otimes \alg{K}_{\sss 2}} (k_{\sss 1} \otimes k_{\sss 2})=
\varepsilon_{\sss \alg{K}_{\sss 1}}(k_{\sss 1}) \varepsilon_{\sss
  \alg{K}_{\sss 2}}(k_{\sss 2})$. Sind $\alg{K}_{\sss 1}$ und
$\alg{K}_{\sss 2}$ sogar \Name{Hopf}-Algebren, so ist $S_{\sss  
  \alg{K}_{\sss 1} \otimes \alg{K}_{\sss 2}}= S_{\sss \alg{K}_{\sss
    1}} \otimes S_{\sss \alg{K}_{\sss 2}}$, und $\alg{K}_{\sss 1}
\otimes \alg{K}_{\sss 2}$ ist ebenfalls eine \Name{Hopf}-Algebra. Die
Notation $\Delta_{\sss 13} \Delta_{\sss 24}$ ist das gleiche wie $(\id
\otimes \tau \otimes \id)\circ (\Delta_{\sss \alg{K}_{\sss 1}} \otimes
\Delta_{\sss \alg{K}_{\sss 2}})$.
\end{beispiel}


\begin{beispiel}[Die Gruppenalgebra ${\ring{C}[G]}$]
\label{Beispiel:GruppenalgebraHopf}
\index{Algebra!Gruppenalgebra|textbf} 
Sei $G$ eine Gruppe. Die Gruppenalgebra
$\ring{C}[G]$ von $G$ "uber einem Ring $\ring{C}$ ist ein freier
$\ring{C}$-Modul, und die \Name{Hopf}-Algebrastruktur ist gegeben
durch: $\eta(1_{\sss \ring{C}})=e$, $\varepsilon(g)=1_{\sss
  \ring{C}}$, $\Delta(g)=g \otimes g$, $S(g)=g^{-1}$. Offensichtlich
handelt es sich um eine \Name{Hopf}-$^\ast$-Algebra, wenn man 
$g^{\ast}=g^{-1}$ definiert. Man sieht ferner, da"s $\ring{C}[G]$ immer
kokommutativ ist, jedoch nur dann kommutativ, wenn die Gruppe $G$
kommutativ ist. 
\end{beispiel}

\begin{beispiel}[Die universell einh"ullende Algebra $\universell{\LieAlg{g}}$]
\label{Beispiel:UniverselleEinhuellendeHopf}
\index{Algebra!universell Einh\"ullende|textbf}
Sei $\LieAlg{g}$ eine \Name{Lie}-Algebra "uber $\ring{R}$. Die universell einh"ullende Algebra
$\universellRingC{\LieAlg{g}} = \universellRingR{\LieAlg{g}} \otimes
\ring{C}$ "uber $\ring{C}$ kann man als \Name{Hopf}-$^\ast$-Algebra
auffassen\footnote{Die genaue Definition einer universell
  einh"ullenden Algebra wird sp"ater im Rahmen von universellen Objekten
  Definition~\ref{Definition:UniverselleObjekte} geben.}. Um die Struktur zu
verstehen, reicht es sich die Strukturabbildungen auf den Elementen
von $\LieAlg{g}$ anzusehen. Sei $\xi \in \LieAlg{g}$: 
$\varepsilon(\xi)=0$, $\Delta (\xi) = 1 \otimes \xi + \xi \otimes 1$,
$S(\xi)=-\xi$. Mit $\xi^{\ast}= -\xi$ wird $\universell{\LieAlg{g}}$ zu
einer \Name{Hopf}-$^\ast$-Algebra.  

Weitere Eigenschaften zu einh"ullenden Algebren findet man in
\citep{dixmier:1977a}. 
\end{beispiel}

\begin{beispiel}[$\ring{C}$-wertige Funktionen auf der Gruppe $G$]
\label{Beispiel:FunktionenAufGruppeHopf}
Seien $\alg{F}(G)$ die $\ring{C}$-wertigen Funktionen auf der Gruppe
$G$. Wir definieren die $\ring{C}$-Modul- und Algebrastruktur
punktweise und f"ur $f\in \alg{F}(G)$ sei 
$\varepsilon(f) = f(e)$, $S(f)(g) =f(g^{-1})$, $\Delta(f)(g_1, g_2) =
f(g_1\cdot g_2)$ 
\end{beispiel}

\begin{beispiel}[Nicht (ko)kommutative \Name{Hopf}-Algebra,
    {\citep[Example 1.3.2]{majid:1995a}}]
\label{Beispiel:NichtKoKommutativeHopfAlgebra}
Wir betrachten die Algebra, die durch die $1$ und die drei Elemente
$X,g,g^{-1}$ "uber einem K"orper $\field{K}$ durch die Relationen 
\begin{align*}
gg^{-1}=1=g^{-1}g, \quad Xg=qgX, \quad Xg^{-1}=q^{-1}g^{-1}X  
\end{align*} 
erzeugt wird. Dabei sei $q$ ein invertierbares Element des K"orpers
$\field{K}$. Diese Algebra wird zu einer \Name{Hopf}-Algebra durch 
\begin{gather*}
\Delta(X)= X\otimes 1+g\otimes X, \quad \Delta(g)=g\otimes g, \quad
\Delta(g^{-1})=g^{-1} \otimes g^{-1},\\ \varepsilon(X) =0,\quad
\varepsilon(g)=\varepsilon(g^{-1})=1, \quad S(X) = -g^{-1}X, \quad
S(g)=g^{-1},\quad S(g^{-1})=g.  
\end{gather*}
\end{beispiel}
Diese Algebra ist weder kommutativ noch kokommutativ und nach
Bemerkung \ref{Bemerkung:Antipode} nicht halbeinfach, da $S\circ S(X)
= qX \neq X$.  

\subsection{Wirkungen von \Name{Hopf}-Algebren auf Algebren}
\label{sec:WirkungenHopfAufAlgebren}
\begin{definition}[Wirkung einer \Name{Hopf}-Algebra]
\label{Definition:AxiomeWirkungHopf}
\index{Hopf-Wirkung@\Name{Hopf}-Wirkung|textbf}
Eine Linkswirkung $\ell:H \otimes \alg{A} \to \alg{A}$ einer
\Name{Hopf}-Algebra $(H,\mu,\Delta,\varepsilon,\eta,S)$ auf eine
Algebra $(\alg{A},\mu_{\sss \alg{A}}, 1_{\sss \alg{A}})$ mit Einselement ist
eine lineare Abbildung, und es gilt f"ur alle $g,h \in H$ und $a,b\in
\alg{A}$: 

\begin{compactenum}
\item $\ell \circ (\mu \otimes \id) (g \otimes h \otimes a) = \ell \circ
    (\id \otimes \ell) (g \otimes h \otimes a) $
\item $\ell \circ (\id \otimes \mu_{\sss \alg{A}})(h \otimes a \otimes b)
    =\mu_{\sss \alg{A}} \circ (\ell \otimes \ell) \circ (\id \otimes \tau
    \otimes \id)\circ(\Delta \otimes \id  \otimes \id)(h \otimes a
    \otimes b)$ 
\item $\ell (1_{\sss H} \otimes a) = a$
\item $\ell (h \otimes 1_{\sss \alg{A}}) = \varepsilon (h) 1_{\sss
      \alg{A}}$ 
\end{compactenum}
 \end{definition}

Die Algebra $\alg{A}$ ist damit eine $H$-Linksmodulalgebra , und wir
bezeichnen die $H$-Linksmodulalgebra mit $(H, \alg{A}, \ell)$
\footnote{Analog kann man auch eine $H$-Rechtsmodulalgebra $(H,
  \alg{A}, r)$ definieren, wobei $r: \alg{A} \otimes H \to
  \alg{A}$ eine Rechtswirkung ist. Bei den sp"ater definierten
  Cross-Produktalgebren kann man "ahnlich verfahren und ebenso eine
  Cross-Produktalgebra betrachten, das auf einem $H$-Rechtsmodul basiert und
  die Form $\lcross{H}{\alg{A}}$ hat.}.   

Wie bereits im letzten Kapitel, werden wir auch diese Notation
vereinfachen und dazu die Wirkung mit $\neact$ bezeichnen, d.~h.~$\ell (h
\otimes a)=: h \triangleright a $. Die obigen Bedingungen werden damit
zu 

\begin{compactenum}
\item $(gh) \act a = g \act(h\act a)$,
\item $h\act(ab)=(h_{\sss (1)} \act a)(h_{\sss (2)} \act b)$, 
\item $1_{\sss H} \act a=a$, 
\item $h\act 1_{\sss \alg{A}}= \varepsilon (h) 1_{\sss \alg{A}}$. 
\end{compactenum}

Sind nun sowohl die \Name{Hopf}-Algebra $H$ als auch die Algebra
$\alg{A}$ mit einer $^\ast$-Struktur ausgestattet, so m"ochten wir
auch eine damit vertr"agliche Wirkung $\neact$ haben.  

\begin{definition}
    Sei nun $(H, \alg{A}, \neact)$ eine $H$-Linksmodulalgebra und
    sowohl $H$ als auch $\alg{A}$ seien mit einer $^\ast$-Struktur
    ausgestattet. Man nennt $\neact$ eine {\em $^\ast$-Wirkung}, falls   
\begin{align}
      \label{eq:SternWirkung}
    (h \act a)^\ast = S(h)^{\ast} \act a^{\ast}
\end{align}
\end{definition}

Wir haben nun eine abstrakte Definition f"ur Wirkungen von
\Name{Hopf}-Algebren auf Algebren gegeben. Nun wollen wir  einige
Beispiele f"ur konkrete Wirkungen angeben -- ohne allerdings die
\Name{Hopf}-$^\ast$-Algebra zu spezifizieren.  

\begin{beispiel}[Triviale Wirkung]
\label{Beispiel:TrivialeHWirkung}
\index{Hopf-Wirkung@\Name{Hopf}-Wirkung!triviale}
Sei nun $\alg{A}$ eine $^\ast$-Algebra. Die {\em triviale Wirkung} einer
\Name{Hopf}-$^\ast$-Algebra ist gegeben durch 
\begin{align}
        \label{eq:TrivialeHWirkung}
	h\act a = \varepsilon(h) a,
\end{align}
f"ur alle Elemente $a\in \alg{A}$. Man kann leicht sehen, da"s es
sich hierbei um eine $^\ast$-Wirkung handelt.
\end{beispiel}

\begin{beispiel}[Adjungierte Wirkung]
\label{Beispiel:AdjungierteWirkungHopfAlgebra}
\index{Hopf-Wirkung@\Name{Hopf}-Wirkung!adjungierte}
Sei $H$ eine \Name{Hopf}-$^\ast$-Algebra. Die {\em adjungierte Wirkung}
einer \Name{Hopf}-$^\ast$-Algebra auf sich selbst ist gegeben durch: 
\begin{align}
    \label{eq:AdjungierteHWirkung}
     \ad(g)(h) := g_{\sss (1)} h S(g_{\sss (2)}).
\end{align}
\end{beispiel}

\begin{proof}
Hier werden wir nun zeigen, da"s es sich wirklich um eine Wirkung
handelt. Dazu mu"s man die Gleichungen aus Definition
\ref{Definition:AxiomeWirkungHopf}, sowie die $^\ast$-Bedingung
(Gleichung (\ref{eq:SternWirkung})) nachrechnen. Seien nun $h,g,g',k
\in H$. Es ist trivial zu
zeigen, da"s $\ad(1_{\sss H})(h)=1_{\sss H}hS(1_{\sss H})=h$ und
$\ad(g)(1_{\sss H})=g_{\sss (1)}1_{\sss H}S(g_{\sss
  (2)})=\varepsilon(g)1_{\sss H}$ ist.  

\begin{align*}
   \ad(gg')(h) &=  (gg')_{\sss (1)} h
            S((gg')_{\sss (2)}) \\ &=  g_{\sss (1)}g'_{\sss (1)} h
            S(g_{\sss (2)}g'_{\sss (2)}) \\ &=  g_{\sss (1)}g'_{\sss
              (1)} h 
            S(g'_{\sss (2)})S(g_{\sss (2)}) \\ &=  g_{\sss
              (1)}\ad(g')(h) S(g_{\sss (2)}) \\ &=
            \ad(g) \circ \ad(g')(h),
\intertext{ferner m"ussen wir zeigen, da"s}
	\ad(g)(hk) &= g_{\sss (1)} hk S(g_{\sss (2)})
        \\ &= g_{\sss (1)} h S(g_{\sss (2)}) g_{\sss (3)} k S(g_{\sss
          (4)}) \\ &= \ad(g)(h) \ad(g)(k) 
\intertext{und nun noch die Vertr"aglichkeit mit der $^\ast$-Struktur}
	    \left(\ad(g)(h)\right)^\ast &=
            \left(g_{\sss (1)} h S(g_{\sss (2)}) \right)^{\ast} \\ &=
            S(g_{\sss (2)})^{\ast} h^{\ast} g_{\sss (1)}^{\ast} \\ & =
            S(g_{\sss (2)})^{\ast} h^{\ast} S(S(g_{\sss
              (1)}^{\ast})^{\ast}) \\ &= S(g)^{\ast}_{\sss (1)}
            h^{\ast} S(S(g)^{\ast}_{\sss (2)}) \\ & =
            \ad(S(g)^{\ast})(h^{\ast}). 
\end{align*}
\end{proof}

\begin{beispiel}[Innere Wirkung auf Algebra mittels $^\ast$-Homomorphismus]
\label{Beispiel:WirkungDerHopfAlgebramittelImpulsabbildung}   
\index{Hopf-Wirkung@\Name{Hopf}-Wirkung!Impulsabbildung@durch Impulsabbildung}
Sei $\alg{A}$ eine $^\ast$-Algebra, $H$ eine
\Name{Hopf}-$^\ast$-Algebra $J:H \to \alg{A}$ ein
$^\ast$-Homomorphismus (Impulsabbildung\index{Impulsabbildung!algebraische}),
d.~h.~es gilt f"ur alle $g,h \in H$ 
   \begin{align}
       \label{eq:ImpulsabbildungHomomorphismus}
       J(g)J(h) = J(gh), \quad J(1_{\sss H})=1_{\sss \alg{A}}, \quad
       J(g^{\ast}) = J(g)^{\ast}.
   \end{align}
Dann ist 
\begin{align}
    \label{eq:WirkungviaImpulsabbildungJ}
    h \act a := J(h_{\sss (1)}) a J(S(h_{\sss (2)}))
\end{align}
eine $^\ast$-Wirkung.
\begin{proof} 
Der Beweis dazu erfolgt analog zu
Beispiel \ref{Beispiel:AdjungierteWirkungHopfAlgebra}. 
\end{proof}
 Die Wirkung in Gleichung \eqref{eq:WirkungviaImpulsabbildungJ} ist trivial, wenn die Algebra
$\alg{A}$ kommutativ ist, da  
\begin{align*}
    J(h_{\sss (1)})aJ(S(h_{\sss (2)})) = J(h_{\sss (1)}S(h_{\sss
      (2)}))a = \varepsilon(h) J(1_{\sss H}) a = a.
\end{align*}

Wir wollen das am Beispiel einer \Name{Lie}-Algebra und einer
Gruppe verdeutlichen. Gleichung
\eqref{eq:WirkungviaImpulsabbildungJ} geht in den beiden F"allen
"uber zu 
\begin{align*}
    \Lie[\xi]a = [J_{\xi},a], \quad & \text{mit $\xi \in
      \universell{\LieAlg{g}}$ und  $J_{\xi} \in \alg{A}$}, \\
    \Phi_{g}a = U_{g} a U_{g}^{-1},  & \quad \text{mit $g \in
      \ring{C}[G]$ und  $U_{g} \in \alg{A}$}.
\end{align*}
\end{beispiel}

\subsection{Die Gruppen $\GR{GL}{H,\alg{A}}$, $\GRn{GL}{H,\alg{A}}$,
  $\GR{U}{H,\alg{A}}$ und $\GRn{U}{H,\alg{A}}$}

\label{sec:GruppenGLundU}

Alle Beweise zu diesem Kapitel findet man ausf"uhrlich in
\citep{jansen.waldmann:2005a}. Weitere Konvolutionsprodukte findet man
beispielsweise in \citep{kassel:1995a,majid:1995a}.   

\subsubsection{Definitionen und grundlegende Eigenschaften}

Seien $H$ eine \Name{Hopf}-$^\ast$-Algebra und $\alg{A}$ eine
$^\ast$-Algebra. Wir schauen uns $\Hom[\ring{C}](H,\alg{A})$ mit dem
{\em Konvolutionsprodukt\index{Konvolutionsprodukt}} 
\begin{align}
    \label{eq:Konvolutionsprodukt}
    (\msf{a} \ast \msf{b})(h)=\msf{a}(h_{\sss (1)}) \msf{b}(h_{\sss (2)})
\end{align}
an, wobei $\msf{a},\msf{b} \in \Hom[\ring{C}](H,\alg{A})$ und $h \in H$. 

\begin{definition}[Definition von $\GR{GL}{H,\alg{A}}$ und
    $\GR{U}{H,\alg{A}}$]  
\label{Definition:GLHAundUHA}
    Ein Element $\msf{a}\in \Hom[\ring{C}](H,\alg{A})$ geh"ort zu
    $\GR{GL}{H,\alg{A}}$, falls f"ur alle $g,h \in H$ und $b\in
      \alg{A}$ gilt
      \begin{compactenum}
      \item $\msf{a}(1_{\sss H}) = 1_{\sss \alg{A}}$ (Normierung),
      \item $\msf{a}(gh) = \msf{a}(g_{\sss (1)}) (g_{\sss (2)} \act
          \msf{a}(h))$ (Wirkungsbedingung),\index{Wirkungsbedingung}
      \item $(h_{\sss (1)} \act b) \msf{a}(h_{\sss (2)}) =
          \msf{a}(h_{\sss (1)})(h_{\sss (2)} \act b)$ (Modulbedingung).\index{Modulbedingung}
      \end{compactenum}
Wir bezeichen mit $\GR{U}{H,\alg{A}}$ all die Elemente, die zus"atzlich noch 
      \begin{compactenum}
      \item[\it iv.)] $\msf{a}(h_{\sss (1)})(\msf{a}(S(h_{\sss
            (2)})^{\ast}))^{\ast} = \varepsilon (h) 1_{\sss \alg{A}}$ (Unitarit"atsbedingung).\index{Unitaritaetsbedingung@Unitarit\"atsbedingung}
        \end{compactenum}
erf"ullen.
\end{definition}

\begin{proposition}
    Die Menge $\GR{GL}{H,\alg{A}}$ ist bez"uglich des
    Konvolutionsprodukts eine Gruppe, und $\GR{U}{H,\alg{A}}$ wird zu
    einer Untergruppe der invertierbaren Elemente von
    $\GR{GL}{\Hom[\ring{C}](H,\alg{A}), \ast}$. Das Inverse eines
    Elements $\msf{a}\in \GR{GL}{H,\alg{A}}$ ist gegeben durch
    \begin{align}
        \label{eq:InversesGLHA}
        \msf{a}^{-1}(h)=h_{\sss (2)} \act \msf{a}(S^{-1}(h_{\sss (1)})).
    \end{align}
\end{proposition}

\begin{bemerkung}
    Die Gruppe $\GR{U}{H,\alg{A}}$ ist f"ur jede Wirkung einer
    \Name{Hopf}-$^\ast$-Algebra $H$ auf eine assoziative $^\ast$-Algebra $\alg{A}$ mit
    Einselement $1_{\sss \alg{A}}$ definiert, bei der die Antipodenabbildung
    $S: H \to H$ invertierbar ist. F"ur die Gruppe
    $\GR{U}{H,\alg{A}}$ ben"otigen wir die $^\ast$-Involution auf $H$
    und $\alg{A}$.   
\end{bemerkung}

\begin{definition}[Die \Name{Abel}sche Gruppen $\GR{GL}{\zentrum{\alg{A}}}$
      und $\GR{U}{\zentrum{\alg{A}}}$] 
\label{Definition:AbelscheGruppenGRZAundUZA}
    Man bezeichnet die Menge aller zentralen und invertierbaren Elemente von
    $\alg{A}$ mit $\GR{GL}{\zentrum{\alg{A}}}$, und alle \em
        unit"aren und zentralen Elemente von $\alg{A}$ mit
      $\GR{U}{\zentrum{\alg{A}}}$. 
        \begin{align}
            \GR{GL}{\zentrum{\alg{A}}} & =\left\{a \in \alg{A} \big| ab=ba \;
                  \forall b\in \alg{A},\; \exists a^{-1}\in \alg{A}, \;
                  \text{so da"s} \; a^{-1}a = aa^{-1} = 1_{\sss
                    \alg{A}} \right\} \\ 
              \GR{U}{\zentrum{\alg{A}}} & =\left\{a \in \alg{A} \big|
                  ab=ba \; \forall b\in \alg{A},\; \exists a^{\ast}
                  \in \alg{A}, \;
                  \text{so da"s} \; a^{\ast}a = aa^{\ast} = 1_{\sss
                    \alg{A}} \right\}
        \end{align}
Desweiteren wollen wir die $H$-invarianten Elemente in
$\GR{GL}{\zentrum{\alg{A}}}$ bzw.~$\GR{U}{\zentrum{\alg{A}}}$ mit
$\GR{GL}{\zentrum{\alg{A}}}^{H}$ bzw.~$\GR{U}{\zentrum{\alg{A}}}^{H}$
  bezeichnen.  
\end{definition}

\begin{proposition}[]
    Sei $c \in \GR{GL}{\zentrum{\alg{A}}}$. Wir definieren via 
    \begin{align}
        \hat{c}(h):=c(h \act c^{-1})
    \end{align}
ein Element $\hat{c} \in \GR{GL}{H,\alg{A}}$, und $c \mapsto \hat{c}$
ist ein Gruppenhomomorphismus\index{Gruppenhomomorphismus}, so da"s 

\begin{align}
\label{eq:ExakteSequenz1}
    1\too \GR{GL}{\zentrum{\alg{A}}}^{H} \too \GR{GL}{\zentrum{\alg{A}}}
    \stackrel{\widehat{~}}{\too} \GR{GL}{H,\alg{A}}
\end{align}

eine exakte Sequenz ist. Analog verh"alt es sich f"ur $c \in
\GR{U}{\zentrum{\alg{A}}}$ und $\hat{c} \in \GR{U}{H,\alg{A}}$. Das
hei"st 

\begin{align}
\label{eq:ExakteSequenz2}
    1\too \GR{U}{\zentrum{\alg{A}}}^{H} \too \GR{U}{\zentrum{\alg{A}}}
    \stackrel{\widehat{~}}{\too} \GR{U}{H,\alg{A}} 
\end{align}
ist auch eine exakte Sequenz von Gruppenhomomorphismen. 
\end{proposition}

\begin{definition}[Die Quotientengruppen\index{Quotientengruppe} $\GRn{GL}{H,\alg{A}}$ und
    $\GRn{U}{H,\alg{A}}$] 
\label{Definition:QuotioentengruppenGLnullUndUnull}
     Wir definieren die Quotientengruppen $\GRn{GL}{H,\alg{A}}$ und
    $\GRn{U}{H,\alg{A}}$ indem wir das Bild von
    $\GR{GL}{\zentrum{\alg{A}}}$ bzw.~$\GR{U}{\zentrum{\alg{A}}}$ unter
    \,${}\hat{}$ 
    herausteilen.
    \begin{align}
        \label{eq:QuotientengruppenGLundU}
        \GRn{GL}{H,\alg{A}} & = \GR{GL}{H,\alg{A}} /
        \widehat{\GR{GL}{\zentrum{\alg{A}}}} \\ \GRn{U}{H,\alg{A}} & =
          \GR{U}{H,\alg{A}} / \widehat{\GR{U}{\zentrum{\alg{A}}}} 
    \end{align}
\end{definition}
Die Definition~\ref{Definition:QuotioentengruppenGLnullUndUnull} ist
m"oglich, da $\GR{GL}{\zentrum{A}}$ und $\GR{U}{\zentrum{A}}$ im
Zentrum liegen und daher Normalteiler sind.
Die Sequenzen \eqref{eq:ExakteSequenz1} und \eqref{eq:ExakteSequenz2}
k"onnen wir damit zu den beiden exakten Sequenzen

\begin{align}
 \label{eq:ExakteSequenz12}          
   1\too \GR{GL}{\zentrum{\alg{A}}}^{H} \too \GR{GL}{\zentrum{\alg{A}}}
    & \stackrel{\widehat{~}}{\too} \GR{GL}{H,\alg{A}} \too
    \GRn{GL}{H,\alg{A}} \too 1 \\ \label{eq:ExakteSequenz22}
    1\too \GR{U}{\zentrum{\alg{A}}}^{H} \too \GR{U}{\zentrum{\alg{A}}}
    & \stackrel{\widehat{~}}{\too} \GR{U}{H,\alg{A}} \too
    \GRn{U}{H,\alg{A}} \too 1    
\end{align}
vervollst"andigen.

\begin{bemerkung}
\label{Bemerkung:TrivialesZentrumGleichheitderGruppen}
    Offensichtlich ist $\GR{GL}{\zentrum{\alg{A}}} =
    \GR{GL}{\zentrum{\alg{A}}}^{H}$ und $\GR{U}{\zentrum{\alg{A}}} =
    \GR{U}{\zentrum{\alg{A}}}^{H}$ falls die Algebra $\alg{A}$ nur das
    triviale Zentrum $\zentrum{\alg{A}} = \ring{C}1_{\sss \alg{A}}$
    hat. Desweiteren sind in diesem Fall die Gruppen
    $\GR{GL}{H,\alg{A}}= \GRn{GL}{H,\alg{A}}$ sowie
    $\GR{U}{H,\alg{A}}= \GRn{U}{H,\alg{A}}$.
\end{bemerkung}

\begin{proposition}
\label{Proposition:HAequivarianterHomomorphismusExakteSequenz}
    Seien $\alg{A}$, $\alg{B}$ und $\alg{C}$ assoziative
    $^\ast$-Algebren mit Einselement, und sei $\phi: \alg{A} \to \alg{B}$ ein
    $H$-"aquivarianter, surjektiver Homomorphismus. 
    \begin{compactenum}
    \item F"ur jedes $\msf{a} \in \GR{GL}{H,\alg{A}}$ ist
        $\phi_{\ast} \msf{a} := \phi \circ \msf{a} \in
        \GR{GL}{H,\alg{B}}$.
    \item Die Abbildung $\phi_{\ast}:\GR{GL}{H,\alg{A}} \to
    \GR{GL}{H,\alg{B}}$ ist ein Gruppenhomomorphismus.
    \item Sei $\psi: \alg{B} \to \alg{C}$ ein weiterer
        $H$-"aquivarianter surjektiver Homomorphismus, dann ist
        $(\psi \circ \phi)_{\ast} = \psi_{\ast} \circ \phi_{\ast}$ und
        $(\id_{\sss \alg{A}})_{\ast} = \id_{\sss \GR{GL}{H,\alg{A}}}$.
    \item Der Gruppenhomomorphismus $\phi_{\ast}: \GR{GL}{H,\alg{A}}
        \to \GR{GL}{H,\alg{B}}$ induziert einen Gruppenhomomorphismus
        $\GRn{GL}{H,\alg{A}} \to \GRn{GL}{H,\alg{B}}$,
        den wir auch mit $\phi_{\ast}$ bezeichnen, so da"s das
        folgende Diagramm kommutiert
 \begin{equation}
            \label{eq:GLZAHZAHANullKommutiert}
            \bfig
            \hSquares|rrrrrrr|/>`>``>`>`>`>/%
            [1`\GR{GL}{\zentrum{\alg{A}}}^H%
            `\GR{GL}{\zentrum{\alg{A}}}%
            `1%
            `\GR{GL}{\zentrum{\alg{B}}}^H%
            `\GR{GL}{\zentrum{\alg{B}}}%
            ;```\phi`\phi``]
            \morphism(1680,0)<300,0>[`;]
            \morphism(1680,500)<300,0>[`;]
            \hSquares(2230,0)|rrrrrrr|/>`>`>`>``>`>/%
            [\GR{GL}{H, \alg{A}}%
            `\GRn{GL}{H,\alg{A}}%
            `1%
            `\GR{GL}{H,\alg{B}}%
            `\GRn{GL}{H,\alg{B}}%
            `1.%
            ;``\phi_{\ast}`\phi_{\ast}```]
            \efig
        \end{equation}

    \item Ist $\phi$ zus"atzlich ein $^\ast$-Homomorphismus, so kann
        man bei den Punkten {\it i.)} bis {\it iv.)} \glqq
        $\mathsf{GL}$\grqq{} durch \glqq$\mathsf{U}$\grqq{} ersetzen,
        insbesondere ist das folgende Diagramm kommutativ
\begin{equation}
            \label{eq:UZAHZAHANullKommutiert}
            \bfig
            \hSquares|rrrrrrr|/>`>``>`>`>`>/%
            [1`\GR{U}{\zentrum{\alg{A}}}^H%
            `\GR{U}{\zentrum{\alg{A}}}%
            `1%
            `\GR{U}{\zentrum{\alg{B}}}^H%
            `\GR{U}{\zentrum{\alg{B}}}%
            ;```\phi`\phi``]
            \morphism(1550,0)<300,0>[`;]
            \morphism(1550,500)<300,0>[`;]
            \hSquares(2050,0)|rrrrrrr|/>`>`>`>``>`>/%
            [\GR{U}{H, \alg{A}}%
            `\GRn{U}{H,\alg{A}}%
            `1%
            `\GR{U}{H,\alg{B}}%
            `\GRn{U}{H,\alg{B}}%
            `1.%
            ;``\phi_{\ast}`\phi_{\ast}```]
            \efig
        \end{equation}   
   \end{compactenum}
\end{proposition}

Man kann die Proposition
\ref{Proposition:HAequivarianterHomomorphismusExakteSequenz} auch
anders auffassen, n"amlich als eine Gruppoid-Wirkung auf eine exakte
Sequenz.

\begin{korollar}
Das Gruppoid\index{Gruppoid} der $H$-\"aquivarianten Isomorphismen $\IsoH$ wirkt auf
die exakte Sequenz \ref{eq:ExakteSequenz1} via Isomorphismen. Die
ganze Sequenz von Gruppen mit seiner $\AutH(\alg{A})$-Wirkung ist eine
Invariante der Algebra $\alg{A}$ mit der $H$-Wirkung.
Analog wirkt $\starIsoH$ durch Isomorphismen auf die exakte Sequenz
\ref{eq:ExakteSequenz2}, wobei die gesamte Sequenz mit ihrer
$\starAutH(\alg{A})$-Wirkung eine Invariante von $\alg{A}$ als
$^\ast$-Algebra mit einer $^\ast$-Wirkung von $H$ ist.
    
\end{korollar}

\subsubsection{Der kokommutative Fall}

In diesem Abschnitt wollen wir nun davon ausgehen, da"s wir eine
kokommutative \Name{Hopf}-$^\ast$-Algebra gegeben haben. In diesem
Fall vereinfacht sich das im letzten Abschnitt
beschriebene. In unserer Konvention bilden Homomorphismen immer
Einselemente auf Einselemente ab.

\begin{proposition}
    Sei $H$ eine kokommutative \Name{Hopf}-$^{\ast}$-Algebra und
    $\alg{A}$ eine assoziative $^{\ast}$-Algebra.
    \begin{compactenum}
    \item $H \act \zentrum{\alg{A}} \subseteq \zentrum{\alg{A}}$,
        d.~h.~zentrale Elemente in $\alg{A}$ bleiben durch eine
        \Name{Hopf}-Wirkung zentral.
    \item F"ur $\msf{a} \in \GR{GL}{H,\zentrum{\alg{A}}}$ ist
        $\msf{a}(h)\in \zentrum{\alg{A}}$ f"ur alle $h \in H$, daher gilt
        $\GR{GL}{H,\alg{A}} = \GR{GL}{H,\zentrum{\alg{A}}}$ und
        $\GRn{GL}{H,\alg{A}} = \GRn{GL}{H,\zentrum{\alg{A}}}$. Analoges
        gilt f"ur $a\in \GR{U}{H,\alg{A}}$.
    \item Die Gruppen $\GR{GL}{H,\alg{A}}$, $\GRn{GL}{H,\alg{A}}$,
        $\GR{U}{H,\alg{A}}$ und $\GRn{U}{H,\alg{A}}$ sind \Name{Abel}sch.
    \item Der Raum der Algebra-Homomorphismen $H \to
        \zentrum[H]{\alg{A}}$ ist eine Untergruppe von $\GR{GL}{H,
          \alg{A}}$, und der Raum der $^\ast$-Homomorphismen  $H \to
        \zentrum[H]{\alg{A}}$ ist eine Untergruppe von
        $\GR{U}{H,\alg{A}}$. Das Inverse eines Homomorphismus $\msf{a}$ ist
        gegeben durch $\msf{a}^{-1}(h) = \msf{a}(S(h))=
        \msf{a}(S^{-1}(h))$.    
     \end{compactenum}
\end{proposition}

\begin{proof}
  Der erste Teil der Proposition ist klar. F"ur den zweiten rechnen
  wir nach
  \begin{align*}
      \msf{a}(h)b & = \msf{a}(h_{\sss (1)})\varepsilon(h_{\sss (2)}) b
      = \msf{a}(h_{\sss (1)}) ((h_{\sss (2)}S(h_{\sss (3)}))\act b) =
      (h_{\sss (1)} \act (S(h_{\sss (3)}) \act b)\msf{a}(h_{\sss (2)})
      = \varepsilon(h_{\sss (1)}) b \msf{a}(h_{\sss (2)})\\ & = b \msf{a}(h),
  \end{align*}
dabei nutzen wir sowohl die Kokommutativit"at der \Name{Hopf}-Algebra
wie auch die Modulbedingung aus Definition
\ref{Definition:GLHAundUHA}. Der dritte Teil ist eine Konsequenz dessen
und f"ur den vierten gilt per Definition $\msf{a}: H \to \zentrum[H]{\alg{A}}$,
$\msf{a}(1_{\sss H}) = 1_{\sss \alg{A}}$ und 
\begin{align*}
    \msf{a}(gh) = \msf{a}(g)\msf{a}(h) = \msf{a}(g_{\sss (1)})
    \varepsilon (g_{\sss (2)}) \msf{a}(h) = \msf{a}(g_{\sss
      (1)})(g_{\sss (2)} \act \msf{a}(h)),
\end{align*}
da $\msf{a}(h)$ $H$-invariant ist. Au"serdem ist 
\begin{align*}
    (h_{\sss (1)} \act b) \msf{a}(h_{\sss (2)}) = \msf{a}(h_{\sss
      (2)})(h_{\sss (1)} \act b) = \msf{a}(h_{\sss (1)})(h_{\sss (2)}
    \act b),
\end{align*}
da $\msf{a}(h_{\sss (2)})$ zentral und $H$ kokommutativ ist. F"ur den
Fall, da"s $\msf{a}$ ein $^\ast$-Homomorphismus ist, gilt
\begin{align*}
    \msf{a}(h_{\sss (1)}) (\msf{a}(S(h_{\sss (2)})^{\ast}))^{\ast} =
    \msf{a}(h_{\sss (1)}) \msf{a}(S(h_{\sss (2)})) = \msf{a}(h_{\sss
      (1)}S(h_{\sss (2)})) = \varepsilon(h) 1_{\sss \alg{A}}.
\end{align*}
Es ist leicht zu sehen, da"s f"ur $\msf{a},\msf{b}: H \to
\zentrum[H]{\alg{A}}$ das Konvolutionsprodukt $\msf{a} \ast \msf{b}$
wieder ein Homomorphismus ist, der Werte in $\zentrum[H]{\alg{A}}$
annimmt. Ferner gilt 
\begin{align*}
    \msf{a}^{-1}(h) = \varepsilon(h_{\sss (1)}) \msf{a}(S^{-1}(h_{\sss
      (2)}) = \msf{a}(S^{-1}(h)) = \msf{a}(S(h)).
\end{align*}
Da aufgrund der Kokommutativit"at von $H$ f"ur die Antipode
$S^{2} = \id$ gilt, und $\msf{a}(h)$ invariant ist, ist auch $\msf{a}^{-1}$ ein
Homomorphismus, falls $\msf{a}$ ein $^\ast$-Homomorphismus ist, da die Antipode $S$ mit der
$^\ast$-Involution vertauscht.
\end{proof}

Die Algebrahomomorphismen 
\begin{align}
\chi: H \to \ring{C},
\end{align}
 also die Charaktere\index{Charakter} von $H$, tragen immer zur Gruppe
 $\GR{GL}{H,\alg{A}}$ bei. Ist $\chi$ zus"atzlich ein $^\ast$-Ho\-mo\-mor\-phis\-mus so nennen
 wir ihn einen {\em unit"aren
   Charakter}\index{Charakter!unitaerer@unit\"arer}. Falls das Zentrum
 der Algebra $\alg{A}$ trivial ist, dann bilden die Charaktere von $H$ die
 ganze Gruppe $\GR{GL}{H,\alg{A}}$. Dies wollen wir in der folgenden
 Proposition festhalten.

\begin{proposition}
    Sei $H$ eine kokommutative \Name{Hopf}-Algebra.
    \begin{compactenum}
    \item "Uber die Abbildung $\chi \mapsto
        \msf{a}^{\sss \chi}$ mit $\msf{a}^{\sss \chi}(h) = \chi(h)
        1_{\sss \alg{A}}$ bilden die Charaktere von $H$ eine Untergruppe von
        $\GR{GL}{H, \alg{A}}$, und die unit"aren Charaktere von $H$
        bilden eine Untergruppe von $\GR{U}{H, \alg{A}}$.
\item Ist das Zentrum von $\alg{A}$ trivial, d.~h.~$\zentrum{\alg{A}} =
    \ring{C}1_{\sss \alg{A}}$, dann ist jedes Element von $\GR{GL}{H,
      \alg{A}} = \GRn{GL}{H, \alg{A}}$ ein Charakter und jedes Element
    von $\GR{U}{H, \alg{A}} = \GRn{U}{H, \alg{A}}$  ist ein unit"arer
    Charakter. 
    \end{compactenum}
\end{proposition}

\section{Cross-Produktalgebren}
\label{sec:CrossProdukte}
\index{Algebra!Cross-Produktalgebra|textbf}
\index{Cross-Produktalgebra}

\begin{definition}[Die Cross-Produktalgebra {$\cross{\alg{A}}{H}$}]
\label{Definition:CrossProdukt}
Gegeben seien eine $H$-Linksmodulalgebra $(H, \alg{A}, \neact)$.
Die assoziative Algebra {$\cross{\alg{A}}{H}=(\alg{A} \otimes H, \cdot$)}
sei das Tensorprodukt  $\alg{A} \otimes H$ ausgestattet mit der
Multiplikation $\cdot$, so da"s f"ur alle $g,h \in H$ und $a,b \in
\alg{A}$ gilt 
\begin{align}
(a \otimes g)\cdot (b \otimes h) = a(g_{\sss (1)}\act b) \otimes
g_{\sss (2)} h.  
\end{align} 
Wir bezeichnen eine solche Algebra als {\em Cross-Produktalgebra}
oder als {\em Cross-Produkt}. In der Literatur findet man auch den
Begriff des {\em Smash-Produkts}.  
\end{definition}

\begin{bemerkungen}
~\vspace{-5mm}
\begin{compactenum}
\item Die Assoziativit"at l"a"st sich leicht nachrechnen. Dazu brauchen
wir die Koassoziativit"at der \Name{Hopf}-Algebra, sowie die
Definitionen des $H$-Linksmoduls $(H,\alg{A},\neact)$. 

\item Ist die Algebra $\alg{A}$ mit einem Einselement $1_{\sss \alg{A}}$
  ausgestattet, so ist $1_{\sss \alg{A}}\otimes 1_{H}$ das Einselement
  in $\cross{\alg{A}}{H}$, denn offensichtlich gilt
  \begin{align*}
     (1_{\sss \alg{A}}\otimes 1_{\sss H})\cdot (a \otimes h) & = 1_{\sss
       \alg{A}} (1_{\sss H} \act a)  \otimes 1_{\sss H} h = a \otimes
     h = a (h \act 1_{\sss \alg{A}}) \otimes h 1_{\sss H} =  (a
     \otimes h)\cdot (1_{\sss \alg{A}}\otimes 1_{\sss H}).      
  \end{align*}
\end{compactenum}
\end{bemerkungen}

Im Rahmen von $^\ast$-Algebren und \Name{Hopf}-$^\ast$-Algebren
m"ochten wir auch eine $^\ast$-Struktur auf einer
Cross-Produktalgebra $\cross{\alg{A}}{H}$ definieren. Diese bekommen
wir auf nat"urliche Weise geschenkt, wenn $\alg{A}$ eine
$^\ast$-Algebra und $H$ eine \Name{Hopf}-$^\ast$-Algebra sind. 

\begin{lemma}[$^\ast$-Struktur auf {$\cross{\alg{A}}{H}$}]
\label{Lemma:SternStrukturAufCrossProdukt}
Sei $(H,\alg{A},\neact)$ eine $H$-Linksmodulalgebra und sowohl
$\alg{A}$ als auch $H$ sei mit einer $^\ast$-Struktur
versehen, und $\neact$ sei eine $^\ast$-Wirkung\footnote{Manchmal bezeichnen wir eine derartige Modulalgebra
  auch als kurz als $^\ast$-Modulalgebra. Aus dem Zusammenhang sollte
  klar sein was gemeint ist.}, so wird
$\cross{\alg{A}}{H}$ zu einer $^\ast$-Algebra mittels 
\begin{align}
(a\otimes h)^\ast = h_{\sss (1)}^{\ast} \act a^{\ast} \otimes h_{\sss
  (2)}^{\ast}. 
\end{align}
\end{lemma}

\begin{proof}
Wir rechnen nach, da"s diese Definition allen Strukturen einer
$^\ast$-Algebra gerecht wird. 
\begin{align*}
   (a\otimes h)^{\ast \ast} & = (h_{\sss (1)}^{\ast} \act a^{\ast}
   \otimes h_{\sss (2)}^{\ast})^\ast \\ &=  h_{\sss (2)(1)}^{\ast
     \ast} \act (h_{\sss (1)}^{\ast} \act a^{\ast} )^{\ast} \otimes
   h_{\sss (2)(2)}^{\ast \ast} \\& =  h_{\sss (2)(1)} \act (S(h_{\sss
     (1)}^{\ast})^\ast \act a^{\ast \ast}) \otimes h_{\sss (2)(2)} \\
   & = \underbrace{S^{-1}(h_{\sss (1)}) h_{\sss
       (2)(1)}}_{\varepsilon(h_{\sss (1)})} \act a \otimes h_{\sss
     (2)(2)} \\ & = a \otimes h. \\
\intertext{Desweiteren gilt}
 \left( (a \otimes g)\cdot(b \otimes h) \right)^\ast &= (a(g_{\sss
   (1)}\act b) \otimes g_{\sss (2)} h)^{\ast} \\ &= (g_{\sss (2)(1)}
 h_{\sss (1)})^{\ast} \act (a(g_{\sss (1)}\act b))^{\ast} \otimes
 h_{\sss (2)}^{\ast} g_{\sss (2)(2)}^{\ast} \\ &= (g_{\sss (2)(1)}
 h_{\sss (1)})^{\ast} \act (g_{\sss (1)}\act b)^{\ast} a^{\ast}
 \otimes h_{\sss (2)}^{\ast} g_{\sss (2)(2)}^{\ast} \\ &= h_{\sss
   (1)}^{\ast} g_{\sss (2)(1)}^{\ast} \act (S(g_{\sss (1)})^{\ast}
 \act b^\ast) a^{\ast} \otimes h_{\sss (2)}^{\ast} g_{\sss
   (2)(2)}^{\ast} \\ &= ((h_{\sss (1)}^{\ast} g_{\sss
   (2)(1)}^{\ast})_{\sss (1)} \act (S(g_{\sss (1)})^{\ast} \act
 b^\ast)) ((h_{\sss (1)}^{\ast}g_{\sss (2)(1)}^{\ast})_{\sss (2)} \act
 a^{\ast}) \otimes h_{\sss (2)}^{\ast} g_{\sss (2)(2)}^{\ast} \\ &=
 (h_{\sss (1)(1)}^{\ast} (\underbrace{S(g_{\sss (1)})g_{\sss
     (2)(1)(1)}}_{\varepsilon(g_{\sss (1)})})^{\ast} \act b^\ast)
 (h_{\sss (1)(2)}^{\ast}g_{\sss (2)(1)(2)}^{\ast} \act a^{\ast})
 \otimes h_{\sss (2)}^{\ast} g_{\sss (2)(2)}^{\ast} \\ &= (h_{\sss
   (1)}^{\ast} \act b^{\ast})(h_{\sss (2)}^{\ast}
 \cc{\varepsilon(g_{\sss (1)})}g_{\sss (2)}^{\ast} \act a^{\ast})
 \otimes h_{\sss (3)}^{\ast} g_{\sss (3)}^{\ast} \\ &= (h_{\sss
   (1)}^{\ast} \act b^{\ast} \otimes h_{\sss (2)}^{\ast}) \cdot
 (g_{\sss (1)}^{\ast} \act a^{\ast} \otimes g_{\sss (2)}^{\ast}) \\& =
 (b \otimes h)^{\ast} \cdot (a \otimes g)^{\ast}.  
\intertext{Zuletzt zeigt man noch, da"s das Einselement unter der
$^\ast$-Involution invariant ist}
(1_{\sss \alg{A}}\otimes 1_{\sss H})^{\ast} & = 1_{\sss H} \act 1_{\sss
   \alg{A}} \otimes 1_{\sss H} = 1_{\sss \alg{A}}\otimes 1_{\sss H},
\end{align*}
womit die Behauptung bewiesen w"are.
\end{proof}

\begin{lemma}
Sei $\Phi: \alg{A} \to \alg{B}$ ein $H$-"aquivarianter
$^\ast$-Homomorphismus, dann ist
\begin{align}
    \label{eq:CrossProduktSternHomomorphismus}
    \Phi \otimes \id: \cross{\alg{A}}{H} \to \cross{\alg{B}}{H}
\end{align}
ein $^\ast$-Homomorphismus, der insbesondere einen
Gruppoidhomomorphismus $\cross{\cdot}{H}: \starIsoH \to \starIso$
induziert, so da"s die Identit"aten $\alg{A}$ in $\starIsoH$ auf ihre
Cross-Produktalgebren $\cross{\alg{A}}{H}$ und die Pfeile $\Phi$ auf
$\Phi \otimes \id$ abgebildet werden.
\end{lemma}

Es liegt nahe, die beiden injektiven (abgesehen von Torsionseffekten\index{Torsionseffekt}
durch das Tensorprodukt $\otimes$ "uber $\ring{C}$, von denen wir aber mal absehen
wollen) $^\ast$-Homomorphismen  
\begin{align}
\label{eq:EinbettungAinAxH}
\imath: \alg{A} \ni a \mapsto  a \otimes 1_{\sss H} \in \cross{\alg{A}}{H}, \\
\label{eq:EinbettungHinAxH}
\jmath:  H \ni h \mapsto 1_{\sss \alg{A}} \otimes h \in  \cross{\alg{A}}{H}
\end{align}
zu definieren. Sie liefern eine Einbettung der Algebren $\alg{A}$ und
der \Name{Hopf}-Algebra $H$ in die Cross-Produktalgebra. Offensichtlich sind die Algebren
$\alg{A}\otimes 1_{\sss H}$ und $1_{\sss 
  \alg{A}} \otimes H$ Unteralgebren der Cross-Produktalgebra
$\cross{\alg{A}}{H}$. 
Auf der Cross-Produktalgebra $\cross{\alg{A}}{H}$ existiert auf nat"urliche
Weise eine $^\ast$-Links\-wir\-kung von $H$, die durch 
\begin{align}
g \act (a\otimes h) = (g_{\sss (1)}\act a) \otimes g_{\sss (2)} h
S(g_{\sss (3)}) 
\end{align} 
gegeben ist. Unter Verwendung von Definition
\eqref{eq:EinbettungHinAxH} kann man die Wirkung als eine \glqq
innere\grqq{} schreiben:  

\begin{align}
g \act (a\otimes h)=\jmath(g_{\sss (1)}) (a\otimes h) \jmath(S(g_{\sss
  (2)})). 
\end{align}

Der $^\ast$-Homomorphismus $\imath:\alg{A}
\to \cross{\alg{A}}{H}$ ist $H$-"aquivariant, d.~h.~es gilt
$h \act \imath(a)=\imath(h \act a)$. 

Die Abbildungen \eqref{eq:EinbettungAinAxH} und \eqref{eq:EinbettungHinAxH} 
 k"onnen wir nun die Cross-Produktalgebra
als ein {\em universelles Objekt} interpretieren. 

\begin{definition}[Universelle Objekte]
\label{Definition:UniverselleObjekte}
\index{Objekt!universelles}
Seien $\kat{A}$ und $\kat{B}$ zwei Kategorien und $F:\kat{A}
\to \kat{B}$ ein Funktor von $\kat{A}$ nach $\kat{B}$, ferner sei
$B\in \Obj (\kat{B})$. Eine {\em Universelle von $B$} bez"uglich
des Funktors $F$ ist das Paar $(U,u)$, wobei $U \in \Obj(\kat{A})$ ein Objekt in
$\kat{A}$ und $u: B \to FU$ ein Morphismus von $B$ nach $FU$ ist, so da"s
falls $g:B \to F \tilde{B}$ ein beliebiger Morphismus von $B$ nach $F\tilde{B}$ ist, ein
eindeutiger Morphismus $\tilde{g}: U \to \tilde{B}$ von $U$ nach $\tilde{B} \in \Obj
(\kat{B})$ existiert, und das Diagramm $$\bfig
\ptriangle<600,600>[B`FU`F{\tilde{B}};u`g`\tilde{g}]\efig$$ 
kommutiert. Man nennt $U$ ein {\em universelles $\kat{A}$-Objekt} f"ur
$B$ und $u$ den korrespondierenden {\em universellen Morphismus}.\index{Morphismus!universeller}  
\end{definition}

\begin{proposition}[Cross-Produktalgebra als universelles Objekt]
Sei $\alg{B}$ eine $^\ast$-Algebra mit Einselement mit den beiden
$^\ast$-Homomorphismen $\imath_{\sss \alg{B}}: \alg{A} \to \alg{B}$
und $\jmath_{\sss \alg{B}}: H \to \alg{B}$, so da"s $\imath_{\sss
  \alg{B}}(g\neact a)= \jmath_{\sss \alg{B}}(g_{\sss (1)}) \imath_{\sss
  \alg{B}}(a) \jmath_{\sss \alg{B}}(S(g_{\sss (2)}))$, dann existiert
ein eindeutiger $^\ast$-Ho\-mo\-mor\-phis\-mus $\phi: \cross{\alg{A}}{H} \to
\alg{B}$, so da"s $\imath_{\sss \alg{B}}=\phi \circ \imath$ und
$\jmath_{\sss \alg{B}}=\phi \circ \jmath$. Dies bedeutet $\phi (a
\otimes g) = \imath_{\sss \alg{B}}(a) \jmath_{\sss \alg{B}} (g)$, und
das folgende Diagramm ist kommutativ. 
  
$$\bfig\Vtrianglepair/>`<-`>`>`>/<600,600>[\alg{A}`\cross{\alg{A}}{H}`H`\alg{B};\imath`\jmath`\imath_{\sss
  \alg{B}}`\phi`\jmath_{\sss \alg{B}}]\efig$$ 
\end{proposition}

Ein weiteres Beispiel f"ur ein universelles Objekt erw"ahnten wir
schon zuvor: die uni\-ver\-sell Ein\-h"ul\-len\-de einer
\Name{Lie}-Algebra \index{Algebra!Lie-Algebra@\Name{Lie}-Algebra} \index{Lie-Algebra@\Name{Lie}-Algebra}
(siehe Beispiel \ref{Beispiel:UniverselleEinhuellendeHopf}).  

\begin{dkonstruktion}[Universell einh"ullende Algebren]
\index{Algebra!universell Einh\"ullende}
Sei $\LieAlg{g}$ eine \Name{Lie}-Algebra "uber einem K"orper
$\field{K}$. Wir nennen eine assoziative Algebra mit Einselement
$\alg{U}(\LieAlg{g})$ "uber $\field{K}$ und eine Abbildung $\alpha:
\LieAlg{g} \to \alg{U}(\LieAlg{g})$ eine {\em universell
  einh"ullende Algebra} falls die folgenden Bedingungen erf"ullt
sind: 
\begin{compactenum}
\item Die Abbildung $\alpha: \LieAlg{g} \to
    L(\alg{U}(\LieAlg{g}))$ ist ein \Name{Lie}-Algebrahomomorphismus. 
\item Falls $\alg{A}$ eine assoziative Algebra mit Einselement "uber
    $\field{K}$ und $h: \LieAlg{g} \to L(\alg{A})$ ein
    Algebrahomomorphismus ist, so existiert ein \Name{Lie}-Algebrahomomorphismus
    $k: \alg{U}(\LieAlg{g}) \to\alg{A}$ so da"s $k(1)=1_{\alg{A}}$ und
    $h=k\alpha$.  
\end{compactenum}
\end{dkonstruktion}

\begin{satz}[Existenz und Isomorphie universeller Einh"ullender]
Zu jeder \Name{Lie}-Algebra $\LieAlg{g}$ existiert eine universell
einh"ullende Algebra $\alg{U}(\LieAlg{g})$. Zwei verschiedene
universell einh"ullende Algebren zu $\LieAlg{g}$ sind stets
isomorph. 
\end{satz}

\section{Verb"ande}
\label{sec:Verbaende}

Wir wollen eine kurze Einf"uhrung in die Verbandstheorie geben. Diese
geht ein wenig "uber die von uns gebrauchten Anwendungen hinaus und  
basieren auf \citep{jacobson:1985a} und den unver"offentlichen
Aufzeichnungen von \citet{bordemann:misc}. 

\begin{definition}[Verband]
\index{Verband|textbf}
Man nennt ein Menge $V$ mit den beiden Abbildungen $\wedge, \vee: V 
\times V \to V$ einen {\em Verband}, wenn die folgenden Eigenschaften
f"ur alle $a,b,c \in V$ erf"ullt sind:

\begin{compactenum}
\item $(a \vee b) \vee c = a \vee (b \vee c)$ und $(a \wedge
        b) \wedge c = a \wedge (b \wedge c)$ (Assoziativit"at der Verkn"upfungen)
\item $a \vee b = b\vee a$ und $a \wedge b = b \wedge a$
    (Kommutativit"at der Verkn"upfungen)
\item $a \wedge a = a \vee a =a$ (Idempotenz der Verkn"upfungen)
\item $a \vee (a \wedge b) =a$ und $a \wedge (a \vee b) = a$
    (Vertr"aglichkeit der beiden Verkn"upfungen).
\end{compactenum}
\end{definition}

\begin{lemma}[Halbgeordnete Menge]
\index{Menge!halbgeordnete}
Jeder Verband $(V,\vee, \wedge)$ wird zu einer halbgeordneten Menge
verm"oge 
\begin{align}
a \le b \quad \Longleftrightarrow \quad a \wedge b = a \quad
\Longleftrightarrow \quad a \vee b = b  
\end{align} 
\end{lemma}

\begin{proof}
~\vspace{0mm}
\begin{compactenum}
\item $a \le a \quad a\wedge a = a = a \vee a$.
\item $a \le b$ und $b \le a$ dann ist $a \wedge b = a$ und ebenfalls
    $a \wedge b =b$ woraus offensichtlich $a=b$ folgt. Analog l"auft
    die Argumentation f"ur $\vee$.
\item F"ur $a \le b \le c$ gilt $a \wedge b =a$ und $b \wedge c
    =b$. Damit kann man schreiben: $a \wedge c = (a \wedge b) \wedge c
    = a \wedge (b \wedge c) =a \wedge b = a$ und damit ist $a \le
    c$. Analog zeigt man auch, da"s $a \vee c = c$ ist.
\end{compactenum}
\end{proof}

\begin{definition}[Orthokomplementierter und orthomodularer Verband]
\index{Verband!orthokomplementierter}
\index{Verband!orthomodularer}
Ein Verband $(V, \vee, \wedge)$ hei"st {\em orthokomplementiert}, falls es
eine Abbildung $^\prime: V \to V$ sowie zwei Elemente $0$ und $1$ 
gibt, so da"s f"ur alle $a,b \in V$ gilt
\begin{compactenum}
\item $0 \le a \le 1$,
\item $^\prime$ ist injektiv,
\item aus $a \le b$ folgt $b^\prime \le a^\prime$,
\item $(a^\prime)^\prime = a$,
\item $a \wedge a^\prime = 0$,
\item $a \vee a^\prime =1$.  
\end{compactenum}

Einen orthokomplementierten Verband $(V, \vee, \wedge, ^\prime)$ f"ur
den zus"atzlich noch die Eigenschaft $(a \vee b^\prime) \wedge b = a$
f"ur alle $a \le n$ gilt, nennt man {\em orthomodular}. 
\end{definition}

\begin{lemma}[Eigenschaften orthokomplementierter Verb"ande]
Sei $(V, \vee, \wedge, ^\prime)$ ein orthokomplementierter Verband, so
gilt
\begin{compactenum}
\item $^\prime : V \to V$ ist bijektiv,
\item $(a \wedge b)^\prime = a^\prime \vee b^\prime$,
\item $(a \vee b)^\prime = a^\prime \wedge b^\prime$,
\item $a \le (a \vee b^\prime ) \wedge b$ falls $a \le b$.
\end{compactenum}

Die Aussagen {\it ii.)} und {\it iii.)} werden als {\em \Name{de Morgan}sche Regeln}
bezeichnet. 
\end{lemma}

\begin{proof}
~\vspace{0mm}
\begin{compactenum}
\item Klar, da $(a^\prime)^\prime =a$.
\item Aus $a \wedge b \ge a$ folgt $(a \wedge b)^\prime \le a$
    bzw.~aus $a \wedge b \ge b$ folgt $(a \wedge b)^\prime \le b$. Und
    daraus folgt dann $(a \wedge b)^\prime = a^\prime \vee b^\prime$. 
\item Analoge Argumentation wie bei {\it ii.)}.
\item Wenn $a \le b$ ist $a \le a \vee b^\prime$. 
\end{compactenum}
Daraus folgt die Behauptung. 
\end{proof}

\begin{beispiel}[Menge der abgeschlossenen $^\ast$-Ideale]
\label{Beispiel:MengeAbgeschlosseneSternIdeale}
\index{Ideal!Stern-Ideal@$^\ast$-Ideal!abgeschlossenes}
\index{Stern-Ideal@$^\ast$-Ideal!abgeschlossenes}
Die Menge der $(\alg{D},H)$-abgeschlossenen $^\ast$-Ideale bilden einen Verband, den
wir mit $\verband[\alg{D},H]$ bezeichnen. 
\end{beispiel}

\begin{proof}[Beweis zu Lemma
    \ref{Lemma:DHAbgeschlosseneSternIdealeVerband} und Beispiel
    \ref{Beispiel:MengeAbgeschlosseneSternIdeale}]
    Es gilt folgendes zu zeigen
    \begin{compactenum}
    \item Es gelten die "ublichen Verbandsregeln,
        d.~h.~Assoziativit"at, Kommutativit"at, Idempotenz und die
        Vertr"aglichkeit. 
 
    \end{compactenum}
Im $H$-"aquivarianten Fall mu"s zus"atzlich noch gezeigt werden,
da"s die durch die Verkn"upfungen $\wedge, \vee$ erzeugten
Elemente des Verbands weiterhin
$H$-abgeschlossen\index{Verband!H-abgeschlossen@$H$-abgeschlossener} sind.
\begin{compactenum}
\stepcounter{enumi}      
\item F"ur beliebige $(\alg{D},H)$-abgeschlossene $^\ast$-Ideale
    $\alg{I}$, $\alg{J}$ ist auch $\alg{J} \wedge \alg{I}$
    $(\alg{D},H)$-abgeschlossen. 
\item F"ur beliebige $(\alg{D},H)$-abgeschlossene $^\ast$-Ideale
    $\alg{I}$, $\alg{J}$ ist auch $\alg{J} \vee \alg{I} := \cap_{i}
    \tilde{\alg{J}}_{i}$, das kleinste $H$-abgeschlossene Ideal mit
    $\alg{J} \cup \alg{I} \subseteq \tilde{\alg{J}}$,  
    $(\alg{D},H)$-abgeschlossen.
\end{compactenum}

Teil {\it i.)} ist einfach und im wesentlichen auch klar. Seien nun $\alg{J}=\ker (\pi)$,
$\alg{I}=\ker (\rho)$ und $\alg{K} = \ker (\kappa)$, dann gilt 
\begin{align*}
    (\alg{J} \wedge \alg{I}) \wedge \alg{K} & = (\ker(\pi \oplus
    \rho)) \wedge \ker (\kappa) \\ & = \ker ((\pi \oplus \rho)
    \oplus \kappa) \\ & =  \ker (\pi \oplus (\rho
    \oplus \kappa)) \\ & =  \alg{J} \wedge (\alg{I} \wedge \alg{K}),
\end{align*}
was die Assoziativit"at zeigt. Die Kommutativit"at folgt aus
\begin{align*} 
   \alg{J} \wedge \alg{I} & = \ker (\pi \oplus \rho) = \ker (\rho
   \oplus \pi) = \alg{I} \wedge \alg{J}.
\end{align*}
Komplett analog verh"alt es sich f"ur $\vee$. F"ur die Idempotenz
ist nun
\begin{align*}
 \alg{J} \wedge \alg{J} & = \ker (\pi \oplus \pi) = \ker (\pi) = \alg{J} 
\quad \text{und} \quad
\alg{J} \vee \alg{J} = \cap_{i} \tilde{\alg{J}}_{i} = \alg{J} \cap
\alg{J} = \alg{J}. 
\end{align*}
In einem letzten Schritt m"ussen wir noch zeigen, da"s die beiden
Verkn"upfungen vertr"aglich sind. Dazu rechnen wir nach
\begin{align*}
    \alg{J} \vee (\alg{J} \wedge \alg{I}) &  =  \alg{J} \vee \ker (\pi
    \oplus \rho) = \ker (\pi) \vee \ker (\pi \oplus \rho) = \ker (\pi)
    = \alg{J},
\end{align*}
und analog nat"urlich $ \alg{J} \wedge (\alg{J} \vee \alg{I}) =
\alg{J}$.

F"ur Teil {\it ii.)}  seien nun $\alg{J}=\ker (\pi)$ und
$\alg{I}=\ker(\rho)$ $(\alg{D},H)$-abgeschlossene $^\ast$-Ideale. Dann
ist $\alg{J} \wedge \alg{I} = \ker (\pi \oplus \rho)$. Da die direkte
Summe zweier $H$-"aquivarianter Darstellungen wieder $H$-"aquivariant ist
ist $\alg{J} \wedge \alg{I}$ auch $H$-"aquivariant und
$(\alg{D},H)$-abgeschlossen.

In Teil {\it iii.)} sind $\tilde{\alg{J}}_{i}$ per Definition
Kern einer $H$-"aquivarianten Darstellung und $(\alg{D},H)$-abgeschlossen, so da"s wir schreiben
k"onnen $\tilde{\alg{J}}_{i} = \ker (\tau_{i})$. Damit ist aber
\begin{align*}
    \alg{J} \vee \alg{I} = \cap_{i} \tilde{\alg{J}}_{i} = \ker
    (\oplus_{i} \tau_{i}),
\end{align*}
was nat"urlich wieder $(\alg{D},H)$-abgeschlossen ist. Damit haben
wir gezeigt, da"s $\verband[\alg{D},H](\alg{A})$ einen Verband bildet.

\end{proof}

\begin{beispiel}[Orthomodularer Verband -- Die Potenzmenge $\mathcal{P}(N)$]
\index{Potenzmenge}
Ein Beispiel f"ur einen orthomodularen Verband ist die Potenzmenge
einer Menge $N$, die wir mit  $(V=\mathcal{P}(N), \cup, \cap,
\backslash)$ bezeichnen. Dabei sind die Verkn"upfungen die
bekannten mengentheoretischen $\cup$ und $\cap$, so da"s $a\vee b := a
\cup b$ und $a\wedge b := a \cap b$ f"ur alle  $a,b \in
\mathcal{P}(N)$, desweiteren ist $0=\varnothing$ und $1=N$.   
\end{beispiel}

\section{Projektive Moduln}

Eine Einf"uhrung in die projektiven Moduln findet man in
\citep{jacobson:1989a}. 

\begin{definition}[Projektiver Modul]
\label{Definition:ProjektiverModul}
\index{Modul!projektiver}
Sei $\alg{A}$ eine assoziative Algebra mit Einselement. Der $\alg{A}$-Rechtmodul
$\rmod{E}{A}$ ist genau dann ein {\em endlich erzeugter, projektiver
  $\alg{A}$-Rechtmodul}, falls
\begin{compactenum}
\item ein Erzeugendensystem $(e_{1},e_{2},\cdots,e_{m})$ existiert, so
    da"s f"ur jedes $x \in \rmod{E}{A}$ Koeffizienten $a^{i}\in \alg{A}$ mit $i=1,\cdots,m$
    existieren, so da"s $x=\sum_{i=1}^{m} e_{i} a^{i}$.
\item ein Projektor $P \in M_{n} (\alg{A})$ mit $P=P^{2}$ existiert,
    so da"s $P\alg{A}^{n} \approx \rmod{E}{A}$.
\end{compactenum}
Dabei ist
\begin{align}
    P\alg{A}^{n} = \bild P = \left\{ (b_{i} \in \alg{A}^{n} \big|
        \exists c_{j} \in \alg{A}^{n}; \sum_{j} P_{ij} c_{j} = b_{i} \right\}.
\end{align}
\end{definition}

Dabei sind folgende Punkte zu bemerken. 
\begin{bemerkungen}[Projektiver Modul]
~\vspace{-5mm}
\begin{compactenum}
\item $(e_{1},e_{2},\cdots,e_{m})$ ist keine Basis, sondern nur ein
    Erzeugendensystem. Damit hat $m$ im allgemeinen nicht die
    Bedeutung einer Dimension. Nur in einer Art \glqq minimalem
    Fall\grqq{} k"onnte man $m$ als Dimension bezeichnen.
\item $\alg{A}^{n}$ ist auf nat"urliche Weise ein $\alg{A}$-Rechtsmodul, indem man
komponentenweise multipliziert. Ebenso ist $P\alg{A}^{n}$ ein
$\alg{A}$-Rechtsmodul, da $(P b_{i})\cdot a = P (b_{i} \cdot a)$,
wobei $a\in \alg{A}$ und $b_{i}$ in $\alg{A}^{n}$. 
\item Die Algebra $M_{n}(\alg{A})$ operiert mittels Matrixmultiplikation auf
$\alg{A}^{n}$, d.~h.~f"ur $a_{ij} \in M_{n}(\alg{A})$ und $b_{k} \in
\alg{A}^{n}$ ist $\alg{A}^{n} \ni c_{j}=\sum_{j} a_{ij} b_{j}$. 
\end{compactenum}
\end{bemerkungen}

\begin{lemma}
   Die $\alg{A}$-linearen Endomorphismen von $\alg{A}^{n}$ sind genau
   die $n\times n$-Matrizen "uber $\alg{A}$,
   d.~h.~$\End[\alg{A}]{\alg{A}^{n}} = M_{n}(\alg{A})$.
\end{lemma}
\begin{proof}
    Sei $\phi \in \End[\alg{A}]{\alg{A}^{n}}$, dann gilt 
    \begin{align*}
        \phi(x) = \phi \left( \sum_{i} e_{i}a^{i}\right) = \sum_{i}
        \phi(e_{i}) a^{i} = \sum_{i} \phi_{ij} a^{i}.
    \end{align*}
\end{proof}

Eine Frage, die sich nun stellt ist, wie sieht
$\End[\alg{A}]{P\alg{A}^{n}}$ aus. Die Antwort gibt folgendes Lemma. 
 
\begin{lemma}
Sei $P\alg{A}^{n}$ ein projektiver, endlich erzeugter $\alg{A}$-Rechtsmodul. Die
$\alg{A}$-linearen Endomorphismen von $P\alg{A}^{n}$ k"onnen wie
folgt beschrieben werden. 
    \begin{align}
        \End[\alg{A}]{P\alg{A}^{n}} = PM_{n}(\alg{A}) P :=\left\{ A
            \in M_{n}(\alg{A}) \big| \exists B \in M_{n}(\alg{A}) \,
            \text{mit} \, A=PBP \right\}
    \end{align}
\end{lemma}

\begin{lemma}[Projektive Moduln]
\label{Lemma:ProjektiveModulnDiverseDefinitionen}
Folgende Aussagen sind "aquivalent
\begin{compactenum}
\item Der Modul $\rmod{E}{A}$ ist projektiv.
\item Seien $\rmod{M}{A}$ und $\rmod{\tilde{M}}{A}$
    $\alg{A}$-Rechtsmoduln und $\phi,\psi$ $\alg{A}$-Rechtsmodulmorphismen, $\phi$ sei
    zudem surjektiv. Dann ist das folgende Diagramm kommutativ, d.~h.~es
    existiert ein $T$ mit $\phi \circ T = \psi$.
  $$\bfig
\btriangle(0,0)|lra|/->`->`->/<500,400>[\rmod{E}{A}`\rmod{M}{A}`\rmod{\tilde{M}}{A};T`\psi`\phi]
\morphism(605,0)||/->/<340,0>[`0.;] \efig$$ 

\item Es existiert ein Modul $\rmod{F}{A}$, so da"s $\rmod{E}{A}
    \oplus \rmod{F}{A} \approx \alg{A}^{n}$.
\item Es existieren $x_{i} \in \rmod{E}{A}$ und $\alg{A}$-lineare $f^{i}:\rmod{E}{A} \to
    \alg{A}$, so da"s $\rmod{E}{A} \ni x = \sum_{i} x_{i}f^{i}(x)$
    darstellbar ist f"ur alle $x \in \rmod{E}{A}$.  
\end{compactenum}
    
\end{lemma}

\begin{bemerkung}
Bedingung {\it iii.)} von Lemma
    \ref{Lemma:ProjektiveModulnDiverseDefinitionen} ist die
    algebraische Formulierung des Satzes von \Name{Serre-Swan}.
\index{Satz!Serre-Swan@von \Name{Serre-Swan}} 
\end{bemerkung}


%% file: geometrie.tex
\label{chapter:GeometrischeGrundlagen}

\section{Faserb"undel}
\label{sec:Faserbuendel}

Wir wollen eine kurze Definition der von uns gebrauchten Strukturen im
Rahmen der B"un\-del\-geo\-me\-trie geben. Einige der Definitionen
entstammen \citep{nakahara:1990a,kobayashi.nomizu:1963a}. 

\subsection{Grundlagen}

\begin{definition}[Faserb"undel]
\label{Definition:Faserbuendel}
\index{Faserbuendel@Faserb\"undel}
Ein (differenzierbares) Faserb"undel $(E,\pi,M,F,G)$ (oder k"urzer
$\bundle{E}{\pi}{M}$) besteht aus den folgenden Elementen: 
\begin{compactenum}
\item Eine differenzierbare Mannigfaltigkeit $E$, der {\em Totalraum}.
\item Eine differenzierbare Mannigfaltigkeit $M$, der {\em
      Basisraum} (oder der {\em Basismannigfaltigkeit}).
\item Eine differenzierbare Mannigfaltigkeit $F$, der {\em Faser} (oder
    der {\em typischen Faser}). 
\item Eine Surjektion $\pi: E \to M$, die {\em Projektion}. Das Urbild
    $\pi^{-1}\equiv F_{p} \cong F$ ist die Faser am Punkte $p$. 
\item Eine Lie-Gruppe $G$, die {\em Strukturgruppe}. Sie wirkt von links auf
    die Faser $F$. 
\item Eine Menge offener "Uberdeckungen $\{U_{i}\}$ von $M$ mit
    einem Diffeomorphismus $\phi_{i}: U_{i} \times F \to
    \pi^{-1}(U_{i})$, so da"s $\pi \circ \phi_{i} (p,f) =
    p$. Man nennt die Abbildungen $\phi_{i}$ eine {\em lokale
    Trivialisierung}. 
\item $\phi_{i}(p,\cdot) =: \phi_{i,p}:F \to F_{p}$ ist
    ein Diffeomorphismus. Auf $U_{i} \cap U_{j} \neq
    \varnothing$ f"uhrt man $t_{ij} := \phi^{-1}_{i,p} \phi
    _{j,p} : F \to F$ ein. Dann sind $\phi _{i}$ und $\phi
    _{j}$ verkn"upft "uber die glatte Abbildung $t_{ij} :
    U_{i} \cap U_{j} \to G$, die Werte in der Lie-Gruppe $G$
    annimmt, und  
\begin{align}
\phi_{j}(p,f)=\phi_{i}(p,t_{ij}(p)f).
\end{align}
Man nennt $\{t_{ij}\}$  die {\em "Ubergangsfunktionen}.
\end{compactenum}  
\end{definition}

\begin{bemerkung}
\label{Bemerkung:Faserbuendel}
\index{Koordinatenbuendel@Koordinatenb\"undel}
Ein Faserb"undel ist unabh"angig von der speziellen Wahl
der "Uberdeckung. Daher bezeichnet man auch $(E,\pi,M,F,G,\{U_{\sss
  i}\}, \{ \phi_{i} \})$ als {\em Koordinatenb"undel}, und die
"Aquivalenzklasse der Koordinatenb"undel als das {\em Faserb"undel}.  
\end{bemerkung}

\begin{definition}["Aquivalenz von Faserb"undeln]
\label{Definition:AequivalenzVonBuendeln}
\index{Aequivalenz@\"Aquivalenz!Faserbuendel@von Faserb\"undeln}
Zwei B"undel $\bundle{E}{\pi}{M}$ und $\bundle{E'}{\pi'}{M}$ sind
genau dann "aquivalent, wenn es einen
B"un\-del\-dif\-feo\-mor\-phis\-mus $\cc{f}:E \to E'$ gibt, so
da"s $f:M \to M$ die Identit"at ist und das folgende Diagramm
kommutiert: 
 $$\bfig \Square[E`E'`M`M;\cc{f}`\pi`\pi'`\id_{\sss M}]\efig$$
\end{definition}

\begin{definition}[Triviales Faserb"undel]
\label{Definition:TrivialesFaserbuendel}
\index{Faserbuendel@Faserb\"undel!triviales}
Ein Faserb"undel $\bundle{E}{\pi}{M}$ nennt man {\em trivial}, falls
es sich als direktes Produkt $E=M\times F$ aus der
Basismannigfaltigkeit $M$ und der Faser $F$ schreiben l"a"st. 
\end{definition}

\begin{definition}[Schnitte]
\label{Definition:Schnitte}
\index{Schnitt}
Gegeben ein Faserb"undel $\bundle{E}{\pi}{M}$. Ein {\em Schnitt} $s:
M \to E$ ist eine Abbildung von der Basismannigfaltigkeit in den
Totalraum, so da"s $\pi\circ s = \id_ {\sss M}$. Die Menge aller
$k$-fach differenzierbaren (bzw.~glatten) Schnitte von
$\bundle{E}{\pi}{M}$ bezeichnet man mit $\Gamma^{k}(M,E)$
(bzw.~$\Gamma^{\infty}(M,E)$) oder einfacher mit $\Gamma^{k}(E)$
(bzw.~$\Gamma^{\infty}(E)$).   
\end{definition}

\begin{beispiele}[Vektorfelder und Einsformen auf $M$]
\label{Beispiele:SchnitteVektorfelderUndEinsformen}
\index{Vektorfeld} \index{Einsform}
Die einfachsten Bespiele f"ur Schnitte sind die Vektorfelder auf einer
Mannigfaltigkeit: $\mathfrak{X}(M)=\Gamma^{\infty}(M,TM)$ und die
Einsformen: $\Omega^{1}(M)=\Gamma^{\infty}(M,T^{\ast}M)$. 
\end{beispiele}

\begin{definition}[Vektorb"undel und Geradenb"undel]
\label{Defintion:Vektorbuendel}
\index{Geradenbuendel@Geradenb\"undel}
Man bezeichnet ein Faserb"undel $\bundle{E}{\pi}{M}$ als {\em 
  Vektorb"undel}, falls die typische Faser ein Vektorraum $V$
ist. Ist die Faser ein eindimensionaler Vektorraum, so spricht man
auch von einem {\em Geradenb"undel}.  
\end{definition}

\subsection{Zusammenhang und Kr"ummung}
\label{sec:ZusammenhaengeUndKruemmung}

\begin{definition}[Zusammenhang auf Vektorb"undel $\bundle{E}{\pi}{M}$]
    \label{Definition:Zusammenhang}
\index{Zusammenhang!linearer|textbf}
    Gegeben sei ein Vektorb"undel $\bundle{E}{\pi}{M}$. Man definiert 
    einen {\em 
      linearen Zusammenhang} (oder eine {\em kovariante Ableitung}) 
     \begin{align}
      \nabla^{E} : \schnitt{TM} \times \schnitt{E} \to \schnitt{E}  
      \end{align}
      "uber die folgenden Eigenschaften: 
      \begin{compactenum}
    \item $\lconnE[fX+gY]s = f \lconnE[X]s + g \lconnE[Y] s$,
    \item $\lconnE[X](fs) = f \lconnE[X] s + X(f) s$,
    \end{compactenum}
f"ur $X,Y \in \schnitt{TM}$, $f,g \in C^\infty(M)$ und $s \in \schnitt{E}$.
\end{definition}

Der Zusammenhang kann aufgrund von Bedingung {\it ii.)} kein
Tensorfeld sein. Allerdings kann man mit Hilfe des Zusammenhangs
interessante Tensorfelder konstruieren.  

\begin{definition}[Kr"ummung eines Zusammenhangs $\nabla^{E}$ auf
    $\bundle{E}{\pi}{M}$] 
\label{Definition:KruemmungEinesZusammenhangs}
\index{Kruemmung@Kr\"ummung!linearer Zusammenhang@eines linearen Zusammenhangs}
Gegeben sei ein Vektorb"undel $\bundle{E}{\pi}{M}$ mit einem
Zusammenhang $\nabla^{E}$. Der {\em Kr"ummungstensor $R \in
  \schnitt{\Lambda^{2}T^{\ast}M \otimes \End{E}}$} des Zusammenhangs
definiert man via 

\begin{align}
R^{E}(X,Y)s = \lconnE[X] \lconnE[Y]s - \lconnE[Y]\lconnE[X]s - \lconnE[{[X,Y]}]s,
\end{align}
 f"ur $X,Y \in \schnitt{TM}$ und $s \in \schnitt{E}$.
\end{definition}

An dieser Stelle m"u"ste man beispielsweise zeigen, da"s die
Kr"ummung wirklich ein Tensorfeld ist, allerdings
wollen wir statt dessen auf geeignete Literatur verweisen,
beispielsweise \citep[Chapter III.5]{kobayashi.nomizu:1963a}.

\begin{lemma}[Zusammenhang auf dem Endomorphismenb"undel
    $\bundle{\End{E}}{\pi'}{E}$]
\label{Lemma:ZusammenhangEndomorphismen}
\index{Zusammenhang!Endomorphismenbuendel@auf Endomorphismenb\"undel}
Sei $\bundle{E}{\pi}{M}$ ein Vektorb"undel mit einem Zusammenhang
$\lconnE$. Auf den Endomorphismenb"undel
$\bundle{\End{E}}{\pi'}{E}$ wird ein Zusammenhang $\lconnEnd{E} :
\schnitt{TM} \times \schnitt{\End{E}} \to \schnitt{\End{E}}$ durch 

\begin{align}
    \label{eq:ZusammenhangEndomorphismen}
\left( \lconnEnd[X]{E}A \right)(s):=\left[\lconnE[X],A \right](s)
\end{align}
f"ur alle $X\in \schnitt{TM}$, $A\in \End{E}$ und $s\in \schnitt{E}$
induziert.  
\end{lemma}

\begin{proof}
Wie rechnen nach, da"s $\lconnEnd{E}$ die Eigenschaften eines
Zusammenhangs hat und funktionenlinear in den Schnitten $s\in
\schnitt{E}$ ist. Die Eigenschaft {\it i.)} in Definition
\ref{Definition:Zusammenhang} ist offensichtlich. Bleibt zu zeigen 
    \begin{align*}
       \left( \lconnEnd[X]{E} (fA)\right)(s) & = \left[ \lconnE[X],fA \right](s) \\ &
        =\lconnE[X](fAs) - fA \lconnE[X]s \\ &=f \lconnE[X](As) + As
        X(f) - fA \lconnE[X]s \\ &= f \left( \lconnE[X] (As)- A
            \lconnE[X] s \right)
        + As X(f) \\ &=f \left[ \lconnE[X],A \right](s) + As X(f) \\ &=f
        (\lconnEnd[X]{E}A)(s) + X(f)A (s), 
\intertext{sowie die Linearit"at in den Schnitten}
      \left( \lconnEnd[X]{E} A \right)(fs) & = \left[ \lconnE[X],A \right] (fs) \\ &=
      \lconnE[X] A (fs) - A \lconnE[X] (fs) \\ &=  f \lconnE[X] (As) -
      (As) X(f) - (Af) \lconnE[X]s + (As) X(f) \\ &= f \left( \lconnE[X] A s - A
      \lconnE[X] s \right) \\ & = f \left[ \lconnE[X],A \right] (s) \\ &= f
      \left( \lconnEnd[X]{E} A\right) (s). 
    \end{align*}
\end{proof}

\section{Symplektische Geometrie}
\label{sec:SymplektischeGeometrie}
\subsection{Allgemeine Definitionen}
\index{Geometrie!symplektische|(}
In diesem Kapitel stellen wir einige 
Grundlagen zur {\em symplektischen Geometrie} zusammen. Die
symplektische Geometrie ist in der Physik von gro"ser Bedeutung, da
eine gro"se Klasse der Phasenr"aume von Teilchen
Kotangentialb"undel und damit symplektische
Mannigfaltigkeiten sind. Die Klasse der Kotangentialb"undel ist im
allgemeinen jedoch nicht ausreichend, da durch Phasenraumreduktion
durchaus symplektische Mannigfaltigkeiten entstehen k"onnen, die keine
Kotangentialb"undel sind. Eine Verallgemeinerung liegt daher auf der Hand.
Andererseits spielen auch \Name{K"ahler}-Mannigfaltigkeiten
$(M,\omega, I, g)$ im Rahmen der Quantisierung (z.~B.~beim
\Name{Wick}-Produkt) eine zentrale Rolle. Auch diese sind
symplektische Mannigfaltigkeiten mit weiteren Strukturen, siehe
z.~B.~\citep{wells:1980a}.      

\begin{definition}[Symplektische Mannigfaltigkeit $M$]
\label{Definition:SymplektischeMannigfaltigkeit}
\index{Mannigfaltigkeit!symplektische|textbf}
Eine {\em symplektische Mannigfaltigkeit $(M,\omega)$} ist eine
Mannigfaltigkeit mit einer punktweise nichtausgearteten geschlossenen
Zweiform $\omega \in \schnitt{\Lambda^{2}T^{\ast}M}$.  
\end{definition}

\begin{bemerkung}
Die Dimension einer symplektischen Mannigfaltigkeit ist immer gerade, d.~h.~$\dim M=2n$.  
\end{bemerkung}

\begin{dlemma}[Symplektomorphismen] 
\index{Symplektomorphismus}
\index{Gruppe!Symplektomorphismen@der Symplektomorphismen}
Ein Diffeomorphismus $\phi:M \to N$ zwischen zwei symplektischen
Mannigfaltigkeiten $(M,\omega)$ und $(N,\omega')$ hei"st {\em
  symplektisch} oder {\em ein Symplektomorphismus}, falls $\phi^{\ast}
\omega' = \omega$. Die symplektischen Diffeomorphismen $M \to M$
bilden eine Gruppe, die {\em Symplektomorphismengruppe $\Sympl (M)$}.
\end{dlemma}

\begin{beispiel}
Ein erstes Beispiel f"ur eine symplektische Mannigfaltigkeit ist der
$\field{R}^{2n}$ mit der kanonischen symplektischen Form $\omega{\sss
  0} = \sum_{i=1}^{n} \de q^{i} \wedge \de p_{i}$. Die symplektische Form ist
offensichtlich geschlossen, d.~h.~$\de \omega_{\sss 0} =0$ und in einer
globalen Karte f"ur den $\field{R}^{2n}$ sogar konstant.  
\end{beispiel}

Das Beispiel des $\field{R}^{2n}$ ist insofern wichtig, als da"s
jede symplektische Mannigfaltigkeit {\em lokal} wie ein
$(\field{R}^{2n}, \omega_{\sss 0})$ aussieht, was die Aussage des \Name{Darboux}-Theorems ist.

\begin{satz}[\Name{Darboux}-Theorem]
    \label{Satz:DarbouxTheorem}
\index{Satz!Darboux@von \Name{Darboux}}
\index{Darboux-Theorem@\Name{Darboux}-Theorem}
Sei $(M,\omega)$ eine symplektische Mannigfaltigkeit der Dimension
$\dim M= 2n$, und $p\in M$. Dann existiert eine offene Umgebung $U
\subseteq M$ von $p$, ein Diffeomorphismus $\phi: U \to V$ sowie eine
offene Umgebung $V \subseteq \field{R}^{2n}$, so da"s   
\begin{align}
\phi:(U,\omega|_{\sss U}) \stackrel{\cong}{\to} (V,\omega_{\sss
  0}|_{\sss V}) 
\end{align}
ein Symplektomorphismus ist.
 \end{satz}

\subsection{Symplektische Zusammenh"ange}
\label{sec:SymplektischeZusammenhaenge}

F"ur die \Name{Fedosov}-Konstruktion ist der Zusammenhang von gro"ser
Bedeutung. Daher werden wir kurz die wesentlichen (ben"otigten)
Eigenschaften von Zusammenh"angen auff"uhren. Im weiteren sind wir an Zusammenh"angen auf $TM$
interessiert. Dies hei"st insbesondere, da"s $E=TM$ und damit
wird der Zusammenhang eine Abbildung 

\begin{align*}
\nabla: \schnitt{TM} \times \schnitt{TM} \to \schnitt{TM}.
\end{align*}

Diese spezielle Art von Zusammenhang erlaubt uns eine neue Gr"o"se, die {\em
Torsion} des Zusammenhangs, zu definieren.  
    
\begin{dlemma}[Torsion eines Zusammenhangs $\nabla$ auf $M$]
\label{Definition:TorsionKruemmungZusammenhang}
\index{Torsion!Zusammenhang@eines Zusammenhangs}
Gegeben eine Mannigfaltigkeit $M$ mit einem beliebigen
Zusammenhang\footnote{Wenn wir von \glqq einer Mannigfaltigkeit $M$
  mit Zusammenhang\grqq{} sprechen, dann ist ein Zusammenhang auf dem
  Tangentialb"undel $TM$ gemeint.} $(M,\nabla)$. Man definiert die {\em Torsion}  
\begin{align}
T_{\nabla} (X,Y)=\lconn[X] Y - \lconn[Y] X -[X,Y] 
\end{align}
f"ur alle $X,Y,Z \in \schnitt{TM}$.
\end{dlemma}

\begin{lemma}[Torsionsfreier Zusammenhang]
\label{Lemma:TorsionsfreierZusammenhang}
Auf jeder Mannigfaltigkeit mit einem Zusammenhang $(M,\tilde{\nabla})$
existiert ein torsionsfreier Zusammenhang $\nabla$,
d.~h.~$T_{\nabla}(X,Y)=0$ f"ur alle $X,Y \in \schnitt{TM}$.       
\end{lemma}

\begin{proof}
    Sei $\tilde{\nabla}$ ein beliebiger Zusammenhang, so ist $\nabla:=
    \tilde{\nabla}-\frac{1}{2}T_{\tilde{\nabla}}$ torsionsfrei, da 
    \begin{equation}
        \begin{aligned}
            T_{\nabla}(X,Y) & =  \tilde{\nabla}_{\!\!\sss X} Y - \frac{1}{2}
            T_{\sss \tilde\nabla} (X,Y) - \tilde\nabla_{\!\!\sss Y} X +
            \frac{1}{2} T_{\sss \tilde\nabla} (Y,X) - [X,Y] \\ &=
            T_{\sss \tilde\nabla} (X,Y) - T_{\sss \tilde\nabla} (X,Y)
            \\&  =0  
        \end{aligned}
    \end{equation}
\end{proof}

\begin{definition}[Symplektischer Zusammenhang]
\label{Definition:SymplektischerZusammenhang}
\index{Zusammenhang!symplektischer|textbf}
Gegeben eine symplektische Mannigfaltigkeit $(M,\omega)$. Man nennt
einen Zusammenhang 
\begin{align*}
\nabla:\schnitt{TM} \times \schnitt {TM} \to \schnitt{TM}
\end{align*}
{\em symplektisch} falls 
\begin{align}
\nabla \omega = 0.
\end{align} 
Ausgeschrieben bedeutet dies, da"s
\begin{align*}
X(\omega(Y,Z)) = \omega (\lconn[X]Y,Z) + \omega (Y,\lconn[X]Z)
\end{align*}
f"ur alle $X,Y,Z \in \schnitt{TM}$.
\end{definition}

\begin{lemma}[Konstruktion symplektischer Zusammenhang]
\label{Lemma:HessTrick}
Sei $(M,\omega,\tilde{\nabla})$ eine symplektische Mannigfaltigkeit
mit einem torsionsfreien Zusammenhang $\tilde{\nabla}$. Der Ausdruck 
\begin{align}
\omega(\lconn[X]Y,Z)=\omega(\tilde{\nabla}_{\sss X}Y,Z) + \frac{1}{3}
\left( \tilde{\nabla}_{\sss X} \omega \right)(Y,Z) + \frac{1}{3}
\left( \tilde{\nabla}_{\sss Y} \omega \right)(X,Z) 
\end{align}
mit $X,Y,Z \in \schnitt{TM}$ definiert einen symplektischen
Zusammenhang $\nabla$.  
\end{lemma}

\begin{proof}
Der Beweis ist eine elementare Rechnung, der die Torsionsfreiheit des
Zusammenhangs $\tilde{\nabla}$, sowie die Geschlossenheit der Zweiform
$\omega$ nutzt.    
\end{proof}

\begin{dlemma}[Symplektische Kr"ummung eines Zusammenhangs]
\label{Lemma:SymplektischeKruemmung}
\index{Kruemmung@Kr\"ummung!symplektische}
Sei auf einer symplektische Mannigfaltigkeit $(M,\omega)$ ein
symplektischer Zusammenhang $\nabla$ gegeben, so ist die
{\em symplektische Kr"ummung} $R$ des Zusammenhangs die mit der 
symplektischen Form kontrahierte Kr"ummung $\hat{R}$, d.~h. 

\begin{align*}
\hat{R}(X,Y)& =\lconn[X] \lconn[Y] - \lconn[Y] \lconn[X] -
\lconn[{[X,Y]}]  \in \schnitt{\Lambda^2 T^{\ast}M \otimes \End{TM}}, \\
R(Z,U,X,Y) & =\omega(Z,\hat{R}(X,Y)U) \in \schnitt{S^{2}T^{\ast}M
  \otimes \Lambda^{2}T^{\ast}M}, 
\end{align*} 

f"ur $X,Y,Z,U \in \schnitt{TM}$.

\end{dlemma}

\begin{proof}
Der Zusammenhang $\nabla$ ist symplektisch, daher ist $R$ in den ersten
beiden Argumenten symmetrisch. Die Schiefsymmetrie in den zweiten
beiden Argumenten bleibt erhalten.     
\end{proof}

\section{Symmetrien -- $G$- und $\LieAlg{g}$-Invarianz}
\label{sec:GundgInvarianz}
\index{Lie-Gruppe@\Name{Lie}-Gruppe}
\index{Lie-Algebra@\Name{Lie}-Algebra}
\index{Algebra!Lie-Algebra@\Name{Lie}-Algebra}

In diesem Kapitel wollen wir einen kurzen "Uberblick "uber die
Wirkung von \Name{Lie}-Gruppen und \Name{Lie}-Algebren
auf Mannigfaltigkeiten geben.

\begin{definition}[$G$-Mannigfaltigkeit]
    \label{Definition:GMannigfaltigkeit}
\index{G-Mannigfaltigkeit@$G$-Mannigfaltigkeit}
\index{Mannigfaltigkeit!G-Mannigfaltigkeit@$G$-Mannigfaltigkeit}    
Eine Mannigfaltigkeit $M$ mit einer Gruppenwirkung $G$ nennt man
    eine {\em $G$-Mannigfaltigkeit}, die man mit $(M,G)$ bezeichnet. Man
    definiert eine Linkswirkung $\phi$ und eine Rechtswirkung $\psi$
    durch 
\begin{align*}
    \phi:  G  \times M & \to M, \quad (g,m)  \mapsto \phi(g,m)=:
    \phi_g(m),  \\ 
    \psi: M \times G & \to M , \quad  (m,g) \mapsto \psi(m,g) =: \psi_g(m).
\end{align*}
sa da"s f"ur alle $g,g' \in G$ gilt:
\begin{align*}
    \phi_{g'} \circ \phi_g =\phi_{g'.g} \quad \text{sowie} \quad
    \phi_{e}=\id,
\end{align*}
beziehungsweise f"ur die Rechtswirkung $\psi$
\begin{align*}
     \psi_{g'} \circ \psi_g =\psi_{g.g'} \quad \text{sowie} \quad
    \psi_{e}=\id.
\end{align*} 
\end{definition}

\begin{bemerkungen}[Notation]
 \label{Bemerkungen:GMannigfaltigkeiten}
 Um die Notation zu vereinfachen werden wir in Zukunft die
    Linkswirkung verk"urzt als $g.m$ schreiben. Die Wirkung mit $g.$ abzuk"urzen
    hat weiter den Vorteil, da"s wir kategoriell unterschiedliche
    Wirkungen (auf Punkte, Funktionen,Vektorfelder, Formen,...)
    vereinheitlicht schreiben k"onnen. Um dies machen zu k"onnen,
    m"ussen wir uns allerdings von der Konsistenz der Wirkungen
    "uberzeugen. Daher werden wir im folgenden einige
    Kompatibilit"aten zeigen.  
\end{bemerkungen}

\begin{definition}[Symplektomorphismus]
    \label{Definition:symplektomorph}
\index{Symplektomorphismus|textbf}
    Gegeben sei eine symplektische $G$-Mannigfaltigkeit
    $(M,\omega, G)$. Man nennt $\phi_{g} :M \to M$ einen {\em
    Symplektomorphismus}, falls die symplektische Form f"ur alle
    $g\in G$ invariant unter der Gruppenwirkung ist, so da"s
    \begin{align}
     \label{eq:Symplektomorphismus}
        {\phi_{g^{-1}}^{\ast}}(\omega) = \omega \quad \mbox{oder kurz} \quad
        g.\omega=\omega \qquad \forall g\in G. 
    \end{align}
\end{definition}

Der n"achste Schritt besteht nun darin die $G$-Wirkung auf ein
Vektorb"undel "uber einer Mannigfaltigkeit zu definieren. Diese
spielen insbesondere bei der "aquivarianten \Name{Morita}-Theorie eine
wichtige Rolle.

\begin{definition}[$G$-Vektorb"undel]
    \label{Definition:GVektorbuendel}
\index{Vektorbuendel@Vektorb\"undel!G-Vektorbuendel@$G$-Vektorb\"undel}
$E\stackrel{\pi}{\to}M$ sei ein Vektorb"undel "uber der
$G$-Mannigfaltigkeit $(M,G)$. Man nennt $E\stackrel{\pi}{\to}M$ ein
{\em $G$ -Vektorb"undel} falls es eine Gruppenwirkung $\Phi$
auf dem Vektorb"undel gibt, so da"s
\begin{enumerate}
\item f"ur alle $v \in E$ 
\begin{equation}
\begin{aligned}
    \label{eq:GruppenwirkungVektorbuendel}
    \Phi: G \times E & \to E,\quad (g,v)  \mapsto \Phi(g,v)=:
    \Phi_g(v)=g.v,
\end{aligned}    
\end{equation}
\item f"ur alle $v\in E$ und $g\in G$ die Gleichung $g.(\pi(v))=\pi(g.v)$ (oder
    argumentfrei $\pi \circ \Phi = \phi \circ \pi$) gilt. Dies ist
    gleichbedeutend mit der Kommutativit"at des folgenden Diagramms
    \begin{equation}
    \begin{CD}
        \Phi:G \times E @>>> E \\
        @VV{\pi}V       @VV{\pi}V \\
        \phi:G \times M @>>> M. 
    \end{CD}
\end{equation}
\end{enumerate}
Man bezeichnet ein $G$-Vektorb"undel mit $(E\stackrel{\pi}{\to}M, G)$ 
oder kurz mit $(E,G)$.  
\end{definition}

\begin{bemerkungen}
~\vspace{-5mm}
\begin{compactenum}
\item Die Abbildung $\Phi_{g}$ ist ein Vektorb"undelautomorphismus "uber $\phi_g$.
\item Die Fasern $E_m=\pi^{-1}(m)$ sind hier komplexe Vektorr"aume
    endlicher Dimension.
\item Die Abbildung $\Phi:E_{m} \to E_{g.m}$ ist ein Vektorraum
    Homomorphismus f"ur alle $m\in M$ und alle $g\in G$.
\end{compactenum}
\end{bemerkungen}

\begin{lemma}[Auf Schnitte induzierte $G$-Wirkung]
\label{Lemma:AufSchnitteInduzierteGWirkung}
Die G-Wirkung auf den Schnitten $s\in\schnitt{E}$ ist auf nat"urliche
Weise gegeben durch
\begin{align}
    (g.s)(m):=\Phi_{g}(s(\phi_{g}^{-1}(m))).
\end{align}
\end{lemma}

\begin{definition}[$G$-invariante Schnitte]
\label{Definition:GInvarianteSchnitte}
\index{Schnitt!G-invarianter@$G$-invarianter}
Gegeben sei ein $G$-Vektorb"undel $(E,G)$. 
Man nennt Schnitte {\em $G$-invariant}, falls $g.s=s$ f"ur
alle $g \in G$ gilt.
\end{definition}

\begin{lemma}[Induzierte $G$-Wirkung auf Funktionen]
\label{Lemma:InduzierteGWirkungAufFunktionen}
 Die Abbildung $\phi: G \times M \to M$ induziert eine
 $G$-Linkswirkung auf den Funktionen $C^\infty(M)$: 
\begin{align}
    g.f := \phi_{g^{-1}}^\ast f = f \circ \phi_{g^{-1}}.
\end{align}
\end{lemma}

\begin{lemma}
    \label{Lemma:FunktionenSchnitte}
    Gegeben sei das $G$-Vektorb"undel $(E\stackrel{\pi}{\to}M,G)$, $f
    \in C^\infty(M)$, $s\in\schnitt{E}$ und $g\in G$, dann gilt:
    \begin{align}
        g.(sf)=(g.s)(g.f)
    \end{align}
\end{lemma}

\begin{proof}
Durch Einsetzen der obigen Definitionen leicht zu sehen.
\end{proof}

\begin{lemma}[Induzierte $G$-Wirkung auf Einsformen und Vektorfeldern]
\label{Lemma:GWirkungAufEinsformenUndVektorfeldern}
Die Gruppenwirkung $\phi$ induziert auch eine $G$-Linkswirkung auf den
Schnitten im Kotangentialb"undel $\schnitt{T^\ast M} \ni \alpha$, den 1-Formen

\begin{equation}
\begin{aligned}
   (g.\alpha)(m):&=(\phi_{g^{-1}}^\ast \alpha)(m) \\
& = \alpha(\phi_{g^{-1}}(m)) \circ T_m \phi_{g^{-1}} \\
& = \alpha(\phi_{g^{-1}}(m)) \circ (T_{\phi_{g^{-1}}(m)} \phi_g)^{-1},
\end{aligned}
\end{equation}

sowie auf den Schnitten im Tangentialb"undel $\schnitt{TM} \ni X$,
den Vektorfeldern

\begin{align}
    (g.X)(m):&=(\phi_{g\ast} X)(m) = T_{\phi_{g^{-1}}(m)}
    \phi_gX(\phi_{g^{-1}}(m)).  
\end{align}
Ferner ist die Gruppenwirkung $\phi$ vertr"aglich mit der Kontraktion. Das
hei"st es gilt: $(g.\alpha)(g.X) = g.(\alpha(X))$.
\end{lemma}

\begin{definition}[Wirkung auf Endomorphismenb\"undel]
Sei $A \in \schnitt{\End{E}}$ und $s \in \schnitt{E}$ so definiert man
eine Gruppenwirkung auf $\schnitt{\End{E}}$ mittels
\begin{align}
    \label{eq:WirkungAufendE}
    (g.A)(s) : = g.(A(g^{-1}a)).
\end{align}
\end{definition}

Die so definierte $G$-Wirkung ist vertr"aglich mit der nat"urlichen
Multiplikation von $A$ und $s$:
    \begin{align}
        g.(As)=(g.A)(g.s).
    \end{align}
Wir rechnen mit $A \in \schnitt{\End{E}}$ und $s \in \schnitt{E}$ nach:
        \begin{align*}
            g.(As)(m) &=\Phi_g(As)\phi_{g^{-1}}(m) \\ &=
            \Phi_g|_{E_{\phi_g^{-1}(m)}} A(\phi_{g^{-1}}s)(m) \\ &=
            \underbrace{\Phi_g|_{E_{\phi_g^{-1}(m)}}
              A(\Phi_g|_{E_{\phi_g^{-1}}(m)})^{-1}}_{(g.A)(m)}
            \underbrace{\Phi_g|_{E_{\phi_{g^{-1}(m)}}}
              s(\phi_g^{-1}(m))}_{ (g.s)(m)} \\ & \Rightarrow \qquad
            g.(A.s)=(g.A)(g.s).  
        \end{align*}

\begin{definition}[$G$- und $\LieAlg{g}$-invarianter Zusammenhang $\nabla$]
\label{Defintion:GundgInvarianterZusammenhang}
\index{Zusammenhang!g-invarianter@$\LieAlg{g}$-invarianter}
\index{Zusammenhang!G-invarianter@$G$-invarianter}
Sei $(M,\omega)$ eine symplektische Mannigfaltigkeit mit einem
linearen Zusammenhang nach Definition
\ref{Definition:SymplektischerZusammenhang}, $G$ eine
\Name{Lie}-Gruppe und $\LieAlg{g}$ eine \Name{Lie}-Algebra, die nicht
notwendigerweise von einer Gruppe $G$ kommt. Man nennt den Zusammenhang
$\nabla$ f"ur alle $X,Y \in \schnitt{TM}$ 
\begin{compactenum}

\item {\em $G$-invariant}, falls f"ur alle $g \in G$ gilt:
\begin{align}
\phi_{g}^{\ast} \lconn[X]{Y} = \lconn[\phi_{g}^{\ast}
X]{\phi_{g}^{\ast} Y}. 
\end{align}
\item {\em $\LieAlg{g}$-invariant}, falls f"ur alle $\xi \in
    \LieAlg{g}$ gilt:
\begin{align}
\Lie[\xi] \left(\lconn[X]{Y} \right) = \lconn[{\Lie[\xi]X}]{Y} +
\lconn[X]{\Lie[\xi] Y}.
\end{align}
\end{compactenum}
\end{definition}

Eine wichtige Frage ist nun, wie man einen $G$-invarianten
Zusammenhang auf einer Mannigfaltigkeit $M$ erh"alt. F"ur
kompakte Gruppen kann man sich einen solchen konstruieren, indem man
den Zusammenhang "uber die Gruppe mittelt.

\begin{lemma}[Mittelung eines Zusammenhangs $\nabla$]
    \label{Lemma:MittelungZusammenhang}
\index{Zusammenhang!gemittelter}
    Gegeben sei eine kompakte \Name{Lie}-Gruppe $G$, die auf der
    Mannigfaltigkeit 
    $(M,\nabla)$ durch Linksaktion agiere. Man erh"alt wie folgt einen
    $G$-invarianten Zusammenhang $\cc{\nabla}$:
    \begin{align}
        \label{eq:MittelungZusammenhang}
        \cc{\nabla}_{\!\!\sss X} Y=\frac{1}{\Vol G} \int_G
        g^{-1}.(\nabla_{\!\!\sss g.X} g.Y) \, \mu(g)
    \end{align}
Dabei ist $\mu$ das
\Name{Haar}-Ma"s\index{Haar-Mass@\Name{Haar}-Ma{\ss}}, und $X,Y \in
\schnitt{TM}$. 
\end{lemma}

\begin{proof}
    Um Lemma \ref{Lemma:MittelungZusammenhang} zu beweisen mu"s man zum
    einen zeigen, da"s es sich bei $\cc{\nabla}$ um einen Zusammenhang
    auf $TM$ handelt und zum anderen, da"s dieser invariant ist. 
    
\begin{compactenum}
    \item Im ersten Schritt zeigen wir die $\Cinf{M}$-Linearit"at im
        ersten Argument, es gilt also $\cc{\nabla}_{\!\!\sss fX}Y =
        f\cc{\nabla}_{\!\!\sss X} Y$ f"ur alle $f \in C^\infty(M)$ und $X,Y \in
        \schnitt{TM}$. 
    \begin{align*}
              \cc{\nabla}_{\!\! \sss fX}Y & =  \frac{1}{\Vol G} \int_G
        g^{-1}.(\nabla_{\!\! \sss g.(fX)}g.Y) \, \mu(g) \\ & =  \frac{1}{\Vol G}
        \int_G g^{-1}.(\nabla_{\!\! \sss (g.f)(g.X)}g.Y) \, \mu(g) \\ & =
        \frac{1}{\Vol G} \int_G g^{-1}.(g.f) g^{-1}.(\nabla_{\!\! \sss
          g.X}g.Y) \, \mu(g) \\ & =f \frac{1}{\Vol G}  \int_G
        g^{-1}.(\nabla_{\!\! \sss g.X}g.Y) \, \mu(g) \\ & = f
        \cc{\nabla}_{\!\! \sss h.X} h.Y  
     \end{align*}
        
    \item Im zweiten Argument ist der Zusammenhang derivativ,
        $\cc{\nabla}_{\!\!\sss X}(fY) = f\cc{\nabla}_{\!\!\sss X}Y +
        X(f)Y$ f"ur alle $f \in C^\infty(M)$  wie eine
        kurze Rechnung zeigt.
    \begin{align*}
        \cc{\nabla}_{\!\! \sss X}(fY) &=  \frac{1}{\Vol G} \int_G
        g^{-1}.(\nabla_{\!\! \sss g.X}g.(fY)) \, \mu(g) \\ &=  \frac{1}{\Vol G}
       \int_G g^{-1}.(\nabla_{\!\! \sss g.X}((g.f)(g.Y)) \, \mu(g) \\
       &=\frac{1}{\Vol G} \int_G g^{-1}.\left(f (\nabla_{\!\! \sss g.X}g.Y) +
           ((g.X)(g.f))Y \right)\, \mu(g) \\ &= \frac{1}{\Vol G}
       \left( \int_G f g^{-1}.\nabla_{\!\! \sss g.X}g.Y \,\mu(g) +
           \int_G g^{-1} 
           ((g.X)(g.f) g.Y) \,\mu(g)\right) \\&=
       f\cc{\nabla}_{\!\! \sss X} Y + X(f)Y. 
     \end{align*}
    \item Der Zusammenhang ist $G$-invariant,
        d.~h.~$h.\cc{\nabla}_{\!\! \sss X} Y= \cc{\nabla}_{\!\!\sss
          h.X} h.Y$ f"ur alle $h \in G$. 
     \begin{align*}
                   h.\cc{\nabla}_{\!\! \sss X} Y &= h. \frac{1}{\Vol G} \int_G
        g^{-1}.(\nabla_{\!\! \sss g.X} g.Y) \, \mu(g) \\ &= \frac{1}{\Vol G}
        \int_G \underbrace{h.g^{-1}}_{=\tilde{g}^{-1}}.(\nabla_{\!\!
          \sss g.X} g.Y) \, \underbrace{\mu(g)}_{=\mu(\tilde{g})} \\
        &= \frac{1}{\Vol 
          G} \int_G \tilde{g}^{-1}.(\nabla_{\!\! \sss \tilde{g}.h.X}
        \tilde{g}.h.Y) \, \mu(\tilde{g}) \\ &=  \cc{\nabla}_{\!\! \sss
          h.X} h.Y. 
             \end{align*}
    \end{compactenum}
\end{proof}

\begin{lemma}[$G$-invarianter, symplektischer Zusammenhang]
    \label{Lemma:SymplektischerZusammenhangGInvarianz}
    W"ahlt man einen beliebigen, torsionsfreien, $G$-invarianten Zusammenhang
    $\tilde{\nabla}$ auf der symplektischen $G$-Mannigfaltigkeit
    $(M,\omega,G)$, und ist die $G$-Wirkung $g.$ ein
    Symplektomorphismus, so erh"alt man mit der Konstruktion eines
    symplektischen Zusammenhangs aus Lemma \ref{Lemma:HessTrick}, einen symplektischen und
    $G$-invarianten Zusammenhang $\nabla$.  
\end{lemma}

\begin{proof}
Wir betrachten dazu folgenden Ausdruck
\begin{equation}
    \begin{aligned}
        \omega(\nabla_{\!\! \sss g.X}g.Y,g.Z) &=
        \omega(\tilde{\nabla}_{\!\! \sss g.X}g.Y,g.Z)
        +\frac{1}{3}(\tilde{\nabla}_{\!\! \sss g.X} \omega)(g.Y,g.Z) +
        \frac{1}{3}(\tilde{\nabla}_{\!\! \sss g.Y} \omega)(g.X,g.Z) \\ &=
        \omega(g.\tilde{\nabla}_{\!\! \sss X}Y,g.Z)
        +\frac{1}{3}(\tilde{\nabla}_{\!\! \sss g.X} g.\omega)(Y,Z) +
        \frac{1}{3}(\tilde{\nabla}_{\!\! \sss g.Y} g.\omega)(X,Z) \\ &=
        g. \omega(\tilde{\nabla}_{\!\! \sss X}Y,Z) +\frac{1}{3}
        g.(\tilde{\nabla}_{\!\! \sss X} \omega)(Y,Z) +
        \frac{1}{3}g.(\tilde{\nabla}_{\!\! \sss Y} \omega)(X,Z) \\ &=
        g.(\omega(\tilde{\nabla}_{\!\! \sss X} Y,Z) +\frac{1}{3}
        (\tilde{\nabla}_{\!\! \sss X} \omega)(Y,Z) + \frac{1}{3} (\tilde{\nabla}_Y
        \omega)(X,Z))\\  &= g.\omega(\nabla_{\!\! \sss X} Y,Z) \\
        &=\omega(g.(\nabla_{\!\! \sss X} Y),g.Z) \\ & \\ \Rightarrow
        \quad & g.(\nabla_{\!\! \sss X} Y)=\nabla_{\!\! \sss g.X}g.Y.  
    \end{aligned}
\end{equation}
\end{proof}

\begin{lemma}[$G$-invarianter Zusammenhang auf Endomorphismenb"undel]
    \label{Lemma:GInvarianterZusammenhangAufEndomorphismenbuendel}
Sei $\bundle{E}{\pi}{M}$ ein B"undel mit einem $G$-invarianten
Zusammenhang $\lconnE$ "uber der Mannigfaltigkeit $M$. Der auf dem
Endomorphismenb"undel $\bundle{\End{E}}{\pi'}{E}$ induzierte
Zusammenhang $\lconnEnd{E}$ ist auch $G$-invariant.
\end{lemma}

\begin{proof}
Der Beweis ist eine einfache Rechnung.
\begin{align*}
g. \left( \lconnEnd[X]{E} A \right) (s) & =g.\left( \left[ 
        \lconnE[X], A \right] \right)(s) \\ & = g. \left( \left[ 
        \lconnE[X], A \right](g^{-1}.s) \right) \\ & = g. \left (
    \lconnE[X](A(g^{-1}.s)) - A \lconnE[X](g^{-1}.s) \right) \\ & =
g. \lconnE[X](A(g^{-1}.s)) - g. \left( A \lconnE[X](g^{-1}.s) \right)
\\ & = \lconnE[g.X](g.(A(g^{-1}.s))) - (g.A) g.\left(
    \lconnE[X](g^{-1}.s) \right) \\ & = \lconnE[g.X]((g.A)(s)) - (g.A)
\lconnE[g.X](s) \\ & =    \left[ \lconnE[g.X], g.A \right](s) \\ & =
\left(\lconnEnd[g.X]{E} (g.A) \right) (s).   
\end{align*}
\end{proof}

\index{Geometrie!symplektische|)}

\section{\Name{Poisson}-Geometrie}
\index{Geometrie!Poisson@\Name{Poisson}|(}
Sich an dieser Stelle ausf"uhrlich mit der \Name{Poisson}-Geometrie
auseinander zu setzen w"urden den Rahmen dieser Arbeit sprengen. Daher
wollen wir uns kurz fassen und nur ein paar grundlegende Definitionen
angeben, und den interessierten Leser auf die folgenden B"ucher
aufmerksam machen. Geeignete Literatur zur \Name{Poisson}-Geometrie sind
\citep{vaisman:1994a}, \citep{marsden.ratiu:2000a},
\citep{cannasdasilva.weinstein:1999a} sowie
\citep{waldmann:2004a}. 

\begin{definition}[\Name{Poisson}-Klammer und \Name{Poisson}-Mannigfaltigkeit]
\label{definition:PoissonMannigfaltigkeit}
\index{Mannigfaltigkeit!Poisson-Mannigfaltigkeit@\Name{Poisson}-Mannigfaltigkeit|textbf}
\index{Poisson-Mannigfaltigkeit@\Name{Poisson}-Mannigfaltigkeit|textbf}
Sei $M$ eine Mannigfaltigkeit. Man nennt eine bilineare Abbildung 
\begin{align*}
\{ \cdot, \cdot\} : \Cinf{M} \times \Cinf{M} \to \Cinf{M}
\end{align*} 
eine  {\em \Name{Poisson}-Klammer}\index{Poisson-Klammer@\Name{Poisson}-Klammer} falls sie die folgenden Eigenschaften hat:
\begin{compactenum}
\item $\{f,g\} = - \{g,f\}$ (Schiefsymmetrie),
\item $\{\alpha_{\sss 1} f_{\sss 1} + \alpha_{\sss 2} f_{\sss 2}, g\}
    = \alpha_{\sss 1} \{f_{\sss 1}, g\} + \alpha_{\sss 2} \{f_{\sss
      2}, g\}$ (Bilininearit"at) 
\item $\{fg,h\} = f\{g,h\} + \{f,g\}h$ (\Name{Leibniz}-Regel),\index{Leibniz-Regel@\Name{Leibniz}-Regel}
\item $\{f , \{g,h \} \} = \{ \{f,g \},h \} + \{g, \{f,h\} \}$
    (\Name{Jacobi}-Identit"at),\index{Jacobi-Identitaet@\Name{Jacobi}-Identit\"at} 
\end{compactenum}
mit $\alpha_{\sss 1}, \alpha_{\sss 2} \in \field{R}$. Man nennt $(M,
\{\cdot, \cdot \})$ eine \Name{Poisson}-Mannigfaltigkeit.  
\end{definition}

\begin{bemerkungen}
~\vspace{-5mm}
\begin{compactenum}
\item Die \Name{Leibniz}-Regel von $\{\cdot, \cdot \}$ ist "aquivalent
dazu, da"s es ein Bivektorfeld $\Lambda \in \schnitt{\Lambda^2 TM}$
gibt, so da"s 
\begin{align}
\{f,g\} = \Lambda (\de f \otimes \de g).
\end{align} 
Man nennt $\Lambda$ den {\em \Name{Poisson}-Bivektor}. In lokalen
Koordinaten $(U,u)$ schreibt man das Bivektorfeld als 
\begin{align}
\Lambda\big|_{\sss U} = \frac{1}{2} \sum_{i,j} \Lambda^{ij}\frac{\del}{\del u^i}
\wedge \frac{\del}{\del u^j}.   
\end{align}
Die \Name{Poisson}-Klammer wird in einer Karte $(U,u)$ dann zu 
\begin{align}
\{f,g\}\big|_{U} = \sum_{i,j} \Lambda^{ij} \frac{\del f}{\del u^i} \frac{\del
  g}{\del u^j}. 
\end{align}
\item Die \Name{Jacobi}-Identit"at kann man dann in einer den lokalen
    Koordinaten $(U,u)$ schreiben als
    \begin{align}
       \sum_{i,j,s,r} \left(\Lambda^{ij}_{,r} \Lambda^{rs} + \Lambda^{js}_{,r} \Lambda^{ri}
        + \Lambda^{si}_{,r} \Lambda^{rj}\right) \frac{\del f}{\del u^i}
        \frac{\del g }{\del u^j} \frac{\del h}{\del u^s} = 0.
    \end{align}
\end{compactenum}
\end{bemerkungen}

\index{Geometrie!Poisson@\Name{Poisson}|)}


%% file: symbolverzeichnis.tex
\chapter*{Symbolverzeichnis}
\fancyhead[CE]{\slshape \nouppercase{Symbolverzeichnis}} 
\fancyhead[CO]{\slshape \nouppercase{Symbolverzeichnis}} 
\addcontentsline{toc}{chapter}{Symbolverzeichnis}

\begin{longtable}{lp{12,4cm}}

\multicolumn{2}{l}{\Large \textbf{Allgemeines}} \\
&  \\
$\hbar$ & \Name{Planck}sches Wirkungsquantum $(\hbar=\frac{h}{2\pi} =
1,0546 \times 10^{-34}$ Js) \\ 
$\im$ & Imagin"are Einheit $\im^2 = -1$ \\
& \\
$\field{N}$ & Nat"urliche Zahlen \\
$\field{Z}$ & Ring der ganzen Zahlen \\
$\field{Z}_{p}$ & Restklassenring $\field{Z}_{p}=\field{Z}/
p\field{Z}$ \\
$\field{Q}$ & K"orper der rationalen Zahlen \\
$\field{R}$ & K"orper der reellen Zahlen \\
$\field{C}$ & K"orper der komplexen Zahlen \\
$\field{K}$ & K"orper: $\field{Q}$, $\field{R}$ oder $\field{C}$ \\
$\fieldf{K}$ & Ring der formalen Potenzreihen in $\lambda$ mit Werten in
$\field{K}$ \\ 
$\field{T}^{n}$ & $n$-dimensionaler Torus
$\field{T}^{n}=\field{R}^{n}/ \field{Z}^{n}$ \\
$S^{n}$ & $n$-dimensionale Sph"are \\
$M_{n}(\cdot) $ & Ring der $n\times n$-Matrizen \\
$\GR{O}{n}$, $\GR{SO}{n}$ & Orthogonale Gruppe, Spezielle Orthogonale
Gruppe in $n$ Dimensionen\\
$\GR{U}{n}$, $\GR{SU}{n}$ & Unit"are Gruppe, Spezielle Unit"are
Gruppe in $n$ Dimensionen \\

$\ring{R}$, $\ring{S}$ & (geordnete) Ringe  \\
$\ring{C}$ & komplexe Erweiterung des geordneten Rings $\ring{R}$,
$\ring{C}=\ring{R}(\im)$ \\

& \\
\multicolumn{2}{l}{\Large \textbf{Kapitel \ref{chapter:Morita}}} \\
&  \\

$\alg{H}$ & Pr"a-\Name{Hilbert}-Raum \\
$\hilbert{H}$ & \Name{Hilbert}-Raum\\
$\SP{\cdot,\cdot}$ & Skalarprodukt auf (Pr"a-)\Name{Hilbert}-Raum \\
$\psi, \phi, \chi$ & Vektoren in einem (Pr"a-)\Name{Hilbert}-Raum \\
$\alg{H}^{\perp}$ & Ausartungsraum von $\SP{\cdot,\cdot}$ \\
$[\psi], [\phi]$ & Vektoren im Quotientenraum $\alg{H} /
\alg{H}^{\perp}$ ("Aquivalenzklasse)   \\
$A,B,C$ & (Adjungierbare) Operatoren \\ 
$\alg{B}(\alg{H})$ & Raum der Operatoren auf $\alg{H}$, deren
Adjungiertes existiert \\
$\PraeHilbert{\ring{C}}$ & Kategorie der Pr"a-\Name{Hilbert}-R"aume
"uber $\ring{C}$ \\
$\Theta_{\sss \phi, \psi}$ & Operator vom Rang Eins \\
$\alg{F}(\alg{H}_{1},\alg{H}_{2})$ & Raum der Operatoren mit endlichem
Rang \\
$\alg{A,B,C,D}$ & $^\ast$-Algebren \\
$\staralg(\ring{C})$ & Kategorie der $^\ast$-Algebren "uber $\ring{C}$ \\
$\starAlg(\ring{C})$ & Kategorie der $^\ast$-Algebren mit Einselement "uber $\ring{C}$ \\
$L^2(\cdot)$ & Raum der quadratintegrablen Funktionen \\
$1_{\sss \alg{A}} \in \alg{A}$ & Einselement der Algebra $\alg{A}$ \\ 
$\alg{A}^{+}$ & Menge der positiven Elemente in $\alg{A}$ \\
$\alg{A}^{++}$ & Menge der quadratischen Elemente in $\alg{A}$ \\
$\omega$ & Positives Funktional \\
$\rmod{H}{D}$, $\lmod{A}{H}$ & $\alg{D}$-Rechtsmodul, $\alg{A}$-Linksmodul \\
$\rSP{\cdot, \cdot}{\alg{D}}$ & $\alg{D}$-wertiges inneres Produkt auf
$\alg{D}$-Rechtsmodul \\
$\lSP{\alg{A}}{\cdot, \cdot}$ & $\alg{A}$-wertiges inneres Produkt auf
$\alg{A}$-Linksmodul \\
$\rmodplus{H}{D}$ & Innerer Produktmodul \\
$\rmodplusne{H}{D}$ & Pr"a-\Name{Hilbert}-Modul \\
$\CSpan{\cdot}$ & $\ring{C}$-lineare H"ulle \\
$h$ & \Name{Hermite}sche Fasermetrik \\ 
 $h_{\sss x} $ &  \Name{Hermite}sche Fasermetrik am Punkt $x$
ausgewertet \\
$(\alg{H},\pi)$ & $^\ast$-Darstellung \\
$T$ & Verschr"ankungsoperator zwischen $^\ast$-Darstellungen \\
$\smod[\alg{D}](\alg{A})$ & Kategorie der $^\ast$-Darstellungen der
$^\ast$-Algebra $\alg{A}$ auf $\alg{D}$-Rechtsmodul \\
$\srep[\alg{D}](\alg{A})$ & Kategorie der $^\ast$-Darstellungen der
$^\ast$-Algebra $\alg{A}$ auf Pr"a-\Name{Hilbert}-$\alg{D}$-Rechtsmoduln \\
$\sMod[\alg{D}](\alg{A})$ & Kategorie der stark nichtentarteten $^\ast$-Darstellungen der
$^\ast$-Algebra $\alg{A}$ auf $\alg{D}$-Rechtsmodul\\
$\sRep[\alg{D}](\alg{A})$ & Kategorie der stark nichtentarteten $^\ast$-Darstellungen der
$^\ast$-Algebra $\alg{A}$ auf Pr"a-\Name{Hilbert}-$\alg{D}$-Rechtsmoduln  \\
$\katrmod{\ring{R}}$  & Kategorie der $\ring{R}$-Rechtsmoduln \\
$\katlmod{\ring{R}}$  & Kategorie der $\ring{R}$-Linksmoduln\\
$\katrMod{\ring{R}}$  & Kategorie der $\ring{R}$-Rechtsmoduln mit
$\rmodo{\alg{E}}{\ring{R}} \cdot \ring{R} = \rmodo{\alg{E}}{\ring{R}}$\\
$\katlMod{\ring{R}}$ & Kategorie der $\ring{R}$-Linksmoduln mit $\ring{R} \cdot
  \lmodo{\ring{R}}{\alg{E}}  = \lmodo{\ring{R}}{\alg{E}}$ \\
 $\bimod{B}{E}{A}$, $\bimod{B}{F}{A}$  &  $\alg{B}$-Links-
$\alg{A}$-Rechtsmodul, d.~h.~$(\alg{B},\alg{A})$-Bimodul \\
$\bimodo{\ring{R}}{\alg{E}^{\ast}}{\ring{S}}$ & dualer Bimodul zu
$\bimodo{\ring{S}}{\alg{E}}{\ring{R}}$ \\
$\bimod{A}{\cc{E}}{B}$ & komplex-konjugierter Bimodul zu $\bimod{B}{E}{A}$ \\
$\rSP[\sss \alg{E}]{\cdot, \cdot}{\alg{A}}$ & $\alg{A}$-wertiges
inneres Produkt auf $\bimod{B}{E}{A}$ \\
$\bimodplus{B}{E}{A}$ & $(\alg{B}, \alg{A})$-Bimodul mit
$\alg{B}$-wertigem und $\alg{A}$-wertigem inneren Produkt \\
$\tensor[\alg{B}]$ & Inneres Tensorprodukt "uber der Algebra $\alg{B}$
\\ 
$\tensorhat[\alg{B}]$ & Inneres Tensorprodukt "uber der Algebra
$\alg{B}$ (ber"ucksichtigt innesres Produkt) \\
$\tensortilde[\alg{B}]$ & Inneres Tensorprodukt "uber der Algebra
$\alg{B}$ (ber"ucksichtigt beide inneren Produkte) \\
$\rieffel{R}{\alg{E}}$ & \Name{Rieffel}-Induktion \\
$\rieffel{S}{\alg{E}}$ & Wechsel der Basisalgebra \\
$\kat{A},\kat{B}, F$ & Kategorien $\kat{A},\kat{B}$ und Funktor
$F:\kat{A}\to \kat{B}$, \\ 
$\Id_{\kat{A}}$ & Identit"atsfunktor auf der Kategorie $\kat{A}$ \\
$\Obj(\kat{A})$ & Objekte der Kategorie $\kat{A}$ \\
$\Morph(\kat{A})$ & Morphismen der Kategorie $\kat{A}$ \\
$J(\ring{R})$ & \Name{Jacobson}-Radikal des Rings $\ring{R}$ \\
$\zentrum{\ring{R}}$ & Zentrum des Rings $\ring{R}$ \\ 
$\Hom[\alg{A}](\alg{B}, \alg{C})$ & $\alg{A}$-lineare Homomorphismen
von $\alg{B}$ nach $\alg{C}$ \\
$\End[\ring{R}]{\rmodo{\alg{E}}{\ring{R}}}$ & $\ring{R}$-lineare
Endomorphismen von $\rmodo{\alg{E}}{\ring{R}}$ \\
$P$ & idempotentes Element oder Projektor \\
$P \ring{R}^{n}$ & $n$-komponentiger projektiver Modul \\

$H$ & \Name{Hopf}-Algebra (siehe auch Anhang~\ref{chapter:AlgebraischeGrundlagen}) \\
$(H, \neact)$ & \Name{Hopf}-Algebra $H$ mit Wirkung $\act$ \\
$(H, \alg{A}, \neact)$ & $H$-Modulalgebra \\ 
$\ccact$ & Wirkung auf dualem Bimodul $\bimod{A}{\cc{E}}{B}$ \\
$\smod[\alg{D},H](\alg{A})$ & Kategorie der $H$-"aquivarianten
$^\ast$-Darstellungen der $^\ast$-Algebra $\alg{A}$ auf
$\alg{D}$-Rechtsmodul \\
$\srep[\alg{D},H](\alg{A})$ & Kategorie der $H$-"aquivarianten
$^\ast$-Darstellungen der $^\ast$-Algebra $\alg{A}$ auf
Pr"a-\Name{Hilbert}-$\alg{D}$-Rechtsmodul \\
$\sMod[\alg{D},H](\alg{A})$ & Kategorie der stark nichtentarteten
$H$-"aquivarianten $^\ast$-Darstellungen der
$^\ast$-Algebra $\alg{A}$ auf $\alg{D}$-Rechtsmodul\\
$\sRep[\alg{D},H](\alg{A})$ & Kategorie der stark nichtentarteten
$H$-"aquivarianten $^\ast$-Darstellungen der
$^\ast$-Algebra $\alg{A}$ auf Pr"a-\Name{Hilbert}-$\alg{D}$-Rechtsmoduln  \\

& \\
\multicolumn{2}{l}{\Large \textbf{Kapitel \ref{chapter:Picard}}} \\
&  \\

$\KPic$ & \Name{Picard}-Bikategorie \\ 
$\Pic$, $\starPic$, $\strPic$ & \Name{Picard}-Gruppoid,
$^\ast$-\Name{Picard}-Gruppoid, starkes \Name{Picard}-Gruppoid \\

$\Pic(\alg{A})$ & \Name{Picard}-Gruppe von $\alg{A}$ \\
$\SPic(\alg{A})$ & Statische (klassische, kommutative) \Name{Picard}-Gruppe von
$\alg{A}$ \\

$\Iso(\alg{A},\alg{B})$ & Isomorphismen von $\alg{A} \to \alg{B}$ \\
$\!\starIso(\alg{A},\alg{B})$ &  $^\ast$-Isomorphismen von $\alg{A} \to
\alg{B}$ \\
$\!\Aut(\alg{A})$ & Automorphismengruppe von $\alg{A}$ \\
$\!\starAut(\alg{A})$ &  $^\ast$-Automorphismen von $\alg{A}$ \\
$\InnAut(\alg{A})$ & Gruppe der inneren Automorphismen der Algebra $\alg{A}$ \\
$\!\starInnAut(\alg{A})$ & Gruppe der inneren $^\ast$-Automorphismen der Algebra
$\alg{A}$  \\
$\OutAut(\alg{A})$ & "au"sere Automorphismen der Algebra $\alg{A}$ \\
$\!\starOutAut(\alg{A})$ & "au"sere $^\ast$-Automorphismen der
Algebra $\alg{A}$ \\ 
$\cdot_{\sss \phi}$, $\cdot_{\sss \psi}$ & durch Automorphismus $\phi$ bzw.~$\psi$
getwistete Modulstruktur \\ 
$\Diff(M)$ & Diffeomorphismengruppe von $M$ \\
$\HCech[n](M,\field{Z})$ & $n$-te integrale \v{C}ech-Kohmologie \\  
$\!\AutH$, $\PicH$, & $H$-"aquivariante Automorphismen, $H$-"aquivariantes
\Name{Picard}-Gruppoid  \\
$G$, $\LieAlg{g}$ & \Name{Lie}-Gruppe $G$, \Name{Lie}-Algebra
$\LieAlg{g}$ \\
$\Lie[\xi]$ & \Name{Lie}-Ableitung \\
$X_{\sss \xi}$ & linksinvariantes Vektorfeld zu $G$ \\
$e_{1}, \ldots, e_{n}$ & Basis von $\LieAlg{g}$ \\
$X_{\sss e_{1}}, \ldots, X_{\sss e_{n}}$ & Modulbasis aller
Vektorfelder $\schnitt{TG}$ von $G$\\
$\nabla$ & linearer Zusammenhang \\
$h \mapsto u_{\sss h}$ & $\ring{C}$-linearer Endomorphismus von
$\bimod{B}{E}{A}$ \\
$\msf{a}, \msf{b}$ & $\ring{C}$-linearer Homomorphismus von
\Name{Hopf}-Algebra $H$  in $^\ast$-Algebra \\
$\neactb$ & mit $\msf{b}$ von links getwistet Wirkung $\neact$ auf Bimodul \\
$\neacta$ & mit $\msf{a}$ von rechts getwistet Wirkung $\neact$ auf Bimodul \\
$\GR{GL}{H,\alg{A}}$ & Gruppe von Elementen aus
$\Hom[\ring{C}](H,\alg{A})$ bez"uglich Konvolutionsprodukt $\ast$   \\ 
$\GR{U}{H,\alg{A}}$ & Unit"are Elemente in $\GR{GL}{H,\alg{A}}$ \\
$\GR{GL}{\zentrum{\alg{A}}}$ & \Name{Abel}sche Gruppe der invertierbaren, zentralen
Elemente der Algebra $\alg{A}$\\ 
$\GR{U}{\zentrum{\alg{A}}}$ & \Name{Abel}sche Gruppe der unit"aren, zentralen Elemente der
$^\ast$-Algebra $\alg{A}$ \\
$\GRn{GL}{H,\alg{A}}$, $\GRn{U}{H,\alg{A}}$ & Quotientengruppen \\
$Z^{n}_{\sss \mathrm{CE}}(\LieAlg{g},\alg{A})$ & $n$-Kozykeln der
  \Name{Chevalley-Eilenberg} Kohomologie mit Werten in $\alg{A}$ \\
$J$ & Impulsabbildung \\
$(\bimod{B}{E}{A}, \neact)$ & Bimodul mit Wirkung \\

$\Bij(M)$ & Gruppe der Bijektionen auf $M$ \\

$\alg{I},\alg{J},\alg{K}$ & $^\ast$-Ideale \\
$\alg{J}^{\mathrm{cl}}$ & kleinstes abgeschlossenes Ideal, das $\alg{J}$
enth"alt \\
$\verband[\alg{D}]$ & Verband der $\alg{D}$-abgeschlossenen Ideale \\
$\verband[\alg{D},H]$ & Verband der $(\alg{D},H)$-abgeschlossenen Ideale \\

$\zentrum[H]{\alg{A}}$ & $H$-invariantes Zentrum der Algebra $\alg{A}$
\\ 

$\KProj(\alg{A})$ & Kategorie der endlich erzeugten, projektiven
$\alg{A}$-Rechtsmoduln \\
$\KstarProj(\alg{A})$  & Kategorie der endlich erzeugten, projektiven
Pr"a-\Name{Hilbert}-$\alg{A}$-Rechtsmoduln \\
$\KstrProj(\alg{A})$  & Kategorie der endlich erzeugten, projektiven
starken Pr"a-\Name{Hilbert}-$\alg{A}$-Rechtsmoduln \\
$\KProjH(\alg{A})$ & Kategorie der $H$-"aquivarianten endlich erzeugten, projektiven
$\alg{A}$-Rechtsmoduln \\
$\KstarProjH(\alg{A})$ &  Kategorie der $H$-"aquivarianten, endlich erzeugten, projektiven
Pr"a-\Name{Hilbert}-$\alg{A}$-Rechtsmoduln \\
$\KstrProjH(\alg{A})$ & Kategorie der $H$-"aquivarianten, endlich erzeugten, projektiven
starken Pr"a-\Name{Hilbert}-$\alg{A}$-Rechtsmoduln \\
$\{x_{i},y_{i} \}_{i=1,...,n}$ & \Name{Hermite}sche duale Basis eines
endlich erzeugten, projektiven Rechtsmoduls \\

$\uu{\kat{A}}$, $\uu{\kat{B}}$, $\uu{\kat{C}}$, $\uu{F}$ & Bikategorien
$\uu{\kat{A}}$, $\uu{\kat{B}}$, $\uu{\kat{C}}$ und Bifunktor $\uu{F}:
\uu{\kat{A}} \to \uu{\kat{B}}$ \\
$\Bimorph(\uu{\kat{C}})$ & $2$-Morphismen der Bikategorie
$\uu{\kat{C}}$ \\

& \\
\multicolumn{2}{l}{\Large \textbf{Kapitel \ref{sec:MoritaAequivalenzVonCrossProdukten}}} \\
&  \\

$\cross{\alg{A}}{H}$, $\cross{\alg{B}}{H}$ & Cross-Produktalgebren \\
$I_{1}$, $I_{2}$, $I_{3}$ & (kanonische) Isomorphismen \\
$\chi$ & Charakter \\

& \\
\multicolumn{2}{l}{\Large \textbf{Kapitel \ref{chapter:Sternprodukte}}} \\
&  \\

$M,N,Q$ & $C^{\infty}$-Mannigfaltigkeiten (\Name{Hausdorff}sch,
zweites Abz"ahlbarkeitsaxiom) \\
$\Cinf{M}$ & glatte komplexwertige Funktionen auf $M$ \\
$\Cinfc{M}$ & glatte komplexwertige Funktionen mit
kompaktem Tr"ager auf $M$ \\
$f,g$ & Funktionen auf Mannigfaltigkeit \\
$f \mapsto \cc{f}$ & punktweise komplexe Konjugation \\
$\supp f$ & Tr"ager von $f$ \\
$\mu$ & \Name{Borel}-Ma"s \\  
$\Comega{\field{C}^{n}}$ & reell-analytische Funktionen auf dem
$\field{C}^{n}$ \\
$(M,\omega)$ & symplektische Mannigfaltigkeit \\
$(M,\omega,I,g)$ & \Name{K"ahler}-Mannigfaltigkeit mit komplexer
Struktur $I^2=-\id$ und Metrik $g$ \\
$(M,\Lambda)$ & \Name{Poisson}-Mannigfaltigkeit \\
$\{\cdot,\cdot\}$ & \Name{Poisson}-Klammer \\
$TQ$ &  Tangentialb"undel der Mannigfaltigkeit $Q$ \\
$T^{\ast}Q$ &  Kotangentialb"undel der Mannigfaltigkeit $Q$ \\
$\omega_{\sss 0}$ & kanonische symplektische Form auf
Kotangentialb"undeln \\ 
$(T^{\ast}Q, \omega_{\sss 0})$ & Kotangentialb"undel mit kanonischer
symplektischer Zweiform $\omega_{\sss 0} = - \de \theta_{\sss 0}$  \\
$q^{1},\cdots, q^{n},p_{1},\cdots p_{n}$ & lokale B"undelkoordinaten
auf $T^{\ast}Q$ \\
$H$ & \Name{Hamilton}-Funktion \\
$(M,G)$ & $G$-Mannigfaltigkeit \\
$(M,\omega,G)$ & symplektische $G$-Mannigfaltigkeit \\
$\hilbert{D}, \hilbert{H}$ & \Name{Hilbert}-R"aume \\
$\field{P}\hilbert{H}$ & projektiver \Name{Hilbert}-Raum \\
$\psi, \psi'$ & Vektoren in \Name{Hilbert}-Raum \\
$\rSP{\cdot, \cdot}{~}$ & inneres $\field{C}$-wertiges Produkt \\
$P, Q$ & Impulsoperator, Ortsoperator \\
$P_{j},Q^{i}$ & $j$-te Komponente des Impulsoperators, $i$-te
Komponente des Ortsoperators \\
$\id_{\sss \hilbert{H}}$ & Identit"at auf \Name{Hilbert}-Raum
$\hilbert{H}$ \\
$A \mapsto A^{\ast}$ & Adjungieren des Operators $A$  \\
$\delta_{ij}$ & \Name{Kronecker}-Delta \\
$[\cdot, \cdot]$ & Kommutator (von Operatoren)\\ 
$\DiffOp(\field{R}^{n})$ & Differentialoperatoren mit glatten Koeffizientenfunktionen
auf $\field{R}^{n}$ \\
$\DiffOpPol(\field{R}^{n})$ & Differentialoperatoren mit polynomialen Koeffizientenfunktionen
auf $\field{R}^{n}$ \\

$\Pol(\field{R}^{2n})$, $\Pol(T^{\ast}\field{R}^{n})$ & Polynome auf dem
$\field{R}^{2n}$ oder $T^{\ast}\field{R}^{n}$ \\
$z^{1}\cdots z^{n}, \cc{z}^{1}\cdots \cc{z}^{n}$ & (globale)
Koordinatenfunktionen auf dem $\field{C}^{n}$ \\ 
$\de \mu$ & \Name{Gauss}-Ma"s \\
$L^{2}(\field{C}^{n}, \de \mu)$ & Raum der quadratintegrablen
Funktionen auf dem $\field{C}^{n}$ \\
$\ordweyl$, $\ordstd$, $\ordkappa$ & \Name{Weyl}-, Standard-,
 $\kappa$-geordnete Darstellung \\
$\ordwick$, $\ordtkappa$ & \Name{Wick}-, $\tilde{\kappa}$-geordnete Darstellung \\
$\symbweyl$, $\symbstd$, $\symbkappa$   &
\Name{Weyl}-, Standard-, $\kappa$-Symbolabbildungen \\
$\symbwick$, $\symbtkappa$ & \Name{Wick}-, $\tilde{\kappa}$-Symbolabbildungen \\
$\star$ & Sternprodukt \\
$\spweyl$, $\spstd$, $\spkappa$ & \Name{Weyl}-, Standard-,
$\kappa$-geordnete Sternprodukte \\
$\spwick, \sptkappa $ &  \Name{Wick}-, $\tilde{\kappa}$-geordnete Sternprodukte \\
$[\cdot, \cdot]_{\sss \star}$ & Kommutator bez"uglich des
Sternprodukts $\star$ \\
$\alg{O}(\lambda^{n})$ & Terme von $\lambda$ mit mindestens der Potenz
$n$ \\
$N_{\sss \kappa}, S_{\sss \tilde{\kappa}}$ & $\kappa$-ordnender
\Name{Neumaier}-Operator, Analogon f"ur $\tilde{\kappa}$-Ordnung \\ 
$\Delta, \tilde{\Delta}$ & \Name{Laplace}-Operator \\
$P,P^{\ast},Q,Q^{\ast}$ & differentielle Abbildungen auf $\Pol(T^{\ast}
\field{R}^{n}) \otimes \Pol(T^{\ast} \field{R}^{n})$  \\
$D$ & allgemeiner Differentialoperator \\
$\mu$ & Multiplikationsabbildung $\mu (a \otimes b) = ab$ \\
$\tau$ & Vertauschungsoperator, Flip $\tau (a \otimes b) = b \otimes
a$ \\
$C$ & komplexe Konjugation $C(a)=\cc{a}$ \\
$Z,\cc{Z}$ & differentielle Abbildungen auf $\Pol (\field{C}^{n}) \otimes
\Pol(\field{C}^{n})$  \\  
$B_{n}$ & Bidifferentialoperator der Stufe $n$ \\

$\alg{A}$ & Algebra $\alg{A}$ "uber dem Ring $\ring{C}$ \\
$\algf{A}$ & formale Potenzreihe in $\lambda$ mit Koeffizienten aus $\alg{A}$ \\
$\defalg{A}$ & deformierte Algebra $\defalg{A} = (\algf{A}, \star)$
bez"uglich $\star$ \\
$1_{\sss \alg{A}}$ & Einselement der Algebren $\alg{A}$ bzw.~$\defalg{A}$ \\
$C_{n}$ & $\ring{C}$-bilineare Abbildung auf Algebra $\alg{A}$ \\
$T$ & "Aquivalenzoperator \\
$[\star]$ & "Aquivalenzklasse der Deformation $\star$ \\
$\Def(\alg{A})$ & Menge der "Aquivalenzklassen von Deformationen der
Algebra $\alg{A}$ \\
$\Def(\alg{A}, \{\cdot, \cdot\})$ & Menge der "Aquivalenzklassen von Deformationen der
Algebra $\alg{A}$ bei fester \Name{Poisson}-Struktur \\
$\Aut(\alg{A})$ & Gruppe der Automorphismen der Algebra $\alg{A}$ \\
$M_{n}(\alg{A})$ & $n\times n$-Matrizen mit Eintr"agen aus $\alg{A}$,
$M_{n}(\alg{A}) = M_{n}(\ring{C}) \otimes \alg{A}$ \\
$\cl$ & klassische Limes-Abbildung \\ 
$(\alg{A}_{\sss \hbar},\spwick)$ & konvergente deformierte Algebra
$\alg{A}_{\sss \hbar} \subset \Comegaf{\field{C}^{n}}$ bez"uglich des
\Name{Wick}-Produkts \\
$(M, \Lambda, \star)$ &  \Name{Poisson}-Mannigfaltigkeit mit Sternprodukt
$\star$ \\
$(M, \omega, \star)$ &  symplektische Mannigfaltigkeit mit Sternprodukt
$\star$ \\
$R,S,T$ & "Aquivalenzoperatoren bei Sternprodukten \\
$\HdeRham[n] (M)$, $\HdeRham[n](M,\field{C})$ &  $n$-te \Name{de Rham}-Kohomologie von $M$\\
$\HdeRham[n] (M,\field{Z})$ & $n$-te integrale \Name{de Rham}-Kohomologie von $M$ \\
$c(\star)$ & charakteristische Klasse des Sternprodukts $\star$ \\

& \\

$x^{1}, \cdots, x^{n}$ & lokale Koordinaten in einer Umgebung $U
\subseteq M$ \\ 
$\deohne x^{1}, \cdots, \de x^{n}$ & Koordinateneinsformen \\
$ \frac{\del}{\del x^{1}}, \cdots, \frac{\del}{\del x^{n}}$ &
Koordinatenvektorfelder \\
$\omega_{ij}$, $\Lambda^{ij}$ & Koeffizienten der symplektischen Form
$\omega$, des \Name{Poisson}-Tensors $\Lambda$ \\

$\Lambda^{n} T^{\ast}M$ &
schiefsymmetrische Tensoren der Stufe $n$ auf $T^{\ast}M$  \\
$S^{n} T^{\ast}M$ & symmetrische Tensoren der Stufe $n$ auf $T^{\ast}M$  \\

$W_{p}$, $W_{p} \otimes \Lambda^{\bullet}$ & formale
\Name{Weyl}-Algebren "uber dem Punkt $p\in M$) \\
$W$, $W \otimes \Lambda^{\bullet}$ & B"undel aller formalen
\Name{Weyl}-Algebren \\
$\dega$, $\degs$ & schiefsymmetrische, symmetrische
Gradabbildung \\ 
$\degl$, $\Deg$ & $\lambda$-Gradabbildung, totaler Grad \\  
$W^{(k)}$,$W^{(k)}_{p}$ & homogene Elemente bez"uglich des totalen
Grades ("uber $p \in M$)\\  

$(W_{k} \otimes \Lambda)$ &  Menge aller Elementen vom $\Deg$-Grad
$\ge k$ \\ 
$a^{(k)}$ & homogenes Element in $W^{(k)}_{p}$ \\ 
$\alg{W} \otimes \mit{\Lambda}^{\bullet}$ & \Name{Weyl}-Algebra,
direktes Produkt der Schnitte im B"undel $W \otimes \Lambda$ \\ 
$\mu$, $\fpweyl$ & undeformiertes Produkt, faserweises
\Name{Weyl-Moyal}-Produkt  \\
$\ad(a)$, $[a, \cdot]_{\sss \fpweyl}$ & $\field{Z}_{2}$-gradierter
Kommutator bez"uglich $\fpweyl$ mit $a$ \\  
$\delta$, $\delta^{\ast}$, $\delta^{-1}$ & Differentiale \\
$\sigma$, $\sigma'$, $\sigma^{\sss E}$ & Projektion auf symmetrischen
und schiefsymmetrischen Grad $0$ \\ 
$\lfloor \cdot \rfloor$ & abrunden auf n"achste ganze Zahl,
d.~h.~$x-1 < \lfloor x \rfloor \le x$ und $\lfloor x \rfloor \in
\field{Z}$ f"ur alle $x \in \field{R}$ \\
$\nabla$ & (symplektischer) Zusammenhang \\
$\alg{W} \otimes \mit{\Lambda}^{\bullet} \otimes \alg{E}$ & B"undel,
($\alg{W} \otimes \mit{\Lambda}^{\bullet} \otimes \END{E},\alg{W}
\otimes \mit{\Lambda}^{\bullet})$-Bimodul  \\
$\alg{W} \otimes \mit{\Lambda}^{\bullet} \otimes \END{E}$ &
Assoziative Algebra, Erweiterung von $\alg{W} \otimes
\mit{\Lambda}^{\bullet}$ \\
$D$, $D$', $D^{\sss E}$ & kovariante Differentiale auf $\alg{W} \otimes
\mit{\Lambda}^{\bullet}$,  $\alg{W} \otimes \mit{\Lambda}^{\bullet}
\otimes \END{E}$, $\alg{W} \otimes \mit{\Lambda}^{\bullet} \otimes
\alg{E}$ \\ 

$\alg{D}$, $\alg{D}'$ ,$\alg{D}^{\sss E}$ & \Name{Fedosov}-Derivationen \\
$\hat{R}$ & Kr"ummung eines Zusammenhangs $\nabla$ \\
$R$ & symplektische Kr"ummung eines symplektischen Zusammenhangs $\nabla$ \\
$R^{\sss E}$ & Kr"ummung des Zusammenhangs $\lconnE$ auf
$\bundle{E}{\pi}{M}$ \\
$R'$ & Kr"ummung des Zusammenhangs  $\lconnEnd{E}$ auf
Endomorphismenb"undel $\bundle{\End{E}}{\pi'}{E}$ \\
$r$, $r^{\sss E}$, $r'$ & spezielle Elemente in \Name{Weyl}-Algebren \\
$\Omega$ & formale Reihe geschlossener Zweiformen  \\
$s$ & Element in $\alg{W}_{3} \otimes \mit{\Lambda}^{0}$ \\ 
$\star_{\sss (\nabla,\Omega,s)}$ & \Name{Fedosov}-Sternprodukt \\
$\bundle{E}{\pi}{M}$ & (komplexes) Vektorb"undel "uber $M$ \\ 
$\bundle{L}{\pi}{M}$ & (komplexes) Geradenb"undel "uber $M$ \\
$\cdot', \cdot$ & undeformierte Modulverkn"upfungen \\
$\circ'$, $\circ$ & \Name{Weyl-Moyal} deformierte Modulverkn"upfungen \\
$\bullet', \bullet$ & deformierte Modulmultiplikationen \\    
$(M,\omega,G,\nabla)$ & symplektische $G$-Mannigfaltigkeit mit
Zusammenhang \\
$g.$ & Wirkung der Gruppe $G$ \\

& \\
\multicolumn{2}{l}{\Large \textbf{Kapitel \ref{sec:MoritaAequivalenzdeformierterAlgebren}}} \\
& \\
$(M_{\sss \mathrm{red}}, \omega_{\sss \mathrm{red}})$ & durch
Phasenraumreduktion reduzierte symplektische Mannigfaltigkeit \\
$J,\mbf{J}$ & Impuslabbildung, Quantenimpulsabbildung \\ 
$T_{\sss \field{C}}^{\bullet}(\LieAlg{g})\FP$ & formale Potenzreihe
der komplexifizierten Tensoralgebra einer \Name{Lie}-Algebra
$\LieAlg{g}$ \\
$\universelll{\LieAlg{g}}$ & $\lambda$-abh"angige modifizierte
universell einh"ullende Algebra der \Name{Lie}-Algebra $\LieAlg{g}$ \\
$\Pol^{\bullet}(\LieAlgd{g})$ & Algebra der Polynome auf $\LieAlgd{g}$ \\
$\spgutt$ & \Name{Gutt}-Sternprodukt auf
$\Pol^{\bullet}(\LieAlgd{g})\FP$ \\
$P_{0}$ & idempotentes Element oder Projektor in $M_{n}(\alg{A})$ \\
$\defpro{P}$ & deformiertes idempotentes Element oder deformierter
Projektor \\
$\SPf{\cdot,\cdot}$ & deformiertes inneres Produkt \\
$\HPoisson[n](M,\field{C})$ & $n$-te \Name{Poisson}-Kohomologie \\
$\HPoisson[n](M,\field{Z})$ & $n$-te integrale \Name{Poisson}-Kohomologie \\

& \\
\multicolumn{2}{l}{\Large \textbf{Anhang \ref{chapter:AlgebraischeGrundlagen}}} \\
& \\
$(G,\cdot)$ & Gruppe \\
$\rmod{M}{\ring{R}}$ & $\ring{R}$-Rechtsmodul\\
$\ring{I}_{\sss L}$, $\ring{I}_{\sss R}$, $I$ & Linksideal,
Rechtsideal, Ideal eines Rings $\ring{R}$ \\ 
$\hat{\ring{R}}$ & Quotientenk"orper des Rings $\ring{R}$ \\
$\alg{A}$ & Assoziative Algebra \\
$\mu, \muop$ &  Multiplikation $\mu (a \otimes b) = ab$,
geflipte Multiplikation $\muop(a \otimes b) = ba$ \\
$\eta$ & Einsabbildung \\
$\zentrum{\alg{A}}$ & Zentrum der Algebra $\alg{A}$ \\ 
$\alg{K}$ & Koassoziative Koalgebra \\
$\Delta, \Deltaop$ & Komultiplikation, geflipte Komultiplikation \\
$\varepsilon$ &  Koeins \\
$\Prim{\alg{K}}$ & Primitive Elemente der Koalgebra $\alg{K}$ \\
$H$ & \Name{Hopf}-Algebra, \Name{Hopf}-$^\ast$-Algebra \\
$S$ & Antipodenabbildung einer \Name{Hopf}-Algebra \\
$I,^\ast$ & Antilinearer Antiautomorphismus / $^\ast$-Involution \\
$\ideal{I}$ & \Name{Hopf}-Ideal \\
$\mathfrak{g}$ & \Name{Lie}-Algebra \\
$\universell{\LieAlg{g}}$ & Universell Einh"ullende von
$\LieAlg{g}$ \\ 
$\ring{C}(G)$ &  Gruppenalgebra "uber Ring $\ring{C}$ \\
$\universellC{\LieAlg{g}}$ & komplexifizierte universell
Einh"ullende $\universellC{\LieAlg{g}}=\universellR{\LieAlg{g}}
\tensor[\field{R}] \field{C}$ \\
$\alg{F}(G)$ & $\ring{C}$-wertige Funktionen auf der Gruppe $G$ \\
$\Hom[\ring{C}](H,\alg{A})$ & $\ring{C}$-linearen Homomorphismen von
$H$ nach $\alg{A}$ \\ 
$\msf{a},\msf{b}$ & Elemente in $\Hom[\ring{C}](H,\alg{A})$ \\

$\cross{\alg{A}}{H}$ & Cross-Produktalgebra \\
$\kat{A},\kat{B},F$ & Kategorien $\kat{A},\kat{B}$ und Funktor
$F:\kat{A}\to \kat{B}$ \\ 
$\lambda$ & Formaler Parameter \\
$\modulf{M}$ & Formale Reihen in $\lambda$ mit Koeffizienten
im Modul $\modul{M}$ \\ 
$\modul{M}[\lambda]$ & Polynome in $\lambda$ mit Koeffizienten in
$\modul{M}$ \\
$\algf{A}$ & Formalen Potenzreihen in $\lambda$ mit Koeffizienten in der
Algebra $\alg{A}$ \\
$o(\cdot)$ & $\lambda$-adische Ordnung \\
$\varphi(\cdot)$ & $\lambda$-adische Bewertung \\ 
$d_{\sss \varphi}$ & von $\varphi$ induzierte Metrik \\

& \\
\multicolumn{2}{l}{\Large \textbf{Anhang \ref{chapter:GeometrischeGrundlagen}}} \\
&  \\

$(E,\pi,M,F,G)$ & Faserb"undel mit Totalraum $E$, Projektion $\pi$,
Basismannigfaltigkeit $M$, Faser $F$ und Strkturgruppe $G$ \\
$t_{ij}$ & "Ubergangsfunktionen \\
$\lconnE$ & Zusammenhang auf Vektorb"undel $\bundle{E}{\pi}{M}$ \\
$\lconnEnd{E}$ & Zusammenhang auf Endomorphismenb"undel $\bundle{\End{E}}{\pi}{E}$ \\
$R^{E}$ & Kr"ummung eines Zusammhangs auf $E$ \\
$\Lambda^{k \ell}_{,j}$ & $\Lambda^{k \ell}$ nach der $j$-ten
Komponente abgeleitet\\
$T_{\nabla}$ & Torsion eines Zusammenhangs $\nabla$ auf $TM$ \\
\end{longtable}


%% file: personenindex.tex
\chapter*{Personenindex}
\fancyhead[CE]{\slshape \nouppercase{Personenindex}} 
\fancyhead[CO]{\slshape \nouppercase{Personenindex}} 
\addcontentsline{toc}{chapter}{Personenindex}

F"ur Biographien von wichtigen Geometern und Algebraikern verweisen wir auf die B"ucher
 \citep{scriba.schreiber:2005a} und \citep{alten.et.al:2003a}. Eine
 "Ubersicht zu den genannten Nobelpreistr"agern liefert beispielsweise
 \citep{brockhaus.nobelpreise:2001}. Weitere hier verwendete
 Informationen stammen von der Internetzseite {\tt http://www.wikipedia.org}. 

\begin{longtable}{lp{85mm}}

\Namebd{Abel, Niels H.}{1802}{1829} & \Name{Niels Henrik Abel} war
norwegischer Mathematiker und starb im Alter von 26 Jahren an Tuberkulose.\\

\Namebd{Archimedes}{287}{212 v.~Chr.} & \Name{Archimedes von Syrakus}
war antiker griechischer Mathematiker, Physiker und Ingenieur. \\ 

\Namebd{Artin, Emil}{1898}{1962} & \Name{Emil Artin} war
"osterreichischer Mathematiker und besch"aftigte sich mit Algebra,
Zahlentheorie sowie der K"orpertheorie. \\ 

\Namebd{Bargmann, Valentine}{1908}{1989} & \Name{Valentin Bargmann}
war ein in Deutschland geborener mathematische Physiker. Er emigrierte
in die USA in der er \Name{Albert Einstein} assistierte.\\


\Namebd{Birkhoff, Gerrett}{1911}{1996} & \Name{Gerrett Birkhoff} war amerikanischer
Mathematiker, der auf dem Gebiet der Algebra forschte. \\

\Namebd{Borel, \'{E}mile}{1871}{1956} & Sein kompletter Name lautetet
\Name{F\'{e}lix \'{E}douard Justin \'{E}mile Borel}. Er war franz"osischer
Mathematiker und Politiker und leistete grundlegende Beitr"age zur Topologie, Ma"s-,
Wahrscheinlichkeits- und Spieltheorie. \\ 



\Namebd{Cartan, \'{E}lie J.}{1869}{1951} & \Name{\'{E}lie Joseph
  Cartan} war ein bedeutender franz"osischer Mathematiker, der
Beitr"age zu \Name{Lie}-Gruppen, der Differentialgeometrie sowie der
mathematischen Physik geliefert hat. \\

\Namebd{Cauchy, Augustin L.}{1789}{1857} & \Name{Augustin Louis
  Cauchy} war franz"osischer Mathematiker und verfa"ste knapp 800
Ver"offentlichungen in fast allen Bereichen der Mathematik,
insbesondere in der Funktionentheorie.\\ 

\Namebd{\v{C}ech, Eduard}{1893}{1960} & \Name{Eduard \v{C}ech} war ein
b"ohmischer Mathematiker.\\ 

\Namebd{Chern, Shiing-Shen}{1911}{2004} & \Name{Shiing-Shen Chern} war
chinesisch-ame\-ri\-ka\-nischer Mathematiker, der mit seinen
Beitr"agen zur Topologie und Geometrie die Mathematik des
20.~Jahrhunderts gepr"agt hat.\\ 

\Namebd{Chevalley, Claude}{1909}{1984} & \Name{Claude Chevalley} war ein franz"osischer
Mathematiker und Gr"undungsmitglied der Bourbaki-Gruppe. Er hat
insbesondere in der \Name{Lie}-Gruppentheorie und
\Name{Lie}-Algebrentheorie wichtige Beitr"age beigesteuert. \\  

\Namebd{Darboux, Jean G.}{1842}{1917} & \Name{Jean Gaston Darboux} war
ein franz"osischer Mathematiker und assistierte \Name{Joseph
  Liouville} an der Sorbonne in Paris. \\

\Namebd{de Morgan, Augustus}{1806}{1871} & \Name{Augustus de Morgan}
war ein in Indien geborener, britischer Mathematiker,
der sich insbesondere der Logik verschrieben hat.\\

\Namebd{de Rham, Georges}{1903}{1990} & \Name{Georges de Rham} war ein
schweizerischer Mathematiker, der sich insbesondere mit der
Differentialtopologie auseinandersetzte. Ihm gelang der Beweis der
Homotopieinvarianz der nach ihm benannten Kohomologie.\\

\Nameb{Deligne, Pierre R.}{1944} & \Name{Pierre Ren\'{e} Deligne} ist
belgischer Mathematiker und erhielt 1978 die
\Name{Fields}-Medaille. Ihm gelang der vollst"andige Beweis der
\Name{Weil}-Vermutungen.   \\

\Namebd{Descartes, Ren\'{e}}{1596}{1650} & \Name{Ren\'{e} Descartes}
war franz"osischer Philosoph, Mathematiker und Wissenschaftler, der auch als \Name{Cartesius}
bekannt war. Er gilt als einer der wichtigsten Denker der Neuzeit, da
er sowohl die Philosophie als auch die Mathematik entscheidend beeinflu"ste. \\

\Namebd{Dirac, Paul A.~M.}{1902}{1984} & \Name{Paul Adrien Maurice
  Dirac} war britischer Physiker. Er erhielt 1933 den \Name{Nobel}preis
f"ur Physik: \glqq für die Entdeckung einer neuen, n"utzlichen Form der Atomtheorie\grqq. \\ 

\Namebd{Einstein, Albert}{1879}{1955} & \Name{Albert Einstein} war Deutsch-schweizerischer
Physiker. Er erhielt 1921 den \Name{Nobel}preis Physik \glqq f"ur
seine Verdienste in der theoretischen Physik, besonders f"ur die
Entdeckung des Photoelektrischen Effekts\grqq. \\  

\Namebd{Eilenberg, Samuel}{1913}{1998} & \Name{Samuel Eilenberg} war
polnischer Mathematiker. Er besch"aftigte sich vorwiegend mit der
algebraischen Topologie sowie der Kategorientheorie. \\

\Namebd{Euler, Leonhard}{1707}{1783} & \Name{Leonhard Euler} war schweizerischer Mathematiker
und Physiker. Es war wahrscheinlich der wichtigste Mathematiker des 18.~Jahrhunderts.\\

\Nameb{Fedosov, Boris}{1938} & \Name{Boris Fedosov} ist ein
russischer Mathematiker, der grundlegende Beitr"age auf dem Gebiet
der Deformationsquantisierung geleistet hat. \\

\Namebd{Fock, Wladimir A.}{1898}{1974} & \Name{Wladimir Alexandrowitsch Fock}
war russischer Physiker, der grundlegende Beitr"age zur
Quantenmechanik und zur Quantenfeldtheorie beigetragen hat.\\

\Namebd{Gauss, Carl F.}{1777}{1855} & \Name{Carl Friedrich Gau"s} war
deutscher Mathematiker, Geod"at und Erfinder. \\

\Nameb{Grothendieck, Alexander}{1928} & Als Sohn eines russischen
Vaters und einer deutschen Mutter in Berlin geboren, entwickelte er
sich zu einem der wichtigsten Mathematikern des 20.~Jahrhunderts und
leistete wichtige Beitr"age in der algebraischen Topologie, der
algebraischen Geometrie sowie der Funktionalanalysis. \\

\Nameb{Gutt, Simone}{1956} & \Name{Simone Gutt} ist belgische
Mathematikerin und an der Universit\'{e} Libre de Bruxelles t"atig. \\

\Namebd{Haar, Alfr\'{e}d}{1885}{1933} & \Name{Alfr\'{e}d Haar} war
ungarischer Mathematiker. Er promovierte bei \Name{David Hilbert}. \\

\Namebd{Hamilton, William R.}{1805}{1865} & Sir \Name{William Rowan
Hamilton} war irisch-englischer Mathematiker und Physiker, hatte ab
1827 eine Professur f"ur Astronomie inne und war k"oniglicher
Astronom f"ur Irland. \\

\Namebd{Hausdorff, Felix}{1868}{1942} & \Name{Felix Hausdorff} war
deutscher Mathematiker und Mitbegr"under der modernen Topologie.\\

\Namebd{Heisenberg, Werner}{1901}{1976} & \Name{Werner Heisenberg} war
deutscher Physiker und erhielt 1932 den Physik \Name{Nobel}preis f"ur seine Arbeiten zur
Aufstellung der Quantenmechanik. \\ 

\Namebd{Hellinger, Ernst D.}{1883}{1950} & \Name{Ernst David Hellinger}
war deutscher Mathematiker. \\

\Namebd{Hermite, Charles}{1822}{1901} & \Name{Charles Hermite} war
franz"osischer Mathematiker. Er arbeitete insbesondere auf dem Gebiet
der Zahlentheorie und der Algebra. \\


\Namebd{Hilbert, David}{1862}{1943} & \Name{David Hilbert} war
ostpreu"sischer Mathematiker. Er axiomatisierte die Geometrie und
formulierte die 23 \Name{Hilbert}-Probleme.
\\

\Nameb{Hochschild, Gerhard P.}{1915} & \Name{Gerhard Paul Hochschild}
ist ein in
Berlin geborener Mathematiker.  \\

\Namebd{Hopf, Heinz}{1894}{1971} & \Name{Heinz Hopf} war
schweizerischer Mathematiker, der insbesondere im Bereich der
algebraischen Topologie arbeitete.\\

\Namebd{Jacobi, Carl G.~J.}{1804}{1851} & \Name{Carl Gustav Jacob
  Jacobi} war ein deutscher Mathematiker und wurde von seinen
Sch"ulern der \glqq \Name{Euler} des 19.~Jahrhunderts\grqq{}
genannt. Er trug zu vielen Bereichen der Mathematik wichtige
Ergebnisse bei und wird daher als einer der vielseitigsten
Mathematiker der Geschichte angesehen.\\

\Namebd{Jacobson, Nathan}{1910}{1999} & \Name{Nathan Jacobson} war
amerikanischer Mathematiker.\\

\Namebd{K"ahler, Erich}{1906}{2000} & \Name{Erich K"ahler} war deutscher Mathematiker. \\

\Namebd{Kepler, F.~Johannes}{1571}{1630} & \Name{Friedrich Johannes
  Kepler} war ein deutscher Naturphilosoph, Mathematiker, Astronom,
Astrologe und Optiker. \\

\Nameb{Kontsevich, Maxim}{1964} & \Name{Maxim Kontsevich} ist
russischer Mathematiker. Er bekam 1998 die
\Name{Fields}-Medaille und ist derzeit am Institut des Hautes
\'{E}tudes Scientifiques (IH\'{E}S) in Bures-sur-Yvette, Frankreich.\\

\Namebd{Kronecker, Leopold}{1823}{1891} & \Name{Leopold Kronecker} war deutscher
Mathematiker. Er besch"aftigte sich mit Algebra, Zahlentheorie,
Analysis und der Funktionentheorie. \\

\Namebd{Lagrange, Joseph L.}{1736}{1813} & \Name{Joseph Louis
  Lagrange}, der als \Name{Giuseppe Luigi Lagrangia} in Turin geboren
wurde, nahm die franz"osische Staatsb"urgerschaft an und ist f"ur seine
Arbeiten in der Astronomie, der Physik sowie der Mathematik bekannt
geworden. \\ 

\Namebd{Laplace, Pierre-Simon}{1749}{1827} & \Name{Pierre-Simon
  Marquis de Laplace} war franz"osischer Mathematiker und Astronom.\\ 

\Namebd{Laurent, Pierre A.}{1813}{1853} & \Name{Pierre Alphonse Laurent} war
franz"osischer Mathematiker und Entdecker der \Name{Laurent}-Reihen. \\

\Namebd{Lebesgue, Henri L.}{1875}{1941} & \Name{Henri L\'{e}on Lebesgue} war
franz"osischer Mathematiker.\\ 

\Namebd{Leibniz, Gottfried W.}{1646}{1716} & \Name{Gottfried Wilhelm
  Leibniz} war ein deutscher Philosoph, Mathematiker, Diplomat, Physiker,
Historiker, Bibliothekar und Doktor des weltlichen und des
Kirchenrechts.\\

\Namebd{Lie, M.~Sophus}{1842}{1899}  & \Name{Marius Sophus Lie} war ein
norwegischer Mathematiker. Er begr"undete die kontinuierlichen
Symmetrien und leistete Beitr"age zur Theorie der Differentialgleichungen.\\ 

\Namebd{Liouville, Joseph}{1809}{1882} & \Name{Joseph Liouville} war
franz"osischer Mathematiker. Er
legte an der \'Ecole Polytechnique u.~a.~Pr"ufungen bei
\Name{Sim\'eon-Denis Poisson}
ab und erhielt dort 1838 eine Professur. Er arbeitete insbesondere in der
Zahlentheorie, Funktionentheorie und der Differentialgeometrie.\\

\Namebd{Morita, Kiiti}{1915}{1995} & \Name{Kiiti Morita} war japanischer Mathematiker. \\

\Namebd{Moyal, Jose E.}{1910}{1998} & \Name{Jose Enrique Moyal} war
australischer Mathematiker und Physiker. \\

\Nameb{Neumaier, Nikolai A.}{1971} & \Name{Nikolai Alexander Neumaier}
ist deutscher Physiker und Mathematiker. Er ist 
z.~Zt.~an der \Name{Albert-Ludwigs}-Universit"at Freiburg t"atig und
arbeitet auf dem Gebiet der Deformationsquantisierung.\\ 

\Namebd{Newton, Isaac}{1643}{1727} & Sir \Name{Isaac Newton} war englischer
Mathematiker, Physiker, Astronom, Alchemist und Philosoph.\\

\Namebd{Noether, Emmy A.}{1882}{1935} & \Name{Emmy Amalie Noether}
habilitierte sich 1919 als erste Frau in Deutschland. Sie gilt als eine der
Mitbegr"underinnen der modernen Algebra. \\

\Namebd{Picard, Charles \'{E}.}{1856}{1941} & \Name{Charles \'{E}mile
Picard} war franz"osischer Mathematiker. \\

\Namebd{Planck, Max}{1858}{1947} & \Name{Max Planck} war deutscher
Physiker. Er erhielt den \Name{Nobel}preis f"ur Physik 1918 \glqq als Anerkennung des
Verdienstes das er sich durch seine Quantentheorie um die Entwicklung
der Physik erworben hat\grqq. \\ 

\Namebd{Poincar\'{e}, Henri}{1854}{1912} & \Name{Henri Poincar\'{e}}
war ein bedeutender franz"osischer Mathematiker und theoretischer Physiker. \\

\Namebd{Poisson, Sim\'eon-Denis}{1781}{1840} & \Name{Sim\'eon-Denis
  Poisson} war ein bedeutender franz"osischer Mathematiker und
Physiker. In der Mathematik arbeitete auf den Gebieten der Differentialgeometrie, der
Wahrscheinlichkeitsrechnung sowie der Infinitesimalrechnung. In der
Physik forschte er insbesondere an der Wellentheorie,
der Akustik, sowie der Elastizit"atstheorie und der W"arme. \\

\Namebd{Riemann, Bernhard}{1826}{1866} & \Name{Bernhard Riemann} war
ein bedeutender deutscher Mathematiker, der
unter anderem bei \Name{Gauss} und \Name{Dirichlet} lernte. Er schuf
die mathematischen Grundlagen f"ur die Allgemeine
Relativit"atstheorie \Name{Einstein}s. Er starb mit
39 Jahren an Tuberkulose. \\

\Namebd{Schr"odinger, Erwin}{1887}{1961} & \Name{Erwin Schr"odinger}
war "osterreichischer Physiker. Er erhielt 1933 den \Name{Nobel}preis
f"ur Physik. \\ 

\Namebd{Schwarz, K.~Hermann A.}{1843}{1921} & \Name{Karl Hermann
  Amandus Schwarz} war ein deutscher Mathematiker, der sich
insbesondere mit der komplexen Analysis auseinandergesetzt hat.\\

\Nameb{Serre, Jean-Pierre}{1926} & \Name{Jean-Pierre Serre} ist
franz"osischer Mathematiker, der die Mathematik des 20.~Jahrhunderts
pr"agte. Er  erhielt u.~a.~die \Name{Fields}-Medaille im Jahre
1954 und den \Name{Abel}-Preis 2003.\\ 

\Name{Swan, Richard G.}{} & \Name{Richard Gordon Swan} ist
amerikanischer Mathematiker. \\

\Name{Sweedler, Moss E.}{} & \Name{Moss Eisenberg Sweedler} ist
amerikanischer Mathematiker.  \\

\Namebd{Taylor, Brook}{1685}{1731} &  \Name{Brook Taylor} war englischer Mathematiker, dem
es gelang, die allgemeine Konstruktion f"ur \Name{Taylor}-Reihen von 
Funktionen, die eine solche Reihe besitzen, anzugeben. \\

\Namebd{Toeplitz, Otto}{1881}{1940} & \Name{Otto Toeplitz} war deutsch-j"udischer
Mathematiker. Er besch"aftigte sich mit Funktionalanalysis und
linearer Algebra und hatte eine Professur in Bonn bevor er 1939 nach
Pal"astina fl"uchtete. \\

\Namebd{Vey, Jacques}{1943}{1979} & \Name{Jacques Vey} war schweizerischer Mathematiker. \\



\Namebd{Weyl, Hermann K.~H.}{1885}{1955} & \Name{Hermann Klaus Hugo Weyl} war
deutscher Mathematiker.\\

\Namebd{Wick, Gian-Carlo}{1909}{1992} & \Name{Gian-Carlo Wick} war italienischer Physiker. \\

\Namebd{Witt, Ernst}{1911}{1991} & \Name{Ernst Witt} war deutscher Mathematiker. Er
studierte in Freiburg im Breisgau und G"ottingen, wo er bei Emmy
Noether promovierte und sich habilitierte. \\
\end{longtable}
